\documentclass{article}

\title{An Analysis of a Virtually Synchronous Protocol}
\author{Dan Arnon and Navindra Sharma \\ EMC Corporation}

\usepackage{geometry}
\usepackage{amsmath}

\usepackage{amsthm}
\usepackage{amsopn}
\usepackage{parskip}   
\usepackage{lmodern}
\usepackage{amssymb}
\usepackage{fixltx2e}  
\usepackage[ruled]{algorithm2e}

\theoremstyle{plain}
\newtheorem{thm}{Theorem}
\newtheorem{lem}{Lemma}
\newtheorem*{lem*}{Lemma}
\newtheorem{cor}{Corollary}

\newtheorem{indhyp}{Inductive Hypothesis}
\newtheorem{defn}{Definition}

\newcommand{\cbcast}[0]{{\ttfamily CBCAST}}
\newcommand{\gms}[0]{{\ttfamily GMS}}
\newcommand{\ulp}[0]{{\ttfamily APP}}
\newcommand{\prot}[0]{{\ttfamily PROTOCOL}}
\newcommand{\gv}[1]{\textit{\textsf{#1}}}
\newcommand{\lv}[1]{\textit{\textsf{#1}}}
\newcommand{\fp}[1]{\textsf{#1}}
\newcommand{\se}[1]{\textsf{#1}}
\newcommand{\orig}[1]{{\scriptstyle{\operatorname{ORIG}}}(#1)}
\newcommand{\mview}[1]{{\scriptstyle{\operatorname{VIEW}}}(#1)}
\newcommand{\mvt}[1]{{\scriptstyle{\operatorname{VT}}}(#1)}
\newcommand{\rset}{\se{ReceiveSet}}
\newcommand{\wset}{\se{WaitSet}}
\newcommand{\fwset}{\se{FwdWaitSet}}
\newcommand{\bwset}{\se{BcastWaitSet}}
\newcommand{\fque}{\se{FwdQueue}}
\newcommand{\mset}{\se{MSet}}
\newcommand{\lset}{\se{LiveSet}}
\newcommand{\cset}{\se{ContactSet}}
\newcommand{\lque}{\se{LaunchQueue}}
\newcommand{\pque}{\se{PendViewQueue}}
\newcommand{\uldata}{\se{ReplicatedData}}
\newcommand{\fvec}{\gv{flush}}
\newcommand{\gvec}{\gv{ghost}}
\newcommand{\cview}{\gv{cur\_view}}
\newcommand{\vgap}{\gv{v\_gap}}
\newcommand{\vtime}{\gv{vt}}
\newcommand{\self}{\gv{self}}
\newcommand{\sgh}{\gv{ghost\_height}}
\newcommand{\sfh}{\gv{flush\_height}}
\newcommand{\mpkout}{\gv{mpkt\_out}}
\newcommand{\mpkin}{\gv{mpkt\_in}}

\newcommand{\ulview}{\gv{ul\_view}}
\newcommand{\ulcount}{\gv{ul\_count}}
\DeclareMathOperator{\magicview}{\operatorname{MAGIC\_VIEW}}
\DeclareMathOperator{\magicmessage}{\operatorname{MAGIC\_MSG}}
\DeclareMathOperator{\magicsize}{\operatorname{MAGIC\_SIZE}}

\DeclareMathOperator{\processset}{\mathbb{P}}
\DeclareMathOperator{\eventset}{\mathbb{E}}
\DeclareMathOperator{\numview}{\mathfrak{V}}
\DeclareMathOperator{\kmsg}{\mathfrak{M}_k}
\DeclareMathOperator{\labelspace}{\mathfrak{L}}
\DeclareMathOperator{\cspace}{\mathfrak{L}_c}
\DeclareMathOperator{\sspace}{\mathfrak{L}_s}
\DeclareMathOperator{\aspace}{\mathfrak{L}_a}
\DeclareMathOperator{\fspace}{\mathfrak{L}_f}
\DeclareMathOperator{\viewset}{\mathbb{S}}
\DeclareMathOperator{\packetset}{\mathbb{K}}
\DeclareMathOperator{\notificationset}{\mathbb{F}}
\DeclareMathOperator{\requestset}{\mathbb{A}}
\DeclareMathOperator{\uncont}{\mathbb{U}}
\DeclareMathOperator{\tstamp}{\operatorname{time}}
\DeclareMathOperator{\vbound}{\operatorname{vb}}

\DeclareMathOperator{\tr}{\operatorname{tr}}

\DeclareMathOperator{\lCODONATE}{\operatorname{CODONATE}}
\DeclareMathOperator{\lDONATE}{\operatorname{DONATE}}
\DeclareMathOperator{\lGHOST}{\operatorname{GHOST}}
\DeclareMathOperator{\lFLUSH}{\operatorname{FLUSH}}
\DeclareMathOperator{\lACK}{\operatorname{ACK}}
\DeclareMathOperator{\lBCAST}{\operatorname{BCAST}}
\DeclareMathOperator{\vcrit}{v_{\operatorname{crit}}}
\DeclareMathOperator{\minusG}{-G}
\DeclareMathOperator{\plusc}{c}
\DeclareMathOperator{\minusc}{-c}
\DeclareMathOperator{\plusz}{z}
\DeclareMathOperator{\minusz}{-z}
\DeclareMathOperator{\plusminusG}{{\pm}G}
\DeclareMathOperator{\origpk}{\operatorname{origin}}
\DeclareMathOperator{\shftmsg}{\Uparrow^{(+m)}}
\DeclareMathOperator{\minusshftmsg}{\Uparrow^{(-m)}}
\DeclareMathOperator{\shftack}{\Downarrow^{(+s)}}
\DeclareMathOperator{\minusshftack}{\Downarrow^{(-s)}}
\DeclareMathOperator{\shftghost}{\Downarrow^{(+g)}}
\DeclareMathOperator{\minusshftghost}{\Downarrow^{(-g)}}
\DeclareMathOperator{\shftflush}{\Downarrow^{(+f)}}
\DeclareMathOperator{\minusshftflush}{\Downarrow^{(-f)}}

\DeclareMathOperator{\lbl}{\operatorname{\lambda}}
\DeclareMathOperator{\heighta}{\operatorname{height_1}}
\DeclareMathOperator{\heightb}{\operatorname{height_2}}

\newcommand{\channel}[2]{\protect\overrightarrow{#1#2}}
\newcommand{\psendevent}[1]{{#1}^{\scriptscriptstyle{\operatorname{QU}}}}
\newcommand{\preceiveevent}[1]{{#1}^{\scriptscriptstyle{\operatorname{PR}}}}
\newcommand{\pappevent}[1]{{#1}^{\scriptscriptstyle{\operatorname{PR}}}}
\newcommand{\notify}[2]{v_{#1}({#2})}
\newcommand{\pnotifyevent}[2]{\notify{#1}{#2}^{\scriptscriptstyle{\operatorname{PR}}}}
\newcommand{\joinevent}[1]{{#1}_{\textsc{run}}}
\newcommand{\haltevent}[1]{{#1}_{\textsc{hlt}}}
\newcommand{\crittime}[2]{{\operatorname{Crit}}({#1} \rightarrow {#2})}

\newcommand{\lblggg}[4]{[{\ell_{#1}}.{#2}.{#3}|{#4}]} 
\newcommand{\lblzgg}[3]{[{\ell_{#1}}.\hat{0}.{#2}|{#3}]}
\newcommand{\lblgzg}[3]{[{\ell_{#1}}.{#2}.\hat{0}|{#3}]}
\newcommand{\lblzzg}[2]{[{\ell_{#1}}.\hat{0}.\hat{0}|{#2}]}
\newcommand{\lblzgz}[2]{[{\ell_{#1}}.\hat{0}.{#2}|\hat{0}]}
\newcommand{\lblggz}[3]{[{\ell_{#1}}.{#2}.{#3}|\hat{0}]} 
\newcommand{\lblgzz}[2]{[{\ell_{#1}}.{#2}.\hat{0}|\hat{0}]} 
\newcommand{\lblzzz}[1]{[{\ell_{#1}}.\hat{0}.\hat{0}|\hat{0}]} 

\newcommand{\lblcggg}[5]{[{\crittime{#1}{#2}}.{#3}.{#4}|{#5}]} 
\newcommand{\lblczzg}[3]{[{\crittime{#1}{#2}}.\hat{0}.\hat{0}|{#3}]} 
\newcommand{\lblcgzg}[4]{[{\crittime{#1}{#2}}.{#3}.\hat{0}|{#4}]}
\newcommand{\lblcgzz}[3]{[{\crittime{#1}{#2}}.{#3}.\hat{0}|\hat{0}]} 
\newcommand{\lblcggz}[4]{[{\crittime{#1}{#2}}.{#3}.{#4}|\hat{0}]} 
\newcommand{\lblczzz}[2]{[{\crittime{#1}{#2}}.\hat{0}.\hat{0}|\hat{0}]} 

\newcommand{\valuepre}[3]{{#1}_{{#2}@{#3}}}
\newcommand{\valuepost}[3]{{#1}^{{#2}@{#3}}}
\DeclareMathOperator{\view}{\textsf{view}}
\DeclareMathOperator{\trig}{\textsf{trig}}
\DeclareMathOperator{\trans}{\textsf{trans}}
\DeclareMathOperator{\cont}{\textsf{cont}}
\DeclareMathOperator{\njop}{\textbf{n}_{\scriptscriptstyle{JOIN}}}
\DeclareMathOperator{\nrop}{\textbf{n}_{\scriptscriptstyle{REM}}}
\DeclareMathOperator{\nnop}{\textbf{n}_{\scriptscriptstyle{START}}}
\DeclareMathOperator{\ndop}{\textbf{n}_{\scriptscriptstyle{STOP}}}
\DeclareMathOperator{\pkmop}{\textbf{p}_{\scriptscriptstyle{MSG}}}
\DeclareMathOperator{\pksop}{\textbf{p}_{\scriptscriptstyle{ACK}}}
\DeclareMathOperator{\pkfop}{\textbf{p}_{\scriptscriptstyle{FLUSH}}}
\DeclareMathOperator{\pkhop}{\textbf{p}_{\scriptscriptstyle{GHOST}}}
\DeclareMathOperator{\pkdop}{\textbf{p}_{\scriptscriptstyle{DONATE}}}
\DeclareMathOperator{\pkcop}{\textbf{p}_{\scriptscriptstyle{CO-DONATE}}}
\DeclareMathOperator{\gap}{\textsf{gap}}
\newcommand{\nj}[2]{\njop\langle#1, #2\rangle}
\newcommand{\nr}[1]{\nrop\langle#1\rangle}
\newcommand{\nn}[1]{\nnop\langle#1\rangle}
\newcommand{\nd}[0]{\ndop}
\newcommand{\pkm}[1]{\pkmop\langle#1\rangle}
\newcommand{\pks}[1]{\pksop\langle#1\rangle}
\newcommand{\pkf}[1]{\pkfop\langle#1\rangle}
\newcommand{\pkh}[1]{\pkhop\langle#1\rangle}
\newcommand{\pkd}[1]{\pkdop\langle#1\rangle}
\newcommand{\pkc}[1]{\pkcop\langle#1\rangle}
\newcommand{\vdeath}[1]{r({#1})}

\begin{document}
\maketitle

\begin{abstract}
Enterprise-scale systems such as those used for cloud computing require a scalable and highly available infrastructure. One crucial ingredient of such an infrastructure is the ability to replicate data coherently among a group of cooperating processes in the presence of process failures and group membership changes. The last few decades have seen prolific research into efficient protocols for such data replication. One family of such protocols are the virtually synchronous protocols. Virtually synchronous protocols achieve their efficiency by limiting their synchronicity guarantee to messages that bear a causal relationship to each other. Such protocols have found wide-ranging commercial uses over the years. One protocol in particular, the {\ttfamily CBCAST} protocol developed by Birman, Schiper and Stephenson in 1991 and used in their ISIS platform was particularly promising due to its unique no-wait properties, but has suffered from seemingly intractable race conditions. In this paper we describe a corrected version of this protocol and prove its formal properties.
\end{abstract}

\newpage
\tableofcontents

\section{Introduction}
Many modern computational tasks are performed by groups of processes that are physically distributed and are prone to failure. Such computations require efficient ways to reliably multicast messages within the group in the presence of group membership changes. These requirements are becoming increasingly common in data center systems such as storage systems and server clusters. Meeting these requirements while maintaining good performance is challenging. There is a need to keep all the members of the group in a synchronized state, whatever that may mean, and there is a need to avoid split brains and subtle race conditions that can occur when the group is reconfigured.

These requirements generally fall into three categories:
\begin{description}
\item[Message synchrony] For the group to work coherently, at least some messages must be delivered at different group members in the same order.
\item[Reconfiguration atomicity] When the group membership changes, there must be some guarantee against split brains and all the members of the new group must reach some kind of agreement on a common initial state.
\item[Progress guarantee] There must be some guarantee that messages get delivered to the whole group.
\end{description}

Meeting such requirements can be costly. To reduce these costs different authors have proposed weaker requirements within these three categories that allow for more performant systems. The original papers describing the process group model \cite{cheriton1985distributed, birman87} envisioned incremental group reconfigurations, possibly limited to one member addition or removal at a time. This makes reconfiguration more expensive. More recent treatments of the subject assume general reconfigurations that occur in bulk. In \cite{birman2010virtually} bulk reconfiguration is joined with a flexible framework of delivery guarantees that allows for the tailoring of these guarantees to the needs of specific applications. However this framework requires all message delivery to stop while the system is being reconfigured. The Rambo system \cite{lynch2002rambo} is designed more specifically for a distributed atomic memory, but it allows messages to flow while bulk reconfiguration is taking place. It also enjoys weaker synchrony guarantees that are tailored to the specific requirements of atomic memory.

One multicast paradigm that proved especially useful is Virtually Synchrony (see \cite{birman87}). Virtual Synchrony achieves high efficiency albeit with some reduced availability (see \cite{birman1993process}) by serializing messages only when they may be causally linked. Messages that are not causally linked may be delivered in different orders to different members of the process group. Virtually synchronous multicast protocols have been used for a long time in many commercial and non-commercial systems, for example Horus (see \cite{birman96}) and Transis (see \cite{amir92}). For a comprehensive specification of group multicast protocols and their properties, see \cite{chockler2001group}.

One truly exceptional proposal for a virtually synchronous protocol was described in \cite{birman1991lightweight}. The \cbcast{} protocol that the authors describe in that paper not only promises virtual synchrony and reconfiguration atomicity without stopping message delivery, but also promises to work without any delivery guarantee - all message broadcasts are ship-and-pray. While this last promise is not explicitly elaborated on in the paper, it is what makes the \cbcast{} protocol exceptionally powerful. Unfortunately, \cbcast{} did not deliver on its promise. Its very complexity meant that race conditions could never be completely wrung out of it. The goal of this paper is to fix that protocol and provide a rigorous proof of its properties.

Following \cite{birman1991lightweight}, we describe the \cbcast{} algorithm within a formal model that contains processes, channels, packets, and an opaque group membership service.
Channels connect pairs of processes, enabling them to send point-to-point packets (we reserve the term "messages" to the entities that are  multicast by \cbcast{}). The group membership service (\gms{}) is assumed to be a primary-component \gms{} (as opposed to a partitionable one, see \cite{chockler2001group}.) \gms{} provides ordered notifications of group membership changes. Each component of the model can fail independently. As a result, failure scenarios can get very complex.

We provide a rigorous proof of two essential properties of the \cbcast{} protocol, the Causal Order Property and the Progress Property. The Causal Order Property says that messages are delivered at each process in an order that respects the causality relationships between messages. The Progress Property says that if two processes in the group never halt, then any message broadcast by one is delivered at the other, provided that only a finite number of processes join the process group. Both of these properties are proved to hold under any pattern of component failure in the cluster.

The proof is divided into a number of parts, each of which is a separate investigation. The outline of the proof plan is as follows:
\begin{itemize}
\item
The first part deals with the formal model only and is independent of \cbcast{}. We analyze the failure scenarios in the model (stop failures only) using an axiomatic approach, and show that under reasonable assumptions on the behavior of the \gms{} all partial failure cases are equivalent to either the failure-free case or to a simultaneous failure of all the components. In essence the group behaves like a single fault domain. This frees us to carry out the rest of the proof under the assumption that no failures ever occur. A critical ingredient in the analysis is the fact that the formal act of removing a process from the group by the \gms{} can mask the stop-failure of a process.
\item
In the second part we give a detailed description of the \cbcast{} protocol. The version described here is not the most general and certainly not the most efficient. Our goal here is to achieve maximum simplicity in order to facilitate a clear analysis. Some parts of the protocol that deal with the admission of processes into the group are new. Specifically the steps of state {\em donation} and {\em co-donation}.
\item
The third part deals with processes that are admitted to the group by \gms{} (as opposed to processes that are in the group from the start). We prove the rather surprising fact that a process admission event can be reduced to a process removal event. To do that we construct an explicit mapping from a history that contains an admission of a process $G$ to a history that contains a removal of an "opposite" process $\minusG$. We show that the two histories carry the same computation and share the same progress and causality order properties. As a result we can get rid of any finite number of \gms{} admission events.
\item
The fourth part is the proof of the Causal Order Property and the Progress Property in the special case where no new processes are ever admitted to the group beyond the original members. A central idea in the proof is the concept of {\em effective routes}. The \cbcast{} protocol allows a message to be transmitted from source to target multiple times and through multiple routes. However only one of the routes is effective and all the other are redundant. Concentrating on the effective routes simplifies the analysis a great deal.
\end{itemize}

Our model makes a number of assumptions on the behavior of the group membership service. There are many different implementations of such a service in the literature. Such a service is sometimes referred to as a {\em reconfiguration} service. See for example \cite{birman2010virtually, lynch2002rambo}. Virtually all such services are based on the Paxos protocol (see \cite{lamport1998part, lamport2001paxos}). Elsewhere we describe a detailed implementation and analysis of a compliant (primary-component) group membership service (see \cite{arnonXX}).

\section{The Underlying Computational Model}
\label{UnderlyingModelSec}
\subsection{Introduction}
In this section we create a detailed, axiomatic computational model in which the \cbcast{} protocol executes. The purpose of this model is to create a context that is rich and precise enough for us to carry two of the core arguments in this paper. First, that the failure model of a group computation with a group membership service and stop faults is equivalent to the failure model of a single process, where the only failure is an instantaneous failure of the whole system; Second, that any group computation that is based on \cbcast{} where a process $G$ joins in the middle of the computation is identical to a computation carried by a similar group that includes the process $G$ from the start. The first argument allows us to carry out our analysis of the properties of \cbcast{} without taking failures into consideration. The second allows us to carry the analysis under the assumption that processes never join the group during the computation.

The first argument is carried out in the current part. We use the interface points between different components in the model as a place to "shift the blame" from channel and \gms{} failures to process failures. Then we hide the process failures by subsuming them into the group membership service itself. To put this last shift in other words, if a process is officially removed by the group membership service and halts at the same time, then we can analyze the event as a pure group membership service event rather than as two separate events (a group membership service event and a component failure). To make this intuitive, process removal and process halting have to occur simultaneously. The notion of simultaneity requires some work since our model does not include a notion of time. We perform a careful analysis of the partial order that exists between events in the model to show that they can be laid along a timeline in such a way that desired events (such as process removal and process halting) occur at the same time.

To make blame shifting possible, we model the channels as including a queuing stage at the sending and receiving ends of the channel - in other words we add a send queue and a receive queue to each channel. A channel that fails to ship a packet to its destination can shift the blame to the sending process by claiming that the faulty packet never left the send queue. We add a similar queue for \gms{} notifications. This queuing construct is not as artificial as it may sound. Queuing is inherent to communications networks. The sender/channel and channel/receiver boundaries are inherently blurry.

The main technical difficulty in the analysis is keeping infinities from creeping in. If we try to shift the blame for an infinite number of failures to a single component of the model, we end up with the absurd conclusion that an infinite number of events happen in a finite period of time. It is for this reason that the analysis proceeds by dealing with each class of failures in one fell swoop rather than dealing with failures one at a time.

\subsection{Model Overview}
\label{ModelOverviewSS}
We assume a group of {\em processes} - the group having a set of initial members to which at various times new processes
can be added, while processes that are already in the group can be removed. We do not care how the decision
to add or remove processes is made. We assume that there is an opaque {\em group membership service} (\gms{} for short) that notifies all the current
member processes of any addition or removal of a process. The notifications for each process are appended by \gms{} to a {\em \gms{} receive queue}. Eventually the process dequeues each notification in order and processes it. We do not assume that the membership service is reliable - notifications can stop arriving at an eligible process.

We assume that the processes can communicate with each other by exchanging packets on point-to-point, unidirectional {\em channels}.
We assume that the packets in each channel arrive in the order they were sent, i.e. the channel between the source and target is FIFO.
The channels are not assumed to be reliable. When a process $P$ wants to send a packet to a process $Q$, it appends the packet to the {\em send queue} of the outbound channel between $P$ and $Q$. Eventually an opaque driver dequeues the packet and sends it through the channel. When $Q$ receives a packet from $P$, the packet is appended to the {\em receive queue} of the inbound channel by an opaque driver. Eventually the process dequeues each packet and processes it. To make the model more symmetrical we assume that there is a channel, called the {\em self-channel}, between each process and itself.

We assume that each process executes some code. This code is made up of a {\em protocol} \prot{} and an {\em application} \ulp{}. The application is an arbitrary execution thread that makes {\em message multicast requests} to communicate with other processes. Each message multicast request is appended to the {\em \ulp{} receive queue} of the process. Eventually the process dequeues each message and multicasts it.

\ulp{} does not directly specify the set of recipient processes of each multicast. Doing so is impossible since the roster of processes changes over time and messages are multicast asynchronously rather than immediately. Instead, \ulp{} must specifiy a {\em functional set} of recipients, and leaves it to \prot{} to decide who are the members of the functional set at the moment that it performs the multicast.

In general many functional sets can be specified for a multicast. For example if a natural order relation exists between all possible processes (e.g. an order on their identifiers) then the smallest member process, largest member process, or two smallest processes can be specified as functional sets. If the processes have distinguishing characteristics (e.g. colors) then the blue processes or ultraviolet processes can be specified.

In our analysis we do not assume any order or distinguishing characteristics. This leaves us with only one interesting functional set, namely the universal set of all member processes. A multicast to the universal set is called a {\em broadcast}. This is not really a limitation because functional sets can be specified by \ulp{} in the body of each message. Once the message is received by a process, the receiving process can decide on its own whether it belongs to the specified set. If it does not it can simply ignore the message.

\prot{} is a set of non-blocking callbacks that the process uses to process the packets, notifications and message broadcast requests that it dequeues from its various receive queues. We assume that the processing of each received {\em item} - packet, notification or message broadcast request - completes without any context switching, meaning that the items are processed one at a time; that items that arrive at the same queue are processed in the order at which they were queued; and that \ulp{} does not progress while an item is processed.

We assume that each process has a {\em state}. The state can be thought of as the totality of values of all the variables that are managed by \prot{}. We assume that \ulp{} can read some of this state, but cannot change any of it directly. This does not mean that \ulp{} has no private state of its own, however we do not model it. We will see that at least in the case of \cbcast{}, and plausibly in the general case, \ulp{} is started from scratch at each new process and therefore we can treat it as stateless for the purpose of our analysis.

Since the state is only changed by \prot{}, the state of a process when it dequeues the next packet, notification or message broadcast request is identical to the state it had after it finished processing the previous packet, notification or message broadcast request. The initial group of processes all start their life with an identical initial state. Any process that is subsequently admitted into the group by \gms{} starts life with a state that is identical to the state of a {\em parent} which must be an existing group member. The parent is selected by \gms{}. To be more accurate, the new process starts life by dequeuing the \gms{} notification of its own joining. At that moment is has the same state that its parent has when it dequeues the same notification from \gms{}. This makes the act of dequeuing the notification similar to the execution of a $\fp{fork}()$ system call in UNIX.

There are a few types of {\em events} in our model. Each event occurs at a specific process and can be either a {\em queuing} event, also called a {\em side effect} event, or a {\em dequeuing event}, also called a {\em trigger} event.

A trigger event occurs when any item is dequeued from some receive queue. This includes packet, notification and message broadcast request dequeuing events. Such an event triggers processing, using the appropriate callback supplied by \prot{}.

A side effect event occurs when the process multicasts information by appending one or more identical packets to the send queues of outbound channels, destined to different processes. Such an event always occurs as a result of the execution of a callback, and is therefore a side effect of that execution. In other words, the side effects of a trigger are determined by the current state of the process and by the implementation of \prot{}. \ulp{} cannot initiate a queuing event directly, but only through the issue of an message broadcast request.
    
The various send queues of outbound channels; the receive queues of inbound channels; and the receive queues of notifications and message broadcast requests form an integral part of a process in our model. The inclusion of queuing as part of the model is in recognition of the fact that queuing is a fundamental property of any communication mechanism and is not an artifact of any particular implementation. The presence of queuing blurs the boundary between a processes and channels and between processes and \gms{}. This insight is crucial in reasoning about failures.

Some events have an actual or potential causal relations between them. This is captured by a {\em partial order} on events. The queuing event of a packet always precedes its dequeuing event. In addition, the events at a specific process are linearly ordered, capturing the assumption of a single threaded processing of trigger events. In fact, at every process the sequence of events can be broken into intervals composed of a single trigger followed by a finite sequence of the zero or more side effects that are caused by the processing of the trigger. We refer to such a sequence as a {\em transaction}.

The structure of the transactions that compose the event set of each process is determined by \prot{}. Each trigger event causes a \prot{}-specific callback to be executed. The callback can generate an arbitrary number of queuing events.

While the number of events may be infinite over the life of a process, there is only a finite number of events preceding each particular event.

\subsection{The Formal Model}
\subsubsection{Ingredients of the model}
\begin{description}
\item[$\processset$] \hfill \\
The set of all {\em processes}. 
\item[$\processset_h$] \hfill \\
The set of all halting processes. $\processset_h \subset \processset$.
\item[$\numview$] \hfill \\
The number of views. $0 < \numview \le \infty$
\item[$\viewset_i$ where $i < \numview$] \hfill \\
The set of members of the $i^{th}$ view. $\viewset_i$ is a finite set of processes. We will freely refer to both $\viewset_i$ and its index $i$ as a "view".
\item[$\packetset$] \hfill \\
The set of all packets. A packet $k$ has a {\em content} which in denoted by $\cont(k)$ and is protocol and application specific. A packet is {\em sent} from a source process and {\em received} by a target process. Due to faults, a packet may fail to be sent or received. So in general a packet is either sent and received, or sent and not received, or not sent. In the second case, where the packet is sent but not received, we say that the packet is {\em dropped}.

By abuse of notation we will say $k = c$ when we mean $\cont(k) = c$.
\item[$\notificationset$] \hfill \\
The set of all \gms{} notifications. For each view $i$ and each process $P$ there is at most one notification for view change $i$ that is supposed to be received at $P$. We denote this notification by $\notify{i}{P}$. A notification has a {\em content} which is made up of the view change type (join or removal), the identity of the process that is joining or is being removed, and possibly the identity of the parent process (in case of a join). We denote the content of the notification by $\cont(\notify{i}{P})$. The possible contents for a notification are:
\begin{description}
\item[$\nr{\fp{pid}}$]  \hfill \\
Process \fp{pid} is removed.
\item[$\nj{\fp{pid}}{\fp{p\_pid}}$] \hfill \\
Process \fp{pid} is joining, and its parent is the existing member \fp{p\_pid}.
\item[$\nn{\fp{pid}}$] \hfill \\
You (the recipient of the notification) are the new member of the group, and your identifier is \fp{pid}.
\item[$\nd{}$] \hfill \\
You (the recipient of the notification) are no longer a member of the group.
\end{description}

Due to faults, a notification may fail to be received by its target process. In such a case we say that the notification is {\em dropped}.

By abuse of notation we will say $\notify{i}{P} = c$ when we mean $\cont(\notify{i}{P}) = c$.
\item[$\requestset$] \hfill \\
The set of all message broadcast requests from \ulp{}. An message broadcast request $r$ has a {\em content} which is denoted by $\cont(r)$ and is application specific. Since these requests are generated locally we assume that they are never dropped.
\item[$\channel{P}{Q}$] \hfill \\
The unidirectional channel from a {\em source} process $P$ to a {\em target} process $Q$. A channel is a set of packets. $\channel{P}{Q} \subset \packetset$.
\item[$\eventset, \prec$] \hfill \\
The set of all events, partially ordered by the $\prec$ order relation. If $e \prec f$ we say that event $e$ {\em precedes} event $f$, or that $e$ is {\em earlier} than $f$ and $f$ is {\em later} than $e$.

The $\prec$ relationship is a {\em weak} partial order, meaning that in some cases both $e \preceq f$ and $f \preceq e$ can hold at the same time for $e \ne f$. We will in indicate such cases by $e \asymp f$ and say that $e$ and $f$ are {\em contemporaneous}.
\item[$\eventset_P$] \hfill \\
The set of all events that occur at process $P$. If this set is empty we say that $P$ is {\em uninitialized}.
\item[$\requestset_P$] \hfill \\
The set of all requests that \ulp{} issues at process $P$.
\item[\gms{}] \hfill \\
A Group Membership Service that delivers view change notifications to the processes.
\item[\prot{}] \hfill \\
A protocol that determines how triggers are processed and what side effects they create.
\item[\ulp{}] \hfill \\
A user application that generates message broadcast requests at various processes.
\end{description}

\paragraph{Events in the model}
\begin{description}
\item[$\psendevent{k}$] \hfill \\
The packet $k \in \channel{P}{Q}$ is appended to the send queue of the channel. This event occurs at $P$.
\item[$\preceiveevent{k}$] \hfill \\
The packet $k \in \channel{P}{Q}$ is removed from the receive queue of the channel and processed. This event occurs at $Q$.
\item[$\pnotifyevent{i}{P}$] \hfill \\
The notification $\notify{i}{P}$ is removed from the receive queue at process $P$ and processed. This event occurs at $P$.
\item[$\pappevent{r}$] \hfill \\
The message broadcast request $r \in \requestset_P$ is removed from the \ulp{} receive queue and processed. This event occurs at $P$.
\end{description}

The set $\eventset_P$ includes all the packets queued by $P$, all the packets dequeued by $P$, all the view notifications dequeued by $P$, and all the message broadcast requests dequeued at at $P$:
\begin{multline*}
\eventset_P = \left\{ \psendevent{k} \in \eventset \vert \exists Q(k \in \channel{P}{Q}) \right\} \,\bigcup\,
              \left\{ \preceiveevent{k} \in \eventset \vert \exists Q(k \in \channel{Q}{P}) \right\} \\
              \bigcup\, \left\{ \pnotifyevent{i}{P} \in \eventset \right\} \,\bigcup\,
              \left\{ \pappevent{r} \in \eventset \vert r \in \requestset_P \right\}
\end{multline*}

\subsubsection{The \prot{} interface}
\label{ProtIntSSS}
We mentioned in the introduction that each notification; dequeued packet; and message broadcast request has to be processed somehow. The processing is mostly controlled by \prot{} which consists of a small number of non-blocking calls that we will describe below.

When \ulp{} wishes to broadcast a message $m$ it invokes the \prot{} call protBroadcast($m$).

When a brand new process group $\mathit{Grp}$ is initialized, the \prot{} call protStart($\mathit{Grp}$, $P$) must be invoked manually at each member process $P \in \mathit{Grp}$.

When a notification $n$ is dequeued at a process $P$ from the notification queue, some pre-processing has to occur before \prot{} can take over. First of all removal notifications have to be separated from join notifications.

If the notification is the removal notification of a process $R$ then the \prot{} call protRemove($R$) is invoked.

If the notification is a join notification of a process $J$ with parent $E$, the process $P$ has to determine whether it is the designated parent of the joining process. If $P$ is not the designated parent ($P \ne E$) then it invokes the \prot{} call protJoin($J$, $E$). If $P$ is the parent then it must first fork the new process and call protRun($J$) in the child process to initialize $J$. This is summarized in the pseudo-code below.

\begin{procedure}[H]
\caption{doNotification($\fp{n}$)}
\label{doNotification}
\SetAlgoNoLine
\Indp
\Indp
\KwIn{$\fp{n}$ is the notification being processed}
\BlankLine
\If{$\fp{n} = \nr{R}$}
{
	protRemove($R$)\;
}
\ElseIf{$\fp{n} = \nj{J}{E}$ and $E \ne \self{}$}
{
	protJoin($J, E$)\;
}
\Else
{
	\tcp{the local process is the parent}
	\Switch{$\operatorname{fork}()$}
	{
		\Case{$\operatorname{parent}$:}
		{
			\tcp{this block is executed when fork() returns in the parent}
			protJoin($J,E$)\;
		}
		\Case{$\operatorname{child}$:}
		{
			\tcp{this block is executed when fork() returns in the child}
			protRun($J$)\;
		}
	}
}
\end{procedure}

When a received packet $k$ is dequeued at a target process $T$ from the receive queue of an incoming channel $\channel{S}{T}$, the process invokes the \prot{} call protPacket($k$, $S$).

All in all the following six non-blocking calls must be implemented by \prot{}:

\begin{description}
\item[protBroadcast(\fp{m})] \hfill  \\
This call is issued by \ulp{} when it wishes to broadcast a message \fp{m}.
\item[protStart(\fp{roster}, \fp{P})] \hfill  \\
This call is issued manually at each initial member process when the group is initialized at the start of view zero. \fp{roster} is the set of initial members and $\fp{P} \in \fp{roster}$ is the identifier of the process at which the call is issued. While \gms{} does not issue join notifications to the initial members at view zero, this call appears as if it were issued in response to such a notification.
\item[protRun(\fp{P})] \hfill  \\
This call is issued at a new process right after it is forked from its parent. \fp{P} is the identifier of the new process. At the moment of forking the new process is identical to its parent and this call is the means by which the new process acquires an independent identity. While \gms{} does not issue a join notification to the joining member, this call appears as if it were issued in response to such a notification.
\item[protRemove(\fp{P})] \hfill  \\
This call is issued in response to a removal notification from \gms{}. \fp{P} is the identifier of the removed process.
\item[protJoin(\fp{P}, \fp{E})] \hfill  \\
This call is issued in response to a join notification from \gms{}. \fp{P} is the identifier of the joining process and \fp{E} is the identifier of its parent.
\item[protPacket(\fp{k}, \fp{S})] \hfill  \\
This call is issued when a packet is dequeued from a receive queue. \fp{k} is the received packet and \fp{S} is the process identifier of the sender of the packet.
\end{description}

\subsubsection{The \ulp{} interface}
\label{UserApplication}

The user application \ulp{} runs in its own thread at each process. Its only means for communicating with other processes is the \ref{BroadcastMessage} call. With the help of \prot{}, \ulp{} can have view change notifications and messages from other processes delivered to it. In order to facilitate these deliveries, \ulp{} must implement a small number of non-blocking callbacks that are executed by \prot{} within the context of \prot{} calls. As a result these callback execute outside the context of the main \ulp{}-thread. 

We assume that these callbacks are used by \ulp{} to manage an opaque (\ulp{}-dependent) data structure \uldata{}. This data structure can be manipulated by the callbacks, but the main \ulp{} thread can only read that structure and not change it. In order to allow \uldata{} to be initialized \ulp{} must provide an initialization callback.

The callbacks are not able to invoke the \ref{BroadcastMessage} call. Only the main thread can do that.

Specifically, \ulp{} must obey the following rules:
\begin{enumerate}
\item \ulp{} must implement the following callback functions:
\begin{itemize}
\item GroundState(): This call creates the initial value of \uldata{}. It is called when a member of view zero is initialized, thus guaranteeing that the replicated data will start with coherent values at all the initial processes (the meaning of "coherent" here is \ulp{}-specific. It can simply mean "identical", but it can also indicate a more complex relationship.)
\item ApplyMessage(\fp{msg}, \fp{originator}): This call applies a message \fp{msg} from process \fp{originator} to \uldata{}. One possible way to apply the message is to append it to a delivery log (see \cite{birman87}).
\item ApplyJoin(\fp{pid}): This call applies the notification that a process with identity \fp{pid} has joined the group. In \cite{birman87} such notifications are appended to the delivery log.
\item ApplyRemoval(\fp{pid}): This call applies the notification that the process with identity \fp{pid} has been removed from the group. In \cite{birman87} such notifications are appended to the delivery log.
\end{itemize}
\item \ulp{} must implement a Main(\fp{pid}) function. This is the main application thread. Its only parameter is the local process identity. This implies that when \ulp{} is started at a new process, it has no context except for the local process identity and the information that is available to it through \uldata{}.
\item The Main() function of \ulp{} may invoke the \ref{BroadcastMessage} procedure but the callbacks may not.
\item  The Main() function has read-only access to \uldata{}. It may manage additional data outside of \uldata{} without any restrictions.
\item The callbacks listed above have read and write access to \uldata{}. They have access to no other information. In particular they do not know the identity of the local process.
\item The Main() function runs in its own thread. The callbacks are called in the context of a critical section, and therefore must not block.
\end{enumerate}

\subsubsection{Model axioms and histories}
\label{ModelAxiomsSSS}
\newcommand{\AxViewI}{View Interval Axiom}
\paragraph{\AxViewI} \hfill \\
A process $P$ is a member of at least one view, and the set of views of which it is a member is an unbroken interval, called the {\em view interval}, which is either finite or infinite. Formally
$$
\left\{ i | P \in \viewset_i \} = \{ i | j(P) \le i < \vdeath{P} \right\}  
$$

If $j(P) > 0$ then we say that $j(P)$ is the {\em join view} of $P$. If $j(P) = 0$ then we say that $P$ is {\em original}.

If $\vdeath{P} < \numview{}$ then we say that $\vdeath{P}$ is the {\em removal view} of $P$. If $\vdeath{P} = \numview{}$ then we say that $P$ is {\em not removed}.

\newcommand{\AxViewII}{View Change Axiom}
\paragraph{\AxViewII} \hfill \\
View zero, which contains the initial group members, is finite. Every subsequent view differs from its predecessor by the addition or removal of a single process. As a result, each view other than view zero is the join view or the removal view of exactly one process.

\newcommand{\AxPackEventI}{Channel Axiom}
\paragraph{\AxPackEventI} \hfill \\
Every packet belongs to exactly one channel. If $k \in \channel{P}{Q}$ we say that $P$ is the {\em source} of $k$ and $Q$ is the {\em target} of $k$.

\newcommand{\AxPackEventII}{Packet Event Axiom}
\paragraph{\AxPackEventII} \hfill \\
Every packet $k \in \channel{P}{Q}$ has a single queuing event $\psendevent{k} \in \eventset_P$ and at most one dequeuing event $\preceiveevent{k} \in \eventset_Q$. If the packet has a dequeuing event then
$$
\psendevent{k} \prec \preceiveevent{k}
$$
in other words $k$ is queued (at the sending process) before it is dequeued (at the receiving process).

If $k_1 \ne k_2$ then $\preceiveevent{k_1} \ne \preceiveevent{k_2}$ and $\preceiveevent{k_1} \ne \psendevent{k_2}$, wherever these events exist.

If $e = \psendevent{k}$ then the set of packets $M_e = \{ k' \,|\, \psendevent{k'} = e \}$ is finite. $M_e$ is called the {\em multicast set} of $e$. All the packets in $M_e$ have identical contents and share the same source $P$, but they must all have different targets. In other words, no two of them belong to the same channel. The set of target processes $T_e = \{ Q \,\Vert \exists k' \in M_e (k' \in \channel{P}{Q}) \}$ is called the {\em target set} of $e$.

\newcommand{\AxPackEventIII}{Packet Order Axiom}
\paragraph{\AxPackEventIII} \hfill \\
Channels are FIFO. Precisely, if
\begin{itemize}
\item packets $k$ and $k'$ belong to the same channel
\item $\psendevent{k} \prec \psendevent{k'}$
\item $\preceiveevent{k'}$ exists
\end{itemize}
then $\preceiveevent{k}$ exists and $\preceiveevent{k} \prec \preceiveevent{k'}$.

\newcommand{\AxGMSI}{\gms{} Axiom}
\paragraph{\AxGMSI} \hfill \\
The following \gms{} notifications exist in the model
\begin{itemize}
\item For each process $P$ and each $j(P) < i < \vdeath{P}$ there is exactly one notification $\notify{i}{P}$.
\begin{itemize}
\item If $i$ is the join view of process $J$ (with parent $E$) then $\notify{i}{P} = \nj{J}{E}$
\item If $i$ is the removal view of process $R$ then $\notify{i}{P} = \nr{R}$. 
\end{itemize} 
\item For each process $P$ with $j(P) > 0$ there is exactly one notification $\notify{j(P)}{P} = \nn{P}$
\item For each process $P$ with $\vdeath{P} < \numview$ there is exactly one notification $\notify{r(P)}{P} = \nd{}$
\end{itemize}

We also add the following artificial notifications that do not relate to actual \gms{} notifications:
\begin{itemize}
\item For each original process $P$ we add a notification $\notify{0}{P} = \nn{P}$
\item For each process that halts and is not removed we add a notification $\notify{\numview{}}{P} = \nd{}$
\end{itemize}

\newcommand{\AxNotEventI}{Notification Event Axiom}
\paragraph{\AxNotEventI} \hfill \\
If $\notify{i}{P}$ exists, there is at most one $\pnotifyevent{i}{P}$ event. If $i = 0$ and $P$ is original then $\pnotifyevent{i}{P}$ exists.

If $\pnotifyevent{j(P)}{P}$ exists, we will use the shorthand $\joinevent{P} \equiv \pnotifyevent{j(P)}{P}$ \\
If $\pnotifyevent{\vdeath{P}}{P}$ exists, we will use the shorthand $\haltevent{P} \equiv \pnotifyevent{\vdeath{P}}{P}$

\newcommand{\AxNotEventII}{Notification Order Axiom}
\paragraph{\AxNotEventII} \hfill \\
View notifications are dequeued in order. Precisely, if $j(P) \le i < i' \le \vdeath{P}$ and $\pnotifyevent{i'}{P}$ exists then
\begin{enumerate}
\item $\pnotifyevent{i}{P}$ exists.
\item $\pnotifyevent{i}{P} \prec \pnotifyevent{i'}{P}$
\end{enumerate}

\newcommand{\AxNotEventIII}{Parent Axiom}
\paragraph{\AxNotEventIII} \hfill \\
If $J$ joins in view $j(J) > 0$ and $E$ is its parent and $\joinevent{J}$ exists then $\pnotifyevent{j(J)}{E}$ also exists and $\pnotifyevent{j(J)}{E} \asymp \joinevent{J}$. In other words
a new process would not instantiate unless its parent has processed the notification that announces its joining but we model the two events as being contemporaneous.

\newcommand{\AxProcI}{Process Order Axiom}
\paragraph{\AxProcI} \hfill \\
The precedence order $\prec$ is a linear order at each process $P$. In other words any two events in $\eventset_P$ are $\prec$-comparable.

\newcommand{\AxProcII}{Process Liveness Axiom}
\paragraph{\AxProcII} \hfill \\
A process $P$ does not queue a packet to send to another process $Q$ unless $Q$ appears live to $P$. In other words, if $k \in \channel{P}{Q}$ then the following two conditions must be met:
\begin{enumerate}
\item $Q \in \viewset_{j(P)}$ or, $\pnotifyevent{j(Q)}{P} \in \eventset_P$ and $\pnotifyevent{j(Q)}{P} \prec \psendevent{k}$
\item If $\pnotifyevent{\vdeath{Q}}{P} \in \eventset_P$ then $\psendevent{k} \prec \pnotifyevent{\vdeath{Q}}{P}$
\end{enumerate}
Take note that this axiom is a statement about the behavior of \prot{}, since its callbacks are the only elements of the model that generate queuing events.

\newcommand{\AxProcIII}{Piggyback Axiom}
\paragraph{\AxProcIII} \hfill \\
Packets are processed in the same or higher view than the one at which they are queued\footnote{For this axiom to hold the implementer has to "piggyback" the latest view change information on every packet that is sent, in case this information is missing at the receiving side}.
In other words, if $k \in \channel{P}{Q}$ and $\preceiveevent{k}$ exists, then for any $i$, if
$i \le j(P)$ or $\pnotifyevent{i}{P} \prec \psendevent{k}$ then $i \le j(Q)$ or $\pnotifyevent{i}{Q}$ exists and $\pnotifyevent{i}{Q} \prec \preceiveevent{k}$.

\newcommand{\AxProcIV}{Self Channel Axiom}
\paragraph{\AxProcIV} \hfill \\
Packets on self channels are processed early. If $P$ is a process, $k \in \channel{P}{P}$ is a packet and $i$ is a view such that $\pnotifyevent{i}{P}$ exists and $\psendevent{k} \prec \pnotifyevent{i}{P}$ then $\preceiveevent{k}$ exists and $\preceiveevent{k} \prec \pnotifyevent{i}{P}$.

\newcommand{\AxAppI}{Request Event Axiom}
\paragraph{\AxAppI} \hfill \\
If request $r \in \requestset_P$ exists, there is at most one $\pappevent{r}$ event.

\newcommand{\AxOrderI}{Order Foundation Axiom}
\paragraph{\AxOrderI} \hfill \\
The $\prec$ order in $\eventset_P$ is {\em very well founded}, meaning that every event is preceded by only a finite number of earlier events. Formally, for any event $e \in \eventset_P$:
$$
\lvert \left\{ f \in \eventset_P \vert f \prec e \right\} \rvert < \infty
$$

If $\eventset_P \ne \emptyset$ then $\joinevent{P}$ exists and is the first element of $\eventset_P$. \\
If $\haltevent{P}$ exists then is the last element of $\eventset_P$.

\newcommand{\AxOrderII}{Minimal Order Axiom}
\paragraph{\AxOrderII} \hfill \\
The order $\prec$ is the minimal order generated by the order relations at each process and by the orders stipulated by the \AxPackEventII{} and the \AxNotEventIII{}.

\newcommand{\AxHaltI}{First Halting Axiom}
\paragraph{\AxHaltI} \hfill \\
A halting process $P$ has a finite event set. In other words, if $P \in \processset_h$ then $|\eventset_P| < \infty$.

\newcommand{\AxHaltII}{Second Halting Axiom}
\paragraph{\AxHaltII} \hfill \\
Let $P$ be a process and let $i$ be a finite integer in the interval $j(P) \le i \le \vdeath{P}$. Then
\begin{itemize}
\item If $\notify{i}{P}$ is dropped then $\pnotifyevent{i}{P}$ does not exist.
\item If $P$ does not halt and $\notify{j}{P}$ is not dropped for any $j \le i$ then $\pnotifyevent{i}{P}$ exists. In other words if $P$ does not halt then it dequeues all the notifications that it is legally allowed to process. 
\end{itemize}

\newcommand{\AxHaltIII}{Third Halting Axiom}
\paragraph{\AxHaltIII} \hfill \\
Let $P$ and $Q$ be processes and let $k \in \channel{P}{Q}$ be a packet. Then
\begin{itemize}
\item If $P$ does not halt then $k$ is sent. In other words, a non-halting process eventually sends all the packets that it queues to its send queues.
\item If $k$ is not received then $\preceiveevent{k}$ does not exist.
\item If
\begin{itemize}
\item $P$ and $Q$ do not halt
\item Every packet $k' \in \channel{P}{Q}$ where $\psendevent{k'} \preceq \psendevent{k}$ is received
\end{itemize}
then $\preceiveevent{k}$ exists. In other words in the absence of a gap or stoppage, a packet in a channel is eventually dequeued and processed.
\end{itemize}

\newcommand{\AxHaltIV}{Fourth Halting Axiom}
\paragraph{\AxHaltIV} \hfill \\
Let $P$ be a process and let $r \in \requestset_P$ be an message broadcast request. If $P$ does not halt then $\pappevent{r}$ exists. In other words a non-halting process eventually dequeues and processes all of its message broadcast requests.

\begin{defn}
\label{HistoryDef}
A particular vector of values $(\processset,\processset_h,\packetset,\notificationset, \requestset, \eventset, \prec)$ that satisfies the model axioms is called a {\bf history}.
\end{defn}

\subsection{Event Order and Time}
\label{AntiChainSSS}
The group computation model that we presented in the previous section does not include a notion of time. But it does include a notion of causality that is embodied by the partial order on events. Furthermore, any instance of group computation, embodied by the notion of a history, does unfold in physical time. How is physical time related to event order? intuitively, effects always follow causes in time. Also, within any finite interval of time only a finite number of events can occur. Beyond that we can say nothing. In other words, given any physical group computation in our model there should be a timestamp mapping
$$
\tstamp: \eventset^H \rightarrow \mathbb{R}
$$

Where $H$ is the history of the computation, $\mathbb{R}$ is physical time as represented by the real numbers, and $\tstamp(e)$ is the physical time at which the event $e$ occurs. The mapping $\tstamp(\cdot{})$ must have the following properties:
\begin{enumerate}
\item $e \asymp f \implies \tstamp(e) = \tstamp(f)$
\label{TimeCondI}
\item $e \prec f \implies \tstamp(e) < \tstamp(f)$
\label{TimeCondII}
\item $\Big| \left\{ e \in \eventset^H | \tstamp(e) < r  \right\} \Big| < \infty$ \quad for all $r \in \mathbb{R}$
\label{TimeCondIII}
\newcounter{timePropertyCounter}
\setcounter{timePropertyCounter}{\value{enumi}}
\end{enumerate}
Conversely, any arbitrary timestamp mapping $\tstamp(\cdot{})$ that meets the above criterion should be realizable by a physical computation that yields the history $H$. One simply has to assume that the the computation unfolds at each process at the speed that is dictated by $\tstamp(\cdot{})$.

With that in mind, we want to show that every history can be realized by a physical computation that unfolds in physical time in a particularly convenient fashion. Namely we want to show that for every history there is a timestamp function that enjoys the additional property
\begin{enumerate}
\setcounter{enumi}{\value{timePropertyCounter}}
\item $\tstamp(\pnotifyevent{i}{P}) = i$ \quad whenever $\pnotifyevent{i}{P}$ exists
\label{TimeCondIV}
\end{enumerate}
In other words, all the notification events for view $i$ occur at exactly the same time. By constructing such a realization we will demonstrate that race conditions where different processes have different ideas about group membership can be ignored and the local view change notification events at the various processes can be collapsed into a single event. This is exactly what we intend to do in section \ref{HistoryReduxS}.

For this kind of timing to be possible we must show that the partial order $\prec$ is stratified by view. In other words we must show that every event $e$ can be assigned a value $\view(e)$ such that
\begin{itemize}
\item $e \preceq f$ only if $\view(e) \le \view(f)$.
\item $\view(\pnotifyevent{i}{P}) = i$.
\end{itemize}

It turns out that this can be done for any history $H$.

\subsubsection{K\"{o}nig's Lemma}

A key tool in this and subsequent investigations of the event order relation in $H$ is K\"{o}nig's Lemma, a well known property of some partially ordered sets. We use the following formulation of the lemma:

\newcommand{\KonigLem}{K\"{o}nig's Lemma}
\begin{lem*}[\KonigLem]
Let $A$ be an infinite partially ordered set with a first element $a_0$ where the following properties hold for any element $a \in A$:
\begin{itemize}
\item $a$ has a finite number of immediate successors.
\item If $b > a$ then there is an immediate successor $b_0$ of $a$ such that $b \ge b_0$.
\end{itemize}
Then $A$ contains an infinite ascending branch.
\end{lem*}

\begin{proof}
Call an element $a \in A$ {\em heavy} if it has an infinite number of successors. Then obviously $a_0$ is a heavy element. We will show that every heavy element has a heavy successor. Once we do that we can choose a heavy successor $a_1$ to $a_0$, a heavy successor $a_2$ to $a_1$, etc. until we get an infinite ascending branch $a_0 < a_1 < a_2 < \dots$.

Let $a$ be a heavy element. By assumption, each successor $b > a$ is mediated by an immediate successor of $a$. Since there are an infinite number of the former and only a finite number of the latter, there must be some immediate successor $c$ of $a$ that has an infinite number of its own successors. In other words, $c$ is a heavy successor of $a$.
\end{proof}

\subsubsection{Stratifying events by view}

We start our investigation by showing that the set of notification events of each view form a maximal semi-antichain in the partial order $\prec$, and that these semi-antichains partition the packet events by the view at which they occur, a notion that we will make precise. 

For each finite view $i$ in the view interval $0 \le i \le \numview$ define
$$
G_i = \left\{ \pnotifyevent{i}{X} | \pnotifyevent{i}{X} \text{ exists} \right\}
$$
The sets $G_i$ form a partition of all the notification events in $H$.

We now extend the partition to all the events.
For all finite $0 \le i \le \numview$ Define
\begin{align*}
\hat{K}_i &=
\left\{ e \in \eventset | g \preceq e \text{ for some } g \in G_i \right\} \\
K_i &=
\hat{K}_i \setminus \bigcup_{i < j \le \numview} \hat{K}_j & \text{for finite } 0 \le i \le \numview 
\end{align*}

We will now investigate the order relations between events in the different $G_i$ sets. We will show that each $G_i$ is a semi-antichain, meaning that there are no strict $\prec$-inequalities between its elements, and that elements in a high $G_i$ never precede elements in a low $G_i$.

\begin{lem}
\label{GOrderLem}
Suppose that there are processes $P$ and $Q$ and views $i, j$ such that
$$
\pnotifyevent{i}{P} \prec \pnotifyevent{j}{Q}
$$
Then $i < j$.
\end{lem}

\begin{proof}
The \AxOrderII{} implies that the relation $\pnotifyevent{i}{P} \prec \pnotifyevent{j}{Q}$ is derived from a sequence of parent/child relations ($\pnotifyevent{j(J)}{E} \asymp \joinevent{J}$) and queuing/dequeuing relations ($\psendevent{k} \prec \preceiveevent{k}$). We will prove the claim by induction on the number of intermediate steps in the shortest derivation.

If the derivation is immediate then the two events must occur at the same process, namely $P = Q$. In this case the \AxNotEventI{} and the \AxNotEventII{} imply that $i < j$ and we are done.

If the derivation is longer, look at the first step. This step can be of packet type or parent/child type.

If the first step is of packet type then there is a process $R$ and a packet $k \in \channel{P}{R}$ such that $\preceiveevent{k}$ exists, $\psendevent{k} \in \eventset_P$ and
$$
\pnotifyevent{i}{P} \prec \psendevent{k} \prec \preceiveevent{k} \prec \pnotifyevent{j}{Q}
$$
By the \AxProcIII{} this implies that $i \le j(R)$ or $\pnotifyevent{i}{R} \prec \preceiveevent{k}$. If $\pnotifyevent{i}{R} \prec \preceiveevent{k}$ then we can create a shorter derivation leading from $\pnotifyevent{i}{R}$ to $\pnotifyevent{j}{Q}$ and conclude by induction that $i < j$.

If $i \le j(R)$ then we can create a shorter derivation leading from $\joinevent{R} = \pnotifyevent{j(R)}{R}$ to $\pnotifyevent{j}{Q}$ and conclude by induction that $j(R) < j$ and therefore $i < j$.

If the first step is of parent/child type then there is an initialized process $J$ that is a child of $P$ and
$$
\pnotifyevent{i}{P} \preceq \pnotifyevent{j(J)}{P} \asymp \joinevent{J} \preceq \pnotifyevent{j}{Q}
$$
and either the rightmost or leftmost inequalities is strict.

The left inequality implies, by the \AxNotEventI{} and the \AxNotEventII{}, that $i \le j(J)$. If the inequality is strict then $i < j(J)$.

The right inequality implies by induction that $j(J) \le j$. If the inequality is strict then $j(J) < j$.

Together these facts imply that $i < j$ and we are done.
\end{proof}

\begin{cor}
$G_i \subset K_i$
\end{cor}

\begin{cor}
\label{PartCor}
Every event $e \in \eventset$ is a member of exactly one $K_i$.
\end{cor}

\begin{proof}
It follows directly from the definition that $e$ cannot belong to more than one $K_i$. The difficult part is showing that $e$ belongs to some $K_i$. The event $e$ occurs at some process $Q$ and therefore $\joinevent{Q} \preceq e$. Therefore $e \in \hat{K}_{j(Q)}$. If $\numview < \infty$ then there is a largest $j$ such that $e \in \hat{K}_j$ and therefore $e \in K_j$ and we are done. To finish the proof we have to show that even when $\numview = \infty$ there is such a largest $j$.

By the \AxOrderI{}, there is a largest $j$ such that $\pnotifyevent{j}{Q} \preceq e$. We will show that $j$ is the maximal value we are looking for. Suppose there is a process $P$ and a view $i$ such that $\pnotifyevent{i}{P} \preceq e$. If the derivation of this relation is immediate then $P = Q$ and by definition $i \le j$.

For a non-immediate relation we proceed by induction on the length of the shortest derivation.

If the derivation starts with a packet type step then there is a process $R$ and a packet $k \in \channel{P}{R}$ such that
$$
\pnotifyevent{i}{P} \prec \psendevent{k} \prec \preceiveevent{k} \preceq e
$$
By the \AxProcIII{} the leftmost inequality implies that $i \le j(R)$ or $\pnotifyevent{i}{R} \prec \preceiveevent{k}$. If $\pnotifyevent{i}{R} \prec \preceiveevent{k}$ then we can create a shorter sequence leading from $\pnotifyevent{i}{R}$ to $e$ and conclude by induction that $i \le j$.

If $i \le j(R)$ then we can create a shorter sequence leading from $\joinevent{R} = \pnotifyevent{j(R)}{R}$ to $e$ and conclude by induction that $j(R) \le j$ and therefore $i \le j$.

If the derivation starts with a parent/child type step then there is a child $J$ of $P$ such that
$$
\pnotifyevent{i}{P} \preceq \pnotifyevent{j(J)}{P} \asymp \joinevent{J} \preceq e
$$
The lefthand inequality occurs within $\eventset_P$ and therefore we know that $i \le j(J)$. Also, since $\joinevent{J} = \pnotifyevent{j(J)}{J}$ we can conclude by induction that the righthand inequality implies that $j(J) \le j$. Therefore $i \le j$ and we are done.
\end{proof}

\begin{defn}
\label{EventViewDefn}
If an event $e \in \eventset$ belongs to $K_i$, we say that the {\bf view} of $e$ is $i$ and denote it by $\view(e) = i$. It follows from the proof of Corollary \ref{PartCor} that all the events in $\eventset_P$ immediately following a notification event $\pnotifyevent{i}{P}$ share the same view $i$.
\end{defn}

Next we investigate the order between the $K_i$ sets. 

\begin{lem}
\label{KOrderLem}
Let $e$ and $f$ be events such that $e \preceq f$. Then $\view(e) \le \view(f)$.
\end{lem}

\begin{proof}
By definition there is an event $g \in G_{\view(e)}$ such that $g \preceq e$. Therefore $g \preceq f$ and therefore $f \in \hat{K}_{\view(e)}$. It is easy to see from the definition of $K_*$ that this implies $\view(e) \le \view(f)$.
\end{proof}

\begin{cor}
\label{WellFoundedCor}
The partial event order $\prec$ is very well founded in $\eventset$, meaning that for each event $e$, the set $\left\{ f \in \eventset | f \prec e \right\}$ is finite.
\end{cor}

\begin{proof}
Let $e$ be an event and let $i$ be the view of $e$. Then for every preceding event $f \prec e$ with view $j$ we have $j \le i$ by Lemma \ref{KOrderLem}. Therefore if a process $P$ joins at a view that is higher than $i$ then $\eventset_P$ does not contain any predecessor of $e$ because for any $f \in \eventset_P$ we have $\joinevent{P} \prec f$ and therefore $\view(f) > i$. So all the events that precede $e$ come from early joining processes, and it follows from the \AxViewII{} that there is only a finite number of such processes.

Suppose that the partially ordered set of predecessors of $e$ (including $e$ itself) is infinite. The \AxOrderII{} implies that each event in the set has at most two immediate predecessors (one at the process and one at the channel. A join event of a child has one strict immediate predecessor at the parent process.) The same axiom, together with the \AxProcI{} and \AxOrderI{} imply that every predecessor is mediated by an immediate predecessor. By inverting the direction of the order \KonigLem{} implies that there is an infinite decreasing sequence
$$
e = e_0 \succ e_1 \succ e_2 \succ \dotsb
$$
All the predecessors of $e$ reside on a finite number of early joining processes, and therefore there is one process $X$ that contains an infinite number of the events in the regression chain $e_{j_1} \succ e_{j_2} \succ \dotsb$. But this means that $e_{j_1}$ is an event in $\eventset_X$ that has an infinite number of predecessors at the same process $X$. This contradicts the \AxOrderI{}. 
\end{proof}

\subsubsection{Creating an event timeline}

We want to show that the properties of the event order in $H$ are sufficient for the creation of a timeline that meets the basic timeline criteria (\ref{TimeCondI})-(\ref{TimeCondIII}) as well as the additional criterion (\ref{TimeCondIV}).

The na\"{i}ve plan is to start by assigning $\tstamp(g) = v$ for each event $g \in G_v$. This should work thanks to Lemma \ref{GOrderLem} and because the sets $G_v$ are finite. Then it would be tempting to squeeze all the events in $K_v \setminus G_v$ into the open time interval $(v, v+1) \in \mathbb{R}$.

This would have worked if we could guarantee the finiteness of $K_v$. But we cannot do that because of a number of permissible pathological situations. For example, a process can be removed and yet not halt. Such a process could generate an infinite number of events  with a fixed view. Another example is a process that is not removed, but stops receiving \gms{} notifications at some point. As a result all the events at the process have a fixed view beyond a certain point. A reasonable group communication protocol would eliminate such pathologies (e.g. through various timeout clocks that would force a process to halt when a pathology is suspected.) In Section \ref{ConfModelSSS} we explore a set of such "reasonableness" assumptions. For now we explore the more general situation that requires more finesse in constructing the timeline.

The idea is to force the events in $G_v$ to occur contemporaneously while allowing some events in $K_v$ to occur indefinitely late, bounding their timing by a measure of their pathology.
We measure the pathology of an event $e$ using the notion of a view bound:

\begin{defn}
Let $K = \bigcup_{0 \le i < \numview} (K_i \setminus G_i)$. Let $e \in K$ be an event. The {\bf view bound} of $e$ is
$$
\vbound(e) =
\begin{cases}
\min \left\{ w | \exists g (e \prec g \in G_w) \right\} & \quad \text{if such views exist}  \\
\numview & \quad \text{otherwise}
\end{cases}
$$
\end{defn}

\begin{defn}
$K^b$ denotes the set of events in $K$ of bound $b < \numview$. $K^{\numview}$ denotes the unbounded events in $K$.
\end{defn}

It follows from the definition of $\view()$ and from Lemma \ref{GOrderLem} that $\vbound(e) > \view(e)$.

It follows from Corollary \ref{WellFoundedCor} and from the fact that the sets $G_i$ are finite that the set $K^b$ is finite for each $b < \numview$. The set $K^{\numview}$ may be infinite.

We can now amend our timing plan by placing the events in the finite set $K^b$ on the time interval $(b-1, b)$ and by carefully distributing the events in $K^{\numview}$ between different time intervals. 
To do this we need the following set-theoretic lemma:

\begin{lem}
\label{SequenceLem}
Let $(A, \prec)$ be a countable, very well founded, partially ordered set. Then the elements of $A$ can be listed in a sequence that respects the partial order $\prec$. In other words, the partial order $\prec$ can be extended to a total order of order type $\omega$, the order type of the natural numbers.
\end{lem}

We will prove the lemma below. We use Lemma \ref{SequenceLem} to order the elements of $K^{\numview}$ into a sequence
$$
K^{\numview} = \{ e^1, e^2, e^3, \dots \}
$$
that respects the $\prec$-order.

Now we can place all the events of $\eventset$ on a timeline.

Obviously, every event in $g \in G_v$ occurs at time $v$. As we mentioned before, the events of $K^b$ occur in the time interval $(b-1, b)$. Lemma \ref{SequenceLem} guarantees that the events in $K^b$ can be arranged along that interval in a way that respects the $\prec$-order. In order to reserve some free time to schedule $K^{\numview}$ events in the interval, we restrict the events in $K^b$ to the smaller interval $(b-1, b-1/2)$.

Now we can place the events of $K^{\numview}$ along the timeline. To make matters simpler, we will place all of the events in this set at times of the form $\tstamp(e^i) = n_i + 3/4$ where $n_i \in \mathbb{N}$. We do that inductively by defining:
$$
\tstamp(e^i) = \max \left\{ \lfloor \tstamp(f) \rfloor | f \prec e^i \right\} + 7/4 \qquad \text{($\lfloor x \rfloor$ stands for the integer part of $x$)}
$$
This inductive formula is well defined because each $e^i$ has a finite number of predecessors (Corollary \ref{WellFoundedCor}) and because every predecessor of $e^i$ in $K^{\numview}$ must come earlier in the sequence than $e^i$ and therefore has its time already defined.

Our construction of a timeline for $\eventset$ clearly satisfies conditions (\ref{TimeCondI}), (\ref{TimeCondIII}) and (\ref{TimeCondIV}). We have to demonstrate that condition (\ref{TimeCondII}), which stipulates that the timeline respects the $\prec$-order, is also satisfied. Take any two events $e_1, e_2$ such that $e_1 \prec e_2$.

If $e_2 \in K^{\numview}$ then by construction $\tstamp(e_1) < \tstamp(e_2)$. Henceforth we will assume that if $e_2 \in K$ then $\vbound(e_2) < \numview$.

If $e_1 \in G_i$ and $e_2 \in G_j$ then Lemma \ref{GOrderLem} implies that $\tstamp(e_1) = i < j = \tstamp(e_2)$.

If $e_1 \in G_i$ and $e_2 \in K$ then Lemma \ref{KOrderLem} implies that $\view(e_2) \ge i$ and therefore $\vbound(e_2) > i$. Since we can also assume $\vbound(e_2) < \numview$ it follows that $\vbound(e_2) - 1 < \tstamp(e_2) < \vbound(e_2)$. Therefore
$$
\tstamp(e_1) = i \le \vbound(e_2) -1 < \tstamp(e_2)
$$ 
If $e_1 \in K$ and $e_2 \in G_j$ then by definition $\vbound(e_1) \le j < \numview$ and therefore
$$
\tstamp(e_1) < \vbound(e_1) \le j = \tstamp(e_2)
$$
We are left with the case where $e_1, e_2 \in K$. The order $e_1 \prec e_2$ implies $\vbound(e_1) \le \vbound(e_2)$ and we are assuming that $\vbound(e_2) < \numview$. If there is a strict inequality $\vbound(e_1) < \vbound(e_2)$ then
$$
\tstamp(e_1) < \vbound(e_1) \le \vbound(e_2) - 1 < \tstamp(e_2)
$$

If there is equality then both $e_1$ and $e_2$ belong to the same set $K^{\vbound(e_1)} = K^{\vbound(e_2)}$ and their timing matches their order by construction.

\begin{proof}[Proof of Lemma \ref{SequenceLem}]
The set $A$ is countable. Let the bijection $h:\mathbb{N} \rightarrow A$ be an arbitrary ordering of $A$ into a "bad" sequence (one that does not necessarily extend the $\prec$-order.)

Let $c_{*} = c_1, c_2, c_3, \dots$ be an arbitrary sequence of natural numbers where each number appears an infinite number of times (for example one could use the sequence $1,2,1,3,2,1,4,3,2,1, \dots$)

To create the sequential extension of the $\prec$-order, Start with an empty "good" sequence. We will append the elements of $A$ to the good sequence one at a time using the following infinite process. At the $n^{th}$ step, look at the $c_n^{th}$ element of the bad sequence, namely the element $h(c_n)$. Now do the following:
\begin{enumerate}
\item If $h(c_n)$ has already been appended to the good sequence, do nothing.
\item If some predecessor $a \prec h(c_n)$ has not yet been appended to the good sequence, do nothing.
\item If $h(c_n)$ has not yet been appended to the good sequence, but all of its predecessors in the $\prec$-order had been appended, append $h(c_n)$ to the good sequence now.
\end{enumerate}
The good sequence that this procedure generates has some obvious good properties. It includes at most one copy of each element of $A$, and the sequence order respects the $\prec$-order. We just have to show that every element of $A$ is eventually appended to the good sequence.

Assume instead that some element of $A$ is never appended to the good sequence. Look at the subset $A_0 \subset A$ of elements that are never appended. This subset is not empty by assumption. Since the $\prec$-order on $A$ is very well founded, there must be a $\prec$-minimal element in $A_0$. Suppose that this minimal element is $a_0 = h(j)$. Since the $\prec$-order is very well founded, the element $a_0$ has a finite number of $\prec$-predecessors $a'_1, a'_2, \dots, a'_k$.

By assumption all of these predecessors are appended to the good sequence, and this must happen by some finite step $n_0$. By our assumption on the sequence $c_{*}$ there is some high number $n > n_0$ such that $c_n = j$. At step $n$ the element $a_0 = h(j) = h(c_n)$ will be inspected and it will be found that all of its predecessors have already been appended to the good sequence. As a result $a_0$ will be appended at this point, contrary to our assumption.
\end{proof}

\subsection{Understanding Faults}

A history represents a possible computation by a group of processes in our model. Such a computation can be plagued by different types of stop faults - halting processes, dropped packets and dropped \gms{} notifications. Our goal is to demonstrate that by adding a small number of reasonable restrictions to the model, we will be able to analyze all histories in our model under the assumption that all faults are eliminated except for one simple fault - the simultaneous halt of all the processes in the group. This is a very desirable property, because it means that the whole group behaves the same way a single process would - albeit at a much higher performance and availability level.

The fault simplifications techniques that we propose all involve the use of timers, timeouts and voluntary process halts when timeouts occur. But we are going to keep our computational model asynchronous and avoid introducing the notion of time explicitly. We achieve this by introducing axiomatic correlations between certain faults. These correlations can be forced to occur through the use of timers and timeouts, but these implementation details are not part of the model itself. 

Let us start by looking at notifications. Suppose that $P$ fails to dequeue a notification that $P$ is entitled to. How can that happen? According to the \AxHaltII{} either some $\notify{j}{P}$ is dropped or $P$ halts. So the blame can be assigned to $P$ or to \gms{}. If we want to simplify the fault model and eliminate \gms{} faults, we must shift the blame to $P$. But this is not possible unless $P$ halts. A reasonable implementation will not allow a process $P$ to continue running indefinitely when \gms{} becomes unresponsive. At some point $P$ will timeout and halt voluntarily. If that happens then we have a hope of shifting the blame to $P$ and away from \gms{}.

A similar observation can be made about packets. Suppose there is a packet $k \in \channel{P}{Q}$ that is queued by $P$ but not dequeued by $Q$. How does that happen? According to the \AxHaltIII{} either $P$ halts (which opens the possibility that $k$ is never sent), or $Q$ halts (which opens the possibility that $k$ is never dequeued even if it is received) or else $k$ - or some packet ahead of $k$ - is dropped. So the blame can be assigned to $P$, to $Q$ or to the channel $\channel{P}{Q}$. If we want to simplify the fault model and eliminate $\channel{P}{Q}$ faults, we must shift the blame either to $P$ or to $Q$. But this is not possible unless one of them halts. A reasonable implementation will not allow $P$ and $Q$ to continue running indefinitely when the channel between them becomes unresponsive. At some point one of them must be successfully evicted, and failing that, both of them will timeout at some point and halt voluntarily. If that happens then we have some hope of shifting the blame to $P$ or $Q$ and away from the channel between them.

This built-in ambiguity of the boundary between a process and its channels and membership service allows us to reinterpret the root cause of visible faults. If a packet is not dequeued, we can shift the blame from the channel to the process and vice versa. If a notification is not dequeued we can shift the blame from the membership service to the process and vice versa. But there are limits to this blame shifting game. If a process halts, we cannot keep on blaming it forever. If other processes keep sending it packets indefinitely, then at some point the channels leading to the halted process must start dropping packets since there is no process on the target side to receive them. If the membership service keeps attempting to inform the halted process of view changes, and never succeeds in evicting it, then its notifications will have to start dropping at some point. 

Well implemented processes will not, however, keep attempting to send packets to a process that has already halted a long time ago. A well implemented membership service will not keep a halted process as a member in view after view indefinitely. A well implemented process will monitor its channels and \gms{} and attempt to detect problems and remove them, and at some point it will have to give up and halt.

If everyone behaves well enough then perhaps the blame for all the faults in the system can be shifted to the processes.

The process faults themselves cannot be eliminated entirely, but they can be greatly simplified. When a process halts, there can be an arbitrary, finite number of unsent packets in its send queues, and an arbitrary finite number of packets, notifications and message broadcast requests in its receive queues. The halt can occur in the middle of the processing of a trigger event, when some of the side effect events have already occurred while others have not. This is very different than what one would expect to happen when a process is removed in an orderly fashion. In an orderly removal you would expect a notification to go to all the processes, including the removed process itself, and you would expect the group to enter a quiescent period during which all on going tasks are completed and all packets in flight are received and processed while the application execution is put on hold so it does not create any new work. Only once the group is fully quiescent will the removed process halt, the group reconfigure itself, and finally resume its normal work under the new configuration.

We will demonstrate that under reasonable assumptions, a general process halt can be simplified to look like an orderly removal - with the aforementioned exception of a simultaneous halt of the whole group. This means that with this one exception, we can view any process halt as the {\em consequence} of a planned process removal, rather than as a failure that was the {\em cause} of a removal. In other words, we can make process failures go away entirely, with the one notable exception of a simultaneous system-wide failure.

First of all we can assume that a halting process completes the execution of its final transaction before halting (see Section \ref{ModelOverviewSS} for a definition of transactions). This can affect its send queues and its internal state, but it does not affect any or the remaining processes as long as the packets that the transaction generates remain stuck in a send queue. Also as long as the newly generated packets are not sent out they do not adversely effect the perceived reliability of the outbound channels of the process.

A further simplification can be achieved by emptying out the receive queues of the process. This is more tricky. The basic idea is pretty simple - just assume that the process dequeues and processes all the items in its receive queues before it halts. This can generate a lot of side effects, but as long as the side effects are stuck in their respective send queues, the surviving members of the group will not be affected.

A final simplification can be achieved by emptying out some of the send queues. This is the most tricky part since it allows a process that is presumably already halted to affect other processes by sending them additional packets. In a sense the halted process can alter the information held by other processes and change history "from the grave". This can be made acceptable if we limit ourselves to only sending packets bound to processes that appear to have halted even earlier. This should confine any new information to the world of the dead and prevent it from altering history.

There are three additional complications.

The first and easiest complication is that the order of processing must comply with the \AxProcIII{}.

The second and more pernicious complication is that the processing must comply with the \AxProcIV{}. The issue is that while most side effects can be left stuck in a send queue or sent "downstream" to an already halted process, the latter axiom sometimes forces packets on the self channel to be sent, received and processed, thus generating more side effects, which themselves could result in more packets on the self channel, ad infinitum. We cannot allow that to happen because this violates the \AxHaltI{}. This problem does not go away by itself. But it does not arise if \prot{} is implemented in a reasonable fashion and does not generate infinite loops for itself in the absence of external stimuli. If the protocol meets this {\em vacuum convergence} condition, then we can assume that a halting process leaves no unprocessed or partially processed items.

The third and perhaps subtlest complication is that processes can give birth to child processes before they halt, so even if a packet is bound to a target that halts earlier than its source did, the information that is conveyed by that packet may live on inside a child of the target process - and the child could live indefinitely. This complication is bound together with a more general issue of race conditions that can arise when a process births a child. As with the previous complication, these race condition issues cannot be wished away. Instead we must require that \prot{} behave in a way that prevents race conditions from occurring. An additional but less critical simplifying requirement is that processes that fail to initialize must not be chosen by \gms{} to be the parents of child processes.
 
These observations lead us to consider a sub-class of histories that arise from {\em conforming models}, which are models that assume well behaving processes; a well behaving membership service; and a well behaving \prot{} layer. We will see that for such histories it is possible to drastically simplify the faults that need to be considered.

\subsubsection{Conforming models and conforming histories} \hfill \\
\label{ConfModelSSS}
We claimed that under some assumptions of "reasonableness" of the way the drivers and \prot{} are implemented in a model we have a hope of simplifying the behavior of faults. We are going to make this notion exact by introducing a number of new axioms that are less general than the axioms of \ref{ModelAxiomsSSS} but represent behaviors that would be expected from a reasonable implementation. As we mentioned before, implementing these reasonable behaviors requires timers, but the new axioms preserve asynchrony by describing the model simplifications as mere correlations between faults

We start with a few definitions.

\begin{defn}
\label{StuntedHistoryDef}
A {\bf stunted history} is a history where the number of processes is finite and all of them halt.
\end{defn}

\begin{defn}
A {\bf finite channel} is a channel where only a finite number of packets are dequeued. In other words a channel $\channel{P}{Q}$ is finite if
$\left| \left\{ k \in \channel{P}{Q} | \preceiveevent{k} \in \eventset \right\} \right| < \infty$.
\end{defn}

\newcommand{\CAxChannel}{Conforming Channel Axiom}
\paragraph{\CAxChannel} \hfill \\
If $\channel{P}{Q}$ is finite, then either $P$ is removed or $Q$ halts.

\newcommand{\CAxPacket}{Conforming Packet Axiom}
\paragraph{\CAxPacket} \hfill \\
A process $P$ does not dequeue packets from process $Q$ after a removal notification for $Q$ is dequeued by $P$.

\newcommand{\CAxNotify}{Conforming Notification Axiom}
\paragraph{\CAxNotify} \hfill \\
If any notification $\notify{i}{P}$ is dropped then $P$ halts.

\newcommand{\CAxView}{Conforming \gms{} Axiom}
\paragraph{\CAxView} \hfill \\
f a process halts then it has a finite view interval.

\newcommand{\CAxParent}{Conforming Parent Axiom}
\paragraph{\CAxParent} \hfill \\
If a process is uninitialized then it has no child processes.

\newcommand{\CAxHalt}{Conforming Halt Axiom}
\paragraph{\CAxHalt} \hfill \\
If a process is removed then it halts.

\begin{defn}
\label{ConformingDef}
A {\bf conforming history} is a history that satisfies all the conforming axioms.
\end{defn}

\subsubsection{Fault equivalence}
\begin{defn} \hfill \\
A {\bf downstream channel} is a channel $\channel{P}{Q}$ where $\vdeath{P} > \vdeath{Q}$ or $P = Q$. In other words it is either a self channel or else it is a channel whose target is removed from the group before its source is removed.
A {\bf upstream channel} is a channel $\channel{P}{Q}$ where $\vdeath{P} \le \vdeath{Q}$ and $P \ne Q$.
A {\bf upstream packet} or {\bf downstream packet} is a packet that belongs to an upstream or downstream channel, respectively.
\end{defn}

\begin{defn} \hfill \\
\label{VacEventDef}
A {\bf vacuum event} is an event $e$ at a halting process that has no lasting effect on the history. To be precise, $e$ is a vacuum event if its successor events
$$
\{
f | f \succeq e
\}
$$
form a finite set, and all them occur at halting processes.

A {\bf vacuum packet} is a packet $k$ with a vacuum queuing event $\psendevent{k}$.
\end{defn}

Obviously, all the successors of a vacuum event are vacuum events as well.
 
\begin{defn}
Two histories are said to be {\bf fault equivalent} if they are identical with the following exceptions:
\begin{itemize}
\item packets that are sent in one need not be sent in the other. Packets that are received in one need not be received in the other.
\item notifications that are dropped in one need not be dropped in the other.
\item the two histories have the same non-vacuum events.
\item the two histories have the same non-vacuum packets.
\end{itemize}
\end{defn}

\begin{defn}
\label{LosslessDef}
A history $H$ is called {\bf lossless} if
\begin{itemize}
\item $H$ has no dropped packets
\item $H$ has no dropped notifications
\item All upstream channels clear their receive queues. In other words all received upstream packets are dequeued.
\item All downstream channels clear their send queues. In other words all queued downstream packets are sent. 
\end{itemize}
\end{defn}

\subsubsection{The Vacuum Loop and vacuum closure}
\begin{defn}
\label{VacuumDef}
Let $H$ be a lossless history and let $P$ be any halting process in $H$. If $j(P) > 0$ let $E$ be the parent of $P$ and assume that $\pnotifyevent{j(P)}{E}$ exists. Look at $P$ at the moment that it halts. Let $v$ be the highest view processed by $P$, namely the value of the highest notification $\notify{v}{P}$ that was processed by $P$ before it halted. By the \AxOrderI{} such a maximal view always exists unless $\eventset_P = \emptyset$. In the latter case set $v = j(P)-1$.

Look at the following loop, called the {\bf vacuum loop}:

\begin{enumerate}
\item If the event $\haltevent{P}$ exists at the end of $\eventset_P$
\begin{itemize}
\item remove it
\item return the notification $\notify{\vdeath{P}}{P}$ to the \gms{} receive queue
\item set $v = \vdeath{P} - 1$
\end{itemize}
\label{VacProcessH}

\item Finish processing the current item. If any new packet multicasts are generated, append the respective packet queuing events at the end of $\eventset_P$ . Queue any new upstream packets to their respective send queues and declare them to be  unsent and unreceived. Queue any new downstream packets to their respective receive queues and declare them to be sent and received.
\label{VacProcessT}

\item If $v \ge j(P)$ and there are any requests in the \ulp{} receive queue, dequeue and process them one by one. For each request, append the generated request dequeuing event at the end of $\eventset_P$. If the processing of a request creates any side effects, append the resulting packet queuing events at the end of $\eventset_P$. Queue any new upstream packets to their respective send queues and declare them to be  unsent and unreceived. Queue any new downstream packets to their respective receive queues and declare them to be sent and received.
\label{VacProcessA}

\item If $v = \vdeath{P}-1$, dequeue and process all the packets in all the receive queues of channels $\channel{Q}{P}$ where $Q \ne P$. For each packet, append the corresponding dequeuing event to $\eventset_P$; and process the packet. If the processing of a packet creates any side effects, append the resulting packet queuing events at the end of $\eventset_P$. Queue any new upstream packets to their respective send queues and declare them to be  unsent and unreceived. Queue any new downstream packets to their respective receive queues and declare them to be sent and received.
\label{VacProcessP}

\item If the receive queue of the self channel $\channel{P}{P}$ is empty, proceed to step (\ref{VacProcessN}). Otherwise, dequeue all the packets there. For each packet, append the corresponding dequeuing event to $\eventset_P$; and process the packet. If the processing creates any side effects, append the resulting packet queuing events at the end of $\eventset_P$. Queue any new upstream packets to their respective send queues and declare them to be  unsent and unreceived. Queue any new downstream packets to their respective receive queues and declare them to be sent and received.
Repeat step (\ref{VacProcessSelf}) until the receive queue of the self channel is empty.
\label{VacProcessSelf}

\item Increment $v$
\label{VacProcessN}
\begin{itemize}
\item if $v = \vdeath{P}$, append a $\haltevent{P}$ event to $\eventset_P$ and exit.
\item dequeue the $\notify{v}{P}$ notification and process it. Append the $\pnotifyevent{v}{P}$ event to $\eventset_P$. If $v = j(P) > 0$, add the order relation $\pnotifyevent{v}{E} \asymp \pnotifyevent{v}{P}$. Notice that in this case $\pnotifyevent{v}{P} = \joinevent{P}$.
\item If the processing of the notification creates any side effects, append the resulting packet queuing events at the end of $\eventset_P$. Queue any new upstream packets to their respective send queues and declare them to be  unsent and unreceived. Queue any new downstream packets to their respective receive queues and declare them to be sent and received.
\item go back to step (\ref{VacProcessA}).
\end{itemize}
\end{enumerate}
\end{defn}

\begin{lem}
\label{VacContLem}
Let $H$ be a lossless history and let $P$ be a halting process in $H$. Assume that $P$ is either a member of view zero, or else the parent of $P$ processes the notification of the joining of $P$ ($\pnotifyevent{j(P)}{E}$ exists, where $E$ is the parent of $P$).

Suppose that the vacuum loop is run against the halting state of $P$ for a finite number of steps. Then the resulting structure is a lossless history that is fault equivalent to $H$. Moreover, if the original history $H$ is conforming then the extended history is conforming as well.
\end{lem}

\begin{proof}
We start by showing that the extended structure is a lossless history. Then we show that it is fault equivalent to the original history $H$.

The proof that the extension of $H$ satisfies the history axioms is by induction. Suppose that all the history axioms remain true, and the history remains lossless after going through a certain number of steps. We will show that the same remains true after running through one more step. For most axioms we do not actually need the inductive hypothesis - they are either trivially true or can be demonstrated directly. But there are a few exceptions.

The \AxPackEventII{} remains true because every new packet that we create comes with a queuing event per multicast. We only add a dequeuing event to packets that are already queued (steps (\ref{VacProcessP}) and (\ref{VacProcessSelf})).

The \AxPackEventIII{} is violated if there are packets $k, k' \in \channel{Q}{Y}$ such that $\psendevent{k} \prec \psendevent{k'}$ and $\preceiveevent{k'}$ exists and yet $\preceiveevent{k}$ does not exist or does not precede $\preceiveevent{k'}$.

Suppose that $\preceiveevent{k'}$ already existed prior to the current step. By induction the extension of $H$ was a history at the conclusion of the last step and therefore $\psendevent{k'}$ already existed and as a result the preceding $\psendevent{k}$ event had existed as well. Therefore by induction $\preceiveevent{k}$ had also existed and preceded $\preceiveevent{k'}$ and we are done.

Assume therefore that $\preceiveevent{k'}$ is generated in the current step. This implies $Y = P$. Moreover, the only steps that create packet dequeuing events are steps (\ref{VacProcessP}) and (\ref{VacProcessSelf}).

If $\preceiveevent{k'}$ is generated in the current step then the packet $k'$ is already received after the previous step. The induction hypothesis implies that $k'$ is a downstream packet in this case, or else it would have been dequeued already due to losslessness (if we are at step (\ref{VacProcessSelf}) then $k'$ is downstream by definition). Therefore, again by induction, losslessness implies that the packet $k$ must have been sent (and therefore received) before the current step. Since $k$ was queued before $k'$, it must have been dequeued before $k'$ is dequeued. Therefore $\preceiveevent{k}$ exists and precedes $\preceiveevent{k'}$.

The \AxNotEventII{} remains true because step (\ref{VacProcessN}) dequeues notifications in order and without gaps.

The \AxNotEventIII{} remains true thanks to the conditions we imposed on $P$.

The \AxProcII{} is a statement about the behavior of \prot{} callbacks, which are the only part of the model that generates queuing events. Since the vacuum loop processes triggers using the appropriate \prot{} callbacks, the axiom remains valid.

Suppose that the \AxProcIII{} is violated. Then there is a packet $k \in \channel{Q}{P}$ that we dequeue in step (\ref{VacProcessP}) or (\ref{VacProcessSelf}) of the vacuum loop even though there is an $i$ such that $i < j(Q)$ or $\pnotifyevent{i}{Q} \prec \psendevent{k}$ and yet $i > j(P)$ and either $\pnotifyevent{i}{P}$ does not exist or it exists and $\preceiveevent{k} \prec \pnotifyevent{i}{P}$.

In the case of the self-channel (step (\ref{VacProcessSelf})) this can be easily seen to be absurd by substituting $Q = P$.

In the case of step (\ref{VacProcessP}) we have by definition $\preceiveevent{k} \succ \pnotifyevent{v}{P}$ and so it must be that $i > v$. Since $v = \vdeath{P}-1$ in this case, we have $i \ge \vdeath{P}$. So we either have $\vdeath{P} < j(Q)$ or we have $\pnotifyevent{\vdeath{P}}{Q} \preceq \pnotifyevent{i}{Q} \prec \psendevent{k}$.  Either of these cases violate the \AxProcII{} which holds by induction prior to the current step.

To see that the \AxProcIV{} remains true notice that the vacuum loop dequeues a notification (in step (\ref{VacProcessN})) only after it verifies that all the packets on the self channel are processed (in step (\ref{VacProcessSelf})). 

The \AxOrderI{} is mostly trivial except for the $\joinevent{P}$ and $\haltevent{P}$ part.

If $\joinevent{P}$ does not already exist then $v = j(P)-1$ and there was no current item to finish processing when the vacuum loop started. By induction, $\eventset_P = \emptyset$. All the steps except step (\ref{VacProcessN}) become no-ops and no new events are added to $\eventset_P$. To see that step (\ref{VacProcessSelf}) is a no-op, notice that a packet can reside in the receive queue of the self-channel only if it was previously queued to the send queue. This would show up as a queuing event which is not possible in our case.

Step (\ref{VacProcessN}) increments $v$ to equal $j(P)$ and adds a $\pnotifyevent{v}{P} = \pnotifyevent{j(P)}{P} = \joinevent{P}$ event as the first element of $\eventset_P$. This proves this case.

Step (\ref{VacProcessH}) of the vacuum loop removes $\haltevent{P}$ from $\eventset_P$ and returns it to the \gms{} receive queue, if necessary. Step (\ref{VacProcessN}) can create or restore the $\haltevent{P}$ event. If that happens then the loop exits, leaving $\haltevent{P}$ as the last event in $\eventset_P$.

All the other axioms are trivial to verify. For the \AxHaltII{} one just needs to remember that we are dealing here exclusively with lossless histories, so there are no dropped notifications.

We still have to show that executing each step of the vacuum loop preserves losslessness (see Definition \ref{LosslessDef}). Since we start with a lossless history we have no dropped notifications. Every packet that is created by the vacuum loop is either unsent and unreceived (if it is an upstream packet) or sent and received (if it is a downstream packet), so the loop never creates a dropped packet and guarantees the losslessness conditions with regard to upstream and downstream channels

We have to show that the extended history is fault equivalent to the original history $H$. This is more or less trivial by construction. All we did was add a finite number of new packets and events. All of the new events occur at the "end of history" in the sense that they do net precede any pre-existing events. All of the new events are added at $P$ which is a halting process. Therefore by definition we only added vacuum events and vacuum packets. Any pre-existing event acquires at most a finite number of new successors and all of them occur at a halting process. Therefore all pre-existing vacuum events remain so in the extended history.

Finally, assume that $H$ is conforming. We must check that the extended history is still conforming.

The vacuum loop does not change the set of removed processes. It does not change the set of halting processes. It does not change the view interval of any process. It adds only a finite number of dequeuing events to $\eventset_P$, and it does not add any dequeuing events to any other process. Therefore the extended history satisfies the \CAxChannel{}, the \CAxView{} and the \CAxHalt{}. The \CAxNotify{} is vacuously true because a lossless history has no dropped notifications.

The \CAxParent{} holds because any uninitialized process in the extended history is uninitialized in the original history and the vacuum loop does not create any new parent/child relationships.

The \CAxPacket{} holds for every process other than $P$ because $H$ is conforming. Suppose that $P$ processes a removal notification of a process $X \ne P$. This implies that $\vdeath{X} < \vdeath{P}$ and therefore the channel $\channel{X}{P}$ is upstream. By losslessness this implies that the receive queue of the channel is empty throughout the execution of the vacuum loop, and therefore the loop does not add any new dequeuing events for packets on that channel.

If the processing of the removal notification of $X$ occurs during the vacuum loop, then we have already shown that all the packets from $X$ are already processed at this point. Since the vacuum loop does not add any new packets to the $X$-receive queue, the axiom holds in this case as well.

If $X = P$ the axiom holds because $\haltevent{P}$ is the last event in $\eventset_{P}$, according to the \AxOrderI{}.
\end{proof}

\begin{defn}
A history extension of the type that is described in Lemma \ref{VacContLem} is called a {\em vacuum continuation} of $H$ at $P$. If the vacuum loop at $P$ terminates then there is a maximal vacuum continuation of $H$ at $P$, called the {\em vacuum closure} of $H$ at $P$. If the vacuum closure of $H$ at $P$ is equal to $H$ (meaning that the vacuum loop did not generate any new events or packets) then we say that $H$ is {\em vacuum closed} at $P$. It is trivial to check that the vacuum closure of $H$ at $P$ is vacuum closed at $P$.
\end{defn}

\begin{cor}
\label{VacuumCleanerCor}
Let $H$ be a lossless history and let $P$ be a halting process in $H$. If $H$ is vacuum closed at $P$, then $P$ processes all the packets that it receives and all the notifications that it is entitled to, from $\joinevent{P}$ to $\haltevent{P}$.
\end{cor}

\begin{proof}
It is easy to see that the vacuum loop does not terminate as long as there is an unprocessed \gms{} notification that $P$ is entitled to. This means that if $P$ is vacuum closed then it must have processed $\haltevent{P}$.

In order to reach the processing of $\haltevent{P}$, the vacuum loop must pass through steps (\ref{VacProcessP}) and (\ref{VacProcessSelf}) while $v = \vdeath{P}-1$. This clears all the receive queues as required.
\end{proof}

\subsubsection{Transactional histories and the Fault Theorem}
\begin{defn}
\label{TransactionalDef}
A {\bf transactional history} is a conforming, lossless history that meets the following additional restrictions:
\begin{enumerate}
\item All notifications are processed to completion.
\item All message broadcast requests are processed to completion.
\item All received packets are processed to completion.
\end{enumerate}
\end{defn}

\begin{defn}
\label{VacConvDef}
A protocol \prot{} is {\bf vacuum convergent} if for any conforming model in which it participates and for any halting process in any lossless history of that model the vacuum loop terminates.
\end{defn}

We are now ready to state the principal finding of this part of the paper - the Fault Theorem.

\newcommand{\FaultThm}{Fault Theorem}
\begin{thm}[\FaultThm]
\label{FaultThm}
Let $H$ be a conforming history, and assume that \prot{} is vacuum convergent. Then $H$ can be extended to a fault equivalent transactional history $\tr(H)$. $\tr(H)$ is called the {\bf transactional closure} of $H$.
\end{thm}

We start by showing that drops can be eliminated, and then we move on to simplifying process faults.

\begin{lem}
\label{CfDropPktLem}
If a packet $k \in \channel{P}{Q}$ is dropped in a conforming history, then either $P$ halts or $Q$ halts.
\end{lem}

\begin{proof}
By the \AxHaltIII{}, $\preceiveevent{k}$ does not exist. By the \AxPackEventIII{}, for any $k' \in \channel{P}{Q}$ with $\psendevent{k'} \succ \psendevent{k}$, $\preceiveevent{k'}$ does not exist either. Therefore, considering that $\eventset_P$ and $\eventset_Q$ are linearly ordered, the only $k' \in \channel{P}{Q}$ packets for which $\preceiveevent{k'}$ exists are the ones for which $\psendevent{k'} \prec \psendevent{k}$, and since (by the \AxOrderI{}) the $\prec$ relation is very well founded at $P$, there are only a finite number of such packets. Therefore the channel is finite. By the \CAxChannel{}, either $P$ is removed or $Q$ halts. In the former case, the \CAxHalt{} guarantees that $P$ halts, and we are done. 
\end{proof}

\begin{lem}
\label{FirstFaultLem}
In a non-stunted conforming history, a process halts if and only if it is removed.
\end{lem}

\begin{proof}
By the \CAxHalt{} we know that every removed process halts. To show the converse, suppose a process $P$ halts but is not removed. By the \CAxView{} process $P$ must have a finite view interval. Since $P$ is not removed, $\vdeath{P} = \numview$ and so $\numview$ is finite. This implies that the number of processes is finite. Let $Q$ be any process. Since $P$ halts, the \AxHaltI{} guarantees that the channel $\channel{P}{Q}$ is finite. By the \CAxChannel{}, $P$ is removed or $Q$ halts. Since $P$ is not removed, $Q$ must halt. We conclude that there is a finite number of processes and all of them halt. In other words the history is stunted.
\end{proof}

\begin{lem}
\label{FinMissLem}
In a conforming history, a halting process $P$ only queues a finite number of packets and only misses a finite number of packets and notifications:
\begin{align*}
&\left| \bigcup_{Q \in \processset} \channel{P}{Q} \right| < \infty \\
&\left| \bigcup_{Q \in \processset} \left\{ k \in \channel{Q}{P} | \preceiveevent{k} \notin \eventset_P \right\} \right| < \infty \\
&\left| \left\{ i \vert j(P) \le i \le \vdeath{P} \text{ and } \pnotifyevent{i}{P} \notin \eventset_P \right\} \right| < \infty
\end{align*}
\end{lem}

\begin{proof}
Let $P$ be a halting process in a conforming history. By the \AxHaltI{}, $P$ has a finite event set. Therefore only a finite number of packets are queued by $P$. This proves the first claim.

By the \CAxView{}, a halting process in a conforming history has a finite view interval. Therefore $P$ is eligible only for a finite number of view notifications. This proves the third claim.

The main difficulty is with the second claim.

Let $Q$ be a process. First we want to show that $P$ misses a finite number of packets from $Q$.

If $Q$ halts, then $Q$ has a finite event set, therefore $Q$ queues a finite number of packets to $\channel{Q}{P}$, therefore $P$ misses finite number of packets from $Q$ and we are done. So assume that $Q$ does not halt. This implies, incidentally, that the history is not stunted.

Since the history is conforming and not stunted, Lemma \ref{FirstFaultLem} implies that $P$ is removed.

If $\vdeath{P} < j(Q)$ then the \AxProcII{} guarantees that $\channel{Q}{P}$ is empty and we are done. So we can assume that
$$
j(Q) < \vdeath{P} < \numview = \vdeath{Q}
$$
And therefore by the \AxGMSI{} the notification $\notify{\vdeath{P}}{Q}$ exists. Since the history is conforming and $Q$ does not halt, there are no dropped notifications at $Q$ (due to the \CAxNotify{}) and therefore $Q$ must process the removal notification of $P$.

By the \AxProcII{}, for every $k \in \channel{Q}{P}$ we have $\psendevent{k} \prec \pnotifyevent{\vdeath{P}}{Q}$. By the \AxOrderI{}, there is only a finite number of events in $\eventset_Q$ that precede $\pnotifyevent{\vdeath{P}}{Q}$ and therefore $\channel{Q}{P}$ is a finite channel. Therefore, $P$ can only miss a finite number of packets from $Q$.

We have established that $P$ misses a finite number of packets from each source process. As long as only a finite number of processes queue packets targeted at $P$, we are done. Since $P$ has a finite view interval, there is only a finite number of processes $Q$ with $j(Q) < \vdeath{P}$. For any $Q$ with $j(Q) > \vdeath{P}$ we have already established that the channel $\channel{Q}{P}$ is empty.
\end{proof}

\begin{thm}[Lossless History Theorem]
\label{LosslessThm}
Every conforming history is fault equivalent to a lossless conforming history that contains the same packets and events.
\end{thm}

\begin{proof}
Let $H$ be a conforming history. We are going to change $H$ into a fault equivalent lossless history by changing the faulting characteristics of notifications and packets, without adding or subtracting any vacuum events and packets. The conforming axioms (see \ref{ConfModelSSS}) are not affected by such changes except for the \CAxNotify{}. However a lossless history has no dropped notifications, so this and all other conforming axioms are going to remain valid.

We start with notifications by simply declaring that none of the notifications are dropped. There are two catches. First, we may have just added an infinite number of notifications into the notification queues of some processes. But it follows from the \CAxNotify{} that dropped notifications only exist at halting processes, and it follows from Lemma \ref{FinMissLem} that halting processes only have a finite number of dropped notifications, so this problem does not occur. The other catch is that we have to prove that $H$ is still a history. There are several assertions in the \AxHaltII{} that are related to dropped notifications
which we now have to verify. Suppose that some notification $\notify{i}{P}$ is dropped in $H$.
\begin{itemize}
\item The \AxHaltII{} claims that if $\notify{i}{P}$ is a dropped notification then $\pnotifyevent{i}{P}$ does not exist. This assertion is not violated when we declare that $\notify{i}{P}$ is not dropped, so we are done in this case.
\item The same axiom claims that if $P$ does not halt and $\notify{i}{P}$ is not dropped then $\pnotifyevent{i}{P}$ does exist. However history $H$ is conforming, and so by the \CAxNotify{} $P$ must halt. Therefore this assertion is not violated either.
\end{itemize}

We now move to packets. We make the following changes in the faulting characteristics of packets in the channel $\channel{P}{Q}$:
\begin{itemize}
\item We declare all the unprocessed upstream packets to be unsent and unreceived.
\item We declare all the unprocessed downstream packets to be sent and received.
\end{itemize}

The same two catches apply here as well. By preventing packets from being sent we may saddle some processes with an infinite number of unsent packets. By forcing packets to be received without being processed we may be saddling some processes with an infinite number of received packets that linger in the process' receive queues indefinitely. In addition, we have to verify all the relevant axioms.

To see that we do not create an infinite number of packets that remain stuck in the send or receive queues of a process $P$, notice that if $P$ halts then Lemma \ref{FinMissLem} guarantees that only a finite number of unprocessed packets exist in $P$'s incoming and outgoing channels and so we cannot create infinities at $P$. If $P$ does not halt then the situation is a little bit more complex and we have to look at the send queues and receive queues of $P$ separately.

It follows from the \CAxHalt{} that $P$ is not removed and therefore $\vdeath{P} = \numview$. As a result most of the outgoing channels of $P$ are downstream. Since we force all the unprocessed downstream packets to be sent we do not create any unsent packets on these outgoing channels. The exceptions the are channels $\channel{P}{Q}$ that lead to some other process $Q$ that is not removed. Lemma \ref{FirstFaultLem} implies that $Q$ does not halt. The \CAxChannel{} implies that $\channel{P}{Q}$ is not finite and therefore the \AxPackEventIII{} implies that all the packets on the channel are processed and as a result we do not change the faulting characteristics of any of packets on these channels.

On the other hand the incoming channels of $P$ are all upstream, with the exception of the self channel $\channel{P}{P}$. Since we force all the unprocessed upstream packets to be unsent, we do not create any received-and-unprocessed packets on these channels. We have already seen that when both ends of a channel do not halt, all the packets on the channel are processed.  Therefore $P$ processes all of the packets on its self channel and as a result we do not change the faulting characteristics of any self channel packets.

As far as axioms go, the only axiom that is related to the fault properties of packets is the \AxHaltIII{}. The first part of this axiom claims that a non-halting process sends all of its queued packets. This part is not violated by our changes because we only declare a packet to be unsent if it emanates from a removed process $P$. By the \CAxHalt{} the process $P$ halts.

The second part of the axiom claims that a packet is not processed unless it is received. We only declare a packet to be unreceived if it is not processed. Therefore this part of the axiom is not violated.

The third and last part of the axiom claims that packets keep getting processed as long as there is no impediment such as a halting source or target, or a previous unreceived packet. Suppose $P$ and $Q$ do not halt. The the \AxHaltIII{} implies that all the packets in $\channel{P}{Q}$ are sent. Lemma \ref{CfDropPktLem} implies that all the packets are received and as a result the \AxHaltIII{} implies that all the packets in the channel are processed. Therefore we do not touch any of these packets and the third part of the axiom remains valid.

We have to show that the revised history is lossless (see Definition \ref{LosslessDef}). This follows directly from our construction. We obviously do not have any dropped packets or notifications anymore. As for channels, since we declared all the unprocessed packets in upstream channels to be unsent and unreceived we have cleared all the receive queues of these channels. Similarly, since we declared all the unprocessed packets in downstream channels to be sent and received, and since all the processed packets must have been sent to begin with, we have cleared all the send queues of these channels, as required.

Our last task is to show that the new history is fault equivalent to the original history. But this is trivial since we did not add or subtract any events.
\end{proof}

\begin{proof}[Proof of the Fault Theorem]
Theorem \ref{LosslessThm} established that the conforming history $H$ is fault equivalent to a lossless conforming history $H_1$. In a two step process, we will improve $H_1$ to a fault equivalent transactional history $H_3$. The intermediate histories will have the following properties:
\begin{itemize}
\item $H_1$ is conforming and lossless.
\item $H_2$ is conforming and lossless. In addition, all processes in $H_2$ are initialized and process all of their notifications.
\item $H_3$ is transactional.
\end{itemize}

The intermediate history $H_2$ is constructed by running the vacuum loop to completion at each halting process. As one might expect, there are some complications.

The first complication is that for Lemma \ref{VacContLem} to apply at a process $P$, we must make sure that its donor $E$, if it exists, had dequeued the join notification of $P$. This can be guaranteed by traversing the processes of $H$ by increasing join view. Since $j(P) > j(E)$, this order guarantees that we run the vacuum loop at $E$ before we run it at $P$. The vacuum convergence property of \prot{} together with Corollary \ref{VacuumCleanerCor} guarantee that $E$ will process the join notification of $P$ before we run the vacuum loop at $P$.

The second complication is that we may be running the vacuum loop an infinite number of times. This only happens if $H_1$ is not stunted, in which case Lemma \ref{FirstFaultLem} implies that every halting process is removed. At each finite step Lemma \ref{VacContLem} guarantees that the resulting structure is a conforming, lossless history that is fault equivalent to the original history $H$. But we have to show that all of these properties are preserved at the limit. This is mostly but not entirely trivial.

The limiting structure $H_2$ is a conforming history because almost all the history axioms and conforming history axioms either deal with views and view intervals, or with the events at a single process, or with events at a pair of processes, or with packets in a single channel. Any such axiom is either not affected by the vacuum loop at all, or is affected only by a finite number of applications of the vacuum loop in our infinite sequence. Therefore all of these axioms follow immediately from Lemma \ref{VacContLem}. The only exception is the \AxOrderII{}, but this axiom is "continuous" in the sense that it naturally commutes with limits. This is because each order relationship that exists at the limit already exists after a finite number of applications of the vacuum loop.

$H_2$ is lossless for the same reason: the requirement that there be no dropped packets or notifications is continuous  - it is fulfilled at the limit if it is fulfilled at each step. The requirements on upstream and downstream channels affect one channel at a time and each channel is affected only by the execution of the vacuum loop at the source and target of the channel.

The non-trivial part is in showing that $H_2$ is fault equivalent to $H_1$. The argument in Lemma \ref{VacContLem} was that the vacuum loop does not create any new non-vacuum events because it only introduces a finite number of events, all of which are vacuum events themselves. Obviously a finiteness argument of this sort cannot simply be carried over to the limit. Instead the fault equivalence of $H_1$ and $H_2$ arises from deeper roots.

Let $e$ be any event in $H_2$ that has an infinite number of successor events in $H_2$. The fault equivalence claim will follow if we can show that $e$ already has an infinite number of successors in $H_1$. From \KonigLem{} it follows that there is an infinite increasing sequence
$$
B = b_1 \preceq b_2 \preceq b_3 \preceq \dots
$$
of successors of $e$ in $H_2$, and we can assume that $B$ contains no gaps, meaning that any two consecutive pair of events in $b_i \preceq b_{i+1}$ in $B$ is a primitive relation, meaning that either
\begin{itemize}
\item $b_i $ and $b_{i+1}$ are adjacent events in a process $P$.
\item There is a packet $k$ such that $b_i = \psendevent{k}$ and $b_{i+1} = \preceiveevent{k}$.
\item There is a parent/child pair of processes $E/J$ such that $b_i = \pnotifyevent{j(J)}{E}$ and $b_{i+1} = \joinevent{J}$
\end{itemize}

All we need to show is that $B$ cannot be made up exclusively of events that are outside of $H_1$, namely events that are added by the extension process. We demonstrate that through a sequence of claims. We assume that $B$ is made up by events that are added by the vacuum loops and use the notation $P(b)$ to indicate the process that event $b$ occurs at. We reach contradiction through a sequence of claims.

Claim: {\em The sequence $P(b_1), P(b_2), \dots$ contains an infinite number of different processes.}

If only a finite number of processes appear in the sequence then there is some process $P_\infty$ that appears an infinite number of times in the sequence. But that means that the sequence $B$ contains an infinite number of events at $P_\infty$ all of which are added, by our assumption, by the vacuum loop at $P_\infty$. This contradicts our assumption that the vacuum loop completes in a finite number of steps at $P_\infty$.

Claim: {\em The sequence $B$ must contain an infinite number of parent/child type pairs $\pnotifyevent{j(J)}{E} \preceq \joinevent{J}$.}

If not, then we can remove an initial segment of $B$ so that it does not contain any such pair. This would leave us only with consecutive pairs that occur at the same process or pairs of the type $\psendevent{k} \prec \preceiveevent{k}$. It follows from Lemma \ref{VacContLem} that the vacuum loop only creates $\preceiveevent{k}$ events for downstream packets. It follows that for all $i$, $\vdeath{P(b_{i+1})} \le \vdeath{P(b_i)}$. As a result the sequence $B$ can only involve a finite number of processes, contradicting the first claim.

Claim: {\em If $B$ contains a pair $\psendevent{k} \prec \preceiveevent{k}$, then $k$ is a packet on a self-channel.}

Suppose that $B$ contains an event $b_i = \preceiveevent{k}$ where $k$ is not on a self-channel. If follows from the previous claim that there is a parent/child pair later on in the sequence. Therefore there is a parent/child pair $E/J$ and a segment in $B$ of the form
$$
\preceiveevent{k} \prec f_1 \prec f_2 \prec \dots \prec f_N \prec \pnotifyevent{j(J)}{E}
$$
where $N \ge 0$; all the events occur at $E$; and $k$ is not on the self-channel $\channel{E}{E}$. Because $k$ is not on the self-channel, the event $\preceiveevent{k}$ must be generated by step (\ref{VacProcessP}) of the vacuum loop, which means that it is generated while $v = \vdeath{E}-1$. This means that the next (and last) notification event that the vacuum loop creates at $E$ is $\haltevent{E}$ and not $\pnotifyevent{j(J)}{E}$.

Now we are ready to draw a contradiction. The sequence $B$ contains an infinite number of parent/child pairs and all other pairs in $B$ are local to a process. Therefore there is a segment in $B$ of the form
$$
\pnotifyevent{j(J)}{E} \preceq \joinevent{J} \prec f_1 \prec f_2 \prec \dots \prec f_N \prec \pnotifyevent{j(K)}{J} \preceq \joinevent{K}
$$
where $N \ge 0$; E/J and J/K are parent/child pairs; and all the events $f_1, f_2, \dots, f_N$ occur at $J$. It follows that the process $J$ is uninitialized in $H_1$, since the $\joinevent{J}$ event is generated by the vacuum loop at $J$. By the \CAxParent{} the process $J$ cannot have a child process $K$. Contradiction.

We are not done yet. We have constructed a conforming, lossless history $H_2$ that is fault equivalent to $H_1$ and which has very good properties. $H_2$ has no uninitialized processes and according to Corollary \ref{VacuumCleanerCor} all the notifications and \ulp{} message broadcast requests in $H_2$ are processed to completion. But $H_2$ is not transactional. The reason is that the vacuum loop at each process may create unprocessed packets in the receive queues of non-self downstream channels. Since we generate $H_2$ from $H_1$ using a single pass over all the halting processes, these packets may never be revisited.

To make sure that we deal with these packets, we perform a second pass, the same way we performed the first pass. As we have already shown, this procedure creates a conforming, lossless history $H_3$ that is fault equivalent to $H_2$. However unlike the previous pass, all the processes in $H_2$ have already processed all of their notifications. Therefore for every downstream non-self channel $\channel{P}{Q}$ the process $P$ already knows that $Q$ is removed. It now follows from the \AxProcII{} that $P$ does not queue any new packets to the channel during the vacuum loop. As a result $H_3$ is transactional.
\end{proof}

\section{The \cbcast{} Algorithm}
\subsection{Introduction}
The algorithm we present here is based on the outline in \cite{birman1991lightweight}. Our description is more detailed than the authors', but the algorithm itself is a bare bones version of the original. Our prime motivation here is to create a description that is easy to check for correctness. Therefore there is no attempt to accommodate frills like allowing multiple clusters, or optimizing time, space or communication complexity. For such considerations refer back to \cite{birman1991lightweight}.

From this point on we consider our model to be conforming and so all the histories that we analyze are assumed to be conforming. The properties of \cbcast{} as presented here do not necessarily hold for non-conforming models.
\subsection{Terminology}
\subsubsection{Messages and delivery}
As mentioned in the introduction, we use the term {\em messages} to refer to the objects being broadcast between processes with the expectation of virtually
synchronous delivery. The \cbcast{} algorithm implements these broadcasts using the underlying multicast of point-to-point packets. When a packet is dequeued at a process
and found to contain a message, the message is not immediately delivered in order to preserve causality constraints. Following \cite{birman1991lightweight}, we
say that the message is {\em received} once the process dequeues its packet up from the receive queue of the channel, but the message within it is {\em delivered} only at the moment when doing so is consistent with causality constraints. The act of delivery is implemented by invoking the ApplyMessage callback, which applies the message to the user's replicated data (see \ref{UserApplication} below). Every message carries with it three pieces of metadata, denoted $\orig{\fp{msg}}$, $\mview{\fp{msg}}$ and $\mvt{\fp{msg}}$ and describing the originator, view and vector time of the message respectively. These notions are explained below.
\subsubsection{Views, installations and view gaps}
The membership service sends a coherent stream of membership change notifications to member processes, as described in the previous section. This creates a
natural sequence of membership {\em views}, starting at $\view(0)$ and progressing through $\view(1)$, $\view(2)$, etc. Each view is a finite set of
processes, and $\view(n)$ is computed from $\view(n-1)$ by taking the $n^{th}$ notification of the membership service and applying it to the earlier view,
either by adding or removing a process.
Each process keeps track of the view notifications as they arrive from the membership service, and then attempts to {\em install} them. Installation involves
waiting a while for packets in flight to arrive at their destinations. This flushing procedure is at the heart of the algorithm and is necessary in order to
guarantee virtually synchronous delivery across views. As a result of this wait, a process may be several views behind as old installations are delayed and new view notifications
keep arriving. This gap is referred to as the {\em view gap}. If a process has received a notification of a new view, we say that the process is {\em aware} of the view, regardless of whether the process has already installed the view.

It should be noted that a process that leaves the group never re-joins. Even if a user re-starts a process, from the point of view of the membership service the re-started process is brand new.
\subsubsection{Instability and forwarding}
A major difficulty in designing a coherent broadcast protocol is that it must make broadcasts look like atomic operations, where in reality message packets are sent to each target process individually and are received (or fail to be received) individually. Creating the perception of atomicity requires careful bookkeeping of the progress of each message by the sender and by each target.

In order to broadcast a message to a set of members of the currently installed view, a sender process creates, for each target, a packet containing the message and then sends the packet through the appropriate channel. The sender process tries to keep track of the arrival of the packets. To do that it creates a set, called the {\em instability set}, that initially contains the identities of all the target processes. When the sender receives an acknowledgment of receipt from a target\footnote{In our implementation, acknowledgments are received directly from the target - one acknowledgment per packet. In more practical implementations with lower communication costs, fewer acknowledgments are used and the sender may have other, more indirect methods of deducing that a packet has been received.}, it removes that target from the instability set of the message. Likewise, if the group membership service notifies the sender that a target has been removed from the group, that target is removed from the instability set. A message with a non-empty instability set is called {\em unstable}. The sender keeps copies of all the unstable messages in a {\em wait set}. If the instability set becomes empty, the message is said to have become {\em stable} and is removed from the wait set.

When a process receives a packet containing a message, it keeps a copy of the message in a {\em receive set} where it waits to be delivered. In addition the receiver has a responsibility to help the sender propagate the message to its intended target set. For that purpose the receiver keeps a copy of the received message in a {\em forwarding queue}. If the receiver learns of the removal of the sender it takes over its duties by re-broadcasting all the messages in the forwarding queue that were received from that sender.

In a practical implementation the receiver tries to keep track of the stabilization of each received message, just like the sender does. Once a message stabilizes there is no more need to help propagate it, even if its sender is removed. Good bookkeeping is essential for keeping forwarding sets small and communication costs low. In our simplified implementation we do not keep track of stabilization on the receiver side, except for the very rudimentary measure of removing obsolete messages from the forwarding queues.

Due to this forwarding mechanism we must differentiate between the {\em originator} of a message, which is the process that originally broadcast a message, and the packet {\em sender} which is the process that happened to send the packet containing the message. Usually these two are the same, but in the presence of process removals, a message may be carried from originator to target through a series of forwarded packets. The forwarding procedure naturally leads to duplicated deliveries. The receiving process removes the duplicates using the vector time (see below). There are simple ways to reduce the number of duplicates and thus reduce the communication cost that is involved in forwarding. We do not include these here for the sake of simplicity.
\subsubsection{Vector time}
Each process keeps track of causality relations between messages using a vector of natural numbers, indexed by member processes, called the {\em vector time}. At each coordinate, the vector time contains the serial number of the latest delivered message that originated at the process that corresponds to that coordinate. The vector
time is reset to zero every time a new view is installed. When a message is originally broadcast by a process the vector time is incremented (at the originator's own coordinate) and the metadata of the message - namely $\orig{\fp{msg}}$, $\mview{\fp{msg}}$ and $\mvt{\fp{msg}}$ - are then set to the id of the process, the currently installed view and current vector time of the process, respectively. These values remain fixed for the lifetime of the message. For a detailed discussion, including proofs, of how the vector time is used to guarantee causality order preservation within each view, see \cite{birman1991lightweight}.
\subsubsection{Cluster Initialization and original processes}
\label{InitializationSSS}
Birman et al (\cite{birman1991lightweight}) assume that the cluster starts at view $1$ with a single member. We have to relax this assumption because a central tool in our analysis of \cbcast{} is the History Reduction Mapping (see Section \ref{HistoryReduxS}) that moves process joins back to the initial view. Therefore we allow an arbitrary finite number of processes to belong to that view (view zero in our exposition). We call these the {\em original} processes. These processes get started through an invocation of the \ref{StartCluster} procedure. We assume that the procedure gets called at each original process at exactly the same time and with the same roster of members that includes exactly the set of original processes. While the assumption of simultaneity is not realistic, we only need to use it with theoretical "reduced" cases and not with actual clusters which can still be assumed to start with a single member.

\subsection{Outline of the Algorithm}
\label{AlgOutlineSS}
Each process, from the moment it joins the group to the moment it leaves, keeps track of the views as notifications are received from the membership service. For each view from the currently installed view to the most recently announced view, the process keeps a list of the members of that view.

Usually, when the process needs to broadcast a message to the group, it fixes the metadata of the message with the current view and vector time, and then sends a packet containing the message to each member of the current view. The process also places the message in a wait set, where it tracks its stabilization as the recipient processes acknowledge the message. However, when there is a view gap (i.e. there are announced views that have not been installed yet) the process refrains from broadcasting messages or fixing their metadata, and it queues them instead. The messages are broadcast whenever the view gap closes.

When a process receives a packet containing a message, it acknowledges its receipt to the sender  (the sender, remember, may be different from the originator). If the message is not a duplicate it is placed in the receive set until it becomes deliverable and a copy of it is created and appended to the tail of the forwarding queue of the sender of the message. When a message becomes deliverable it is removed from the receive set and applied to the user's replicated data. At the time of delivery the process updates its own vector time by incrementing the coordinate that corresponds to the originator of the message.

When a process is notified that another member process has been removed, it takes each message in the forwarding queue of that process and forwards it to all the live processes. These messages are forwarded with their original metadata (originator, view and vector time) unchanged. A copy of each message is placed in the wait set, to await stabilization.

Whenever a view gap exists, such as after a new view notification is received, the process must wait for its wait set to empty out before it can install the next view. Once the wait set becomes empty, the process sends a {\em flush} packet to all the live processes. This packet contains the value of the latest view known to the process (i.e. the current view plus the view gap). The process then waits to receive similar flush packets from all the live processes. Once that happens, the process installs the next view. It applies a view installation notification to the user's replicated data and removes any obsolete messages from the receive set and the forwarding queue.

Our implementation contains, in addition to flush packets, a related type of packet called a {\em ghost} packet. These packets are not necessary in a practical implementation. We use them to facilitate our reasoning about joining processes. A ghost packet is the "ghost" of a flush packet that would have been sent by a child of an existing process, had that child already been born.

When a new process joins the group, it must somehow synchronize its state with the state of the existing processes. This is done in two stages. Initially the new process starts life as a perfect replica of an existing process, the {\em parent}. We do not describe how this is done, and subsume it into the opaque membership service. In addition the new process must compensate for the natural race conditions that occur as a result of the fact that packets have been in flight between its parent and the other members of the group at the moment that it is born. This compensation is performed using the {\em donation} protocol, whereby each existing process other than the parent exchanges instability information with the new process.

\subsection{Variable and Function Definitions}
\subsubsection{Global variables - the state of a process}
\label{ProcessState}
\begin{description}
\item[\cview{}]
The number of the current view.
\item[\vgap{}]
The number of yet-uninstalled views of which we have been notified by the membership service.
\item[\self{}]
Local process identifier.
\item[\mset{}]
The set of identifiers of the member processes of the current view.
\item[\pque{}]
A queue of pending view changes. Each view change is either a joining of a new process or the removal of an existing process.
\item[\lset{}]
The set of identifiers of all the live processes. This includes every process that is a member of the current view or a known future view, excluding all known removed processes.
\item[\cset{}]
A subset of \lset{} that excludes all the processes that joined before the local process but from which a donation has not yet been received.
\item[{\vtime{}[]}]
A vector of natural numbers, indexed by the process identifiers of the members of the current view. This is the vector time of the local process.
\item[\rset{}]
The set of non-duplicate messages that were received and not yet discarded or delivered.
\item[{\fque{}[]}]
A vector of message queues, indexed by process identifiers. For each process identifier, the queue contains copies of all the messages that were sent from and acknowledged to that process - excluding duplicates - in the order they were received. The queue includes messages that were merely forwarded by the sending process and did not originate from it. Each queue includes both delivered and undelivered messages.
\item[\wset{}]
The set of all the messages that the process broadcast or forwarded during the current view (note that forwarded messages may have a $\mview{\fp{msg}}$ of a higher view even if they are forwarded during the current view). Each message in the wait set is paired with an index and an instability set. The index indicates how many messages were broadcast or forwarded out of the process prior to the current message. The instability set contains, for every process that has not yet acknowledged the message, an index that indicates how many broadcast and forwarded messages were received from that process prior to the broadcasting or forwarding of the current message. \wset{} is organized as the union of two data structures:
\begin{enumerate}
\item \bwset{} contains only the messages that were broadcast by the current process.
\item \fwset{} contains only the messages that were forwarded by the current process.
\end{enumerate}
\item[\lque{}]
A queue of all the unsent messages that need to be broadcast once the view gap closes.
\item[\uldata{}]
An opaque object containing the replicated user data. This data is managed by the user application in an application-specific way, subject to the rules listed in subsection \ref{UserApplication}.
\item[\sgh{}]
A number indicating the highest ghost value sent by the process so far. Ghost values are sent out in a strictly increasing sequence.
\item[\sfh{}]
A number indicating the highest flush value sent by the process so far. Flush values are sent out in a strictly increasing sequence.
\item[{\gvec{}[]}]
A vector of view numbers, indexed by process identifiers. It keeps, for each process, the highest ghost value that was received from that process.
\item[{\fvec{}[]}]
A vector of view numbers, indexed by process identifiers. It keeps, for each process, the highest flush value that was received from that process.
\item[\mpkout{}]
A counter of outbound messages, made up of two fields:
\begin{itemize}
\item $\mpkout{}.b$ is the number of original messages broadcast by the process up to this point.
\item $\mpkout{}.f$ is the number of messages forwarded by the process up to this point.
\end{itemize}
\item[{\mpkin{}[]}]
A vector of number pairs, indexed by process identifiers, that counts how many messages have been received from each process so far. The two fields are:
\begin{itemize}
\item $\mpkin{}[P].b$ is the number of original $P$-messages received from $P$.
\item $\mpkin{}[P].f$ is the number of forwarded messages received from $P$
\end{itemize}
\end{description}

\subsubsection{Packet and notification types}
\begin{description}
\item[$\nj{\fp{pid}}{\fp{p\_pid}}$]
Notification that a new process with identifier \fp{pid} joined the group as a clone of the parent process with identifier \fp{p\_pid}.
\item[$\nr{\fp{pid}}$]
Notification that current member process with identifier \fp{pid} was removed from the group.
\item[$\pkm{\fp{msg}}$]
A message packet carrying a message \fp{msg}.
\item[$\pks{\fp{msg}}$]
An acknowledgement packet carrying an acknowledgment of receipt of message \fp{msg}.
\item[$\pkh{\fp{v}}$]
A ghost packet indicating ghost value \fp{v}.
\item[$\pkf{\fp{v}}$]
A flush packet indicating flush up to view \fp{v}.
\item [$\pkh{\ge \fp{v}}$ or $\pkf{\ge \fp{v}}$] Stand for any packet $\pkh{\fp{v'}}$ or $\pkf{\fp{v'}}$ where $\fp{v'} \ge \fp{v}$.
\item[$\pkd{\fp{donation}}$]
A donation packet containing a donation of instability information from an existing process to a newly joined process.
\item[$\pkc{\fp{co-donation}}$]
A co-donation packet containing a donation of instability information from a newly joined process to an existing process.
\end{description}

\subsubsection{Message metadata}
Each message is fixed with three pieces of metadata
\begin{description}
\item[$\orig{\fp{msg}}$] Originator, namely the process that broadcast the message originally.
\item[$\mview{\fp{msg}}$] View, which is the view of the originator when the message is broadcast. 
\item[$\mvt{\fp{msg}}$] Vector Time, which is (roughly) the vector time of the originator when the message is broadcast.
\end{description}
The pair $\left< \mview{\fp{msg}}, \mvt{\fp{msg}} \right>$ uniquely identifies the message \fp{msg}.

\vfill 

\subsection{Detailed \cbcast{} Algorithm Pseudo Code}
\label{PseudoCode}

Our pseudo code implements the various \prot{} interfaces (see \ref{ProtIntSSS})

\begin{tabbing}
protStart(\fp{roster}, \fp{P}) \= at page \pageref{StartCluster} \kill  
protBroadcast(\fp{m}) \> at page \pageref{BroadcastMessage} \\
protStart(\fp{roster}, \fp{P}) \> at page \pageref{StartCluster} \\
protRun(\fp{P}) \> at page \pageref{Run}  \\
protRemove(\fp{P}) \> at page \pageref{RemovalNotification} \\
protJoin(\fp{P}, \fp{E}) \> at page \pageref{JoinNotification} \\
protPacket(\fp{k}, \fp{S}) \> at page \pageref{protPacket}
\end{tabbing}

The protPacket(\fp{k}, \fp{S}) interface implementation uses the following procedures to process the various types of packets that are defined in the \cbcast{} protocol:

\begin{tabbing}
ReceiveCoDonation($\fp{co\_donation}, \fp{sender}$) \= at page \pageref{ReceiveCoDonation} \kill  
ReceiveMessage($\fp{msg}, \fp{sender}$) \> at page \pageref{ReceiveMessage} \\
ReceiveAck($\fp{msg}, \fp{sender}$) \> at page \pageref{ReceiveAck} \\
ReceiveGhost($\fp{view}, \fp{sender}$) \> at page \pageref{ReceiveGhost} \\
ReceivefFlush($\fp{view}, \fp{sender}$) \> at page \pageref{ReceiveFlush} \\
ReceiveDonation($\fp{donation}, \fp{sender}$) \> at page \pageref{ReceiveDonation} \\
ReceiveCoDonation($\fp{co\_donation}, \fp{sender}$) \> at page \pageref{ReceiveCoDonation}
\end{tabbing}

In addition there are three service routines that are called from several places that deal with view installation and message delivery:

\begin{tabbing}
TryToInstall() \= at page \pageref{TryToInstall} \kill  
CheckFlush() \> at page \pageref{CheckFlush} \\
TryToInstall() \> at page \pageref{TryToInstall} \\
Scan() \> at page \pageref{ScanCall} 
\end{tabbing}

\SetKw{Append}{append}
\SetKw{Pop}{pop}
\SetKw{Let}{let}
\SetKw{Add}{add}
\SetKw{Remove}{remove}
\SetKw{Increment}{increment}
\SetKw{Decrement}{decrement}
\SetKw{Reset}{reset}
\SetKw{Queue}{queue}

\begin{procedure}[H]
\SetAlgoProcName{Interface}{Procedure}
\caption{protBroadcast($\fp{msg}$)}
\label{BroadcastMessage}
\SetAlgoNoLine
\Indp
\Indp
\KwIn{$\fp{msg}$ is the message that is being broadcast}
\BlankLine
\If {$\vgap{} > 0$}
{
  \nl\Append \fp{msg} to the tail of \lque{}\;\label{BM:LaunchQueue}
}
\Else
{
  \Increment $\mpkout{}.b$\;
  \tcp{calculate the message vector time}
  \Let $\vtime{}' = \vtime{}$\;
  \Let $\vtime{}'[\self{}] = \vtime{}[\self]  + \mpkout.b - \mpkin{}[\self].b$\;
  \tcp{fix message metadata before broadcasting}
  $\orig{\fp{msg}} \Leftarrow \self{}$\;
  $\mview{\fp{msg}} \Leftarrow \cview{}$\;
  $\mvt{\fp{msg}} \Leftarrow \vtime{}'$\;
  \nl\label{BM:Members}\Queue $\pkm{\fp{msg}}$  to \cset{}\tcp*[l]{multicast the message packets}
  \Let $\se{index} = \mpkout{}$\tcp*[l]{locates \fp{msg} in the outgoing message sequence}
  \Let $\se{iset}[] = \mpkin{}[]$\tcp*[l]{the initial instability set for \fp{msg}}
  \Add $\langle \fp{msg}, \se{index}, \se{iset}[] \rangle$ to \bwset{}\;
}
\end{procedure}

\newpage

\begin{procedure}[H]
\SetAlgoProcName{Interface}{Procedure}
\caption{protStart($\fp{roster}, \fp{pid}$)}
\label{StartCluster}
\SetAlgoNoLine
\Indp
\Indp
\KwIn{$\fp{roster}$ is the set of original members of the group (view zero members). $\fp{pid}$ is the process identifier of the local process}
\BlankLine
GroundState()\tcp*[l]{create the initial value of \uldata{}}
\Let $\cview{} = 0$\;
\Let $\vgap{} = 0$\;
\Let $\self{} = \fp{pid}$\;
\Let $\mset{} = \fp{roster}$\;
\Let $\pque{} = \emptyset$\;
\Let $\lset{} = \cset{} = \fp{roster}$\;
\Let $\vtime{} = \emptyset$\;
\Let $\rset{} = \emptyset$\;
\Let $\fque{} = \emptyset$\;
\Let $\wset{} = \emptyset$\;
\Let $\lque{} = \emptyset$\;
\Let $\sgh{} = \sfh{} = 0$\;
\Let $\gvec{}[] = \fvec{}[] = \emptyset$\;
\Let $\mpkout{} = \{ f = 0; b = 0 \}$\;
\Let $\mpkin{}[] = \emptyset$\;
\ForEach{$\fp{id} \in \fp{roster}$}
{
  create $\vtime{}[\fp{id}] = 0$\;
  create $\fque{}[\fp{id}] = \emptyset$\;
  create $\mpkin{}[\fp{id}] = \{ f = 0; b = 0 \}$\;
  create $\gvec{}[\fp{id}] = \fvec{}[\fp{id}] = 0$\;
  ApplyJoin(\fp{id})\;
}
\tcp{We launch the main \ulp{} thread asynchronously}
\tcp{It will start executing at some indeterminate point in the future}
execute Main(\self{})\;
\end{procedure}

\newpage

\begin{procedure}[H]
\SetAlgoProcName{Interface}{Procedure}
\caption{protRun(\fp{pid})}
\label{Run}
\SetAlgoNoLine
\Indp
\Indp
\KwIn{\fp{pid} is the process identifier of the new process}
\BlankLine
\Increment \vgap{}\;
\Append $\langle \operatorname{JOIN}, \fp{pid} \rangle$ to the tail of \pque{}\;
\Add \fp{pid} to \lset{}\;
\Let $\cset{} = \{ \fp{pid} \}$\;
create $\fque{}[\fp{pid}] = \emptyset$\;
\Let $\bwset{} = \emptyset$\;
\ForEach{$\langle \fp{msg}, \fp{index}, \fp{iset}[] \rangle \in \fwset{}$}
{
	\Let $\fp{index}.b = 0$\;
}
\Let $\lque{} = \emptyset$\;
\Let $\sfh{} = \sgh{}$\;
create $\gvec{}[\fp{pid}] = \sgh{}$\;
create $\fvec{}[\fp{pid}] = \sgh{}$\;
\Let $\mpkout{}.b = 0$\;
create $\mpkin{}[\fp{pid}] = \mpkout{}$\;

\Let $\self = \fp{pid}$\;
\ref{CheckFlush}()\;
\end{procedure}

\newpage

\begin{procedure}[H]
\SetAlgoProcName{Interface}{Procedure}
\caption{protRemove($\fp{rem\_proc}$)}
\label{RemovalNotification}
\SetAlgoNoLine
\Indp
\Indp
\KwIn{$\fp{rem\_proc}$ is the identifier of the removed process}
\BlankLine
\nl\Increment \vgap{}\;\label{RRN:Vgap}
\nl\Append $\langle \operatorname{REMOVE}, \fp{rem\_proc} \rangle$ to the tail of \pque{}\;\label{RRN:NewView}
\Remove \fp{rem\_proc} from \lset{}\;
\Remove \fp{rem\_proc} from \cset{}\;
\ForEach{$\langle \fp{msg}, \fp{index}, \fp{iset} \rangle \in \wset{}$}
{
	discard \fp{iset}[\fp{rem\_proc}]\;
  \If{$\fp{iset} = \emptyset$}
  {
    \Remove $\langle \fp{msg}, \fp{index}, \fp{iset} \rangle$ from \wset{}\tcp*[l]{message is stable}
  }
}
discard $\mpkin{}[\fp{rem\_proc}]$\;
\nl\While{$\fque{}[\fp{rem\_proc}] \ne \emptyset$}
{
	\Pop \fp{msg} from the head of $\fque{}[\fp{rem\_proc}]$\;
  \tcp{create an instability set that contains}
  \tcp{all the live processes}
  \Increment $\mpkout{}.f$\;
  \nl\Let $\se{index} = \mpkout{}$\;\label{RRN:Members}
  \Let $\se{iset}[] = \mpkin{}[]$\;
  \nl\label{RRN:Fwd}\Queue $\pkm{\fp{msg}}$ to \cset{}\tcp*[l]{multicast the message packets}
  \nl\Add $\langle \fp{msg}, \se{index}, \se{iset}[] \rangle$ to \fwset{}\;\label{RRN:WaitSet}
}\label{RN:DiscardFque}
discard $\fque{}[\fp{rem\_proc}]$\;
discard $\gvec{}[\fp{rem\_proc}]$\;
discard $\fvec{}[\fp{rem\_proc}]$\;
\ref{CheckFlush}()\;
\end{procedure}

\newpage

\begin{procedure}[H]
\SetAlgoProcName{Interface}{Procedure}
\caption{protJoin($\fp{jn\_proc}$, $\fp{p\_proc}$)}
\label{JoinNotification}
\SetAlgoNoLine
\Indp
\Indp
\KwIn{$\fp{jn\_proc}$ is the identifier of the joining process. $\fp{p\_proc}$ is the identifier of the parent process.}
\BlankLine
\Increment \vgap{}\;
\Append $\langle \operatorname{JOIN}, \fp{jn\_proc} \rangle$ to the tail of \pque{}\;
\Add \fp{jn\_proc} to \lset{}\;
\Add \fp{jn\_proc} to \cset{}\;
create $\fque{}[\fp{jn\_proc}] = \emptyset$\;
\ForEach{$\langle \fp{msg}, \fp{index}, \fp{iset}[] \rangle \in \wset{}$}
{
	\If{$\fp{iset}[\fp{p\_proc}]$ exists}
	{
  	\nl{create} $\fp{iset}[\fp{jn\_proc}] = \left\{ f = \fp{iset}[\fp{p\_proc}].f ; b = 0 \right\}$;\label{JN:NewInst}
  }
}
create $\gvec{}[\fp{jn\_proc}] = \gvec{}[\fp{p\_proc}]$\;
create $\fvec{}[\fp{jn\_proc}] = \gvec{}[\fp{p\_proc}]$\tcp*[l]{The received flush value of the new process is inherited from the received ghost height of the parent}
create $\mpkin{}[\fp{jn\_proc}] = \left\{ f = \mpkin[\fp{p\_proc}].f ; b = 0 \right\}$\;
\Let $\se{donation} = \langle \wset{}, \mpkin{}[], \sgh{}, \sfh{} \rangle$\;
\Queue $\pkd{\se{donation}}$ to \fp{jn\_proc}\;
\ref{CheckFlush}()\;
\end{procedure}

\newpage

\begin{procedure}[H]
\SetAlgoProcName{Interface}{Procedure}
\caption{protPacket($k$, $\fp{sender}$)}
\label{protPacket}
\SetAlgoNoLine
\Indp
\Indp
\KwIn{$k$ is the packet being received. \fp{sender} is the process identifier of the sender of the packet}
\BlankLine
\Switch{$\cont(k)$}
{
	\Case{$\pkm{\fp{msg}}$:}
	{
		\ref{ReceiveMessage}($\fp{msg}, \fp{sender}$)\;
	}
	\Case{$\pks{\fp{msg}}$:}
	{
		\ref{ReceiveAck}($\fp{msg}, \fp{sender}$)\;
	}
	\Case{$\pkh{\fp{view}}$:}
	{
		\ref{ReceiveGhost}($\fp{view}, \fp{sender}$)\;
	}
	\Case{$\pkf{\fp{view}}$:}
	{
		\ref{ReceiveFlush}($\fp{view}, \fp{sender}$)\;
	}
	\Case{$\pkd{\fp{donation}}$:}
	{
		\ref{ReceiveDonation}($\fp{donation}, \fp{sender}$)\;
	}
	\Case{$\pkc{\fp{co\_donation}}$:}
	{
		\ref{ReceiveCoDonation}($\fp{co\_donation}, \fp{sender}$)\;
	}
}
\end{procedure}

\newpage

\begin{procedure}[H]
\caption{ReceiveMessage($\fp{msg}, \fp{sender}$)}
\label{ReceiveMessage}
\SetAlgoNoLine
\Indp
\Indp
\KwIn{$\fp{msg}$ is the message being received. \fp{sender} is the process identifier of the sender of the message packet}
\BlankLine
\Queue$\pks{\fp{msg}}$ to \fp{sender}\tcp*[l]{acknowledge receipt of the message}
\If{$\orig{\fp{msg}} = \fp{sender}$}
{
	\Increment $\mpkin{}[\fp{sender}].b$\;
}
\Else
{
	\Increment $\mpkin{}[\fp{sender}].f$\;
}
\nl\tcp{Check for duplicates:}\label{RMP:NoDuplicates}
\If{$\mview{\fp{msg}} < \cview{}$}
{
 	\nl\label{RMP:DiscardObsolete}discard $\pkm{\fp{msg}}$\tcp*[l]{obsolete messages are duplicates (Lemma \ref{UniqueEffectiveLem})}
}
\ElseIf{$\mview{\fp{msg}} = \cview{}$ and $\vtime{}[\orig{\fp{msg}}] \ge \mvt{\fp{msg}}[\orig{\fp{msg}}]$}
{
  discard $\pkm{\fp{msg}}$\tcp*[l]{duplicate - message already delivered}
}
\ElseIf{$\fp{msg} \in \rset{}$}
{
  discard $\pkm{\fp{msg}}$\tcp*[l]{duplicate - message already received}
}
\Else
{
  \nl \Add \fp{msg} to \rset{}\;\label{RMP:ReceiveEntry}
  \nl \Append \fp{msg} to the tail of \fque{}[\fp{sender}]\;\label{RMP:Sender}
  \nl\ref{ScanCall}()\tcp*[l]{scan \rset{} and deliver all the deliverable messages}\label{RMP:Scan}
}
\end{procedure}

\begin{procedure}[H]
\caption{ReceiveAck($\fp{msg}, \fp{sender}$)}
\label{ReceiveAck}
\SetAlgoNoLine
\Indp
\Indp
\KwIn{$\fp{msg}$ is the message being acknowledged. \fp{sender} is the process identifier of the sender of the acknowledgement packet}
\BlankLine
\tcp{the following if statement will always succeed}
\If{$\langle \fp{msg}, \fp{index}, \fp{iset} \rangle \in \wset{}$ exists}
{
  discard \fp{iset}[\fp{sender}]\;
  \If{$\fp{iset} = \emptyset$}
  {
    \Remove $\langle \fp{msg}, \fp{index}, \fp{iset} \rangle$ from \wset{}\tcp*[l]{message is stable}
    \ref{CheckFlush}()\;
  }
}
\end{procedure}

\newpage

\begin{procedure}[H]
\caption{ReceiveGhost($\fp{view}, \fp{sender}$)}
\label{ReceiveGhost}
\SetAlgoNoLine
\Indp
\Indp
\KwIn{$\fp{view}$ is the ghost height of the sender. \fp{sender} is the process identifier of the sender of the packet}
\BlankLine
\Let $\gvec{}[\fp{sender}] = \fp{view}$\tcp*[l]{\gvec{}[\fp{sender}] always increases}
\end{procedure}

\begin{procedure}[H]
\caption{ReceiveFlush($\fp{view}, \fp{sender}$)}
\label{ReceiveFlush}
\SetAlgoNoLine
\Indp
\Indp
\KwIn{$\fp{view}$ is the flush height of the sender. \fp{sender} is the process identifier of the sender of the packet}
\BlankLine
\Let $\fvec{}[\fp{sender}] = \fp{view}$\tcp*[l]{\fvec{}[\fp{sender}] always increases}
\ref{TryToInstall}()\;
\end{procedure}

\newpage

\begin{procedure}[H]
\caption{ReceiveDonation($\fp{donation}, \fp{sender}$)}
\label{ReceiveDonation}
\SetAlgoNoLine
\Indp
\Indp
\KwIn{$\fp{donation}$ is the donation being received. \fp{sender} is the process identifier of the sender of the donation packet}
\BlankLine
\Add \fp{sender} to \cset{}\;
\Let $\se{co\_donation} = \langle \wset{}, \mpkin{}[], \sgh{}, \sfh{} \rangle$\;
\Queue $\pkc{\se{co\_donation}}$ to \fp{sender}\tcp*[l]{Co-donate local state to the sender}
\tcp{Process, in order, all the untimely packets}

\Let $\se{UNT}_g = \left\{ \langle \fp{msg}, \fp{index}, \fp{iset}[] \rangle \in \wset{} \,|\, \fp{iset}[\fp{sender}] \text{ exists} \right\}$\;
Define $\heighta(\langle \fp{msg}, \fp{index}, \fp{iset}[] \rangle \in \se{UNT}_g) = \fp{index}.b + \fp{index}.f$\;
Define $\heightb(\langle \fp{msg}, \fp{index}, \fp{iset}[] \rangle \in \se{UNT}_g) = 0$\;
\Let $\se{UNT}_p = \left\{ \langle \fp{msg}, \fp{index}, \fp{iset}[] \rangle \in \fp{donation}.\wset{} \,|\, \fp{iset}[\self] \text{ exists} \right\}$\;
Define $\heighta(\langle \fp{msg}, \fp{index}, \fp{iset}[] \rangle \in \se{UNT}_p) = \fp{iset}[\self].b + \fp{iset}[\self].f$\;
Define $\heightb(\langle \fp{msg}, \fp{index}, \fp{iset}[] \rangle \in \se{UNT}_p) = \fp{index}.b + \fp{index}.f$\;
\Let $\se{UNT} = \se{UNT}_g \,\bigcup\, \se{UNT}_p$\;
sort $\se{UNT}$ using the lexicographical order $(\heighta, \heightb)$\;
\tcp{we process the elements of $\se{UNT}$ in order}
\ForEach{$\langle \fp{msg}, \fp{index}, \fp{iset}[] \rangle \in \se{UNT}$}
{
  \If{$\langle \fp{msg}, \fp{index}, \fp{iset}[] \rangle \in \se{UNT}_p$}
  {
    \If{$\fp{index}.b + \fp{index}.f > \mpkin[\fp{sender}].b + \mpkin[\fp{sender}].f$}
    {
      \tcp{we found an untimely message packet from the sender to the parent, and we process its clone now}
      \nl\ref{ReceiveMessage}($\fp{msg}, \fp{sender}$)\;\label{RD:STMessage}
    }
  }
  \If{$\langle \fp{msg}, \fp{index}, \fp{iset}[] \rangle \in \se{UNT}_g$}
  {
    \If{$\fp{index}.b + \fp{index}.f \le \fp{donation}.\mpkin[\self].b + \fp{donation}.\mpkin[\self].f$}
    {
      \tcp{we found a message packet from the parent whose acknowledgement packet was untimely, so we process its clone now}
      \nl\ref{ReceiveAck}($\fp{msg}, \fp{sender}$)\;\label{RD:STAck}
    }
  }
}
\Let $\gvec{}[\fp{sender}] = \fp{donation}.\sgh{}$\;
\Let $\fvec{}[\fp{sender}] = \fp{donation}.\sfh{}$\;
\end{procedure}

\newpage

\begin{procedure}[H]
\caption{ReceiveCoDonation($\fp{co\_donation}, \fp{sender}$)}
\label{ReceiveCoDonation}
\SetAlgoNoLine
\Indp
\Indp
\KwIn{$\fp{co\_donation}$ is the co-donation being received. \fp{sender} is the process identifier of the sender of the co-donation packet}
\BlankLine
\tcp{Process, in order, all the untimely and post-critical packets}

\Let $\se{UNT}_g = \left\{ \langle \fp{msg}, \fp{index}, \fp{iset}[] \rangle \in \fp{co\_donation}.\wset{} \,|\, \fp{iset}[\self] \text{ exists} \right\}$\;
Define $\heighta(\langle \fp{msg}, \fp{index}, \fp{iset}[] \rangle \in \se{UNT}_g) = \fp{iset}[\self].b + \fp{iset}[\self].f$\;
Define $\heightb(\langle \fp{msg}, \fp{index}, \fp{iset}[] \rangle \in \se{UNT}_g) = \fp{index}.b + \fp{index}.f$\;
\Let $\se{UNT}_p = \left\{ \langle \fp{msg}, \fp{index}, \fp{iset}[] \rangle \in \wset{} \,|\, \fp{iset}[\fp{sender}] \text{ exists} \right\}$\;
Define $\heighta(\langle \fp{msg}, \fp{index}, \fp{iset}[] \rangle \in \se{UNT}_p) = \fp{index}.b + \fp{index}.f$\;
Define $\heightb(\langle \fp{msg}, \fp{index}, \fp{iset}[] \rangle \in \se{UNT}_p) = 0$\;
\Let $\se{UNT} = \se{UNT}_g \,\bigcup\, \se{UNT}_p$\;
sort $\se{UNT}$ using the lexicographical order $(\heighta, \heightb)$\;
\tcp{we process the elements of $\se{UNT}$ in order}
\ForEach{$\langle \fp{msg}, \fp{index}, \fp{iset}[] \rangle \in \se{UNT}$}
{
  \If{$\langle \fp{msg}, \fp{index}, \fp{iset}[] \rangle \in \se{UNT}_g$}
  {
    \If{$\fp{index}.b + \fp{index}.f > \mpkin[\fp{sender}].b + \mpkin[\fp{sender}].f$}
    {
      \tcp{we found one of two things here:}
      \tcp{either an untimely forwarded message packet from the parent that we process now as a message from the sender}
      \tcp{or a post-critical, pre-donation forwarded message packet from the sender that we process now}
      \nl\ref{ReceiveMessage}($\fp{msg}, \fp{sender}$)\;\label{RCD:STMessage}
    }
  }
  \If{$\langle \fp{msg}, \fp{index}, \fp{iset}[] \rangle \in \se{UNT}_p$}
  {
    \If{$\fp{index}.b + \fp{index}.f \le \fp{co\_donation}.\mpkin[\self].b + \fp{co\_donation}.\mpkin[\self].f$}
    {
      \tcp{we found a timely message packet from us to the parent whose acknowledgement was untimely, so we process its clone now}
      \nl\ref{ReceiveAck}($\fp{msg}, \fp{sender}$)\;\label{RCD:STAck}
    }
  }
}
\Let $\gvec{}[\fp{sender}] = \fp{co\_donation}.\sgh{}$\;
\Let $\fvec{}[\fp{sender}] = \fp{co\_donation}.\sfh{}$\;
\nl\ref{TryToInstall}()\;\label{RCD:TryToInstall}
\end{procedure}

\newpage

\begin{procedure}[H]
\caption{CheckFlush()}
\label{CheckFlush}
\SetAlgoNoLine
\Indp
\Indp
\BlankLine
\If{$\fwset{} \ne \emptyset$}
{
    \Return{}\tcp*[l]{there are unstable forwarded messages, do not send any ghosts or flushes}
}
\If{$\sgh{} < \cview{}+\vgap{}$}
{
  \Let $\sgh{} = \cview{} + \vgap{}$\;
  \nl\label{CF:SendGhost}\Queue $\pkh{\sgh{}}$ to \cset{}\tcp*[l]{multicast the ghost packets}
}
\If{$\bwset{} \ne \emptyset$}
{
    \Return\tcp*[l]{there are unstable original messages, do not send any flushes}
}
\If{$\sfh{} < \cview{}+\vgap{}$}
{
  \Let $\sfh{} = \cview{} + \vgap{}$\;
  \nl\label{CF:SendFlush}\Queue $\pkf{\sfh{}}$ to \cset{}\tcp*[l]{multicast the flush packets}
}
\end{procedure}

\newpage

\begin{procedure}[H]
\caption{TryToInstall()}
\label{TryToInstall}
\SetAlgoNoLine
\Indp
\Indp
\BlankLine
\tcp{check whether all the members are fully flushed}
\nl\ForEach{$\fp{pid} \in \lset{}$}
{
  \If{$\fvec{}[\fp{pid}] < \cview{}+\vgap{}$}
  {
    \Return\tcp*[l]{some members are not flushed - wait}
  }
}\label{TTI:CheckFlush}
\While(\qquad\tcp*[h]{this loop installs all the pending views}){$\vgap > 0$}
{
  \nl\label{TTI:RemoveRset}\tcp{remove obsolete messages from \rset{}}
  \ForEach{$\fp{msg} \in \rset{}$}
  {
    \If{$\mview{\fp{msg}} = \cview{}$}
    {
      \Remove \fp{msg} from \rset{}\;
    }
  }
  \nl\label{TTI:RemoveFset}\tcp{remove obsolete messages from \fque{}}
  \ForEach{$\fp{pid} \in \lset{}$}
  {
    \ForEach{$\fp{msg} \in \fque{}[\fp{pid}]$}
    {
      \If{$\mview{\fp{msg}} = \cview{}$}
      {
        \Remove \fp{msg} from $\fque{}[\fp{pid}]$\;
      }
    }
  }
  \Increment \cview{}\;
  \Decrement \vgap{}\;
  \Pop \fp{notification} from the head of \pque{}\;
  \If{$\fp{notification} = \langle \operatorname{JOIN}, \fp{pid} \rangle$}
  {
    \Add \fp{pid} to \mset{}\;
    ApplyJoin(\fp{pid})\tcp*[l]{deliver notification to \ulp{}}
    \If{\fp{pid} = \self{}}
    {
      execute Main(\self{})\tcp*[l]{launch the main \ulp{} thread asynchronously}
    }
  }
  \ElseIf{$\fp{notification} = \langle \operatorname{REMOVE}, \fp{pid} \rangle$}
  {
    \Remove \fp{pid} from \mset{}\;
    ApplyRemoval(\fp{pid})\tcp*[l]{deliver notification to \ulp{}}
  }
  \nl\Reset \vtime{}\tcp*[l]{New view now installed. \vtime{} coordinates reflect new membership}\label{TTI:Installation}
  \ref{ScanCall}()\tcp*[l]{high view messages may now become deliverable}
}
\tcp{$\vgap{} = 0$ - time to broadcast all pending messages}
\While{$\lque{} \ne []$}
{
  \Pop \fp{msg} from the head of \lque{}\;
  \nl\ref{BroadcastMessage}(\fp{msg})\;\label{TTI:NoGapLaunch}
}
\end{procedure}

\newpage

\begin{procedure}[H]
\caption{Scan()}
\label{ScanCall}
\SetAlgoNoLine
\Indp
\Indp
\BlankLine
\tcp{look for deliverable messages. A message is deliverable if all the following are true:}
\tcp{1. It is a current-view message}
\tcp{2. It is the next expected message from its originator}
\tcp{3. All the messages on which it depends have been delivered already}
\Let $\lv{deliverable\_messages\_found} = \lv{false}$\;
\ForEach{$\fp{msg} \in \rset{}$}
{
  \If{$\mview{\fp{msg}} = \cview{}$}
  {
    \tcp{\fp{msg} is a current-view message}
    \If{$\mvt{\fp{msg}}[\orig{\fp{msg}}] = \vtime{}[\orig{\fp{msg}}]+1$}
    {
      \tcp{\fp{msg} is the next expected message from its originator}
      \Let $\lv{all\_dependents\_delivered} = \lv{true}$\;
      \ForEach{$\fp{pid} \in \mset{}$ and $\fp{pid} \ne \orig{\fp{msg}}$}
      {
        \If{$\mvt{\fp{msg}}[\fp{pid}] > \vtime{}[\fp{pid}]$}
        {
          \Let $\lv{all\_dependents\_delivered} = \lv{false}$\;
        }
      }
      \If{$\lv{all\_dependents\_delivered} = \lv{true}$}
      {
        \tcp{\fp{msg} is deliverable}
        \Let $\lv{deliverable\_messages\_found} = \lv{true}$\;
        \Increment \vtime{}[$\orig{\fp{msg}}$]\;
        \Remove \fp{msg} from \rset{}\;
        \Let $\lv{originator} = \orig{\fp{msg}}$\;
        strip out metadata stamps $\mview{\fp{msg}}$, $\mvt{\fp{msg}}$ and $\orig{\fp{msg}}$\;
        \nl ApplyMessage(\fp{msg}, \lv{originator})\tcp*[l]{deliver message to \ulp{}}\label{SC:Delivery}        
      }
    }
  }
}
\If{$\lv{deliverable\_messages\_found} = \lv{true}$}
{
  \ref{ScanCall}()\tcp*[l]{try to see if more messages can now be delivered}
}
\end{procedure}

\newpage

\section{Basic Properties Of The \cbcast{} Algorithm}
In subsequent sections we will analyze the \cbcast{} protocol in depth. Right now we want to highlight some of its important basic properties.
\subsection{Some \cbcast{} invariants}

\begin{defn} \hfill
\label{TransDef}
\begin{itemize}
\item Let $T$ be a transaction. The trigger of $T$ is denoted $\trig(T)$
\item Let $e \in \eventset_P$. Then $e$ belongs to a unique transaction $T$. We denote $T = \trans(e)$ and use $\trig(e)$ as shorthand for $\trig(\trans(e))$.
\item Let $T$ be a transaction. The {\bf view} of $T$ is $\view(\trig(T))$ and denoted by $\view(T)$. Since the side effects of $T$ cannot contain notification events, all the events in $T$ share the same view $\view(T)$.
\end{itemize}
\end{defn}

\begin{defn}
\label{PrePostDef}
Let $P$ be any process in a group that executes the \cbcast{} protocol. Let $\fp{var}$ be any state variable (see \ref{ProcessState}) and let $e \in \eventset_P$ be any event other than the join event of $P$, in other words $e \ne \pnotifyevent{j(P)}{P}$.

If $e$ is a trigger event we use the notation $\valuepre{\fp{var}}{P}{e}$ to denote the value of the variable \fp{var} at process $P$ at the onset of the transaction $\trans(e)$. We use the notation $\valuepost{\fp{var}}{P}{e}$ to denote the value of the variable \fp{var} at process $P$ at the conclusion of the transaction $\trans(e)$.

If $e$ is a queuing event we use the notation $\valuepre{\fp{var}}{P}{e}$ to denote the value of the variable \fp{var} at process $P$ at the moment when the queuing event occurs. Since queuing events do not change state variables, there is no distinction here between the pre and post values.

When $e = \pnotifyevent{j(P)}{P}$, the processing of $e$ causes the execution of the \ref{StartCluster} procedure or the \ref{Run} procedure, according as $P$ is an original process (a member of view zero) or a late joining process. We use the same definition of $\valuepost{\fp{var}}{P}{e}$ that we use for any other trigger. However for an original process we define
$$
\valuepre{\fp{var}}{P}{e} = \valuepost{\fp{var}}{P}{e}
$$
and for a late joining process we define $\valuepre{\fp{var}}{P}{e}$ to be the value of $\fp{var}$ right before the invocation of the \ref{CheckFlush} procedure at the end of the \ref{Run} procedure.

This last part of the definition is admittedly not elegant. However it does have some intuitive justification in the sense that the endpoint of \ref{StartCluster} and the pre-\ref{CheckFlush} point in the \ref{Run} procedure are the first points in the life of a process where it is fully initialized as a \cbcast{} process.
\end{defn}

\begin{defn}
\label{UncontactedDef}
Let $P$ be a process and let $e \in \eventset_P$ be any trigger event. Let $\uncont_e$ be the set of processes that had not yet contacted $P$ at the time that $e$ occurred. These are processes that joined the group before $P$ did, but for which $P$ has not yet processed a donation packet. Formally
$$
\uncont_e = \lbrace Q \in \processset \,\Vert\, j(Q) < j(P) \text{ and if } d = \pkd{} \in \channel{Q}{P} \text{ then either } \preceiveevent{d} \succ e \text{ or } \preceiveevent{d} \text{ does not exist} \rbrace
$$
In particular if $P$ is a member of view zero then $\uncont_e = \emptyset$.

The set $\uncont_e$ is called the {\bf uncontacted set} of $e$.

\end{defn}

\begin{lem}[\cbcast{} Omnibus Lemma] \hfill \\
\label{OmnibusCBCASTLem}
Let $P$ be any process and let $e$ be any trigger event in $P$. Then the following relations hold at $P$:
\begin{enumerate}
\item $\valuepost{(\cview{} + \vgap{})}{P}{e} = \view(e)$
\label{OCL:height}
\item
\begin{align*}
\valuepost{\lset}{P}{e} & = \lbrace Q \in \processset \,\Vert\, j(Q) \le \valuepost{(\cview{} + \vgap{})}{P}{e} < r(Q) \rbrace \\
\valuepost{\cset{}}{P}{e} & = \valuepost{\lset}{P}{e} \setminus \uncont_e
\end{align*}
\label{OCL:lset}
\item The entries in the vectors $\valuepost{\fque{}[]}{P}{e}$, $\valuepost{\mpkin{}[]}{P}{e}$, $\valuepost{\gvec{}[]}{P}{e}$ and $\valuepost{\fvec{}[]}{P}{e}$ correspond exactly to the members of $\valuepost{\lset{}}{P}{e}$.
\label{OCL:lsetvec}
\item The entries in the vector $\valuepost{\vtime{}[]}{P}{e}$ correspond exactly to the members of $\valuepost{\mset{}}{P}{e}$.
\label{OCL:msetvec}
\item For any $X \in \valuepost{\lset{}}{P}{e}$
$$
\valuepost{\fvec{}[X]}{P}{e} \le \valuepost{\gvec{}[X]}{P}{e} \le \valuepost{(\cview{} + \vgap{})}{P}{e}
$$
If $X \notin \valuepost{\cset{}}{P}{e}$ then the right inequality is strict. If $\valuepost{\vgap{}}{P}{e} = 0$ then the inequalities are actually equalities.
\label{OCL:gf}
\item For any  $X \in \valuepre{\lset{}}{P}{e} \,\bigcap\, \valuepost{\lset{}}{P}{e}$
\begin{align*}
\valuepre{\gvec{}[X]}{P}{e} & \le \valuepost{\gvec{}[X]}{P}{e}  \\
\valuepre{\fvec{}[X]}{P}{e} & \le \valuepost{\fvec{}[X]}{P}{e}
\end{align*}
in other words the values of $\gvec{}[X]$ and $\fvec{}[X]$ are non-decreasing.
\label{OCL:gfmon}
\item If $e = \pnotifyevent{i}{P}$ and $X \in \valuepost{\lset{}}{P}{e}$ then
\begin{align*}
\valuepost{\gvec{}[X]}{P}{\pnotifyevent{i}{P}} & \le \valuepre{\sgh{}}{X}{\pnotifyevent{i}{X}}  \\
\valuepost{\fvec{}[X]}{P}{\pnotifyevent{i}{P}} & \le \valuepre{\sfh{}}{X}{\pnotifyevent{i}{X}} 
\end{align*}
\label{OCL:gfxmitrcv}
\item $\valuepost{\fvec{}[P]}{P}{e} \le \valuepost{\sfh{}}{P}{e} \le \valuepost{\sgh{}}{P}{e} \le \valuepost{(\cview{} + \vgap{})}{P}{e}$
\label{OCL:sendgfA} 
\item If $\valuepost{\vgap{}}{P}{e} > 0$ then
\begin{align*}
\valuepost{\sgh{}}{P}{e} & = \valuepost{(\cview{} + \vgap{})}{P}{e} \text{  if and only if  } \valuepost{\fwset{}}{P}{e} = \emptyset \\
\valuepost{\sfh{}}{P}{e} & = \valuepost{(\cview{} + \vgap{})}{P}{e} \text{  if and only if  } \valuepost{\wset{}}{P}{e} = \emptyset
\end{align*}
\label{OCL:sendgfB}
\item If $\valuepost{\vgap{}}{P}{e} = 0$ then $\valuepost{\lque{}}{P}{e} = \emptyset$.
\label{OCL:lque}
\end{enumerate}
\end{lem}

\begin{proof}
The proof proceeds by induction on $e$, where we assume by induction that the lemma holds for $f$ if either $f \prec e$ or if $\view(f) < \view(e)$. This is possible thanks to Corollary \ref{WellFoundedCor} and Lemma \ref{KOrderLem}.

We have three types of triggers to consider: message broadcast request events, notification events and packet dequeuing events. In the latter case we will use the following notation throughout. $k$ is the packet that is being dequeued (so $e = \preceiveevent{k}$) and $X$ is the process that queued the packet $k$.

We start by picking off the easy cases. We show that if $e$ is a packet dequeuing event or a message broadcast request event then claims (\ref{OCL:height}), (\ref{OCL:lset}) and (\ref{OCL:lsetvec}) all hold. claim (\ref{OCL:gfxmitrcv}) holds vacuously for $e$ because it only pertains to notification events.

For the remaining cases, since every trigger event causes some \cbcast{} procedure to be executed, we simply go over each procedure and show that if the inductive hypothesis is assumed then the lemma holds at the end of the execution of the procedure.

To prove claim (\ref{OCL:height}) when $e$ is a packet dequeuing event we need two facts. First, the proof of Corollary \ref{PartCor} demonstrates that $\view(e) = \view(e')$, where $e' \prec e$ is the immediate predecessor of $e$ in $\eventset_P$. Second, a lengthy but routine inspection of the pseudo-code shows that the \ref{protPacket} procedure does not change the sum $\cview{} + \vgap{}$ (the \ref{TryToInstall} utility procedure increments $\cview{}$ and decrements $\vgap{}$ zero or more times, but does not change their sum). These facts taken together with the inductive hypothesis give
$$
\view(e) = \view(e') = \valuepost{(\cview{} + \vgap{})}{P}{e'} = \valuepre{(\cview{} + \vgap{})}{P}{e} = \valuepost{(\cview{} + \vgap{})}{P}{e}
$$

The exact same argument holds when $e$ is a message broadcast request event (with \ref{BroadcastMessage} replacing \ref{protPacket}).

To prove the first part of claim (\ref{OCL:lset}) when $e$ is a packet dequeuing event or a message broadcast request event, notice that neither \ref{protPacket} nor \ref{BroadcastMessage} change the value of \lset{}, and as we already saw these procedures do not change the value of $\cview{} + \vgap{}$ either. As a result this part of the claim follows by induction.

An immediate corollary is that if $e$ is a packet dequeuing event then $X \in \valuepost{\lset}{P}{e}$. This is because the definition of $\view(e)$ and the \CAxPacket{} imply that $j(X) \le \view(e) < \vdeath{X}$. It follows from claim (\ref{OCL:height}) and the first part of claim (\ref{OCL:lset}) that $X \in \valuepost{\lset}{P}{e}$.

The second part of the claim is a bit more complicated. As long as $e$ is not a donation packet it is easy to check that \cset{} does not change and $\uncont_e = \uncont_{e'}$, where $e'$ is the immediate predecessor of $e$ in $\eventset_P$ and so this part of the claim follows by induction.

If $e$ is the donation packet from $X$ then it is easy to see that $X \notin \uncont_e$ and that $\uncont_e = \uncont_{e'} \setminus \{ X \}$. We already demonstrated that $X \in \valuepost{\lset}{P}{e}$. The \ref{protPacket} procedure executes the \ref{ReceiveDonation} procedure which in turn adds $X$ to \cset{} and so by induction
\begin{multline*}
\valuepost{\cset}{P}{e} = \valuepre{\cset}{P}{e} \cup \{ X \} = \valuepost{\cset}{P}{e'} \cup \{ X \} =  \\ 
= (\valuepost{\lset}{P}{e'} \setminus \uncont_{e'}) \cup \{ X \}  = (\valuepost{\lset}{P}{e} \setminus \uncont_{e'}) \cup \{ X \} = \\
= (\valuepost{\lset}{P}{e} \cup \{ X \}) \setminus (\uncont_{e'} \setminus \{ X \}) =  \valuepost{\lset}{P}{e} \setminus \uncont_e
\end{multline*}

Proving claim (\ref{OCL:lsetvec}) when $e$ is a packet dequeuing event or a message broadcast request event amounts to a routine check that only notification related procedures, namely \ref{StartCluster}, \ref{Run}, \ref{JoinNotification} and \ref{RemovalNotification} actually add or remove entries from the mentioned vectors or change \lset{}.

We prove the remaining claims for non-notification events by examining the \ref{BroadcastMessage} procedure and each of the service procedures that are invoked by the \ref{protPacket} procedure. A claim has to be tested against a procedure only if the procedure changes one or more of the variables that are mentioned in the claim.

\begin{description}
\item[\ref{BroadcastMessage}] \hfill \\
This procedure only affects claims (\ref{OCL:sendgfB}) and (\ref{OCL:lque}). If $\vgap{} = 0$ it adds a record to \wset{}, but in that case claim (\ref{OCL:sendgfB}) is vacuously true. If $\vgap{} > 0$ it adds a message to \lque{}, but in this case claim (\ref{OCL:lque}) is vacuously true.
\item[\ref{ReceiveMessage}] \hfill \\
This procedure does not affect any of the claims because it does not change any of the relevant variables (this includes the invocation of the \ref{ScanCall} procedure which also does not change any of the relevant variables).
\item[\ref{ReceiveAck}] \hfill \\
This procedure affects claims (\ref{OCL:sendgfA}) and (\ref{OCL:sendgfB}) by making changes to \wset{}, \sgh{} and \sfh{} but no other variable.

Claim (\ref{OCL:sendgfA}) is true by induction before \ref{CheckFlush} is called. The claim remains true if \ref{CheckFlush} sets $\sgh{} = \cview{} + \vgap{}$. There are two possible impediments to this action. If $\fwset{} \ne \emptyset$ then \ref{CheckFlush} does nothing and the claim remains true by induction. If $\sgh{}$ is already high before \ref{CheckFlush} is called the claim remains true regardless of whether \ref{CheckFlush} raises $\sfh{}$ or not. So (\ref{OCL:sendgfA}) remains true in all cases. 

Claim (\ref{OCL:sendgfB}) is vacuously true if $\vgap{} = 0$ so assume that $\vgap{} > 0$. The procedure may shrink, but does not enlarge, either \fwset{} or \bwset{} and if it removes any record, it invokes the \ref{CheckFlush} procedure. If \wset{} does not lose a record then \ref{CheckFlush} is not called and nothing changes. If \wset{} loses a record, we have to look at the following cases:
\begin{description}
\item[\fwset{} remains non-empty after the record loss] \hfill \\
In this case \ref{CheckFlush} does nothing and the claim remains true by induction.
\item[\fwset{} becomes empty while \bwset{} remains non-empty] \hfill \\
In this case it follows from the inductive hypothesis that
\begin{align*}
\valuepre{\sgh{}}{P}{e} & < \cview{} + \vgap{}  \\
\valuepre{\sfh{}}{P}{e} & < \cview{} + \vgap{}
\end{align*}
and therefore \ref{CheckFlush} sets $\valuepost{\sgh{}}{P}{e} = \cview{} + \vgap{}$ and does not touch \sfh{}. These changes preserve the claims of (\ref{OCL:sendgfB}).
\item[\fwset{} becomes empty while \bwset{} remains empty] \hfill \\
In this case it follows from the inductive hypothesis that
\begin{align*}
\valuepre{\sgh{}}{P}{e} & < \cview{} + \vgap{}  \\
\valuepre{\sfh{}}{P}{e} & < \cview{} + \vgap{}
\end{align*}
and therefore \ref{CheckFlush} sets
$$
\valuepost{\sgh{}}{P}{e} = \valuepost{\sfh{}}{P}{e} = \cview{} + \vgap{}
$$
These changes preserve the claims of (\ref{OCL:sendgfB}).
\item[\fwset{} remains empty while \bwset{} remains non-empty] \hfill \\
In this case it follows from the inductive hypothesis that
\begin{align*}
\valuepre{\sgh{}}{P}{e} & = \cview{} + \vgap{}  \\
\valuepre{\sfh{}}{P}{e} & < \cview{} + \vgap{}
\end{align*}
and therefore \ref{CheckFlush} does nothing and the claim remains true by induction.
\item[\fwset{} remains empty while \bwset{} becomes empty] \hfill \\
In this case it follows from the inductive hypothesis that
\begin{align*}
\valuepre{\sgh{}}{P}{e} & = \cview{} + \vgap{}  \\
\valuepre{\sfh{}}{P}{e} & < \cview{} + \vgap{}
\end{align*}
and therefore \ref{CheckFlush} sets $\valuepost{\sfh{}}{P}{e} = \cview{} + \vgap{}$ and does not touch \sgh{}. These changes preserve the claims of (\ref{OCL:sendgfB}).
\end{description}

\item[\ref{ReceiveGhost}] \hfill \\
This procedure affects claims (\ref{OCL:gf}) and (\ref{OCL:gfmon}), but only with respect to the sender process $X$. We know that $\valuepre{\lset}{P}{e} = \valuepost{\lset}{P}{e}$ and that $X \in \valuepost{\lset}{P}{e}$ so both claims must be verified for $X$.

Also note that it follows from claim (\ref{OCL:lsetvec}) that for every process in $\valuepost{\lset}{P}{e}$ all the fields in claims (\ref{OCL:gf})  and (\ref{OCL:gfmon})  are well defined.

Let $k = \pkh{v}$.

To prove claim (\ref{OCL:gfmon}) we first have to note that the value of \sgh{} (and \sfh{}) is non-decreasing, as is easy to verify by looking at the pseudo-code and specifically at the \ref{CheckFlush} procedure.

Assume first that $k$ is not the first packet from $X$ that carries ghost information (in addition to ghost packets, donation packets and co-donation packets also carry ghost information). It follows from the \AxPackEventIII{} and from the monotonicity of \sgh{} that $v \ge \valuepre{\gvec[X]}{P}{e}$ and we are done.

If $k$ is the first such packet then $\valuepre{\gvec[X]}{P}{e}$ is the initial value assigned by the \ref{StartCluster}, \ref{Run} or the \ref{JoinNotification} procedure, where each of the cases occurs when $j(X) = j(P) = 0$; $j(X) \le j(P) \ne 0$; and $j(X) > j(P)$, respectively.

In case $j(P) = 0$ we have $\valuepre{\gvec[X]}{P}{e} = 0 \le v$ and we are done.

In case $j(X) = j(P) \ne 0$ we have $X = P$ and one can verify by looking at the \ref{Run} procedure that
\begin{multline*}
\valuepre{\gvec[X]}{P}{e} = \valuepre{\gvec{}[P]}{P}{e} = \valuepost{\gvec{}[P]}{P}{\pnotifyevent{j(P)}{P}} = {} \\
{} = \valuepost{\sgh{}}{P}{\pnotifyevent{j(P)}{P}} \le \valuepre{\sgh{}}{P}{\psendevent{k}} = v
\end{multline*}
where the last inequality follows from the monotonicity of \sgh{}.

Look at the case $j(X) < j(P)$. By induction (on claim (\ref{OCL:gfxmitrcv}))
$$
\valuepre{\gvec[X]}{P}{e} = \valuepost{\gvec{}[X]}{P}{\pnotifyevent{j(P)}{P}} \le \valuepre{\sgh{}}{X}{\pnotifyevent{j(P)}{X}}
$$
and by the monotonicity of \sgh{}
$$
\valuepre{\sgh{}}{X}{\pnotifyevent{j(P)}{X}} \le \valuepre{\sgh{}}{X}{\psendevent{k}} = v
$$
and we are done.

Now look at the case $j(X) > j(P)$. In this case $X$ is a late joining process. Let $E$ be the parent of $X$. Then
$$
\valuepre{\gvec{}[X]}{P}{e} = \valuepost{\gvec{}[E]}{P}{\pnotifyevent{j(X)}{P}} \le \valuepre{\sgh{}}{E}{\pnotifyevent{j(X)}{E}}
$$
and since $X$ inherits its values of \sgh{} from the value of \sgh{} in $E$ we have
$$
\valuepre{\sgh{}}{E}{\pnotifyevent{j(X)}{E}} = \valuepre{\sgh{}}{X}{\pnotifyevent{j(X)}{X}}
$$
And by monotonicity we have
$$
\valuepre{\sgh{}}{X}{\pnotifyevent{j(X)}{X}} \le \valuepre{\sgh{}}{X}{\psendevent{k}} = v
$$
and we are done in this case as well.

To prove claim (\ref{OCL:gf}) let $e'$ be the trigger of the $X$-transaction that queued the packet $k$. Then $e' \prec \psendevent{k} \prec \preceiveevent{k} = e$. It follows from Lemma \ref{KOrderLem} that $\view(e') \le \view(e)$.

By claim (\ref{OCL:gfmon}) that we just proved and by induction we know that
\begin{multline*}
\valuepost{\fvec{}[X]}{P}{e} = \valuepre{\fvec{}[X]}{P}{e} \le \valuepre{\gvec{}[X]}{P}{e} \le \valuepost{\gvec{}[X]}{P}{e} = v \le \valuepost{\sgh{}}{X}{e'} \le \\
{} \le \valuepost{(\cview{} + \vgap{})}{X}{e'} = \view(e') \le \view(e) = \valuepost{(\cview{} + \vgap{})}{P}{e}
\end{multline*}
We have to prove the additional assertions in claim (\ref{OCL:gf}) in the cases where $\valuepost{\vgap{}}{P}{e} = 0$ and $X \notin \valuepost{\cset{}}{P}{e}$.

In the case $\valuepost{\vgap{}}{P}{e} = 0$ we have by induction
$$
\valuepost{\fvec{}[X]}{P}{e} = \valuepre{\fvec{}[X]}{P}{e} = \valuepre{(\cview{} + \vgap{})}{P}{e} = \valuepost{(\cview{} + \vgap{})}{P}{e}
$$
The case $X \notin \valuepost{\cset{}}{P}{e}$ cannot occur here. To show that, we only have to prove that $X \notin \uncont_e$. The rest follows from the fact that $X \in \valuepost{\lset}{P}{e}$ and from claim (\ref{OCL:lset}).

If $j(X) \ge j(P)$ then $X \notin \uncont_e$ by definition. If $j(X) < j(P)$ then the first packet that $X$ sends to $P$ is a donation packet which must precede $k$. Therefore by definition $X \notin \uncont_e$ and so it must be in \cset{} at this point. This takes care of claim (\ref{OCL:gf}).

\item[\ref{ReceiveFlush}] \hfill \\
This procedure updates $\fvec{}[X]$ and then calls the \ref{TryToInstall} service procedure. We will start by ignoring \ref{TryToInstall} and show that the inductive hypothesis still holds before \ref{TryToInstall} is invoked. Later we show that \ref{TryToInstall} preserves all the claims.

This procedure affects claims (\ref{OCL:gf}), (\ref{OCL:gfmon}) and (\ref{OCL:sendgfA}), but for the first two claims only with respect to the sender process $X$. We know that $\valuepre{\lset}{P}{e} = \valuepost{\lset}{P}{e}$ and that $X \in \valuepost{\lset}{P}{e}$ so both claims must be verified for $X$.

Also note that it follows from claim (\ref{OCL:lsetvec}) that for every process in $\valuepost{\lset}{P}{e}$ all the fields in claims (\ref{OCL:gf})  and (\ref{OCL:gfmon})  are well defined.

Let $k = \pkf{v}$.

We start with claim (\ref{OCL:gf}). To prove it, we first show that $\valuepost{\fvec{}[X]}{P}{e} \le \valuepost{\gvec{}[X]}{P}{e}$. This requires a bit of digging. The packet $k$ is queued by $X$ through the execution of the \ref{CheckFlush} procedure. This procedure may or may not queue a ghost packet of the same height, to the same target set, immediately prior to queuing the flush packet. If a ghost packet is queued then it follows from the \AxPackEventIII{} that $P$ processes the ghost packet immediately prior to the current flush packet, resulting in an equality $\valuepost{\fvec{}[X]}{P}{e} = \valuepost{\gvec{}[X]}{P}{e}$. The only difficulty arises if a ghost packet is not queued.

The \ref{CheckFlush} procedure is invoked by $X$ as part of the execution of a notification transaction or an acknowledgement packet processing transaction. An inspection of the pseudo-code easily shows that in the notification case a queuing of a flush packet is always preceded by the queuing of a ghost packet because all three procedures - \ref{RemovalNotification}, \ref{JoinNotification} and \ref{Run} - increment $\vgap{}$ which results, according to claim (\ref{OCL:sendgfA}), in \sgh{} being low.

Let $e'$ be the trigger of the $X$-transaction that queued the packet $k$. We can assume that $e'$ is an acknowledgement packet processing event. Since $e'$ results in the queuing of a flush packet of height $v$ destined to $P$ without the queuing of a ghost packet of the same height we know by pseudo-code inspection that
\begin{align*}
\valuepre{\sfh{}}{X}{e'} & < \valuepre{(\cview{} + \vgap{})}{X}{e'} = v \quad \text{and therefore } v > 0 \\
\valuepre{\sgh{}}{X}{e'} & = \valuepre{(\cview{} + \vgap{})}{X}{e'} = v \\
P & \in \valuepre{\cset}{X}{e'}
\end{align*}
We know from claim (\ref{OCL:height}) that
$$
\view(e') = \valuepost{(\cview{} + \vgap{})}{X}{e'} = \valuepre{(\cview{} + \vgap{})}{X}{e'} = v
$$

Let $f = \pnotifyevent{v}{X}$ be the most recent view change notification preceding $e'$. If the notification $\notify{v}{X}$ is a removal or joining of some process $Q \ne X$ then $f$ is not the first event in $\eventset_X$ and we know by induction from claim (\ref{OCL:sendgfA}) that
$$
\valuepre{\sgh{}}{X}{f} \le \valuepre{(\cview{} + \vgap{})}{X}{f} = v - 1 < \view(e')
$$

If $f = \joinevent{X}$ then $j(X) = v > 0$ and $X$ has a parent $E$. The \ref{Run} procedure does not change the value of \sgh{} until \ref{CheckFlush} is called and therefore by induction
$$
\valuepre{\sgh{}}{X}{f}  = \valuepre{\sgh{}}{E}{\pnotifyevent{v}{E}} \le \valuepre{(\cview{} + \vgap{})}{E}{\pnotifyevent{v}{E}} = v - 1 < \view(e')
$$ 
Taking care to interpret $\valuepre{\fp{var}}{X}{e'}$ correctly for \ref{Run} (see Definition \ref{PrePostDef}).

Therefore there is some trigger event $f \preceq f' \prec e'$ at $X$ such that
$$
\valuepre{\sgh{}}{X}{f'} < \valuepost{\sgh{}}{X}{f'} = v
$$ 
Code inspection shows that the transaction $\trans(f')$ must invoke the \ref{CheckFlush} procedure and must result in the queuing of a ghost packet of height
$v$. If $P$ is in the target set of this multicast then we are done. If it is not, then $P$ must join \cset{} sometime between $\trans(f')$ and $\trans(e')$. Code inspection shows that $P$ can join \cset{} either after a join notification or as a result of sending a donation to $X$.

By the definition of $f$ there cannot be any notification events between $f'$ and $e'$ and therefore process $X$ must receive a donation packet $d$ from $P$ at some point between these two transactions. This in turn causes $X$ to queue a co-donation packet to $P$ that includes its current \sgh{} value.

Since $f' \prec \preceiveevent{d} \prec e'$ and since we know by direct code inspection that \sgh{} is non-decreasing, it follows that the ghost height that $X$ co-donates is equal to $v$. From the \AxPackEventIII{} it follows that the co-donation packet is processed by $P$ before $k$ is processed. As a result
$$
\valuepre{\gvec{}[X]}{P}{e} = v
$$
And we are done showing that $\valuepost{\fvec{}[X]}{P}{e} \le \valuepost{\gvec{}[X]}{P}{e}$. By induction we can conclude (using claim (\ref{OCL:gf})) that
$$
\valuepost{\gvec{}[X]}{P}{e} = \valuepre{\gvec{}[X]}{P}{e} \le \valuepre{(\cview{} + \vgap{})}{P}{e} = \valuepost{(\cview{} + \vgap{})}{P}{e}
$$
Which concludes the proof of the first part of claim (\ref{OCL:gf}).

We have to prove the additional assertions in claim (\ref{OCL:gf}) in the cases where $\valuepost{\vgap{}}{P}{e} = 0$ and $X \notin \valuepost{\cset{}}{P}{e}$. The case $X \notin \valuepost{\cset{}}{P}{e}$ cannot occur here for the same reason as in the case of \ref{ReceiveGhost}.

If $\vgap{} = 0$ then we can assume by induction that
$$
\valuepre{\fvec{}[X]}{P}{e} = \valuepre{(\cview{} + \vgap{})}{P}{e}
$$
We now use claim (\ref{OCL:gfmon}) (which we will prove shortly) to conclude that
 \begin{multline*}
\valuepost{(\cview{} + \vgap{})}{P}{e} = \valuepre{(\cview{} + \vgap{})}{P}{e} = \valuepre{\fvec{}[X]}{P}{e} \le {}  \\
{} \le \valuepost{\fvec{}[X]}{P}{e} \le \valuepost{(\cview{} + \vgap{})}{P}{e}
\end{multline*}
And we are done.

The proof that the \ref{ReceiveFlush} procedure preserves (\ref{OCL:gfmon}) is almost identical to the same proof for the \ref{ReceiveGhost} procedure. The only difference is that in the case where $k$ is the first packet that carries flush information and where $j(X) > j(P)$, the process $X$ inherits its value of \sfh{} from the value of \sgh{} at the parent $E$, and not the value of \sfh{} of the parent. Similarly process $P$ initializes $\fvec{}[X]$ from $\gvec{}[E]$ and not from $\fvec{}[E]$. Therefore we get
$$
\valuepre{\fvec{}[X]}{P}{e} = \valuepost{\gvec{}[E]}{P}{\pnotifyevent{j(X)}{P}} \le \valuepre{\sgh{}}{E}{\pnotifyevent{j(X)}{E}}
$$
and since $X$ inherits its values of \sfh{} from the value of \sgh{} in $E$ we have
$$
\valuepre{\sgh{}}{E}{\pnotifyevent{j(X)}{E}} = \valuepre{\sfh{}}{X}{\pnotifyevent{j(X)}{X}}
$$
And by monotonicity we have
$$
\valuepre{\sfh{}}{X}{\pnotifyevent{j(X)}{X}} \le \valuepre{\sfh{}}{X}{\psendevent{k}} = v
$$
and we are done.

Proof that the procedure preserves (\ref{OCL:sendgfA}) is only required when X = P, i.e. when $k$ is a self packet. In that case we can use the monotonicity of \sfh{} to conclude
$$
\valuepost{\fvec{}[P]}{P}{e} = v = \valuepre{\sfh{}}{P}{\psendevent{k}} \le \valuepost{\sfh{}}{P}{\preceiveevent{k}} = \valuepost{\sfh{}}{P}{e} 
$$

Finally the procedure invokes the \ref{TryToInstall} procedure which preserves all the claims as we show further on.
\item[\ref{StartCluster}] \hfill \\
It is trivial to check that the procedure initializes $P$ to a state that conforms to the requirements of the lemma.

\item[\ref{RemovalNotification}] \hfill \\
Let $R$ be the process that is being removed. Then by definition $\vdeath{R} = \view(e)$.

To see that this procedure preserves (\ref{OCL:height}) notice that its event $e$ is a notification event, meaning that if $e'$ is the preceding event in $\eventset_P$ then $\view(e) = \view(e') + 1$ while the procedure increments \vgap{}, maintaining the equality.

It preserves the first part of claim (\ref{OCL:lset}) because it removes $R$ from \lset{}. $R$ is the only process in \lset{} that no longer meets the condition $j(Q) \le \cview{} + \vgap{} < r(Q)$ as \vgap{} is incremented. No other process is affected because no other process satisfies either $\vdeath{Q} = \vdeath{R}$ or $j(Q) = \vdeath{R}$.

It preserves the second part of claim (\ref{OCL:lset}) because it removes $R$ from \lset{} and \cset{}, while not affecting $\uncont_e$ since $e$ is not a donation packet dequeuing event.

It preserves (\ref{OCL:lsetvec}) because it removes the $R$ coordinate from all the required vectors. It does not affect (\ref{OCL:msetvec}).

It preserves (\ref{OCL:gf}) because it shrinks \lset{}, makes $\vgap{} > 0$ and increments the right hand side of the inequality while not affecting the left hand side.

To see why it preserves claims (\ref{OCL:sendgfA}) and (\ref{OCL:sendgfB}), notice that by incrementing $\vgap{}$ it forces \sgh{} and \sfh{} to be low without touching them. As a result the call to \ref{CheckFlush} at the end has the following effects:
\begin{itemize}
\item \sgh{} becomes high if and only if $\fwset{} = \emptyset{}$.
\item \sfh{} becomes high if and only if $\wset{} = \emptyset$.
\end{itemize}
Therefore \sfh{} becomes high only if \sgh{} becomes high, and therefore by induction $\sfh{} \le \sgh{}$ in all cases.

The procedure increments \vgap{}, making (\ref{OCL:lque}) vacuously true.

It preserves (\ref{OCL:gfmon}) because it does not change the values of $\gvec{}[X]$ and $\fvec{}[X]$ for $X \ne R$.

Finally we have to show that the procedure preserves claim (\ref{OCL:gfxmitrcv}). Let $i = \view(e)$ and let $X \ne R$ be any process that remains live after the view change. If $P$ does not dequeue any ghost (flush), donation or co-donation packet from $X$ between $\pnotifyevent{i-1}{P}$ and $e = \pnotifyevent{i}{P}$ then we get by induction, for either ghost or flush
\begin{align*}
\valuepost{\gvec{}[X]}{P}{\pnotifyevent{i}{P}} & = \valuepost{\gvec{}[X]}{P}{\pnotifyevent{i-1}{P}} \le \valuepre{\sgh}{X}{\pnotifyevent{i-1}{X}} \le \valuepre{\sgh}{X}{\pnotifyevent{i}{X}}  \\
\valuepost{\fvec{}[X]}{P}{\pnotifyevent{i}{P}} & = \valuepost{\fvec{}[X]}{P}{\pnotifyevent{i-1}{P}} \le \valuepre{\sfh}{X}{\pnotifyevent{i-1}{X}} \le \valuepre{\sfh}{X}{\pnotifyevent{i}{X}}
\end{align*}
and we are done. Otherwise, let $e' = \preceiveevent{k}$ be the last dequeuing event of a ghost (flush), donation or co-donation packet from $X$ at $P$ with $\view(e') = i-1$. Let $v$ be the ghost (flush) height carried by the packet. Then we have in each case respectively
\begin{align*}
\valuepost{\gvec{}[X]}{P}{\pnotifyevent{i}{P}} & = \valuepost{\gvec{}[X]}{P}{e'} = v = \valuepre{\sgh{}}{X}{\psendevent{k}}  \\
\valuepost{\fvec{}[X]}{P}{\pnotifyevent{i}{P}} & = \valuepost{\fvec{}[X]}{P}{e'} = v = \valuepre{\sfh{}}{X}{\psendevent{k}}
\end{align*}
and from the \AxProcIII{} it follows that $\psendevent{k} \prec \pnotifyevent{i}{X}$ and so in each case respectively
\begin{align*}
\valuepre{\sgh{}}{X}{\psendevent{k}} & \le \valuepre{\sgh{}}{X}{\pnotifyevent{i}{X}}  \\
\valuepre{\sfh{}}{X}{\psendevent{k}} & \le \valuepre{\sfh{}}{X}{\pnotifyevent{i}{X}}
\end{align*}
and we are done.

\item[\ref{JoinNotification}] \hfill \\
Let $J$ be the joining process and let $E$ be its parent. Then by definition $j(J) = \view(e)$.

This procedure preserves (\ref{OCL:height}), (\ref{OCL:sendgfA}), (\ref{OCL:sendgfB}) and (\ref{OCL:lque}) for the exact same reasons as \ref{RemovalNotification}.

it satisfies the first part of claim (\ref{OCL:lset}) because it adds $J$ to \lset{}. $J$ is the only process that newly meets the condition $j(Q) \le \cview{} + \vgap{} < r(Q)$ as \vgap{} is incremented. No other process is affected, because no other process satisfies either $j(Q) = j(J)$ or $\vdeath{Q} = j(J)$.

It preserves the second part of claim (\ref{OCL:lset}) because it adds $J$ to \lset{} and \cset{}, while not affecting $\uncont_e$ since $e$ is not a donation packet dequeuing event.

It preserves (\ref{OCL:lsetvec}) because it adds a $J$ coordinate to all the required vectors. It does not affect (\ref{OCL:msetvec}).

It preserves (\ref{OCL:gf}) because it makes $\vgap{} > 0$ and increments the right hand side of each existing inequality while not affecting the left hand side. The new inequalities for the $J$ coordinate are inherited from the original ghost inequality for its parent $E$.

It preserves (\ref{OCL:gfmon}) because it does not change the values of $\gvec{}[X]$ and $\fvec{}[X]$ for $X \ne J$.
 
Finally we have to show that the procedure preserves claim (\ref{OCL:gfxmitrcv}). Let $i = \view(e)$ and let $X$ be any process that is live after the view change. If $X \ne J$ then the claim holds following the exact same argument that we used in the \ref{RemovalNotification} case. In the case $X = J$ we have
$$
\valuepost{\fvec{}[J]}{P}{\pnotifyevent{i}{P}} = \valuepost{\gvec{}[J]}{P}{\pnotifyevent{i}{P}} = \valuepost{\gvec{}[E]}{P}{\pnotifyevent{i}{P}}
$$
and since we have already proven the case $X \ne J$ we know that
$$
\valuepost{\gvec{}[E]}{P}{\pnotifyevent{i}{P}} \le \valuepre{\sgh}{E}{\pnotifyevent{i}{E}}
$$
Since $J$ inherits its values \sgh{} and \sfh{} from its parent value of \sgh{} we have
$$
\valuepre{\sgh}{E}{\pnotifyevent{i}{E}} = \valuepre{\sgh}{J}{\pnotifyevent{i}{J}} = \valuepre{\sfh}{J}{\pnotifyevent{i}{J}}
$$
and we are done.

Notice that in the last equation by definition the value $\valuepre{\sfh}{J}{\pnotifyevent{i}{J}}$ reflects the fact that the \ref{Run} procedure copies the initial value of \sgh{} into \sfh{} (see Definition \ref{PrePostDef}).

\item[\ref{Run}] \hfill \\
Let $J$ be the new process and let $E$ be its parent. Then by definition $j(J) = \view(e)$. Let $e_E = \pnotifyevent{j(J)}{E}$ and let $e' \prec e_E$ be the event immediately preceding $e_E$ in $E$. It follows from the \AxNotEventIII{} that $e_E$ exists and therefore $e'$ exists as well since $e_E$ is not the first event in $\eventset_E$.

Process $J$ starts life with the exact same state that its parent had when it dequeued the $\nj{J}{E}$ notification, which is the same state it had at the conclusion of the $\trans(e')$. By induction, $E$ satisfied all the claims at that point. For most claims this means that they are automatically satisfied at that point in $J$ as well, as long as we interpret the expressions $\valuepost{\fp{var}}{J}{e'}$ to simply mean the initial value of \fp{var}, ignoring the fact that the value of $e'$ there is undefined in $J$. However there are a number of exceptions.

claim (\ref{OCL:height}) becomes ill-defined because $e'$ is not defined in $J$.
The second part of claim (\ref{OCL:lset}) is not satisfied because it depends on the value of $\uncont_{e'}$ which is ill-defined at $J$, and any rate is not inherited from $E$.
Claim (\ref{OCL:gfmon}) becomes ill-defined because $\valuepre{\lset}{J}{e'}$ has no clear interpretation in $J$. Similarly, claim (\ref{OCL:gfxmitrcv}) depends on the meaning of $e'$ and so is hard to interpret in $J$. Claim (\ref{OCL:sendgfA}) references the self value $\fvec[P]$ where $P$ is the local process. Whenever we need rely on any of these ill-defined inductive claims as we go, we will present an explicit calculation that relies directly on the well-defined $E$ version of these claims.

The \ref{Run} procedure preserves (\ref{OCL:height}) because
\begin{multline*}
\view(e) = j(J) = \view(e') + 1 = \valuepost{(\cview + \vgap)}{E}{e'} + 1 =  \\
= \valuepost{(\cview + \vgap)}{J}{e'} + 1 = \valuepost{(\cview + \vgap)}{J}{e}
\end{multline*}
Where the last equality follows because the \ref{Run} procedure increments \vgap{}.

claims (\ref{OCL:sendgfB}), (\ref{OCL:lque}) and most of (\ref{OCL:sendgfA}) are preserved by \ref{Run} for the exact same reasons as in the \ref{RemovalNotification} case. The only missing piece is that the inequality
$$
\valuepost{\fvec{}[J]}{J}{e} \le \valuepost{\sfh{}}{J}{e}
$$
follows because the self flush height $\fvec{}[J]$ is initialized to be equal to $\valuepost{\sgh{}}{J}{e'}$. The same value is used to initialize \sfh{}, and thereafter \sfh{} can only increase.

\ref{Run} satisfies the first part of claim (\ref{OCL:lset}) for the same reason that \ref{JoinNotification} does.

\ref{Run} satisfies the second part of claim (\ref{OCL:lset}) because at the outset $\uncont_e$ contains every process $X$ whose join view is lower than $j(J)$. This includes every member of \lset{} except $J$ itself. Since the \ref{Run} procedure adds $J$ to \lset{} while setting $\cset{} = \{ J \}$, this makes the equations true.

It preserves (\ref{OCL:lsetvec}) because it adds a $J$ coordinate to all the required vectors. It does not affect (\ref{OCL:msetvec})

It preserves (\ref{OCL:gf}) because it makes $\vgap{} > 0$ and increments the right hand side of each existing inequality while not affecting the left hand side. The new inequalities for the $J$ coordinate hold because the values of $\fvec{}[J]$ and $\gvec{}[J]$ are both initialized to $\valuepost{\sgh{}}{J}{e'} \le \valuepost{\sgh{}}{J}{e}$ and because we have shown that claim (\ref{OCL:sendgfA}) is preserved.

It preserves (\ref{OCL:gfmon}) because the \ref{Run} procedure does not change the initial values of the $\gvec{}[]$ and $\fvec{}[]$ vectors.

Finally we have to show that the procedure preserves claim (\ref{OCL:gfxmitrcv}). Let $i = \view(e)$ and let $X$ be any process that is live after the view change. If $X \ne J$ then the value of $\gvec{}[X]$ and $\fvec{}[X]$ is inherited from $E$ without change
\begin{align*}
\valuepost{\gvec{}[X]}{J}{\pnotifyevent{i}{J}} & = \valuepre{\gvec{}[X]}{E}{\pnotifyevent{i}{E}}  \\
\valuepost{\fvec{}[X]}{J}{\pnotifyevent{i}{J}} & = \valuepre{\fvec{}[X]}{E}{\pnotifyevent{i}{E}}
\end{align*}
and since we have already proved all the claims  for $E$ (the \ref{JoinNotification} case) we know that
\begin{align*}
\valuepre{\gvec{}[X]}{E}{\pnotifyevent{i}{E}} & \le \valuepost{\gvec{}[X]}{E}{\pnotifyevent{i}{E}} \le \valuepre{\sgh{}}{X}{\pnotifyevent{i}{X}} \\
\valuepre{\fvec{}[X]}{E}{\pnotifyevent{i}{E}} & \le \valuepost{\fvec{}[X]}{E}{\pnotifyevent{i}{E}} \le \valuepost{\sfh{}}{X}{\pnotifyevent{i}{X}}
\end{align*}
and we are done. The case $X = J$ follows because the \ref{Run} procedure sets
$$
\valuepost{\gvec{}[J]}{J}{\pnotifyevent{i}{J}} = \valuepost{\fvec{}[J]}{J}{\pnotifyevent{i}{J}} = \valuepre{\sgh{}}{J}{\pnotifyevent{i}{J}} = \valuepre{\sfh{}}{J}{\pnotifyevent{i}{J}}
$$
\item[\ref{ReceiveDonation}] \hfill \\
Let $S$ be the sender of the donation packet. Since $e$ is a donation packet dequeuing event in this case, $S$ falls out of $\uncont_e$. This is compensated for by adding $S$ to \cset{}. This preserves claim (\ref{OCL:lset}). Then the procedure makes a sequence of invocations of the \ref{ReceiveMessage} and \ref{ReceiveAck} procedures which preserve all the claims as we have already shown.

Finally the procedure updates $\gvec{}[S]$ and $\fvec{}[S]$ from the donated values of \sgh{} and \sfh{} respectively. Remember that the donation packet is sent by the \ref{JoinNotification} procedure as part of the $j(P)$ view change notification processing. Therefore
\begin{align*}
\valuepost{\gvec{}[S]}{P}{e} & = \valuepre{\sgh{}}{S}{\pnotifyevent{j(P)}{S}}  \\
\valuepost{\fvec{}[S]}{P}{e} & = \valuepre{\sfh{}}{S}{\pnotifyevent{j(P)}{S}}
\end{align*}
It follows from claim \ref{OCL:sendgfA} of the inductive hypothesis that
\begin{multline*}
\valuepre{\sfh{}}{S}{\pnotifyevent{j(P)}{S}} \le \valuepre{\sgh{}}{S}{\pnotifyevent{j(P)}{S}} \le \valuepre{(\cview{} + \vgap{})}{S}{\pnotifyevent{j(P)}{S}} = {}  \\
{} = j(P) - 1 < \view(e) = \valuepost{(\cview{} + \vgap{})}{P}{e}
\end{multline*}
This proves most of claim \ref{OCL:gf}. Since the inequalities are strict we have to show that $\vgap{} > 0$. It follows from claim (\ref{OCL:gfmon}) which we will prove shorty that $\fvec{}[S]$ is non-decreasing, and therefore the inequality was strict at the start of the transaction. Therefore by induction $\vgap{} > 0$.

Claim (\ref{OCL:gfmon}) follows by the exact same reasoning that we used for the \ref{ReceiveGhost} and \ref{ReceiveFlush} procedures, while claim (\ref{OCL:gfxmitrcv}) is not relevant at a non-notification transaction.
\item[\ref{ReceiveCoDonation}] \hfill \\
Let $S$ be the sender of the co-donation packet, and let $e'$ be the trigger of the donation transaction at which $S$ queues the co-donation packet.

The \ref{ReceiveCoDonation} procedure makes a sequence of calls that preserve all the claims and then updates $\gvec{}[S]$ and $\fvec{}[S]$ from the co-donated values. Then

\begin{align*}
\valuepost{\gvec{}[S]}{P}{e} & = \valuepre{\sgh{}}{S}{e'}  \\
\valuepost{\fvec{}[S]}{P}{e} & = \valuepre{\sfh{}}{S}{e'}
\end{align*}

From claim \ref{OCL:sendgfA} of the inductive hypothesis and the \AxProcIII{} we know that

\begin{multline*}
\valuepre{\sfh{}}{S}{e'} \le \valuepre{\sgh{}}{S}{e'} \le \valuepre{(\cview{} + \vgap{})}{S}{e'} \le {}  \\
{} \le \view(e') \le \view(e) = \valuepost{(\cview{} + \vgap{})}{P}{e}
\end{multline*}
which proves most of claim (\ref{OCL:gf}). Process $S$ must be contacted at this point, so we do not have to show that the inequalities are strict. But if one of the inequalities happens to be strict, then we have to show that $\vgap{} > 0$. This follows from the same argument as in the case of \ref{ReceiveDonation}.

Claim (\ref{OCL:gfmon}) follows by the exact same reasoning that we used for the \ref{ReceiveGhost} and \ref{ReceiveFlush} procedures, while claim (\ref{OCL:gfxmitrcv}) is not relevant at a non-notification transaction.

Finally the procedure calls \ref{TryToInstall} which also preserves all the claims as we will see. 
\item[\ref{TryToInstall}] \hfill \\
This is a service procedure that is called by other procedures. However it preserves (\ref{OCL:height}) and (\ref{OCL:gf}) because it always increments \cview{} and decrements \vgap{} together. Also, the procedure makes no changes to $\vgap{}$ unless all the $\fvec[]$ values are high. It preserves (\ref{OCL:msetvec}) because it resets $\vtime{}$ whenever it changes $\mset{}$.

It preserves (\ref{OCL:sendgfA}) because it does not touch either \sgh{}, \sfh{} or the self flush height $\fvec{}[P]$ while keeping $\cview{} + \vgap{}$ fixed (this includes possible invocations of the \ref{BroadcastMessage} and \ref{ScanCall} procedures).

It preserves (\ref{OCL:sendgfB}) because it either results in $\vgap{} = 0$ which makes (\ref{OCL:sendgfB}) vacuously true, or else it fails to install views in which case it does not change any variables and therefore preserves (\ref{OCL:sendgfB}) by induction.

It preserves (\ref{OCL:lque}) because it either results in $\vgap{} > 0$ which makes (\ref{OCL:lque}) vacuously true, or else it starts out with $\vgap{} = 0$ in which case (\ref{OCL:lque}) is true by induction, or else it starts out with $\vgap{} > 0$ and ends with $\vgap{} = 0$ in which case it empties out \lque{}, making (\ref{OCL:lque}) true.

\ref{TryToInstall} does not affect claims (\ref{OCL:lset}) and (\ref{OCL:lsetvec}).
\item[\ref{ScanCall}] \hfill \\
This is a service procedure that is called by other procedures. It does not affect any of the relevant variables so it has no effect on any of the claims.
\end{description}
\end{proof}

\begin{cor}
\label{ContactSetCor}
Let $P$ be a process and let $e \in \eventset_P$ be any trigger event. Suppose that
$$
\valuepre{\lset{}}{P}{e} \ne \valuepre{\cset{}}{P}{e}
$$
Then there is a process $X \in \lset{}$ such that $\fvec{}[X] < j(P)$.
\end{cor}

\begin{proof}
It follows from Lemma \ref{OmnibusCBCASTLem}(\ref{OCL:lset}) that there is a live process $X \in \lset{} \,\cap\, \uncont_e$, and it follows from the definition of $\uncont_e$ that $j(X) < j(P)$. Therefore $P$ is not an original process and it has a parent $E$. $P$ starts life with a state identical to the state of $E$ at $\pnotifyevent{j(P)}{E}$.

Process $X$ can only be in $\lset{}$ if it is there originally or if $P$ receives a join notification for $X$. Since the latter does not happen in this case we must have $X \in \valuepre{\lset{}}{E}{\pnotifyevent{j(P)}{E}}$ and by Lemma \ref{OmnibusCBCASTLem}(\ref{OCL:lsetvec}) $\valuepre{\fvec{}[X]}{E}{\pnotifyevent{j(P)}{E}}$ exists. By Lemma \ref{OmnibusCBCASTLem}(\ref{OCL:gfmon}, \ref{OCL:gfxmitrcv}, \ref{OCL:sendgfA} and \ref{OCL:height})
\begin{multline*}
\valuepre{\fvec{}[X]}{E}{\pnotifyevent{j(P)}{E}} \le \valuepost{\fvec{}[X]}{E}{\pnotifyevent{j(P)}{E}} \le  \\
\le \valuepre{\sfh{}}{X}{\pnotifyevent{j(P)}{X}} \le \valuepre{(\cview{} + \vgap{})}{X}{\pnotifyevent{j(P)}{X}} < j(P)
\end{multline*}
Therefore the initial value of $\fvec{}[X]$ at $P$ is smaller than $j(P)$.

Scanning the pseudo-code reveals that this initial value of $\fvec{}[X]$ does not change in $P$ unless $P$ receives a notification of the removal of $X$ (which does not happen in this case) or if $P$ processes a co-donation packet from $X$ (which also does not happen, because co-donations are always sent from a late joining process to an existing group member), or if $P$ processes a flush or a donation packet from $X$. The donation packet to $P$ is the first packet that $X$ queues to $\channel{X}{P}$ and therefore the first packet from $X$ that $P$ processes. Since $e$ occurs before the donation packet is processed, the initial value of $\fvec{}[X]$ is still extant and we are done.
\end{proof}

\begin{lem}
\label{OrigPCountLem}
Let $P$ and $Q$ be two original processes. Suppose that $P$ sends a message packet $\pkm{\fp{msg}}$ to $Q$. When $P$ queues the packet, it increments its \mpkout{} and places \fp{msg} in its \wset{} together with an \fp{index} value equal to its updated value of \mpkout{} (see the \ref{BroadcastMessage} procedure for original message broadcasts and the \ref{RemovalNotification} procedure for forwarded broadcasts). If $Q$ processes the message it increments the value of $\mpkin{}[P]$. Then:
\begin{enumerate}
\item Before $Q$ processes the message, its value of $\mpkin{}[P]$ is lower than $\fp{index}$.
\item After $Q$ processes the message, its value of $\mpkin{}[P]$ becomes equal to $\fp{index}$.
\end{enumerate}
Moreover, if $P = Q$ then the conclusion holds without the requirement that $P$ be original.
\end{lem}

\begin{proof}
Both $P$ and $Q$ are original. Therefore $Q$ is a member of $\cset{}(P)$ from the start (see the \ref{StartCluster} procedure). Therefore every message that was queued by $P$ prior to \fp{msg} had $Q$ in its recipient list and since channels are FIFO, all of these messages are processed by $Q$ before \fp{msg} is processed. The \mpkout{} variable in $P$ is incremented every time a message packet is multicast by $P$ (at the $.b$ or $.f$ component, according as the message is original or forwarded), and similarly $\mpkin{}[P]$ is incremented by $Q$ every time a message from $P$ is processed (at the $.b$ or $.f$ component, according as the message is original or forwarded). Initially both variables are equal to zero (at both components) at both $P$ and $Q$ (see the \ref{StartCluster} procedure). Therefore they remain at lockstep as claimed.

Almost the same argument works when $P = Q$ and $P$ is not an original process. We just have to verify two things. One, that $P \in \cset{}(P)$ from the start, as one can verify by looking at the \ref{Run} procedure. Two, that initially both \mpkout{} and $\mpkin{}[P]$ are equal at $P$. This can also be verified by looking at the \ref{Run} procedure.
\end{proof}

\begin{lem}
\label{FCountLem}
Let $P$ and $Q$ be two original processes. Suppose that $P$ sends a flush packet $k = \pkf{v}$ to $Q$. When $P$ queues the packet, it increments its \sfh{} (see the \ref{CheckFlush} procedure). If $Q$ processes the flush it increments the value of $\fvec{}[P]$. Then:
$$
\valuepre{\fvec{}[P]}{Q}{\preceiveevent{k}} < \valuepre{\sfh{}}{P}{\psendevent{k}} = \valuepost{\fvec{}[P]}{Q}{\preceiveevent{k}} = v
$$
Moreover, if $P = Q$ then the conclusion holds without the requirement that $P$ be original.

The exact same claim is true if flush is replaced by ghost throughout.
\end{lem}

\begin{proof}
We prove the lemma for the flush case. The ghost case is identical using the appropriate substitutions. Both $P$ and $Q$ are original. Therefore $Q$ is a member of $\cset{}(P)$ from the start (see the \ref{StartCluster} procedure). Therefore every flush that is queued by $P$ prior to $\pkf{v}$ had $Q$ in its recipient list and since channels are FIFO, all of these flushes are processed by $Q$ before $\pkf{v}$ is processed. The \sfh{} variable in $P$ is increased to be equal to $w$ every time a flush of height $w$ is broadcast by $P$ (see the \ref{CheckFlush} procedure), and therefore $\fvec{}[P]$ is increased by $Q$ every time a flush from $P$ is processed (see the \ref{ReceiveFlush} procedure). Initially both variables are equal to zero  at both $P$ and $Q$ (see the \ref{StartCluster} procedure). Therefore they remain at lockstep as claimed.

Almost the same argument works when $P = Q$ and $P$ is not an original process. We just have to verify two things. One, that $P \in \cset{}(P)$ from the start, as one can verify by looking at the \ref{Run} procedure. Two, that initially both \sfh{} and $\fvec{}[P]$ are equal at $P$. This can also be verified by looking at the \ref{Run} procedure.
\end{proof}

\subsection{Side Effects of \cbcast{} Triggers}
As we have seen, each trigger event - a notification event, a packet processing event, or a message broadcast request event - causes a \cbcast{} callback to be invoked. Each invocation can cause zero or more packets to be queued on various channels - in other words the invocation causes side effects (see \ref{ModelOverviewSS}). We are now going to characterize the side effects of each type of trigger in detail.

\subsubsection{Side effects of message broadcast request events}
\label{FirstSideEffectSSS}
Message broadcast requests are processed through the \ref{BroadcastMessage} procedure. The following lemma details the possible side effects of this procedure call.

\begin{lem}
\label{AppEventSELem}
An invocation of the \ref{BroadcastMessage} procedure results in exactly one of the following outcomes:
\begin{itemize}
\item No additional packet queuing, if $\vgap > 0$.
\item A message packet multicast if $\vgap = 0$.
\end{itemize}
\end{lem}

\begin{proof}
Obvious from the code of \ref{BroadcastMessage}.
\end{proof}

\subsubsection{Side effects of view change notifications}
A view change notification from \gms{} is processed either through the \ref{JoinNotification} or the \ref{RemovalNotification} procedures, depending on the type of view change. In addition, a joining process starts out life with an exact replica of the state of its parent, and immediately invokes the \ref{Run} procedure. The following lemma details the possible effects of these calls.

\begin{lem}
\label{ViewChangeSELem}
\begin{enumerate}
\item An invocation of the \ref{JoinNotification} procedure results in the queuing of a donation packet, followed by exactly one of the following outcomes:
\begin{itemize}
\item No additional packet queuing, if \fwset{} is not empty.
\item A ghost packet multicast if \fwset{} is empty and \bwset{} is not empty.
\item A ghost packet multicast followed by a flush packet multicast, if both \bwset{} and \fwset{} are empty.
\end{itemize}
\item An invocation of the \ref{RemovalNotification} procedure when process $P$ is removed, results in exactly one of the following outcomes:
\begin{itemize}
\item A sequence of message broadcasts, one per message in $\fque{}[P]$, if $\fque{}[P]$ is not empty.
\item No additional packet queuing if $\fque{}[P]$ is empty, and \fwset{} is non-empty when \ref{CheckFlush} is called.
\item A ghost packet multicast if \fwset{} is empty and \bwset{} is not empty when \ref{CheckFlush} is called.
\item A ghost packet multicast followed by a flush packet multicast if both \bwset{} and \fwset{} are empty when \ref{CheckFlush} is called.
\end{itemize}
\item An invocation of the \ref{Run} procedure results in exactly one of the following outcomes:
\begin{itemize}
\item No additional packet queuing if \fwset{} is non-empty.
\item A ghost packet multicast followed by a flush packet multicast if \fwset{} is empty.
\end{itemize}
\end{enumerate}
\end{lem}

\begin{proof}
All these outcomes are easy to verify by tracing the code path in the respective calls. In the case of \ref{RemovalNotification} it is important to notice that if $\fque{}[P]$ is not empty, then the forwarded messages are placed in \fwset{}, which as a result is not empty when \ref{CheckFlush} is called.
\end{proof}

\subsubsection{Side effects of message and acknowledgement packet receipts}
A message receipt always results in the sending of an acknowledgement packet in response. An acknowledgement packet receipt causes a stabilization of a message with respect to the process that sent the packet. If this stabilization causes either \fwset{} or \bwset{} to empty out, it can cause the multicasting of ghost and flush packets.

\begin{lem}
\label{MessageAckSELem}
\begin{enumerate}
\item An invocation of the \ref{ReceiveMessage} procedure results in the queuing of an acknowledgement packet targeted at the sender of the message.
\item An invocation of the \ref{ReceiveAck} procedure results in exactly one of the following outcomes:
\begin{itemize}
\item No additional packet queuing if $\vgap{} = 0$.
\item No additional packet queuing if the acknowledgement does not cause either \fwset{} or \bwset{} to empty out, even if one or both were already empty.
\item No additional packet queuing if the acknowledgement causes \bwset{} to empty out and \fwset{} is non-empty.
\item A ghost packet multicast if $\vgap{} > 0$ and the acknowledgement causes \fwset{} to empty out and \bwset{} is non-empty.
\item A flush packet multicast if $\vgap{} > 0$ and the acknowledgement causes \bwset{} to empty out and \fwset{} is empty.
\item A ghost packet multicast followed by a flush packet multicast, if $\vgap{} > 0$ and the acknowledgement causes \fwset{} to empty out and \bwset{} is empty.
\end{itemize}
\end{enumerate}
\end{lem}

\begin{proof}
This follows directly from direct observation and from Lemma \ref{OmnibusCBCASTLem}(\ref{OCL:gf} and \ref{OCL:sendgfB}).
\end{proof}

\subsubsection{Side effects of ghost and flush packet receipts}
A ghost packet receipt does not cause any other packets to be sent. A flush packet, however, may cause one or more views to be installed. If that happens then one or more original messages from \lque{} may be broadcast.

\begin{lem}
\label{GhostFlushSELem}
\begin{enumerate}
\item An invocation of the \ref{ReceiveGhost} procedure does not result in additional packet queuing.
\item An invocation of the \ref{ReceiveFlush} procedure results in exactly one of the following outcomes:
\begin{itemize}
\item No additional packet queuing if no views are installed or if \lque{} is empty.
\item One or more message packet multicasts if one or more pending views are installed and \lque{} is not empty.
\end{itemize}
\end{enumerate}
\end{lem}

\begin{proof}
This follows from Lemma \ref{OmnibusCBCASTLem}(\ref{OCL:lque}) and from direct inspection of the code of the relevant procedures.
\end{proof}

\subsubsection{Side effects of donation and co-donation packet receipts}
\label{LastSideEffectSSS}
The side effects of invoking the \ref{ReceiveDonation} procedure are pretty straightforward - first a co-donation packet is sent and then a sequence of \ref{ReceiveMessage} and \ref{ReceiveAck} invocations are performed, each with its side effects that have already been characterized. The side effects of invoking the \ref{ReceiveCoDonation} procedure are a bit more subtle, because this procedure has an additional call to \ref{TryToInstall} at the end and there is an interplay between its side effects and the side effects of the rest of the procedure.

\begin{lem}
\label{DonationSELem}
\begin{enumerate}
\item An invocation of the \ref{ReceiveDonation} procedure results in the queuing of a co-donation packet back to the sender of the donation packet, followed by a sequence of side effects for each of the \ref{ReceiveMessage} and \ref{ReceiveAck} invocations.
\item An invocation of the \ref{ReceiveCoDonation} procedure results in exactly one of the following outcomes:
\begin{itemize}
\item Zero or more invocations of \ref{ReceiveMessage} or \ref{ReceiveAck} occur with their side effects, and \ref{TryToInstall} fails to install a new view and has no side effects.
\item No invocations of either \ref{ReceiveMessage} or \ref{ReceiveAck} occur, and \ref{TryToInstall} succeeds in installing all the views. As a result zero or more original messages from \lque{} are broadcast.
\end{itemize}
\end{enumerate}
\end{lem}

\begin{proof}
The donation case is self evident. In the co-donation case, if $\vgap{} = 0$  (we will see later that this case does not actually happen) then \ref{TryToInstall} does not install any views and it follows from Lemma \ref{OmnibusCBCASTLem}(\ref{OCL:lque}) that \lque{} is empty and so \ref{TryToInstall} has no side effects.

Therefore the only non-trivial case has to do with a co-donation that starts with $\vgap{} > 0$ and results in a successful new view installation. Suppose that the sender of the co-donation is $G$ and the receiver is $P$. Then a successful installation requires, at $P$:
\begin{align*}
\fvec{}[G] & = \cview{} + \vgap{} \\
\fvec{}[P] & = \cview{} + \vgap{}
\end{align*}

From Lemma \ref{OmnibusCBCASTLem}(\ref{OCL:sendgfA}) we know that at $P$
$$
\fvec{}[P] \le \sfh{} \le \cview{} + \vgap{}
$$
Therefore $\fvec{}[P] = \sfh{} = \cview{} + \vgap{}$ and since $\vgap{} > 0$ it follows from the same lemma (part \ref{OCL:sendgfB}) that \wset{} is empty and as a result $\se{UNT}_p$ is empty.

Looking at the \ref{ReceiveCoDonation} procedure one sees that $P$ updates its value of $\fvec{}[G]$ just before invoking \ref{TryToInstall}, setting it to be equal to the co-donated value of \sfh{} in $G$. Therefore, at the moment that $G$ sends the co-donation, it has $\sfh{}(G) = \cview{}(P) + \vgap{}(P)$. In other words, if we denote by $d$ the donation packet that $P$ sends to $G$ and by $c$ the co-donation packet that $G$ sends to $P$ then
$$
\valuepre{\sfh{}}{G}{\preceiveevent{d}} = \valuepre{(\cview{} + \vgap{})}{P}{\preceiveevent{c}}
$$
The \AxProcIII{} and Lemma \ref{OmnibusCBCASTLem}(\ref{OCL:height}) imply that the co-donation packet cannot be processed when $P$ has an overall view height that is lower than that of $G$ and therefore
$$
\valuepre{(\cview{} + \vgap{})}{P}{\preceiveevent{c}} \ge \valuepre{(\cview{} + \vgap{})}{G}{\preceiveevent{d}}
$$
therefore
$$
\valuepre{\sfh{}}{G}{\preceiveevent{d}} \ge \valuepre{(\cview{} + \vgap{})}{G}{\preceiveevent{d}}
$$
From Lemma \ref{OmnibusCBCASTLem}(\ref{OCL:sendgfA}) it follows that
$$
\valuepre{\sfh{}}{G}{\preceiveevent{d}} \le \valuepre{(\cview{} + \vgap{})}{G}{\preceiveevent{d}}
$$
and so we can conclude that
$$
\valuepre{\sfh{}}{G}{\preceiveevent{d}} = \valuepre{(\cview{} + \vgap{})}{G}{\preceiveevent{d}}
$$
From Lemma \ref{OmnibusCBCASTLem}(\ref{OCL:lset} and \ref{OCL:height}) it follows that $P \notin \valuepre{\cset{}}{G}{\preceiveevent{d}}$. From Corollary \ref{ContactSetCor} it follows that $\valuepre{\fvec{}[P]}{G}{\preceiveevent{d}} < \valuepre{(\cview{} + \vgap{})}{G}{\preceiveevent{d}}$. From Lemma \ref{OmnibusCBCASTLem}(\ref{OCL:gf}) it follows that $\valuepre{\vgap{}}{G}{\preceiveevent{d}} > 0$.

Now we can use Lemma \ref{OmnibusCBCASTLem}(\ref{OCL:sendgfB}) to conclude that $\valuepre{\wset{}}{G}{\preceiveevent{d}}$ is empty. Therefore $P$ receives an empty \wset{} from $G$ as part of its co-donation, and therefore $\se{UNT}_g$ is empty and we are done.
\end{proof}

\subsection{\cbcast{} is Vacuum Convergent}
In our analysis of \cbcast{} we want to take advantage of the main finding of Section \ref{UnderlyingModelSec}, namely the \FaultThm{} (Theorem \ref{FaultThm}), and limit our attention to transactional histories. In order to do that we have to prove that \cbcast{} is vacuum convergent (see Definitions \ref{VacConvDef}).
\begin{thm}
The \cbcast{} protocol is vacuum convergent.
\end{thm}

\begin{proof}
Let $P$ be a halting process in a \cbcast{} based conforming history. Look at step (\ref{VacProcessN}) of the vacuum loop (see Definition \ref{VacuumDef}). By the \CAxView{} process $P$ has a finite view interval and therefore this step can only increment $v$ a finite number of times. Afterwards either the loop exits or step (\ref{VacProcessN}) is not executed again. Assume that the loop never exits. Once $v$ stabilizes, steps (\ref{VacProcessA}) and (\ref{VacProcessP}) are executed once and are not executed again. Therefore after a finite number of steps the vacuum loop degenerates to repeated executions of step (\ref{VacProcessSelf}) which consist of: processing packets on the self channel; queuing side-effect packets to their respective channels, including the self-channel; sending and receiving the queued packets on downstream channels, including the self-channel; processing the newly received packets on the self-channel; etc. Since donation and co-donation packets are never queued to the self channel, at this point such packets are not processed by the vacuum loop anymore.

From Lemma \ref{OmnibusCBCASTLem} we know that the values of \sgh{} and \sfh{} in $P$ cannot rise above $v$. Since the \ref{CheckFlush} procedure only generates ghost and flush packets when the ghost and flush height rise, the vacuum loop can only generate a finite number of such packets. Therefore after a finite number of steps no such packets are generated anymore. Therefore after some more time passes the vacuum loop processes the last of these packets and does not process any ghost or flush packets afterwards.

Message packets are generated as a result of the processing of
\begin{itemize}
\item a message broadcast request (when $\vgap{} = 0$)
\item a flush packet (when the flush causes view installations and \lque{} is not empty)
\item a co-donation packet (ditto)
\item a removal notification (when \fque{} is not empty)
\end{itemize}
Since at this point the vacuum loop does not process any additional items of these types, it only accumulates a finite number of message packets and as a result after some point it processes the last message packet and does not process any more such packets afterwards.

Acknowledgement packets are generated as a result of the processing of message packets, donation packets and co-donation packets. Therefore the vacuum loop processes a finite number of those packets as well.

So at some point the vacuum loop runs out of packets to process and is forced to fall through to step (\ref{VacProcessN}), contrary to our assumption.
\end{proof}

\section{The History Reduction Mapping}
\label{HistoryReduxS}
\subsection{Introduction}
In this section we demonstrate the rather surprising fact that any transactional history of the \cbcast{} protocol that contains at least one join notification can be reduced to an alternate conforming history of \cbcast{} that performs the same computation and where the first join notification is replaced with a removal notification. This allows us to restrict our analysis to histories that contain no join notifications.

We will construct the reduced history explicitly. We will start with the original history, and make careful changes to it. The construction will revolve around the first joining process $G$, its parent $D$, and the join view of $G$, which we call the {\em critical view}. The idea is to declare that $G$ is an original group member (member since view zero) and that it has a doppelg\"{a}nger $\minusG$ that is also a member of view zero. Then instead of having $G$ join the group, we have $\minusG$ leave the group. We need to accomplish this while not violating the \cbcast{} protocol, and without affecting the user application. In fact if we have any hope of success, the user application must be completely oblivious to the change. To create the reduced history we will have to solve two problems. First we will have to create a whole new history for $\plusminusG$ during the pre-critical interval, namely the interval that precedes the critical view change. Then we will have to resolve the race conditions that occur as a result of {\em untimely} packets, namely packets that are sent before the critical view change notification and arrive afterwards.

We will solve the first problem by using $D$ as the pre-critical template for $\plusminusG$. This means that every reactive move by $D$, like receiving or acknowledging a message, will be copied by $\plusminusG$ verbatim. However we will not copy proactive moves by $D$, namely original message broadcasts that are initiated by the application at $D$. During the pre-critical period, $\plusminusG$ will be passive participants. We will ensure that the \ulp{} thread does not run there (and therefore does not generate message queuing requests) by artificially delaying the launch of the \ulp{} thread until after the critical moment. This is possible to arrange because the launch of the \ulp{} thread in \ref{StartCluster} is asynchronous.

The second problem will be solved with the help of the donation/co-donation protocol. This protocol is carefully tailored to provide precise compensation for critical boundary race conditions. In the reduced history the critical donation and co-donation packets will be eliminated. All the simulated packet processing that occurs during the \ref{ReceiveDonation} and \ref{ReceiveCoDonation} procedure calls will be replaced with the receipt and processing of actual, newly-minted untimely packets.

The new packets and events will have to be added and timed in a very precise manner so that we neither violate \cbcast{} nor affect the user application. For example, pre-critical flush packets from $\plusminusG$ will have to arrive slightly earlier than their counterparts from $D$, in order to prevent them from causing the receiving process to install a new view. But pre-critical forwarded messages from $\plusminusG$ will have to arrive slightly {\em later} than their counterparts from $D$, in order to guarantee that they are redundant, and therefore ignored by the receiving process. We do not want $\plusminusG$ to rock the boat prematurely.

To manage this careful surgery we need quite a bit of control. We will gain this control by creating a 4-coordinate label for each event. The label will describe which transaction the event belongs to and whether it is the transaction trigger or a side effect. It will describe whether the event was moved slightly forward or backward to insure redundancy. For events that occur during a donation or co-donation transaction, the label will also describe which simulated sub-transaction the event is related to. The most important fact about the labels is that they all come from a single partially ordered {\em label space} that is common to the original and reduced histories. The use of a single label space for both histories will allow us to relate the order in both, and ultimately to use induction over the shared label space to prove that both histories track each other closely and ultimately converge.

Throughout this part of the paper $H$ is a fixed transactional history that includes a first join view $\vcrit$. The process that joins at view $\vcrit$ is denoted $G$ throughout, and the parent process of $G$ is denoted by $D$. We will sometimes refer to $H$ as the {\em original history}. We will sometimes refer to the synthesized history $H^r$ as the {\em reduced history}.

\subsection{Preliminaries}

\subsubsection{Transactions}
\label{TransactionsSSS}
A process executing the \cbcast{} protocol invokes a procedure in reaction to each triggering event. A view change notification event triggers an invocation of the \ref{StartCluster}, \ref{Run}, \ref{JoinNotification} or \ref{RemovalNotification} procedures. A packet processing event causes the invocation of the \ref{ReceiveMessage}, \ref{ReceiveAck}, \ref{ReceiveGhost}, \ref{ReceiveFlush}, \ref{ReceiveDonation} or \ref{ReceiveCoDonation} procedures. A message broadcast request event causes an invocation of the \ref{BroadcastMessage} procedure. The procedure call in turn causes zero or more side effects in the form of packet queuing events. This sequence of events, starting with a trigger and continuing with a finite number (possibly zero) of side effects is a {\em transaction}. Because each process $P$ runs its \cbcast{} procedures in a critical section, the events at each transaction occur as a contiguous sequence in $\eventset_P$ and therefore transactions can be read out of the history $H$ directly. Compare this with the model-based definition of the notion of transaction in Section \ref{ModelOverviewSS}.

\subsubsection{A clean event order}
The partial event order $\prec$ in $H$ is a bit too weak for our labeling needs. Our analysis in \ref{AntiChainSSS} shows that this ordering is compatible with event views: a high-view event cannot precede a low-view event. But it does not have to succeed it either. Also, given the fact that all the view change notification events of a single view can be viewed as occurring at the same time, it would be convenient to collapse them into a single event. This is exactly what we will do now.

\begin{defn}
\label{CleanEventDef}
The {\bf Clean Event Set} $(\dot{\eventset}, \dot{\prec})$ is the partially ordered set obtained from the set $\eventset$ of events in $H$ by collapsing all the notification events of each view $i$ into a single element $\ell_i$, with the induced partial order. Formally:
\begin{align*}
\dot{\eventset} & = \lbrace \ell_i \rbrace_{0 \le i < \numview+1} \,\bigcup\,  \left( \bigcup_{0 \le i < \numview+1} (K_i \setminus G_i) \right) \\
\dot{e} \dot{\prec} \dot{f} \; & \Longleftrightarrow \;
\begin{cases}
\dot{e} = \ell_i \text{ and } \dot{f} = \ell_j & \text{ and } i < j \\
\dot{e} = \ell_i \text{ and } \dot{f} \in K_j \setminus G_j & \text{ and } i \le j \\
\dot{e} \in K_i \setminus G_i \text{ and } \dot{f} = \ell_j & \text{ and } i < j \\
\dot{e} \in K_i \setminus G_i \text{ and } \dot{f} \in K_j \setminus G_j & \text{ and } i < j \\
\dot{e}, \dot{f} \in K_i \setminus G_i & \text{ and } \dot{e} \prec \dot{f}
\end{cases}
\end{align*}
\end{defn}

Notice that since notification events are necessarily trigger events, the elements of $\dot{\eventset}$ can be divided into trigger events (which include all the $\ell_i$ events) and side-effect events.

The transactions that are triggered by the critical donation and co-donation packets have a special place in our analysis, so we add some specific notation for them.

\begin{defn} \hfill 
\label{CritDonDef}
\begin{itemize}
\item Let $d \in \channel{P}{G}^H$ be the critical donation packet sent by $P$. We denote $\crittime{P}{G} = \preceiveevent{d}$ when $\preceiveevent{d}$ exists.
\item Let $c \in \channel{G}{P}^H$ be the critical co-donation packet sent to $P$. We denote $\crittime{G}{P} = \preceiveevent{c}$ When $\preceiveevent{c}$ exists. 
\end{itemize}
\end{defn}

\subsubsection{Donation and co-donation sub-transactions}
\label{DonCoDonSSS}
The main concern of the \ref{ReceiveDonation} and \ref{ReceiveCoDonation} procedures is the execution of a sequence of \ref{ReceiveMessage} and \ref{ReceiveAck} calls. These calls or {\em sub-transactions} compensate for race conditions. When a new process joins, it is not possible to deliver to it an up-to-date consistent copy of \uldata{} that will allow it to participate in the distributed computation without any further correction. The problem is that at the moment where the new process joins, some of the information required by it is trapped in packets that are in transit, carrying information that will arrive to the packet's destination too late to be of help.

What are these missing packets? the answer to this question is not easy. Na\"{i}vely, most of these are packets that were sent to or from the parent of the new process and failed to arrive by the moment of the view change, but the general picture is a lot more complicated because the parent itself may have joined the group in a previous view change, and the packet of interest may have been sent to/from an ancestor of the parent. Luckily, we only have to worry about the critical view change, which is the first process join in the history. This case, while not trivial, is considerably simpler. Our event labeling will reflect this fact, by labeling the sub-transactions of critical donations and co-donation transactions differently from non-critical ones. This way we can avoid an in-depth analysis of non-critical donations and co-donations, and let their meaning remain entangled in the magic of induction.

In the proceeding discussion we will use the shorthand notation
$$
\lVert \fp{vect} \rVert = \fp{vect}.f + \fp{vect}.b
$$
For vectors like \mpkout{}, \mpkin{}[] and \fp{index}.

\begin{lem}
\label{DonOrderMsgLem}
Let $P$ be an original process (a member of view zero in $H$) and suppose that $G$ processes a critical donation packet from $P$. Then there is a one-to-one, order preserving correspondence between the following two sequences:
\begin{enumerate}
\item The sequence of \ref{ReceiveMessage} invocations by $G$ at labeled step \ref{RD:STMessage} of the \ref{ReceiveDonation} procedure as it processes the donation packet from $P$.
\item\label{DMsgLemCond} The queuing events of untimely message packets in channel $\channel{P}{D}$, ordered by $\prec$. Formally these are the events $\psendevent{k} \in \eventset$ that meet the following criteria:
\begin{align*}
k \in & \channel{P}{D} \\
k = & \pkm{\fp{msg}} \\
\psendevent{k} \prec & \,\pnotifyevent{\vcrit}{P} \\
\pnotifyevent{\vcrit}{D} \prec & \,\preceiveevent{k} \qquad \text{if } \preceiveevent{k} \text{ exists}
\end{align*}
\end{enumerate}
\end{lem}

\begin{proof}
Let $k$ be a packet that meets all the criteria of (\ref{DMsgLemCond}). Then it is a message packet $k = \pkm{\fp{msg}}$ that had been sent from $P$ to $D$ prior to the critical view change. When $k$ is queued, $P$ places a record
$
\langle
\fp{msg}, \fp{index}, \fp{iset}
\rangle
$
in its \wset{}. Since the packet does not make it to $D$ prior to the critical moment, the message does not stabilize with respect to $D$ and remains in \wset{} with a $D$ instability at the critical moment. Moreover, it follows from Lemma \ref{OrigPCountLem} that at the critical moment at $D$,
$\lVert \mpkin{}[P] \rVert < \lVert \fp{index} \rVert$.
Indeed we can say more than that: if $k$ is the $i^{th}$ packet in the sequence of untimely message packets in the $\channel{P}{D}$ channel, then
$\lVert \mpkin{}[P] \rVert + i = \lVert \fp{index} \rVert$.

Process $G$ starts life with the same value of $\mpkin{}[P]$, inherited from $D$. Since the \ref{Run} procedure does not touch this value, and since the critical donation packet from $P$ is the first packet that $G$ receives from $P$, this value remains unchanged until $G$ executes the \ref{ReceiveDonation} procedure.

At the critical moment $P$ executes the \ref{JoinNotification} procedure. It adds $G$ instability to the record, with $\fp{iset}[G].f = \fp{iset}[D].f$ and $\fp{iset}[G].b = 0$. $P$ then donates its state to $G$. When $G$ processes the donation by invoking the \ref{ReceiveDonation} procedure, it finds the record with its $G$-instability and copies it to the $\se{UNT}_p$ set. $G$ sorts the records in such a way that the records of untimely message packets from $P$ are sorted in their broadcast order. 

To see why that is true, let $k$ and $k'$ be untimely message packets in $\channel{P}{D}$ and suppose that $\psendevent{k'} \,\succ\, \psendevent{k}$. Then $(\fp{iset}[D].f)_{@k'} \ge (\fp{iset}[D].f)_{@k}$ because this value is initialized from $\mpkin{}[D].f$ (see the \ref{BroadcastMessage} and \ref{RemovalNotification} procedures) which is monotone increasing. As a result $(\fp{iset}[G].f)_{@k'} \ge (\fp{iset}[G].f)_{@k}$ in the donated state from $P$ and so
$$
\heighta(\langle \fp{msg'}, \fp{index'}, \fp{iset'}[] \rangle) \ge \heighta(\langle \fp{msg}, \fp{index}, \fp{iset}[] \rangle)
$$
If these heights are equal the sorting proceeds based on $\heightb$. This function is derived from the \fp{index} component of the record, which is strictly increasing with each message broadcast at $P$. Therefore
$$
\heightb(\langle \fp{msg'}, \fp{index'}, \fp{iset'}[] \rangle) > \heightb(\langle \fp{msg}, \fp{index}, \fp{iset}[] \rangle)
$$
which proves that the records of the untimely message packets are sorted in the order at which those message packets were queued by $P$.

Process $G$ examines the records in the sorted $\se{UNT}$ set one by one. Every time a record from $\se{UNT}_p$ is examined, its
$\lVert \fp{index} \rVert$
is compared to
$\lVert \mpkin{}[P] \rVert$.
If $\lVert \fp{index} \rVert$ is bigger, the \ref{ReceiveMessage} procedure is invoked, causing $\lVert \mpkin{}[P] \rVert$ to be incremented by $1$. It follows from Lemma \ref{OrigPCountLem} that if the message is timely then the comparison will fail and the call will not be invoked. So the call is invoked only for untimely message packets, in the right order, and causes $\lVert \mpkin{}[P] \rVert$ to be incremented each time. If the call is invoked for all the untimely message packets up to the $i^{th}$ one, then $\lVert \mpkin{}[P] \rVert$ grows by $i-1$ up to that point, which is not enough to prevent the comparison from succeeding for the $i^{th}$ message packet. Therefore the call is invoked exactly for the records of untimely message packets, in the correct order, as claimed.
\end{proof}

\begin{lem}
\label{DonOrderAckLem}
Let $P$ be an original process (a member of view zero in $H$) and suppose that $G$ processes a critical donation packet from $P$. Then there is a one-to-one, order preserving correspondence between the following two sequences:
\begin{enumerate}
\item The sequence of \ref{ReceiveAck} invocations by $G$ at labeled step \ref{RD:STAck} of the \ref{ReceiveDonation} procedure as it processes the donation packet from $P$.
\item\label{DAckLemCond} The queuing events of untimely forwarded acknowledgement packets in channel $\channel{P}{D}$, ordered by $\prec$. Formally these are the events $\psendevent{k} \in \eventset$ that meet the following criteria:
\begin{align*}
k \in & \channel{P}{D} \\
k = & \pks{\fp{msg}} \\
\orig{\fp{msg}} \ne & D \\
\psendevent{k} \prec & \,\pnotifyevent{\vcrit}{P} \\
\pnotifyevent{\vcrit}{D} \prec & \,\preceiveevent{k} \qquad \text{if } \preceiveevent{k} \text{ exists}
\end{align*}
\end{enumerate}
\end{lem}

\begin{proof}
Let $k$ be a packet that meets all the criteria of (\ref{DAckLemCond}). Then this packet had been sent from $P$ to $D$ prior to the critical view change, in reaction to the receipt of a message \fp{msg} from $D$. The message must be forwarded because it did not originate with $D$. When $D$ sends the message packet to $P$, it places a record
$
\langle
\fp{msg}, \fp{index}, \fp{iset}
\rangle
$
in its \fwset{} (see labeled step \ref{RRN:WaitSet} of the \ref{RemovalNotification} procedure). Let $F_{\fp{msg}} = \fp{index}.f$.

As the critical parent, $D$ is an original process and Lemma \ref{OrigPCountLem} applies. Therefore after $P$ executes the \ref{ReceiveMessage} procedure, its value of $\mpkin{}[D].f$ is equal to $F_{\fp{msg}}$. Let $F_{\fp{crit}}$ be the value of $\mpkin{}[D].f$ at $P$ at the critical moment. Since the receipt of the message \fp{msg} at $P$ takes place pre-critically, $F_{\fp{crit}} \ge F_{\fp{msg}}$. At the critical time $P$ copies $\mpkin{}[D].f$ to $\mpkin{}[G].f$ (while zeroing out the $.b$ component) before sending its donation to $G$ (see the \ref{JoinNotification} procedure).

Since the acknowledgement packet for this message does not arrive by the critical moment, the record remains with $P$ in its instability set (\fp{iset}) and $G$ starts life with the same record in its own \fwset{}, after inheriting it from $D$. Right away $G$ zeroes out $\fp{index}.b$ in the record (see the \ref{Run} procedure). Since the donation packet is the first packet ever received at $G$ from $P$, the record must remain in \fwset{} at $G$ until that moment, because there is no way until that time for the $P$-instability to be resolved.

Following the execution of the \ref{ReceiveDonation} procedure the record of \fp{msg} in \fwset{} makes it into the $\se{UNT}_g$ set. Moreover, the value of the $\fp{index}.f = F_\fp{msg}$ component of the record is no higher than the critical value of $\mpkin{}[G].f = F_{\fp{crit}}$ at $P$, which is exactly the value that arrives at $G$ together with $P$'s donation. The $.b$ field in both values is zero. Therefore the \ref{ReceiveAck} procedure gets invoked on behalf of \fp{msg}.   

If $k'$ is a packet that meets the same criteria and has a later queuing event than $k$, then $k' = \pks{\fp{msg'}}$ where \fp{msg'} is sent from $D$ to $P$ after \fp{msg} is sent. Therefore the value of $\lVert \fp{index} \rVert$ for \fp{msg'} is higher and therefore the related sub-transaction is executed later, so the mapping in this direction is order preserving.

The inverse correspondence is constructed in a similar way. Let
$
\langle
\fp{msg}, \fp{index}, \fp{iset}
\rangle
$
Be a record in $\se{UNT}_g$ that passes the comparison test
$$
\lVert \fp{index} \rVert \le \lVert \fp{donation}.\mpkin[G] \rVert
$$
From the construction of $\se{UNT}_g$ it follows that $\fp{iset}[P]$ exists. Therefore the message \fp{msg} got into \wset{} after a message packet queuing event that included $P$ in its target set. When and where did this queuing event happen? If it happened at $G$, then $\fp{index}$ must be strictly bigger than the initial value of $\mpkout{}$ at $G$. This initial value is equal to zero in its $.b$ component (see the \ref{Run} procedure) while its $.f$ component is equal to the critical value of $\mpkout{}.f$ at $D$. It follows from Lemma \ref{OrigPCountLem} that $F_{\fp{crit}}$, the critical value of $\mpkin{}[D].f$ at $P$ is at most as high as the critical value of $\mpkout{}.f$ at $D$. Therefore $\fp{index}.b + \fp{index}.f > F_{\fp{crit}}$. But as we have already seen, $F_{\fp{crit}} = \fp{donation}.\mpkin[G].f$ and $\fp{donation}.\mpkin[G].b = 0$. This is a contradiction. Therefore the queuing event does not happen at $G$ but rather at $D$, and the \wset{} record in question is inherited by $G$ when it is created. From this it follows that the message did not originate with $D$. If it did, then the record would go to \bwset{} in $D$, and this part of \wset{} is not inherited by $G$ but instead zeroed out (see the \ref{Run} procedure). So $\orig{\fp{msg}} \ne D$.

So the message \fp{msg} is forwarded from $D$ to $P$, and is not stabilized with respect to $P$ by the critical moment. Why is it not stabilized? either it does not arrive at $P$ on time (in which case the message packet from $D$ to $P$ is untimely), or else it does arrive on time but the acknowledgement packet from $P$ to $D$ is untimely. If the first case holds then by Lemma \ref{OrigPCountLem} $F_{\fp{crit}} < \fp{index}.b + \fp{index}.f$. As we saw, this leads to a contradiction. Therefore the acknowledgement packet $k = \pks{\fp{msg}}$ that was sent from $P$ to $D$ was untimely. This means that the packet $k$ meets all the criteria listed in the statement of the lemma.

It is obvious that these two correspondences are inverses of each other and so we are done. 
\end{proof}

\begin{cor}
\label{DonOrderCor}
Let $P$ be an original process and suppose that $G$ processes a critical donation packet from $P$. Then there is a one-to-one, order preserving correspondence between the following two sequences:
\begin{enumerate}
\item The sequence of \ref{ReceiveMessage} and \ref{ReceiveAck} invocations by $G$ at labeled steps \ref{RD:STMessage} and \ref{RD:STAck} of the \ref{ReceiveDonation} procedure as it processes the donation packet from $P$.
\item The queuing events of untimely message packets and forwarded acknowledgement packets in channel $\channel{P}{D}$, ordered by $\prec$.
\end{enumerate}
\end{cor}

\begin{proof}
Most of the claim follows directly from Lemmas \ref{DonOrderMsgLem} and \ref{DonOrderAckLem}. But we still have to show that if $k, k' \in \channel{P}{D}$ are an untimely message packet and an untimely acknowledgement packet for a forwarded message respectively, then $\psendevent{k} \,\prec\, \psendevent{k'}$ if and only if the corresponding records in $\se{UNT}_p$ and $\se{UNT}_g$ are processed in the same order.

Let $k = \pkm{\fp{msg}}$ and $k' = \pks{\fp{msg}'}$. Let $\fp{index}'$ be the value of $\mpkout{}$ in $D$ when it queues $\pkm{\fp{msg}'}$. This is the value of the \fp{index} field of the record of $\fp{msg}'$ in $\wset{}(D)$.

Suppose first that $\psendevent{k} \,\prec\, \psendevent{k'}$. Then $P$ queues $k$ before it processes $\pkm{\fp{msg}'}$. By Lemma \ref{OrigPCountLem}, at the time that $P$ queues $k$, we have
$
\mpkin{}[D].f \le \fp{index}'.f
$
at $P$. But in fact the inequality is strict here because the message $\fp{msg}'$ is forwarded by $D$, not originated. Because of that it is $\mpkout{}.f$ that is incremented when $\pkm{\fp{msg}'}$ is queued and $\mpkin{}[D].f$ that is incremented when the same packet is processed. Let $r, r' \in \se{UNT}$ be the records that correspond to $k$ and $k'$ respectively. We have already seen that $\heighta(r) = \mpkin{}[D].f$ and $\fp{index}(r').f = \fp{index}'.f$. Therefore
$$
\heighta(r) = \mpkin{}[D].f < \fp{index}'.f \le \heighta(r')
$$
which proves one direction of the claim.

To prove the other direction, suppose that $\psendevent{k} \,\succ\, \psendevent{k'}$. Then $P$ queues $k$ after it processes $\pkm{\fp{msg}'}$. By Lemma \ref{OrigPCountLem}, at the time that $P$ queues $k$, we have
$
\mpkin{}[D].f \ge \fp{index}'.f
$
and following the same logic as before we conclude that $\heighta(r) \ge \heighta(r')$. In this case however we cannot guarantee a strict inequality. In case there is equality, however, the order between $r$ and $r'$ is determined by $\heightb$. Since $r' \in \se{UNT}_g$ we have $\heightb(r') = 0$ while $r \in \se{UNT}_p$ and therefore
$$
\heightb(r) = \lVert \fp{index}(r) \rVert \ge \fp{index}'.f > 0 = \heightb(r')
$$
where the rightmost inequality is strict because $\fp{msg}'$ is a forwarded message.
\end{proof}

\begin{lem}
\label{CoDonOrderMsgLem}
Let $P$ be an original process (a member of view zero in $H$) that processes a critical co-donation packet from $G$. Then there is a one-to-one, order preserving correspondence between the following two sequences:
\begin{enumerate}
\item The sequence of \ref{ReceiveMessage} invocations by $P$ at labeled step \ref{RCD:STMessage} of the \ref{ReceiveCoDonation} procedure as it processes the co-donation packet from $G$.
\item\label{CMsgLemCond} The union of the following two sets of events, ordered by $\prec$:
\begin{enumerate}
\item\label{CMsgLemCondA} The queuing events of untimely forwarded message packets in channel $\channel{D}{P}$. Formally these are the events $\psendevent{k} \in \eventset$ that meet the following criteria:
\begin{align*}
k \in & \channel{D}{P} \\
k = & \pkm{\fp{msg}} \\
\orig{\fp{msg}} \ne & D \\
\psendevent{k} \prec & \, \pnotifyevent{\vcrit}{D} \\
\pnotifyevent{\vcrit}{P} \prec & \, \preceiveevent{k} \qquad \text{if } \preceiveevent{k} \text{ exists}
\end{align*}
\item \label{CMsgLemCondB} The post-critical, pre-$P$-donation queuing events of message packets in channel $\channel{G}{G}$. Formally these are the events $\psendevent{k} \in \eventset$ that meet the following criteria: 
\begin{align*}
k \in & \channel{G}{G} \\
k = & \pkm{\fp{msg}} \\
\pnotifyevent{\vcrit}{G} \prec & \, \psendevent{k} \prec\, \crittime{P}{G}
\end{align*}
\end{enumerate}
Note: The union of these two sets of events is linearly ordered by $\prec$ because according to the \AxNotEventIII{} $\pnotifyevent{\vcrit}{D} \asymp \joinevent{G} = \pnotifyevent{\vcrit}{G}$.
\end{enumerate}
\end{lem}

\begin{proof}
Let $k$ be a packet that meets all the criteria of (\ref{CMsgLemCondA}). This packet was sent from $D$ to $P$ prior to the critical view change, containing a forwarded message \fp{msg}. When $D$ queued the packet it placed a record
$
\langle
\fp{msg}, \fp{index}, \fp{iset}
\rangle
$
in its \fwset{}. Let $F_{\fp{msg}} = \fp{index}.f$. Let $F_{\fp{crit}}$ be the value of $\mpkin{}[D].f$ at $P$ at the critical moment. Then the fact that the packet is untimely, together with Lemma \ref{OrigPCountLem} imply that $F_{\fp{msg}} >  F_{\fp{crit}}$. We can say more than that: if $k$ is the $i^{th}$ packet in the sequence of untimely forwarded message packets in the $\channel{D}{P}$ channel, then $F_{\fp{msg}} = F_{\fp{crit}} + i$. At the critical moment process $P$ creates $\mpkin{}[G].f = F_{\fp{crit}}$. At the critical moment the value of $\mpkout{}.f$ at $D$ is equal to $F_{\fp{crit}} + u$ where $u$ is the number of untimely forwarded message packets in the $\channel{D}{P}$ channel. As a result $G$ starts life with $\lVert \mpkout{} \rVert = F_{\fp{crit}} + u$ (the \ref{Run} procedure zeroes out $\mpkout{}.b$).

Since the packet $k$ is untimely, the record remains in \fwset{}, with $P$ in its instability set, until the critical moment. Therefore $G$ starts life with the same record in its own \fwset{}, and the \ref{Run} procedure leaves $\fp{index}.f$ untouched. Since $G$ does not receive any packet from $P$ until it receives its donation, the record remains in \fwset{} at $G$ until $\crittime{P}{G}$ and therefore it is sent to $P$ as part of the co-donation from $G$. As a result, when $P$ processes the co-donation from $G$ it finds the record with its $P$-instability, and places it in $\se{UNT}_g$.

$P$ sorts the records in such a way that the records of untimely message packets from $D$ are sorted in their broadcast order. To see why that is true, let $k$ and $k'$ be untimely forwarded message packets in $\channel{D}{P}$ such that $\psendevent{k'} \,\succ\, \psendevent{k}$. Then $\lVert (\fp{iset}[P])_{@k'} \rVert \ge \lVert (\fp{iset}[P])_{@k} \rVert$ at $D$ because this value is initialized from $\lVert \mpkin{}[P] \rVert$ (see the \ref{RemovalNotification} procedure) which is monotone increasing. The value of \fp{iset} is not changed by the \ref{Run} procedure so the same relationship holds in $G$ and so this is what $P$ sees in the co-donated state that it receives from $G$. Therefore
$$
\heighta(\langle \fp{msg'}, \fp{index'}, \fp{iset'}[] \rangle) \ge \heighta(\langle \fp{msg}, \fp{index}, \fp{iset}[] \rangle)
$$
If these heights are equal the sorting proceeds based on $\heightb$. This function is derived from the \fp{index} component of the record. The \ref{Run} procedure zeroes out $\fp{index}.b$ but does not touch $\fp{index}.f$. The original value of $\fp{index}.f$ is derived from the value of $\mpkout{}.f$ at $D$. Since $k$ is a forwarded message packet, the value of $\mpkout{}.f$ is increased after $D$ queues $k$, and therefore $\fp{index}'.f > \fp{index}.f$. Therefore at $G$, where the $.b$ component is zero for the records of both $k$ and $k'$ we have
$
\lVert \fp{index}' \rVert > \lVert \fp{index} \rVert
$
and therefore
$$
\heightb(\langle \fp{msg'}, \fp{index'}, \fp{iset'}[] \rangle) > \heightb(\langle \fp{msg}, \fp{index}, \fp{iset}[] \rangle)
$$
which proves that the records of the untimely message packets are sorted in the order at which those message packets were queued by $D$.

Now let us examine the other subset of packets. Let $k$ be a packet that meets all the criteria of (\ref{CMsgLemCondB}). When $G$ queues $k$ it does not have $P$ in \cset{}, because \cset{} is initialized in $G$ to contain only itself (see the \ref{Run} procedure), and $P$ is only added to the set when $G$ processes the critical donation from $P$. As a result $G$ does not send a message packet to $P$ even though it does create an $\fp{iset}[P]$ entry in the instability set of the message in its \wset{} (to see that, check the \ref{BroadcastMessage} and \ref{RemovalNotification} procedures to see that message packets are only sent to contacted processes, while instability is created to each process that has an entry in the \mpkin{} vector. According to Lemma \ref{OmnibusCBCASTLem}(\ref{OCL:lsetvec}) \mpkin{} has entries exactly for the live processes.) This makes it impossible for the record
$
\langle
\fp{msg}, \fp{index}, \fp{iset}
\rangle
$
of the message to stabilize with respect to $P$ by the time that $G$ processes the donation from $P$, and as a result the record survives in the co-donated state that $G$ sends to $P$, together with its $P$-instability. As a result $P$ copies the record into $\se{UNT}_g$.

The value of \fp{index} in the record is initialized from the value of $\mpkout{}$ at $G$. As we noted before $G$ starts life with $\lVert \mpkout{} \rVert = F_{\fp{crit}} + u$ where $u$ is the number of packets meeting the criteria of \ref{CMsgLemCondA}. Since $\mpkout{}$ is incremented by $G$ every time it queues message packets, we can conclude that the $j^{th}$ packet that meets the criteria of \ref{CMsgLemCondB} has a record with $\lVert \fp{index} \rVert = F_{\fp{crit}} + u + j$. In other words if we take the two sequences \ref{CMsgLemCondA} and \ref{CMsgLemCondB} as one, then the record of the $i^{th}$ packet in that sequence has $\lVert \fp{index} \rVert = F_{\fp{crit}} + i$.

We want to show that these records are also sorted by $P$ according to their $\prec$-order with respect to each other and also with respect to the records of the untimely $\channel{D}{P}$ packets.

Let $k$ and $k'$ be two packets in $\channel{G}{G}$ that meet the \ref{CMsgLemCondB} criteria, and let $k^0$ be a packet in $\channel{D}{P}$ that meets the \ref{CMsgLemCondA} criteria. $k$ and $k'$ are post-critical and both have records that retain their $P$ instability. Moreover, $\fp{iset}[P]$ is initialized in both cases from $\mpkin{}[P]$ at $G$, a value that starts out being equal to the critical value of $\mpkin{}[P]$ at $D$. This value is untouched by the \ref{Run} procedure, and remains unchanged until the moment that $G$ processes the donation from $P$, because there are no incoming message packets from $P$ to $G$ until that time. Therefore $\fp{iset}[P] = \fp{iset}'[P]$. Similarly, $k^0$ gets its value of $\fp{iset}[P]$ from a pre-critical value of $\mpkin{}[P]$ at $D$. Therefore
$$
\heighta(\langle \fp{msg'}, \fp{index'}, \fp{iset'}[] \rangle) = \heighta(\langle \fp{msg}, \fp{index}, \fp{iset}[] \rangle) \ge \heighta(\langle \fp{msg}^0, \fp{index}^0, \fp{iset}^0[] \rangle)
$$
So the $\heighta$ function does not resolve the order between the records of $k$ and $k'$, and if it resolves it between $k^0$ and $k$ it does so in the desired way.

Both $k$ and $k'$ derive their $\fp{index}$ value from the value of \mpkout{} at $G$, which increases with each message packet queuing event. therefore
$$
\heightb(\langle \fp{msg'}, \fp{index'}, \fp{iset'}[] \rangle) > \heightb(\langle \fp{msg}, \fp{index}, \fp{iset}[] \rangle)
$$
and we have the desired order relation in this case. The record for packet $k^0$ at $D$ derives its \fp{index} value from a pre-critical value of \mpkout{} at $D$. $G$ inherits \mpkout{} from $D$ with the $.b$ component zeroed out. $G$ inherits the $k^0$ record and similarly zeroes out its $\fp{index}.b$ value. Post-critically, \mpkout{} at $G$ continues to grow from this initial value with each message broadcast. Therefore 
$$
\heightb(\langle \fp{msg}, \fp{index}, \fp{iset}[] \rangle) > \heightb(\langle \fp{msg}^0, \fp{index}^0, \fp{iset}^0[] \rangle)
$$
and we have the desired order in this case as well.

Process $P$ examines the records in the sorted $\se{UNT}$ set one by one. Every time a record from $\se{UNT}_g$ is examined, its
$\lVert \fp{index} \rVert$
is compared to
$\lVert \mpkin{}[G] \rVert$.
If $\lVert \fp{index} \rVert$ is bigger, the \ref{ReceiveMessage} procedure is invoked, causing $\lVert \mpkin{}[G] \rVert$ to be incremented by $1$.

There are exactly three types of records in $\se{UNT}_g$. There are records for message packets that were queued by $G$ itself, namely post-critical message packets that meet the criteria of \ref{CMsgLemCondB}. Then there are records that were inherited from $D$. These are records of pre-critical message packets. Some of these packets are untimely and meet the criteria of \ref{CMsgLemCondA}, while others are timely.

It follows from Lemma \ref{OrigPCountLem} that if the message is timely then the comparison will fail and the procedure will not be invoked. So the procedure is invoked only for untimely and post-critical message packets, in the right order, and causes $\lVert \mpkin{}[G] \rVert$ to be incremented each time. If the procedure is invoked for all the untimely and post-critical message packets up to the $i^{th}$ one, then $\lVert \mpkin{}[G] \rVert$ grows by $i-1$ up to that point, which is not enough to prevent the comparison from succeeding for the $i^{th}$ message packet. Therefore the procedure is invoked exactly for the records of message packets that meet one of the two criteria, in the correct order, as claimed.
\end{proof}

\begin{lem}
\label{CoDonOrderAckLem}
Let $P$ be an original process (a member of view zero in $H$) that processes a critical co-donation packet from $G$. Then there is a one-to-one, order preserving correspondence between the following two sequences:
\begin{enumerate}
\item The sequence of \ref{ReceiveAck} invocations by $P$ at labeled step \ref{RCD:STAck} of the \ref{ReceiveCoDonation} procedure as it processes the co-donation packet from $G$.
\item\label{CAckLemCond} The queuing events of untimely acknowledgement packets in channel $\channel{D}{P}$, ordered by $\prec$. Formally these are the events $\psendevent{k} \in \dot{\eventset}$ that meet the following criteria:
\begin{align*}
k & \in \channel{D}{P} \\
k & = \pks{\fp{msg}} \\
\psendevent{k} & \prec \pnotifyevent{\vcrit}{D} \\
\pnotifyevent{\vcrit}{P} & \prec \preceiveevent{k} \qquad \text{if } \preceiveevent{k} \text{ exists}
\end{align*}
\end{enumerate}
\end{lem}

\begin{proof}
Let $I_{\fp{crit}}$ be the value of $\mpkin{}[P]$ at $D$ at the critical moment. $G$ inherits $\mpkin{}[P]$ from $D$. This value is not changed by the \ref{Run} procedure, and it remains unchanged until the receipt of a donation packet from $P$, since this is the first packet of any kind that $G$ receives from $P$. Upon receipt of the donation packet from $P$, $G$ immediately sends a co-donation packet to $P$. Therefore when $P$ executes the \ref{ReceiveCoDonation} procedure, we have:
$$
\lVert \fp{co\_donation}.\mpkin[P] \rVert = \lVert I_{\fp{crit}} \rVert 
$$

Let $k$ be a packet that meets all the criteria of (\ref{CAckLemCond}). Then this packet was sent from $D$ to $P$ prior to the critical view change, in reaction to the receipt of a message \fp{msg} from $P$. When $P$ sends the message packet to $D$, it places a record
$
\langle
\fp{msg}, \fp{index}, \fp{iset}
\rangle
$
in its \wset{}. Let $I_{\fp{msg}} = \fp{index}$.

Since the acknowledgement packet for this message does not arrive by the critical moment, the record remains in \wset{} at $P$, with $D$ in its instability set until then. At that moment, $P$ adds a $G$ instability to the record (see labeled step \ref{JN:NewInst} of the \ref{JoinNotification} procedure). This additional instability ensures that the record will remain in \wset{} until the first packet from $G$ arrives. This packet is the co-donation from $G$. Therefore when $P$ processes the co-donation packet it discovers the record in its \wset{} and copies it into $\se{UNT}_p$.

As the critical parent, $D$ is an original process and Lemma \ref{OrigPCountLem} applies. Therefore after $D$ executes the \ref{ReceiveMessage} procedure, its value of $\mpkin{}[P]$ is equal to $I_{\fp{msg}}$. Since the receipt of the message \fp{msg} at $D$ takes place pre-critically, $\lVert I_{\fp{crit}} \rVert \ge \lVert I_{\fp{msg}} \rVert$. Therefore:
$$
\lVert \fp{co\_donation}.\mpkin[P] \rVert = \lVert I_{\fp{crit}} \rVert \ge \lVert I_{\fp{msg}} \rVert = \lVert \fp{index} \rVert 
$$
therefore the record results in the invocation of the \ref{ReceiveAck} procedure at labeled step \ref{RCD:STAck} of the \ref{ReceiveCoDonation} procedure.

If $k'$ is a packet that meets the same criteria and has a later queuing event than $k$, then $k' = \pks{\fp{msg'}}$ where \fp{msg'} is sent from $P$ to $D$ after \fp{msg} is sent. Therefore the value of $\lVert \fp{index} \rVert$ for \fp{msg'} is higher and therefore the related sub-transaction is executed later, so the mapping in this direction is order preserving.

The inverse correspondence is constructed in a similar way. Let
$
\langle
\fp{msg}, \fp{index}, \fp{iset}
\rangle
$
Be a record in $\se{UNT}_p$ that passes the comparison test
$$
\lVert \fp{index} \rVert \le \lVert \fp{co\_donation}.\mpkin[P] \rVert
$$
From the construction of $\se{UNT}_p$ it follows that $\fp{iset}[G]$ exists. Therefore the message \fp{msg} got into \wset{} after a message packet queuing event that included a packet being sent to $G$ (post-critically) or to $D$ (pre-critically). We must eliminate the post-critical case. If the queuing event is post-critical then $\fp{index}$ must be strictly bigger than the critical value of $\mpkout{}$ at $P$. Call this critical value $O_{\fp{crit}}$. It follows from Lemma \ref{OrigPCountLem} that $O_{\fp{crit}} \ge I_{\fp{crit}}$. Therefore
$$
\fp{index} > \lVert O_{\fp{crit}} \rVert \ge \lVert I_{\fp{crit}} \rVert = \lVert \fp{co\_donation}.\mpkin[P] \rVert
$$
contradicting our assumption. It follows that \fp{msg} must have been sent pre-critically to $D$. We know that $D$ processed the message packet pre-critically because the relation
$$
\lVert \fp{index} \rVert \le \lVert \fp{co\_donation}.\mpkin[P] \rVert = \lVert I_{\fp{crit}} \rVert
$$
implies it according to Lemma \ref{OrigPCountLem}. 
It is obvious that these two correspondences are inverses of each other and so we are done. 
\end{proof}

\begin{cor}
\label{CoDonOrderCor}
Let $P$ be an original process that processes a critical co-donation packet from $G$. Then there is a one-to-one, order preserving correspondence between the following two sequences:
\begin{enumerate}
\item The sequence of \ref{ReceiveMessage} and \ref{ReceiveAck} invocations by $G$ at labeled steps \ref{RD:STMessage} and \ref{RD:STAck} of the \ref{ReceiveDonation} procedure as it processes the donation packet from $P$.
\item The queuing events of untimely forwarded message packets in channel $\channel{P}{D}$, the queuing events of untimely acknowledgement packets in channel $\channel{P}{D}$ and the queuing events of post-critical, pre-$P$-donation message packets in channel $\channel{G}{G}$, ordered by $\prec$.
\end{enumerate}
\end{cor}

\begin{proof}
Most of the claim follows directly from Lemmas \ref{CoDonOrderMsgLem} and \ref{CoDonOrderAckLem}. But we still have to show that if $k, k' \in \channel{D}{P}$ are an untimely forwarded message packet and an untimely acknowledgement packet respectively, and if $k'' \in \channel{G}{G}$ is a pre-$P$-donation message packet then
\begin{itemize}
\item $\psendevent{k} \, \prec \, \psendevent{k'}$ if and only if the corresponding records in $\se{UNT}_g$ and $\se{UNT}_p$ are processed in the same order.
\item The record for $k''$ in $\se{UNT}_g$ is processed after the record for $k'$ in $\se{UNT}_p$.
\end{itemize}

Let $k = \pkm{\fp{msg}}$ and $k' = \pks{\fp{msg}'}$. Let $\fp{index}'$ be the value of $\mpkout{}$ in $P$ when it queues $\pkm{\fp{msg}'}$. This is the value of the \fp{index} field of the record of $\fp{msg}'$ the $\wset{}(P)$.

Suppose first that $\psendevent{k} \,\dot{\prec}\, \psendevent{k'}$. Then $D$ queues $k$ before it processes $\pkm{\fp{msg}'}$. By Lemma \ref{OrigPCountLem}, at the time that $D$ queues $k$, we have
$
\lVert \mpkin{}[P] \rVert < \lVert \fp{index}' \rVert
$
at $D$. Let $r, r' \in \se{UNT}$ be the records that correspond to $k$ and $k'$ respectively. We have already seen that $\heighta(r) = \lVert \mpkin{}[P] \rVert$ and $\fp{index}(r') = \fp{index}'$. Therefore
$$
\heighta(r) = \lVert \mpkin{}[P] \rVert < \lVert \fp{index}' \rVert = \heighta(r')
$$
which proves one direction of the claim for $k$ and $k'$.

To prove the other direction, suppose that $\psendevent{k} \,\succ\, \psendevent{k'}$. Then $D$ queues $k$ after it processes $\pkm{\fp{msg}'}$. By Lemma \ref{OrigPCountLem}, at the time that $D$ queues $k$, we have
$
\lVert \mpkin{}[P] \rVert \ge \lVert \fp{index}' \rVert
$
and following the same logic as before we conclude that $\heighta(r) \ge \heighta(r')$. In this case however we cannot guarantee a strict inequality. In case there is equality, however, the order between $r$ and $r'$ is determined by $\heightb$. Since $r' \in \se{UNT}_p$ we have $\heightb(r') = 0$ while $r \in \se{UNT}_g$ and therefore
$$
\heightb(r) = \lVert \fp{index}(r) \rVert > 0 = \heightb(r')
$$

We still have to show that the record for $k''$ is placed later than the record for $k'$ in the sorted order. As we saw before, $G$ inherits the value of $\mpkin{}[P]$ without change, and since $D$ processes $\pkm{\fp{msg}'}$ pre-critically, the critical value of $\lVert \mpkin{}[P] \rVert$ at $D$ is at least as high as $\lVert \fp{index}' \rVert$. Therefore at the moment that $G$ queues $k''$ we have
$
\lVert \mpkin{}[P] \rVert \ge \lVert \fp{index}' \rVert
$
and therefore $\heighta(r'') \ge \heighta(r')$ where $r''$ is the record corresponding to $k''$. The rest of the argument is the same as the similar argument for $\psendevent{k} \,\succ\, \psendevent{k'}$.
\end{proof}

\begin{lem}
\label{CoDonFlushLem}
Let $P$ be an original process (a member of view zero in $H$) that processes a critical co-donation packet from $G$, and suppose that the invocation of \ref{TryToInstall} at labeled step \ref{RCD:TryToInstall} of the \ref{ReceiveCoDonation} procedure results in the installation of the pending views at $P$. Let $v = \cview{} + \vgap{}$ be the view height at $P$ at the moment that it processes the critical co-donation. Then $G$ must have queued a $\pkf{v}$ packet prior to dequeuing the donation from $P$. Moreover, after queuing that packet and until the donation from $P$ is dequeued, $G$ does not queue any ghost, flush or message packets. 
\end{lem}

\begin{proof}
The proof of Lemma \ref{DonationSELem} shows that if \ref{TryToInstall} actually installs the pending views, then the value of $\sfh{}$ at $G$ must have been equal to $v$ at the time that $G$ processed the donation from $P$. Let $v_0 = \valuepre{\sfh{}}{G}{\pnotifyevent{\vcrit}{G}}$ be the initial flush height at $G$ (see Definition \ref{PrePostDef}). $v_0$ is equal to $\valuepre{\sgh{}}{D}{\pnotifyevent{\vcrit}{D}}$ and therefore $v_0 < \vcrit$ according to Lemma \ref{OmnibusCBCASTLem}(\ref{OCL:height} and \ref{OCL:sendgfA}). Lemma \ref{OmnibusCBCASTLem}(\ref{OCL:height}) also implies that $v \ge \vcrit$ because the co-donation is processed at $P$ after the critical view change. This implies that $\sfh{}$ at $G$ increases between the time $G$ is created and the time that it processes the donation from $P$. The only way for that to happen is for $G$ to queue one or more flush packets. The last one of the flush packets before the receipt of the donation from $P$ would reflect the value of $\sfh{}$ at the time. Therefore this last flush packet must be $k_{\fp{last}} = \pkf{v}$. We have to show that $G$ does not queue any additional regular packets between $k_{\fp{last}}$ and the processing of the donation from $P$. We already know that $G$ does not queue any additional flush packets in that interval. From Lemma \ref{OmnibusCBCASTLem}(\ref{OCL:sendgfA}) and the \AxProcIII{} we know that $G$ also does not queue any ghost packets. It follows from the proof of Lemma \ref{DonationSELem} that $G$ does not queue any message packets either.
\end{proof}

\subsection{The Label Space}
\label{LabelSpaceSS}
We construct an infinite, very well founded, partially ordered set $\labelspace$, that we call the {\em label space}. We use this space as a medium for inducing an ordering on the events in $H^r$. We do it the following way. First we assign a label from $\labelspace$ to each event in $H$ in a $\prec$-order preserving manner. Then we expand the assignment to events in $H^r$. Then we use the latter labeling to induce an order on the events in $H^r$.

Each label is made up of four coordinates: the constellation coordinate, the sub-transaction coordinate, the adjustment coordinate and the side-effect coordinate. Each coordinate comes from its own very well founded partially ordered set, and the ordering on $\labelspace$ as a whole is left-handed lexicographic, with the constellation coordinate being the major one.

All the events in a single transaction share the same constellation coordinate, but the converse is not true. It is possible for multiple transactions to share the same constellation coordinate. The set of all the transactions that share each constellation coordinate is called a {\em constellation}. In the original history $H$ only notification transactions are grouped into constellations that contain more than one transaction. In $H^r$ however this is not the case. We will prove by induction that $H$ and $H^r$ converge. The proof will proceed by induction over constellations (normal induction can be carried over any very well-founded partially ordered set, not just over natural numbers). 

\subsubsection{The constellation coordinate set}
This set is made up of all the trigger events in $\dot{\eventset}$ with the $\dot{\prec}$ partial order (see Definition \ref{CleanEventDef}).

\begin{defn}
\label{CSpaceDef}
We use the symbol $\cspace$ to denote the constellation set. For any transaction $T$ the clean trigger $\dot{\trig}(T) \in \dot{\eventset}$ is an element of $\cspace$. We denote that element by $\ell_T$. For any event $e \in \eventset^H$ we denote $\ell_e = \ell_{\trans(e)}$. If $T$ is a notification transaction for view $i$ then by definition $\ell_T = \ell_i$. 
\end{defn}

\subsubsection{The sub-transaction coordinate set}
\label{SubTransSSS}
This coordinate is only used for differentiating sub-transactions within donation and co-donation transactions. For all other events we use a special zero symbol for this coordinate. These sub-transactions simulate the processing of packets that were supposed to be sent to or from a newly joining process, but that were not sent due to race conditions. For the purpose of constructing $H^r$ the labeling of the critical donation and co-donation sub-transactions is crucial. As for the non-critical donations and co-donations (i.e. those that involve processes other than $G$ that join after the critical view change), their labeling is a lot less important. It would be nice to use the same labeling scheme for all donations and co-donation, and this is indeed possible, but very complicated and not really worth the effort. So reluctantly we use two separate systems of labeling sub-transactions. One for the critical ones and one for the non-critical ones.

The analysis in \ref{DonCoDonSSS} shows that the critical sub-transactions in both donation and co-donation transactions can be ordered using related events from $\dot{\eventset}$. These events are all queuing events of packets, ordered by $\dot{\prec}$. This is true for non-critical sub-transactions as well but we have not shown that. So we will simply use positive integers to label those sub-transactions. Together with the additional special zero element this yields:
$$
\sspace = \{ \hat{0} \} + \left( \dot{\eventset} \coprod \mathbb{N} \right)
$$

\subsubsection{The adjustment coordinate set}
This coordinate is used for slight adjustments of trigger events in $H^r$, forcing them to happen either slightly earlier or slightly later than related triggers. The set contains nine elements:
$$
\aspace = \{ \minusshftghost < \shftghost < \minusshftflush < \shftflush < \minusshftack < \shftack < \hat{0} < \minusshftmsg < \shftmsg \}
$$
Where the elements refer to the time adjustment in processing the following packets:
\begin{itemize}
\item $\minusshftghost$: A ghost packet from $\minusG$
\item $\shftghost$: A ghost packet from $G$
\item $\minusshftflush$: A flush packet from $\minusG$
\item $\shftflush$: A flush packet from $G$
\item $\minusshftack$: An acknowledgement packet from $\minusG$
\item $\shftack$: An acknowledgement packet from $G$
\item $\minusshftmsg$: A message packet from $\minusG$
\item $\shftmsg$: A message packet from $G$
\end{itemize}
\subsubsection{The side-effect coordinate set}
This coordinate is used for differentiating events within a transaction (or a sub-transaction, in the case of a donation or a co-donation that includes multiple sub-transactions). The set contains an infinite sequence of elements
\begin{multline*}
\fspace = \{ \hat{0} < \lCODONATE < \lDONATE < \lACK < \lGHOST < \lFLUSH < \\
< \lBCAST_1 < \lBCAST_2 < \dots \}
\end{multline*}
Where $\hat{0}$ is indicates a trigger event and the meaning of each of the other symbols is self explanatory.

\subsubsection{Labeling the events in $H$}
\label{LabelingHSSS}
As we mentioned, each label contains four coordinates, one of each type. Formally
$$
\labelspace = \cspace \times \sspace \times \aspace \times \fspace
$$
If $\ell_T \in \cspace$, $s \in \sspace$, $a \in \aspace$ and $f \in \fspace$, we use the notation $\lblggg{T}{s}{a}{f}$ to denote the label with those coordinates. If $T$ is a critical donation or co-donation transaction we use the notations $\lblcggg{G}{P}{s}{a}{f}$ and $\lblcggg{P}{G}{s}{a}{f}$. 

We want to create a map $\lbl : \eventset^H \rightarrow \labelspace$ that assigns a label to each event in $H$ and preserves order in the sense that if $e \prec f$ then $\lbl(e) < \lbl(f)$.

Let $e$ be any event in $H$. We create the label of $e = \lblggg{T}{s}{a}{f}$ coordinate by coordinate:
\begin{description}
\item[The constellation coordinate:] \hfill \\
We simply use the constellation $\ell_{\trans(e)}$.
\item[The sub-transaction coordinate:] \hfill
\begin{itemize}
\item If $\ell_{\trans(e)} = \crittime{P}{G}$ and $e$ is a side effect of one of the sub-transactions (labeled steps \ref{RD:STMessage} and \ref{RD:STAck} of \ref{ReceiveDonation}) then Corollary \ref{DonOrderCor} implies that there is a corresponding event $f = \psendevent{k} \in \dot{\eventset}$ where $k \in \channel{P}{D}$ is an untimely packet. We use $f$ as the sub-transaction coordinate for $e$.
\item If $\ell_{\trans(e)} = \crittime{G}{P}$ and $e$ is a side effect of one of the sub-transactions (labeled steps \ref{RCD:STMessage} and \ref{RCD:STAck} of \ref{ReceiveCoDonation}) then Corollary \ref{CoDonOrderCor} implies that there is a corresponding event $f = \psendevent{k} \in \dot{\eventset}$ where either $k \in \channel{D}{P}$ is an untimely packet or $k \in \channel{G}{G}$ is a post-critical, pre-$P$-donation packet. We use $f$ as the sub-transaction coordinate for $e$.
\item If $\ell_{\trans(e)} = \crittime{G}{P}$ and $e$ is a side effect of the \ref{TryToInstall} procedure invocation at labeled step \ref{RCD:TryToInstall} of \ref{ReceiveCoDonation} then Lemma \ref{CoDonFlushLem} implies that there is a corresponding event $f = \psendevent{k} \in \dot{\eventset}$ where $k \in \channel{G}{G}$ is the last pre-$P$-donation flush packet, and $k = \pkf{v}$ where $v$ is equal to the view height $\cview{} + \vgap{}$ at $P$ at the time that it processes the co-donation. We use $f$ as the sub-transaction coordinate for $e$. 
\item If $\ell_{\trans(e)}$ is a non-critical donation or co-donation transaction, and $e$ is a side effect of one of the sub-transactions then we label $e$ with a serial number, with all the side effects of the first sub-transaction receiving a value of $1$, the side effects of the second sub-transaction receiving a value of $2$, et cetera. 
\item If $e$ is any other event then we label $e$ with the special value $\hat{0}$. This includes the cases where $e$ is
\begin{itemize}
\item a trigger event.
\item a side effect of a transaction that is neither a donation nor a co-donation.
\item the queuing event of the co-donation packet in a donation transaction.
\item a side effect of the \ref{TryToInstall} invocation in a non-critical co-donation transaction.
\end{itemize}
\end{itemize}
\item[The adjustment coordinate:] \hfill \\
In $H$ this coordinate is always $\hat{0}$. We only make non-trivial adjustments for events that are added in $H^r$.
\item[The side-effect coordinate:] \hfill \\
We assign this coordinate according to the side-effect. If $e$ is a trigger event, we use the special zero value $\hat{0}$. Otherwise $e = \psendevent{k}$ for some packet. In that case we look at that side effect that produced $k$:
\begin{itemize}
\item If $k = \pkc{\fp{co\_donation}}$ then we assign the value $\lCODONATE$.
\item If $k = \pkd{\fp{donation}}$ then we assign the value $\lDONATE$.
\item If $k = \pks{\fp{msg}}$ then we assign the value $\lACK$.
\item If $k = \pkh{v}$ then we assign the value $\lGHOST$.
\item If $k = \pkf{v}$ then we assign the value $\lFLUSH$.
\item If $k = \pkm{\fp{msg}}$ then we have to differentiate between several cases:
\begin{itemize}
\item If $k$ is an original message packet, in the context of a message broadcast request transaction, we assign the value $\lBCAST_1$.
\item If $k$ is an original message packet $\pkm{\fp{msg}}$ in the context of a flush packet processing transaction (see Lemma \ref{GhostFlushSELem}), we assign the value $\lBCAST_i$ if $\fp{msg}$ is the $i^{th}$ original message out of \lque{}.
\item If $k$ is an original message packet $\pkm{\fp{msg}}$ in the context of a co-donation packet processing transaction (see Lemma \ref{DonationSELem}), we assign the value $\lBCAST_i$ if $\fp{msg}$ is the $i^{th}$ original message out of \lque{}.
\item If $k$ is a forwarded message packet $\pkm{\fp{msg}}$ in the context of a notification transaction (removal of process $P$), we assign the value $\lBCAST_i$ if $\fp{msg}$ is the $i^{th}$ forwarded message out of $\fque{}[P]$
\end{itemize}
\end{itemize}
\end{description}

\subsubsection{Some labeling examples}
Suppose that the application at some process $P$ in $H$ decides to broadcast a message. This decision constitutes a message broadcast request event $e_t$ that results in the invocation of the \ref{BroadcastMessage} procedure. If $\vgap = 0$ the message gets broadcast, which means that a set of message packets
$$
k_{P_1}, k_{P_2}, \dots, k_{P_n}
$$
one packet per contacted process, are queued to their respective channels. All of these packets share a single queuing event $e_q$:
$$
\psendevent{k_{P_1}} = \psendevent{k_{P_2}} = \dots = \psendevent{k_{P_n}} = e_q
$$
If $\vgap > 0$ then there is no queuing event. Therefore we have a transaction
$$
T = \trans_{e_t} =
\begin{cases}
\{ e_t \prec e_q \} & \text{if } \vgap = 0 \\
\{ e_t \} & \text{if } \vgap > 0
\end{cases}
$$
And the labeling yields
\begin{align*}
\lbl(e_t) & = \lblzzz{T} \\
\lbl(e_q) & = \lblzzg{T}{\lBCAST_1}
\end{align*}

In another example, suppose that the critical joining process $G$ processes a donation from $P$. The donated state from $P$ includes a record
$
\langle
\fp{msg}, \fp{index}, \fp{iset}
\rangle
$
in \wset{}. This record was placed in \wset{} by $P$ when it sent an untimely message packet $k_m = \pkm{\fp{msg}}$ to $D$. Suppose that the record causes $G$ to invoke the \ref{ReceiveMessage} procedure. When $G$ executes that call it queues an acknowledgement packet $k_s = \pks{\fp{msg}}$ to $P$. The queuing event $\psendevent{k_s}$ is labeled
$$
\lbl(\psendevent{k_s}) = \lblcgzg{P}{G}{\psendevent{k_m}}{\lACK}
$$

\begin{thm}
\label{LabelOrderThm}
The map
$
\lbl : \eventset^H \rightarrow \labelspace
$
is order preserving
\end{thm}

\begin{proof}
Let $e_1, e_2 \in \eventset^H$ be any events such that $e_1 \prec e_2$. Let $\lbl(e_1) = \lblgzg{e_1}{s_{e_1}}{f_{e_1}}$ and $\lbl(e_2) = \lblgzg{e_2}{s_{e_2}}{f_{e_2}}$. Let $e_1 \in \eventset_{P_1}$ and $e_2 \in \eventset_{P_2}$. 

If $P_1 \ne P_2$ then by the minimality of $\prec$ (\AxOrderII{}) there must be a trigger event $e_t$ at $P_2$ such that $e_1 \prec e_t \preceq e_2$. The events at $P_2$ are completely ordered by $\prec$ (\AxProcI{}). It follows from Definition \ref{TransDef} that $e_t \preceq \trig(e_2)$. By the same definition $\trig(e_1) \preceq e_1$. It follows that $\trig(e_1) \prec \trig(e_2)$ and therefore $\ell_{e_1} \,\dot{\prec}\, \ell_{e_2}$. Since $\lbl$ is left lexicographic it follows that $\lbl(e_1) < \lbl(e_2)$ and we are done in this case.

If $P_1 = P_2$ then it follows from Definition \ref{TransDef} that $\trig(e_1) \preceq \trig(e_2)$ and $\ell_{e_1} \dot{\preceq} \ell_{e_2}$. If the inequality is strict then we are done. So assume that $\ell_{e_1} = \ell_{e_2}$. This implies that $e_1$ and $e_2$ belong to the same transaction $T$ at process $P_1$.

We will go over each transaction type and verify that the events in the transaction are labeled with monotonically increasing labels. We rely on the analysis of side effects (see \ref{FirstSideEffectSSS} - \ref{LastSideEffectSSS}). We will denote by $e_t$ the trigger of $T$, and denote by $e_{s_1}, e_{s_2}, \dots e_{s_n}$ the side effects of $T$. Remember throughout the following analysis that $\lbl(e_t) = \lblzzz{T}$.

\begin{description}
\item[{\bf Message broadcast request transaction}] \hfill \\
According to Lemma \ref{AppEventSELem} an application transaction $T$ has one of the following forms:
\begin{itemize}
\item No side effects: $T = \{ e_t \}$ and there is nothing to prove.
\item A message packet multicast: $T = \{ e_t \prec e_{s_1} \}$ where $e_{s_1} = \psendevent{\pkm{\fp{msg}}}$
$$
\lblzzz{T} < \lblzzg{T}{\lBCAST_1}
$$
\end{itemize}

\item[{\bf Notification transaction}] \hfill \\
According to Lemma \ref{ViewChangeSELem} a notification transaction $T$ has one of the following forms:
\begin{itemize}
\item No side effects: $T = \{ e_t \}$ and there is nothing to prove.
\item A ghost packet multicast: $T = \{ e_t \prec e_{s_1} \}$ where $e_{s_1} = \psendevent{\pkh{v}}$
$$
\lblzzz{T} < \lblzzg{T}{\lGHOST}
$$
\item A ghost packet multicast followed by a flush packet multicast: $T = \{ e_t \prec e_{s_1} \prec e_{s_2} \}$ where $e_{s_1} = \psendevent{\pkh{v}}$ and $e_{s_1} = \psendevent{\pkf{v}}$
$$
\lblzzz{T} < \lblzzg{T}{\lGHOST} < \lblzzg{T}{\lFLUSH}
$$
\item A sequence of message packet multicasts: $T = \{ e_t \prec e_{s_1} \prec \dots \prec e_{s_n} \}$ where $e_{s_i} = \psendevent{\pkm{\fp{msg}_i}}$
$$
\lblzzz{T} < \lblzzg{T}{\lBCAST_1} < \dots < \lblzzg{T}{\lBCAST_n}
$$
\end{itemize}

\item[{\bf Message packet processing transaction}] \hfill \\
According to Lemma \ref{MessageAckSELem} a message packet processing transaction $T$ has the form $T = \{ e_t, e_{s_1} \}$ where $e_{s_1} = \psendevent{\pks{\fp{msg}}}$
$$
\lblzzz{T} < \lblzzg{T}{\lACK}
$$

\item[{\bf Acknowledgement packet processing transaction}] \hfill \\
According to Lemma \ref{MessageAckSELem} an acknowledgement packet processing transaction $T$ has one of the following forms:
\begin{itemize}
\item No side effects.
\item A ghost packet multicast.
\item A flush packet multicast.
\item A ghost packet multicast followed by a flush packet multicast.
\end{itemize}
The monotonicity of the labeling here works just as in the notification transaction case.

\item[{\bf Ghost packet processing transaction}] \hfill \\
According to Lemma \ref{GhostFlushSELem} a ghost packet processing transaction $T$ has no side effect so there is nothing to prove in this case.

\item[{\bf Flush packet processing transaction}] \hfill \\
According to Lemma \ref{GhostFlushSELem} a flush packet processing transaction $T$ has one of the following forms:
\begin{itemize}
\item No side effects: $T = \{ e_t \}$ and there is nothing to prove.
\item A sequence of message packet multicasts: $T = \{ e_t \prec e_{s_1} \prec \dots \prec e_{s_n} \}$ where $e_{s_i} = \psendevent{\pkm{\fp{msg}_i}}$
$$
\lblzzz{T} < \lblzzg{T}{\lBCAST_1} < \dots < \lblzzg{T}{\lBCAST_n}
$$
\end{itemize}

\item[{\bf Donation packet processing transaction}] \hfill \\
According to Lemma \ref{DonationSELem} a donation packet processing transaction $T$ has the form of a co-donation queuing event followed by a sequence of zero or more sub-transactions
$$
T = \lbrace e_t, e_c, e_{s_{1, 1}}, \dots, e_{s_{1, k_1}}, e_{s_{2, 1}}, \dots, e_{s_{2, k_2}}, \;\;\;\dots,\;\;\; e_{s_{n, 1}}, \dots, e_{s_{n, k_n}} \rbrace 
$$
where $e_c$ is the queuing event of the co-donation packet and $e_{s_{i, 1}}, \dots, e_{s_{i, k_i}}$ are the side effects of the $i^{th}$ sub-transaction.

Within each sub-transaction the sub-transaction coordinate of the label is the same for all side effects. Let
$$
\lbl(e_{i, j}) = \lblgzg{T}{s_i}{f_{i, j}}
$$
Then it follows from the previous cases that for each $i$ the labels are monotonic:
$$
\lblgzg{T}{s_i}{f_{i,1}} < \lblgzg{T}{s_i}{f_{i,2}} < \dots < \lblgzg{T}{s_i}{f_{i,k_i}}
$$
where $s_i$ is the sub-transaction coordinate of the $i^{th}$ sub-transaction. If the donation is not critical then $s_i = i$, and therefore $s_i < s_{i+1}$ and the whole sequence of sub-transactions is monotonically labeled. If the donation is critical then $s_i = \psendevent{k_i}$ where $k_i$ is the untimely packet that is related to the $i^{th}$ sub-transaction by Corollary \ref{DonOrderCor}. The same corollary guarantees that if $i < j$ then $\psendevent{k_i} \,\dot{\prec}\, \psendevent{k_j}$. Therefore in this case as well the whole sequence of sub-transactions is monotonically labeled. In addition
$$
\lbl(e_t) = \lblzzz{T} < \lblzzg{T}{\lCODONATE} = \lbl(e_c)
$$
and both packets have a sub-transaction coordinate that is equal to $\hat{0}$. Therefore the whole donation transaction is labeled monotonically in all cases.

\item[{\bf Co-donation packet processing transaction}] \hfill \\
According to Lemma \ref{DonationSELem} a co-donation packet processing transaction $T$ has one of the following forms:
\begin{itemize}
\item A sequence of zero or more sub-transactions
$$
T = \lbrace e_t, e_{s_{1, 1}}, \dots, e_{s_{1, k_1}}, e_{s_{2, 1}}, \dots, e_{s_{2, k_2}}, \;\;\;\dots,\;\;\; e_{s_{n, 1}}, \dots, e_{s_{n, k_n}} \rbrace 
$$
where $e_{s_{i, 1}}, \dots, e_{s_{i, k_i}}$ are the side effects of the $i^{th}$ sub-transaction. This case is resolved the same way as a donation transaction, with Corollary \ref{CoDonOrderCor} guaranteeing that the order of the sub-transaction coordinates is correct in the critical case.
\item A sequence of one or more message packet multicasts: $T = \{ e_t \prec e_{s_1} \prec \dots \prec e_{s_n} \}$ where $e_{s_i} = \psendevent{\pkm{\fp{msg}_i}}$. If the co-donation is non-critical the labeling looks as follows
$$
\lblzzz{T} < \lblzzg{T}{\lBCAST_1} < \dots < \lblzzg{T}{\lBCAST_n}
$$
If the co-donation is critical then according to Lemma \ref{CoDonFlushLem} there is a special post-critical flush packet queuing event $f \in \dot{\eventset}$ and the labeling is
$$
\lblzzz{T} < \lblgzg{T}{f}{\lBCAST_1} < \dots < \lblgzg{T}{f}{\lBCAST_n}
$$
\end{itemize}
\end{description}
\end{proof}

\subsection{Defining The History Reduction Mapping}
\label{HistReduxMapSS}
\subsubsection{Preliminaries}
\label{HistReduxPrelim}
We now show how to construct a history $H^r$ that carries the same computation as $H$ but where the critical join event of $G$ is replaced by a removal event of a different process $\minusG$. The interesting part is the construction of $\eventset^{H^r}$ which will rely on the label space that we constructed earlier. We start with the construction of most of the other components of $H^r$.
\begin{align*}
\processset^{H^r} & = \processset^H \cup \{ \minusG \} \\
\processset_h^{H^r} & = \processset_h^H \cup \{ \minusG \} \\
\numview^{H^r} & = \numview^H \\
\viewset_i^{H^r} & =
\begin{cases}
\viewset_i^H & \text{if } i \ge \vcrit  \\
\viewset_i^H \cup \{ G, \minusG \} & \text{if } i < \vcrit
\end{cases} \\
\requestset^{H^r} & = \requestset^H
\end{align*}

The main task is constructing the packets and related events on channels that involve the processes $G$ and $\minusG$. We build up these channels starting with the history $H$ as a base, and then adding and removing packets as necessary in an attempt to simulate what would have happened if $G$ and $\minusG$ were original members of the group and just happened, by some rare chance, to proceed (almost) exactly as the parent $D$ would up to the point of the critical view change, after which $\minusG$ is removed and $G$ remains to evolve as it originally did in the history $H$.

The rest of this section is dedicated to the construction of the channels that involve $G$ and $\minusG$. The construction involves a detailed description of the packets that are added to these channels, as well as the few packets that are removed. Together with each packet we add its associated events and their labels. In addition we describe the requisite changes to and additions of notification events. This construction amounts to a description of most of the missing ingredients in $H^r$ including all the channels and all the events. To complete the definition of $H^r$ we just need to add in the partial order on the events.

We construct the events in $H^r$ with the goal of simulating a reality where the following things happen:
\begin{itemize}
\item $G$ and $\minusG$ are group members from the start. However, their upper layer application does not get started until the critical view change. Therefore $G$ does not originate any message broadcasts up to the critical point, and $\minusG$ never originates any message broadcasts at all.
\item Other than message origination (and the related flush packets - see below), $G$ and $\minusG$ proceed exactly like $D$ up to the critical view change. This happens by luck, as multicasts arrive at $G$ and $\minusG$ at about the same time that they arrive at $D$, and multicasts and responses from $G$ and $\minusG$ arrive at other processes at about the same time that similar multicasts and responses arrive from $D$.
\item Despite their participation in group communications, $G$ and $\minusG$, by sheer luck, have no effect on the state of the group. This happens because acknowledgements from $G$ and $\minusG$ arrive too early (just before similar acknowledgements arrive from $D$) and therefore are never decisive in stabilizing messages; forwarded messages from $G$ and $\minusG$ arrive too late (just after similar forwarded messages arrive from $D$) and as a result they are discarded as duplicates; and flush packets from $G$ and $\minusG$ arrive too early (sometimes long before similar packets arrive from $D$) to be the deciding factor in view installation decisions at the receiving processes.
\item One area where $G$ and $\minusG$ are allowed to diverge from $D$ in a controlled fashion is with the multicasting of flushes. This is unavoidable because $G$ and $\minusG$ do not originate messages and therefore their respective wait sets tend to be emptier than the one at $D$. Therefore $G$ and $\minusG$ may be forced by the \cbcast{} protocol to multicast a flush long before $D$ is ready to do so. Managing this difference is the whole purpose of the ghost multicasts. These multicasts do not depend on the instability of original messages and therefore can happen at the same time at $G$, $\minusG$ and $D$. Flush multicasts out of $G$ and $\minusG$, however, do not happen at the same time that similar multicasts occur at $D$. Rather they immediately follow each ghost multicast.
\end{itemize}

\subsubsection{Some notation}
It should hopefully be intuitive at this point that the simulation is achieved, to a large degree, by adding {\em clone} packets that mimic on $\plusminusG$-bound channels what the original packets did on similar $D$-bound channels. To accommodate the fact that the new flush packets mimic original ghost packets rather than original flush packets, we also introduce the notion of a {\em zombie} packet. A zombie packet, as the name implies, is a clone of ghost a packet that comes back to life as a flush packet, but otherwise remains unchanged.

The label space $\labelspace$ plays a crucial role in ordering the new packet events that are created along with the new clone and zombie packets. Instead of ordering these events directly (which can get very complicated) we label them first. Then we introduce an order on all the events in $H^r$ by inducing it back from the labels.

The central challenge in building $H^r$ is the seam between the pre-critical period and the post-critical period. This problem presents itself in the form of packets that are sent pre-critically, but arrive post-critically (including packets that never arrive). Recall that such packets are called {\em untimely}. We deal with untimely packets using the donation and co-donation transactions. In $H^r$, the clones and zombies of untimely packets arrive exactly at the time that the receiving process begins the sub-transaction that simulates the arrival of the original packets (see labeled steps \ref{RD:STMessage} and \ref{RD:STAck} of the \ref{ReceiveDonation} procedure and labeled steps \ref{RCD:STMessage} and \ref{RCD:STAck} of the \ref{ReceiveCoDonation} procedure at \ref{PseudoCode}). This makes the dequeuing events of donation and co-donation packets pivotal in the definition of the reduction mapping, warranting a specific notation that we introduced in Definition \ref{CritDonDef} to describe them.

\begin{defn}
\label{EquivDef}
Let \fp{msg} be a stamped message in $H$. The {\bf reduction} of \fp{msg} is an identical message $\fp{msg}^r$, but with a slightly different stamping 
\begin{align*}
\orig{\fp{msg}^r} & = \orig{\fp{msg}}  \\
\mview{\fp{msg}^r} & = \mview{\fp{msg}}   \\
\mvt{\fp{msg}^r}[] & =
\begin{cases}
\mvt{\fp{msg}}[] \,\bigcup\, \left\{ [G] = 0, [\minusG] = 0 \right\} & \mview{\fp{msg}} < \vcrit \\
\mvt{\fp{msg}}[] & \mview{\fp{msg}} \ge \vcrit \\
\end{cases}
\end{align*}
For any data structure $A$ in $H$, the {\bf reduction} of $A$ is an identical data structure $A^r$ where all the messages contained in $A$ or its substructures have been replaced by their reductions.

If $A$ is a data structure in $H$ with reduction $A^r$ and $B$ is a data structure in $H^r$ such that $B = A^r$, we say that $B$ is {\bf equivalent} to $A$ and we denote it by $B \cong A$.
\end{defn}

\begin{defn}
\label{CloneZombieDef}
\begin{enumerate}
\item A {\bf clone} of a packet $k$ in $H$ is a new packet in $H^r$, on a different channel, whose type is identical to the type of $k$ and whose content is the reduction of the content $\cont(k)$. A {\bf clone} of an event $e$ in $H$ is a new event in $H^r$ of the same type as $e$. A clone event can occur at the same process or at a different one.
\label{CZD:CloneDef}
\item A {\bf zombie} of a ghost packet is a clone where, unlike a normal clone, the type of the packet changes from ghost to flush. A {\bf zombie} event is a clone of a ghost packet queuing event in $H$ whose type in $H^r$ a flush packet queuing event.
\item A {\bf non-$G$} process is any $P \ne \plusminusG$.
\end{enumerate}
\end{defn}

We now turn to the heart of the construction of $H^r$, which involves the construction of its channels and events. The bulk of the work is in constructing the channels and the related packet events, and then there is a bit of additional work in constructing the notification events. We proceed in multiple steps. We start with constructing new queuing events and their labels. We continue with channels where the source is $\plusminusG$ and the target is non-$G$, then channels where the source is non-$G$ and the target is $\plusminusG$, and then the four channels $\channel{\plusminusG}{\plusminusG}$. We construct new trigger events and their labels in tandem with the channel construction. Then we construct the new notification events and their labels, and finally we construct the $\prec^{H^r}$ order out of the labels.

\subsubsection{Constructing the new queuing events}
\label{NewQEventsSSS}
Most of the new queuing events are added to $\plusminusG$. The only new queuing events that are added to non-$G$ processes are the queuing events of acknowledgement packets that are sent in response to forwarded messages from $\plusminusG$. This should not be a surprise because we deliberately construct $H^r$ in such a way that the additional packets emanating out of $\plusminusG$ are largely ignored and therefore have no side effects to speak of.

We only create clone and zombie events for pre-critical queuing events in $H$. We are going to add post-critical {\em packets} to $H^r$, but their queuing events will all be existing queuing events of existing multicasts. In fact many, but not all, of the pre-critical cloned packets will become parts of existing multicasts and will not require a new queuing event.

If $e = \psendevent{k}$ is a pre-critical queuing event at $D$ in $H$ labeled $\lbl(e) = \lblzzg{T}{f}$ we add the following new events in $H^r$:
\begin{enumerate}
\item If $k$ is a forwarded message packet, or an acknowledgement packet in response to any message packet other than a forwarded message from $D$
\begin{enumerate}
\item A clone queuing event $e_{\plusc} \in \eventset_G$ with label $\lbl(e_{\plusc}) = \lbl(e)$
\item A clone queuing event $e_{\minusc} \in \eventset_{\minusG}$ with label $\lbl(e_{\minusc}) = \lbl(e)$
\end{enumerate}
\item If $k\ \in \channel{D}{D}$ is an acknowledgement packet in response to a forwarded message packet from $D$
\begin{enumerate}
\item Two clone queuing events $e^{\shftmsg}, e^{\minusshftmsg} \in \eventset_D$ with labels
\begin{align*}
\lbl(e^{\shftmsg}) & = \lblzgg{T}{\shftmsg{}}{\lACK} \\
\lbl(e^{\minusshftmsg}) & = \lblzgg{T}{\minusshftmsg{}}{\lACK}
\end{align*}
\item Three clone queuing events $e_{\plusc}$, $e^{\shftmsg}_{\plusc}, e^{\minusshftmsg}_{\plusc} \in \eventset_G$ with labels
\begin{align*}
\lbl(e_{\plusc}) & = \lbl(e) \\
\lbl(e^{\shftmsg}_{\plusc}) & = \lblzgg{T}{\shftmsg{}}{\lACK} \\
\lbl(e^{\minusshftmsg}_{\plusc}) & = \lblzgg{T}{\minusshftmsg{}}{\lACK}
\end{align*}
\item Three clone queuing events $e_{\minusc}$, $e^{\shftmsg}_{\minusc}, e^{\minusshftmsg}_{\minusc} \in \eventset_{\minusG}$ with labels
\begin{align*}
\lbl(e_{\minusc}) & = \lbl(e) \\
\lbl(e^{\shftmsg}_{\minusc}) & = \lblzgg{T}{\shftmsg{}}{\lACK} \\
\lbl(e^{\minusshftmsg}_{\minusc}) & = \lblzgg{T}{\minusshftmsg{}}{\lACK}
\end{align*}
\end{enumerate}
\item If $k$ is a ghost packet and $\lbl(e) = \lblzzg{T}{\lGHOST}$
\begin{enumerate}
\item A clone queuing event $e_{\plusc} \in \eventset_G$ with label $\lbl(e_{\plusc}) = \lbl(e) = \lblzzg{T}{\lGHOST}$
\item A clone queuing event $e_{\minusc} \in \eventset_{\minusG}$ with label $\lbl(e_{\minusc}) = \lbl(e) = \lblzzg{T}{\lGHOST}$
\item A zombie queuing event $e_{\plusz} \in \eventset_G$ with label $\lbl(e_{\plusz}) = \lblzzg{T}{\lFLUSH}$
\item A zombie queuing event $e_{\minusz} \in \eventset_{\minusG}$ with label $\lbl(e_{\minusz}) = \lblzzg{T}{\lFLUSH}$
\end{enumerate}
\item If $k$ is an original message packet or a flush packet, do not create  any clones
\end{enumerate}

If $e = \psendevent{k}$ is a pre-critical queuing event in $H$ at $P \ne D$ with label $\lbl(e) = \lblzzg{T}{\lACK}$ and if $k \in \channel{P}{D}$ is an acknowledgement packet, we add the following new events in $H^r$:
\begin{enumerate}
\item If $k$ is an acknowledgement packet in response to a forwarded message packet
\begin{enumerate}
\item A clone queuing event $e^{\shftmsg} \in \eventset_P$ with label $\lbl(e^{\shftmsg}) = \lblzgg{T}{\shftmsg{}}{\lACK}$
\item A clone queuing event $e^{\minusshftmsg} \in \eventset_P$ with label $\lbl(e^{\minusshftmsg}) = \lblzgg{T}{\minusshftmsg{}}{\lACK}$
\end{enumerate}
\item If $k$ is an acknowledgement packet in response to an original message packet, do not create any clones
\end{enumerate}

\subsubsection{Constructing the channels $\channel{P}{Q}$ for non-$G$ processes $P$ and $Q$}
\label{ChannelPQSSS}
The original $H$ channels are copied to $H^r$ almost without change. The only change is that each channel is reduced (see Definition \ref{EquivDef}). This means that the content $\cont(k)$ of each packet $k$ in the channel is replaced with the reduced content $\cont(k)^r$.

\subsubsection{Constructing the channels $\channel{P}{G}$ and $\channel{P}{(\minusG)}$ for a non-$G$ process $P$}
\label{ChannelPGSSS}
We start with an empty channel for $\channel{P}{(\minusG)}$ and with the reduction of the original channel $\channel{P}{G}^H$ for $\channel{P}{G}$ and {\em remove} the critical donation packet and its events, if they exist. Then we add the following new packets and events:
\begin{description}
\item[For] any timely packet $k \in \channel{P}{D}$
\begin{description}
\item[If] $k$ is a ghost, flush or message packet
\begin{enumerate}
\item A clone packet $c \in \channel{P}{G}$ together with $\psendevent{c} = \psendevent{k}$ and a new $\preceiveevent{c}$ event labeled as
$\lbl(\preceiveevent{c}) = \lbl(\preceiveevent{k})$
\item A clone packet $c' \in \channel{P}{(\minusG)}$ together with $\psendevent{c'} = \psendevent{k}$ and a new $\preceiveevent{c'}$ events labeled as
$\lbl(\preceiveevent{c'}) = \lbl(\preceiveevent{k})$
\end{enumerate}
\item[If] $k$ is an acknowledgement packet in response to a forwarded message
\begin{enumerate}
\item A clone packet $c \in \channel{P}{G}$ together with $\psendevent{c} = (\psendevent{k})^{\shftmsg}$ and a new $\preceiveevent{c}$ event labeled as
$\lbl(\preceiveevent{c}) = \lbl(\preceiveevent{k})$
\item A clone packet $c' \in \channel{P}{(\minusG)}$ together with $\psendevent{c'} = (\psendevent{k})^{\minusshftmsg}$ and a new $\preceiveevent{c'}$ event labeled as
$\lbl(\preceiveevent{c'}) = \lbl(\preceiveevent{k})$
\end{enumerate}
\item[If] $k$ is an acknowledgement packet in response to an original message do not create any clones.
\end{description}
\item[For] any untimely packet $k \in \channel{P}{D}$
\begin{description}
\item[If] $k$ is a ghost, flush or message packet
\begin{enumerate}
\item A clone packet $c \in \channel{P}{G}$  together with
\begin{enumerate}
\item $\psendevent{c} = \psendevent{k}$
\item If $\crittime{P}{G}$ exists, a new $\preceiveevent{c}$ event labeled as
$$
\lbl(\preceiveevent{c}) = \lblcgzz{P}{G}{\psendevent{k}}
$$
\end{enumerate}
\item A clone packet $c' \in \channel{P}{(\minusG)}$  together with $\psendevent{c'} = \psendevent{k}$ but no $\preceiveevent{c'}$ event
\end{enumerate}
\item[If] $k$ is an acknowledgement packet in response to a forwarded message
\begin{enumerate}
\item A clone packet $c \in \channel{P}{G}$ together with
\begin{enumerate}
\item $\psendevent{c} = (\psendevent{k})^{\shftmsg}$
\item If $\crittime{P}{G}$ exists, a new $\preceiveevent{c}$ event labeled as
$$
\lbl(\preceiveevent{c}) = \lblcgzz{P}{G}{\psendevent{k}}
$$
\end{enumerate}
\item A clone packet $c' \in \channel{P}{(\minusG)}$ together with $\psendevent{c'} = (\psendevent{k})^{\minusshftmsg}$ but no $\preceiveevent{c'}$ event

\end{enumerate}
\item[If] $k$ is an acknowledgement packet in response to an original message do not create any clones.
\end{description}
\end{description}

\subsubsection{Constructing the channels $\channel{G}{P}$ and $\channel{(\minusG)}{P}$ for a non-$G$ process $P$}
\label{ChannelGPSSS}
We start with the empty channel for $\channel{(\minusG)}{P}$ and with the reduction of the original channel $\channel{G}{P}^H$ for $\channel{G}{P}$ and {\em remove} the critical co-donation packet and its related events. Then we add the following new packets and events:
\begin{description}
\item[For] any timely packet $k \in \channel{D}{P}$ with $e = \psendevent{k}$ and $\lbl(\preceiveevent{k}) = \lblzzz{T'}$
\begin{description}
\item[If] $k$ is a forwarded message packet
\begin{enumerate}
\item A clone packet $\plusc \in \channel{G}{P}$ together with $\psendevent{\plusc} = e_{\plusc}$ and a new $\preceiveevent{\plusc}$ event labeled as
$\lbl(\preceiveevent{\plusc}) = \lblzgz{T'}{\shftmsg}$
\item A clone packet $\minusc \in \channel{(\minusG)}{P}$  together with $\psendevent{\minusc} = e_{\minusc}$ and a new $\preceiveevent{\minusc}$ event labeled as
$\lbl(\preceiveevent{\minusc}) = \lblzgz{T'}{\minusshftmsg}$
\end{enumerate}
\item[If] $k$ is a ghost packet
\begin{enumerate}
\item A clone packet $\plusc \in \channel{G}{P}$ together with $\psendevent{\plusc} = e_{\plusc}$ and a new $\preceiveevent{\plusc}$ event labeled as
$\lbl(\preceiveevent{\plusc}) = \lblzgz{T'}{\shftghost}$
\item A zombie packet $\plusz \in \channel{G}{P}$ together with $\psendevent{\plusz} = e_{\plusz}$ and a new $\preceiveevent{\plusz}$ event labeled as
$\lbl(\preceiveevent{\plusz}) = \lblzgz{T'}{\shftflush}$

\item A clone packet $\minusc \in \channel{(\minusG)}{P}$ together with $\psendevent{\minusc} = e_{\minusc}$ and a new $\preceiveevent{\minusc}$ event labeled as
$\lbl(\preceiveevent{\minusc}) = \lblzgz{T'}{\minusshftghost}$
\item A zombie packet $\minusz \in \channel{(\minusG)}{P}$ together with $\psendevent{\minusz} = e_{\minusz}$ and a new $\preceiveevent{\minusz}$ event labeled as
$\lbl(\preceiveevent{\minusz}) = \lblzgz{T'}{\minusshftflush}$
\end{enumerate}
\item[If] $k$ is an acknowledgement packet
\begin{enumerate}
\item A clone packet $\plusc \in \channel{G}{P}$ together with $\psendevent{\plusc} = e_{\plusc}$ and a new $\preceiveevent{\plusc}$ event labeled as
$\lbl(\preceiveevent{\plusc}) = \lblzgz{T'}{\shftack}$
\item A clone packet $\minusc \in \channel{(\minusG)}{P}$ together with $\psendevent{\minusc} = e_{\minusc}$ and a new $\preceiveevent{\minusc}$ event labeled as
$\lbl(\preceiveevent{\minusc}) = \lblzgz{T'}{\minusshftack}$
\end{enumerate}
\item[If] $k$ is a flush or an original message packet do not create any clones.
\end{description}

\item[For] any untimely packet $k \in \channel{D}{P}$ with $e = \psendevent{k}$
\begin{description}

\item[If] $k$ is an acknowledgement packet or a forwarded message packet
\begin{enumerate}
\item A clone packet $\plusc \in \channel{G}{P}$ together with $\psendevent{\plusc} = e_{\plusc}$ and if $\crittime{G}{P}$ exists, a new $\preceiveevent{\plusc}$ event labeled as $\lbl(\preceiveevent{\plusc}) = \lblcgzz{G}{P}{\psendevent{k}}$
\item A clone packet $\minusc \in \channel{(\minusG)}{P}$  together a $\psendevent{\minusc} = e_{\minusc}$ and no $\preceiveevent{\minusc}$ event
\end{enumerate}

\item[If] $k$ is a ghost packet
\begin{enumerate}
\item A clone packet $\plusc \in \channel{G}{P}$ together with $\psendevent{\plusc} = e_{\plusc}$ and if $\crittime{G}{P}$ exists, a new $\preceiveevent{\plusc}$ event labeled as $\lbl(\preceiveevent{\plusc}) = \lblcggz{G}{P}{\psendevent{k}}{\shftghost}$
\item A zombie packet $\plusz \in \channel{G}{P}$ together with $\psendevent{\plusz} = e_{\plusz}$ and if $\crittime{G}{P}$ exists, a $\preceiveevent{\plusz}$ event labeled as $\lbl(\preceiveevent{\plusz}) = \lblcggz{G}{P}{\psendevent{k}}{\shftflush}$
\item A clone packet $\minusc \in \channel{(\minusG)}{P}$ together with $\psendevent{\minusc} = e_{\minusc}$ and no $\preceiveevent{\minusc}$ event
\item A zombie packet $\minusz \in \channel{(\minusG)}{P}$ together with $\psendevent{\minusz} = e_{\minusz}$ and no $\preceiveevent{\minusz}$ event
\end{enumerate}

\item[If] $k$ is a flush or an original message packet do not create any clones.
\end{description}

\item[For] any post-critical packet $k \in \channel{G}{G}$ with $\lbl(\psendevent{k}) = \lblgzg{T}{j}{f}$ where $\view(T) < \vdeath{P}$ and, if $\crittime{P}{G}$ exists, $\ell_T \,\dot{\prec}\, \crittime{P}{G}$
\begin{description}

\item[If] $k$ is a ghost, flush or message packet, a clone packet $c \in \channel{G}{P}$ with $\psendevent{c} = \psendevent{k}$ and if $\crittime{G}{P}$ exists, a new $\preceiveevent{c}$ event labeled as $\lbl(\preceiveevent{c}) = \lblcgzz{G}{P}{\psendevent{k}}$
\item[If] $k$ is an acknowledgement, donation or co-donation packet, do not create any clones.

\end{description}
\end{description}

\subsubsection{Constructing the channels $\channel{(\plusminusG)}{(\plusminusG)}$}
\label{ChannelGGSSS}
We start with the empty channels for the three channels involving $\minusG$ and with the reduction of the original channel $\channel{G}{G}^H$ for $\channel{G}{G}$. Then we add the following new packets and events (notice that by the \AxProcIV{} there are no untimely packets in this case):
\begin{description}
\item[For] any timely packet $k \in \channel{D}{D}$ with $e = \psendevent{k}$ and $\lbl(\preceiveevent{k}) = \lblzzz{T'}$
\begin{description}
\item[If] $k$ is a forwarded message packet
\begin{enumerate}
\item A clone packet $\plusc \in \channel{G}{G}$ together with $\psendevent{\plusc} = e_{\plusc}$ and a new $\preceiveevent{\plusc}$ event labeled as
$\lbl(\preceiveevent{\plusc}) = \lblzgz{T'}{\shftmsg}$
\item A clone packet $\plusc' \in \channel{G}{(\minusG)}$ together with $\psendevent{\plusc'} = e_{\plusc}$ and a new $\preceiveevent{\plusc'}$ event labeled as
$\lbl(\preceiveevent{\plusc'}) = \lblzgz{T'}{\shftmsg}$
\item A clone packet $\minusc \in \channel{(\minusG)}{G}$ together with $\psendevent{\minusc} = e_{\minusc}$ and a new $\preceiveevent{\minusc}$ event labeled as
$\lbl(\preceiveevent{\minusc}) = \lblzgz{T'}{\minusshftmsg}$
\item A clone packet $\minusc' \in \channel{(\minusG)}{(\minusG)}$ together with $\psendevent{\minusc'} = e_{\minusc}$ and a new $\preceiveevent{\minusc'}$ event labeled as
$\lbl(\preceiveevent{\minusc'}) = \lblzgz{T'}{\minusshftmsg}$
\end{enumerate}

\item[If] $k$ is an acknowledgement packet for a forwarded message
\begin{enumerate}
\item A clone packet $\plusc \in \channel{G}{G}$ together with $\psendevent{\plusc} = e^{\shftmsg}_{\plusc}$ and a new $\preceiveevent{\plusc}$ event labeled as
$\lbl(\preceiveevent{\plusc}) = \lblzgz{T'}{\shftack}$
\item A clone packet $\plusc' \in \channel{G}{(\minusG)}$ together with $\psendevent{\plusc'} = e^{\minusshftmsg}_{\plusc}$ and a new $\preceiveevent{\plusc'}$ event labeled as
$\lbl(\preceiveevent{\plusc'}) = \lblzgz{T'}{\shftack}$
\item A clone packet $\minusc \in \channel{(\minusG)}{G}$ together with $\psendevent{\minusc} = e^{\shftmsg}_{\minusc}$ and a new $\preceiveevent{\minusc}$ event labeled as
$\lbl(\preceiveevent{\minusc}) = \lblzgz{T'}{\minusshftack}$
\item A clone packet $\minusc' \in \channel{(\minusG)}{(\minusG)}$ together with $\psendevent{\minusc'} = e^{\minusshftmsg}_{\minusc}$ and a new $\preceiveevent{\minusc'}$ event labeled as
$\lbl(\preceiveevent{\minusc'}) = \lblzgz{T'}{\minusshftack}$
\end{enumerate}

\item[If] $k$ is a ghost packet
\begin{enumerate}
\item A clone packet $\plusc \in \channel{G}{G}$ together with $\psendevent{\plusc} = e_{\plusc}$ and a new $\preceiveevent{\plusc}$ event labeled as
$\lbl(\preceiveevent{\plusc}) = \lblzgz{T'}{\shftghost}$
\item A zombie packet $\plusz \in \channel{G}{G}$ together with $\psendevent{\plusz} = e_{\plusz}$ and a new $\preceiveevent{\plusz}$ event labeled as
$\lbl(\preceiveevent{\plusz}) = \lblzgz{T'}{\shftflush}$

\item A clone packet $\plusc' \in \channel{G}{(\minusG)}$ together with $\psendevent{\plusc'} = e_{\plusc}$ and a new $\preceiveevent{\plusc'}$ event labeled as
$\lbl(\preceiveevent{\plusc'}) = \lblzgz{T'}{\shftghost}$
\item A zombie packet $\plusz' \in \channel{G}{(\minusG)}$ together with $\psendevent{\plusz'} = e_{\plusz}$ and a new $\preceiveevent{\plusz'}$ event labeled as
$\lbl(\preceiveevent{\plusz'}) = \lblzgz{T'}{\shftflush}$

\item A clone packet $\minusc \in \channel{(\minusG)}{G}$ together with $\psendevent{\minusc} = e_{\minusc}$ and a new $\preceiveevent{\minusc}$ event labeled as
$\lbl(\preceiveevent{\minusc}) = \lblzgz{T'}{\minusshftghost}$
\item A zombie packet $\minusz \in \channel{(\minusG)}{G}$ together with $\psendevent{\minusz} = e_{\minusz}$ and a new $\preceiveevent{\minusz}$ event labeled as
$\lbl(\preceiveevent{\minusz}) = \lblzgz{T'}{\minusshftflush}$

\item A clone packet $\minusc' \in \channel{(\minusG)}{(\minusG)}$ together with $\psendevent{\minusc'} = e_{\minusc}$ and a new $\preceiveevent{\minusc'}$ event labeled as
$\lbl(\preceiveevent{\minusc'}) = \lblzgz{T'}{\minusshftghost}$
\item A zombie packet $\minusz' \in \channel{(\minusG)}{(\minusG)}$ together with $\psendevent{\minusz'} = e_{\minusz}$ and a new $\preceiveevent{\minusz'}$ event labeled as
$\lbl(\preceiveevent{\minusz'}) = \lblzgz{T'}{\minusshftflush}$

\end{enumerate}

\end{description}
\end{description}

\subsubsection{Constructing the notification events}
The notification events and their labels are mostly left unchanged from $H$. We start out with $\notificationset^{H^r} = \notificationset^H$ and then apply the following additions and substitutions:
\begin{description}

\item[We] add notifications $\notify{0}{G}$ and $\notify{0}{\minusG}$ to $\notificationset^{H^r}$ with
\begin{align*}
\cont(\notify{0}{G}) & = \nn{G}  \\
\cont(\notify{0}{\minusG}) & = \nn{\minusG}
\end{align*}
\item[For] any view $0 < i < \vcrit$ we add notifications $\notify{i}{G}$ and $\notify{i}{\minusG}$ to $\notificationset^{H^r}$ with
$$
\cont(\notify{i}{G}) = \cont(\notify{i}{\minusG}) = \cont(\notify{i}{D})
$$
\item[For] any non-$G$ process $P$ we replace the contents of the critical notification
$$
\cont(\notify{\vcrit}{P}) = \nj{G}{D}
$$
with
$$
\cont(\notify{\vcrit}{P}) = \nr{\minusG}
$$
\item[We] replace the contents of the critical notification
$$
\cont(\notify{\vcrit}{G}) = \nn{G}
$$
with
$$
\cont(\notify{\vcrit}{G}) = \nr{\minusG}
$$
\item[We] add a notification $\notify{\vcrit}{\minusG}$ to $\notificationset^{H^r}$ with
$$
\cont(\notify{\vcrit}{\minusG}) = \nd{}
$$
\item[For] any view $0 \le i < \vcrit$ we add notification events $\pnotifyevent{i}{G}$ and $\pnotifyevent{i}{\minusG}$ to $\eventset_G$ and $\eventset_{\minusG}$ respectively, labeled as $\lbl(\pnotifyevent{i}{G}) = \lbl(\pnotifyevent{i}{\minusG}) = \lblzzz{i}$
\end{description}

\subsubsection{The partial order $\prec^{H^r}$ and dropped items}
\label{ReducedHistoryOrder}
We now construct the order $\prec$ in $H^r$. We will occasionally denote it by $\prec^{H^r}$. As we mentioned before, we use the partial order on labels in the construction.

The relation $\prec^{H^r}$ is defined as the transitive closure of the following primitive order relations:
\begin{enumerate}
\item For each $k \in \packetset^{H^r}$:
$$
\psendevent{k} \prec^{H^r} \preceiveevent{k}
$$
whenever the dequeuing event exists.
\item For each process $X \in \processset^{H^r}$ and any two events $e_1, e_2 \in \eventset^{H^r}_X$:
$$
e_1 \prec^{H^r} e_2
$$
if $\lbl(e_1) < \lbl(e_2)$
\item For any parent/child pair $E/J$ in $H$, other than the critical pair $D/G$:
$$
\pnotifyevent{j(J)}{E} \asymp^{H^r} \joinevent{J}
$$
\end{enumerate}

Please note that this definition produces a well defined {\em relation} $\prec^{H^r}$, but it will require some work to show that it is actually a weak partial order. We will pursue that investigation a little later.

We need to define for each new packet whether it is sent and received, and for each new notification whether it is dropped. We declare that there are no dropped notifications in $H^r$. As for packets, the packets for which we added a dequeuing event must be received, and therefore must be sent. But there are packets for which we did not add a dequeuing event. Since we do not need $H^r$ to be transactional or even lossless, but only conforming, we simply declare that any clone or zombie for which there is no dequeuing event is dropped, meaning that it is sent but not received.

\subsubsection{Implementing the user application interface}
\label{ReducedUAI}
To complete the construction of $H^r$ we have to insulate the user application from any knowledge of the fact that group history has changed. In order to do that we have to change the implementation of the user application interface (see \ref{UserApplication}), essentially creating a thin wrapper around it that hides the effects of the history changes. As part of the wrapper we add two new variables, \ulview{} and \ulcount{}, to the \uldata{} structure, and we make use of three "magic numbers" that are defined below. The new variables count how many views were installed and how many messages were delivered so far, in order to discover the boundary between the pre-critical and post-critical time that would otherwise be invisible to the user application.

\begin{defn}
\label{MagicDef}
\begin{description}
\item[$\magicview$] is the value of \cview{} at $D$ in $H$ at the critical time.
\item[$\magicmessage$] is the number of messages that are delivered at $D$ in $H$ by the critical time.
\item[$\magicsize$] is the number of members of view zero in $H$.
\end{description}
\end{defn} 

We define two different implementations of the interface in $H^r$, one for $\plusminusG$ and one for all other processes. As usual, we will use $P$ to denote any process which is not $\plusminusG$. We use the suffixes "@P" and "@G" to differentiate between the two implementations. We also use the marker "$^r$" to indicate the $H^r$ version of an up-call.

\begin{procedure}[H]
\caption{GroundState$^r$@P()}
\label{GroundStateP}
\SetAlgoNoLine
\Indp
\Indp
\BlankLine
\Let $\ulview{} = {-\magicsize{}} - 2$\;
\Let $\ulcount{} = 0$\;
GroundState@P()\tcp*[l]{call the original $H$ version}
\end{procedure}

\begin{procedure}[H]
\caption{GroundState$^r$@G()}
\label{GroundStateG}
\SetAlgoNoLine
\Indp
\Indp
\BlankLine
\Let $\ulview{} = -\magicsize{} - 2$\;
\Let $\ulcount{} = 0$\;
GroundState@D()\tcp*[l]{call the original $H$ version at process $D$} 
\end{procedure}

\begin{procedure}[H]
\caption{ApplyMessage$^r$@P($\fp{msg}, \fp{originator}$)}
\label{ApplyMessageP}
\SetAlgoNoLine
\Indp
\Indp
\BlankLine
ApplyMessage@P(\fp{msg}, \fp{originator})\;
\end{procedure}

\newpage

\begin{procedure}[H]
\caption{ApplyMessage$^r$@G($\fp{msg}, \fp{originator}$)}
\label{ApplyMessageG}
\SetAlgoNoLine
\Indp
\Indp
\BlankLine
\Increment \ulcount{}\;
\If{$\ulcount{} \le \magicmessage$}
{
  ApplyMessage@D(\fp{msg}, \fp{originator})\tcp*[l]{call the original $D$ version}
}
\Else
{
  ApplyMessage@G(\fp{msg}, \fp{originator})\tcp*[l]{call the original $G$ version}
}
\end{procedure}

\begin{procedure}[H]
\caption{ApplyJoin$^r$@P(\fp{pid})}
\label{ApplyJoinP}
\SetAlgoNoLine
\Indp
\Indp
\BlankLine
\Increment \ulview{}\;
\If{$\ulview{} \ne 0$  and $\ulview \ne -1$}
{
  ApplyJoin@P(\fp{pid})\;
}
\end{procedure}

\begin{procedure}[H]
\caption{ApplyJoin$^r$@G(\fp{pid})}
\label{ApplyJoinG}
\SetAlgoNoLine
\Indp
\Indp
\BlankLine
\Increment \ulview{}\;
\If{$\ulview{} > \magicview$}
{
  ApplyJoin@G(\fp{pid})\;
}
\ElseIf{$\ulview{} \ne 0$  and $\ulview \ne -1$}
{
  ApplyJoin@D(\fp{pid})\;
}
\end{procedure}

\begin{procedure}[H]
\caption{ApplyRemoval$^r$@P(\fp{pid})}
\label{ApplyRemovalP}
\SetAlgoNoLine
\Indp
\Indp
\BlankLine
\Increment \ulview{}\;
\If{$\ulview{} = \vcrit$}
{
  ApplyJoin@G(G)\tcp*[l]{\fp{pid} is equal to $\minusG$}
}
\Else
{
  ApplyRemoval@P(\fp{pid})\;
}
\end{procedure}

\newpage

\begin{procedure}[H]
\caption{ApplyRemoval$^r$@G(\fp{pid})}
\label{ApplyRemovalG}
\SetAlgoNoLine
\Indp
\Indp
\BlankLine
\Increment \ulview{}\;
\If{$\ulview{} \le \magicview$}
{
  ApplyRemoval@D(\fp{pid})\;
}
\ElseIf{$\ulview{} = \vcrit$}
{
  ApplyJoin@G(G)\tcp*[l]{\fp{pid} is equal to $\minusG$}
}
\Else
{
  ApplyRemoval@G(\fp{pid})\;
}
\end{procedure}

\subsection{Basic Properties Of $H^r$}
Now that we have defined the reduced history $H^r$, we establish its basic properties. We show that $H^r$ is a history according to Definition \ref{HistoryDef}, and indeed a conforming history according to Definition \ref{ConformingDef}. Only after we do that can we show that $H^r$ is the history of a valid \cbcast{} execution and that $H^r$ carries the same computation as $H$.

\begin{defn}
For any packet $k \in \packetset^{H^r}$ define
$$
\origpk(k) =
\begin{cases}
k \quad \text{if } k \text{ is the reduction of an original } \packetset^H \text{ packet} \\
j \quad \text{if } k \text{ is a clone or zombie of the original packet } j
\end{cases}
$$
\end{defn}

\begin{thm}
\label{ReducedHistoryThm}
$H^r$ is a history.
\end{thm}

We start with a lemma about the relation $\prec^{H^r}$
\begin{lem}
\label{ReducedOrderLem}
If $e$ and $f$ are packet events in $H^r$ and $e \prec f$ then $\lbl(e) < \lbl(f)$
\end{lem}
\begin{proof}
The relation $\prec$ on the packet events in $H^r$ is the transitive closure of the three primitive relations defined in \ref{ReducedHistoryOrder}. If we show that each of the primitive relations is compatible with the label order then we are done.

The second primitive relation is compatible with label order by construction, and the third primitive relation is compatible because both notification events share the same label $\lblzzz{j(J)}$ so the only difficulty is showing that for every packet $k' \in \packetset^{H^r}$ for which $\preceiveevent{k'}$ exists, $\lbl(\psendevent{k'}) < \lbl(\preceiveevent{k'})$.

If the packet $k'$ is an original $H$ packet, then this property follows from Theorem \ref{LabelOrderThm}. We only need to verify this property for packets $k'$ that are either clones or zombies of an original packet $k$. Looking through the construction of $H^r$ packets, one can easily observe that one of the following three possibilities must occur:
\begin{itemize}

\item $k$ is a timely packet with $\lbl(\psendevent{k}) = \lblggg{T}{*}{*}{*}$ and $\lbl(\preceiveevent{k}) = \lblggg{T'}{*}{*}{*}$ where $\ell_T \,\dot{\prec}\, \ell_{T'}$. In that case for any clone or zombie $k'$ of $k$, $\lbl(\psendevent{k'}) = \lblggg{T}{*}{*}{*}$ and $\lbl(\preceiveevent{k'}) = \lblggg{T'}{*}{*}{*}$. Therefore $\lbl(\psendevent{k'}) < \lbl(\preceiveevent{k'})$.

\item $k$ is an untimely packet with $\lbl(\psendevent{k}) = \lblggg{T}{*}{*}{*}$ where $\ell_T \,\dot{\prec}\, \ell_{\vcrit}$. In that case $\lbl(\psendevent{k'}) = \lblggg{T}{*}{*}{*}$ and if $\preceiveevent{k'}$ exists then $\lbl(\preceiveevent{k'}) = \lblggg{T'}{*}{*}{*}$ where $\ell_{T'} = \crittime{P}{G}$ or $\ell_{T'} = \crittime{G}{P}$ for some original process $P$. In either case $\ell_{T'} \,\dot{\succeq}\, \ell_{\vcrit}$ and therefore $\lbl(\psendevent{k'}) < \lbl(\preceiveevent{k'})$.

\item $k$ is packet in $\channel{G}{G}$ with $\lbl(\psendevent{k}) = \lblggg{T}{*}{*}{*}$ where
$\ell_{\vcrit} \,\dot{\prec}\, \ell_T \,\dot{\prec}\, \crittime{P}{G}$
for some original process $P$. In that case $\lbl(\psendevent{k'}) = \lblggg{T}{*}{*}{*}$  and $\lbl(\preceiveevent{k'}) = \lblcgzz{G}{P}{\psendevent{k}}$. If $c$ is the co-donation packet sent from $G$ to $P$ in the original history $H$ then
\begin{align*}
 & \lbl(\psendevent{k'}) < \lblczzz{P}{G} < \lblczzg{P}{G}{\lCODONATE} = \lbl(\psendevent{c}) < \\
< & \lbl(\preceiveevent{c}) = \lblczzz{G}{P} < \lblcgzz{G}{P}{\psendevent{k}} = \lbl(\preceiveevent{k'})
\end{align*}
\end{itemize}
This concludes the proof.
\end{proof}

\begin{cor}
\label{CompEventCor}
Let  $e,f \in \eventset_P^{H^r}$ be any two events at a process $P$ in $H^r$. Then $e \prec f$ if and only if $\lbl(e) < \lbl(f)$.
\end{cor}

\begin{proof}
Lemma \ref{ReducedOrderLem} proves one direction. The other direction follows immediately from the definition of $\prec^{H^r}$ in \ref{ReducedHistoryOrder}.
\end{proof}

\begin{lem}
\label{SEContactSetLem}
Let $e \in \eventset^{H^r}$ be the queuing event of a message, ghost or flush packet $k$ and let $T_e$ be its target set (refer to the \AxProcI{} for the definition of target set). Let $e_0 = \psendevent{\origpk(k)}$. Then
\begin{enumerate}
\item if $e_0$ is a pre-critical event that occurred at a process $R$ then
$$
T_e = \valuepre{\cset{}}{R}{e_0} \cup \{ \plusminusG \}
$$
\item if $e_0$ is a post-critical event that occurred at $G$ then
$$
T_e = \valuepre{\lset{}}{G}{e_0}
$$
\item if $e_0$ is a post-critical event that occurred at $R \ne G$ then
$$
T_e = \valuepre{\cset{}}{R}{e_0}
$$
\end{enumerate}
\end{lem}

\begin{proof}
First take a quick look at the pseudo code to verify that every multicast of message, ghost and flush packets has a target set that is equal to \cset{}.

The claims follow from a careful examination of the construction of events in $H^r$. 

Let $R \in \processset^H$ and $R \ne G$. Post-critically we do not add any new clones or zombies to packets emanating from $R$ so in this case $\origpk(k) = k$ and $T_e$ is the same as the original target set, which proves the third part.

Pre-critically we do add clones and zombies to message, ghost and flush packets emanating from $R$. If $R \ne D$ then all of these clones (there are no zombies) also emanate from $R$. If $R = D$ then some of the clones and zombies emanate from $\plusminusG$.

Every pre-critical multicast set of message, ghost or flush packets in $R$ contains exactly one packet that is targeted at $D$. This packet gets cloned with two new targets $\plusminusG$ and added to the same multicast set. This proves the first part in this case when $k$ emanates from $R$.

If $R = D$ then for every forwarded message multicast there is a cloned multicast emanating from $\plusminusG$ that contains the complete original target set with the addition of $\plusminusG$. Original message and flush multicasts do not get cloned to $\plusminusG$. Each ghost multicast is cloned and zombied to two consecutive multicasts, both of which have the required target set. This takes care of the first part.

As for $G$, to every post-critical multicast we add packets destined to the uncontacted processes. It follows from Lemma \ref{OmnibusCBCASTLem}(\ref{OCL:lset}) that this makes the target set equal to \lset{}.
\end{proof}

The following is a technical lemma that is required for proving that the channels in $H^r$ are FIFO.

\begin{lem}
\label{FIFOLem}
Let $X$ and $Y$ be processes in $H^r$ and let $k_1, k_2 \in \channel{X}{Y}^{H^r}$ be distinct pre-critical packets such that $\psendevent{k_1} \prec \psendevent{k_2}$. Then
$\psendevent{\origpk(k_1)} \preceq \psendevent{\origpk(k_2)}$, with equality occurring exactly when $\origpk(k_1)=\origpk(k_2)$ is a pre-critical ghost packet and $X = \plusminusG$.
\end{lem}

\begin{proof}
If both $k_1$ and $k_2$ are original, there is nothing to prove. The case where one of them is original and one is a clone or zombie is not possible because all the pre-critical clones and zombies belong to channels that involve $\plusminusG$ while none of the original pre-critical packets do. Therefore we can assume that both $k_1$ and $k_2$ are either clones or zombies of their original packets.

Let $\lbl(\psendevent{k_i}) = \lblggg{T_i}{a_i}{b_i}{c_i}$. By Corollary \ref{CompEventCor} we know that $\lbl(\psendevent{k_1}) < \lbl(\psendevent{k_2})$ and therefore $\trig(T_1) \preceq \trig(T_2)$ . By going through the different cases of \ref{ChannelPGSSS}, \ref{ChannelGPSSS} and \ref{ChannelGGSSS} it is easy to verify that for any pre-critical packet $k$ with $\lbl(\psendevent{k}) = \lblggg{T}{a}{b}{c}$, the event $\psendevent{\origpk(k)}$ has the same constellation label $\ell_T$. We also know that that for any event $e$ in $H$ the constellation label is $\ell_{\trans(e)}$ (see \ref{LabelingHSSS}). Therefore $\trig(T_i) = \trig(\psendevent{\origpk(k_i)})$.

If $\trig(T_1) \prec \trig(T_2)$ then by Definition \ref{TransDef} $\psendevent{\origpk(k_1)} \prec \psendevent{\origpk(k_2)}$ and we are done. Otherwise
$T_1 = T_2$ so $\origpk(k_1)$ and $\origpk(k_2)$ are queued in the same transaction. Denote $T = T_1 = T_2$.
We also know that $\origpk(k_1)$ and $\origpk(k_2)$ belong to the same channel. This channel is obtained from $\channel{X}{Y}$ by replacing any occurrence of $\plusminusG$ with $D$.

Going over the complete list of different possibilities of transaction content for $T$ we have:

\begin{itemize}
\item $T$ comprises a trigger plus one or more multicasts of original or forwarded message packets. This happens if $T$ is an message broadcast request transaction, or if $T$ is a notification transaction that resulted in the forwarding of messages out of \fque{}[P], or if $T$ is a flush packet transaction that caused the installation of a new view and the broadcasting of one or more messages out of \lque{}. In these cases $\origpk(k_i)$ are both pre-critical message packets. Going over the different cases in \ref{ChannelPGSSS} - \ref{ChannelGGSSS} shows that in this case $\lbl(\psendevent{k_i}) = \lbl(\psendevent{\origpk(k_i)})$ and the lemma follows.

\item $T$ comprises a trigger and the queuing of exactly one acknowledgement packet $k$. This happens if $T$ is triggered by the processing of a message packet. In this case it must be that $\origpk(k_1) = \origpk(k_2) = k$. But an examination of \ref{ChannelPGSSS} - \ref{ChannelGGSSS} shows that distinct clones of pre-critical acknowledgement packets belong to different channels, contrary to the assumption that $k_1$ and $k_2$ belong to the same channel. So this case does not occur.

\item $T$ comprises a trigger followed by a single multicast of ghost packets. This happens if $T$ is triggered by a the processing of an acknowledgement packet that cleared out \fwset{} while \bwset{} remained non-empty, or if $T$ is triggered by a $\nr{P}$ notification event that occurred when \fque{}[P] and \fwset{} were empty while \bwset{} was not empty. In this case there is a single packet at each relevant channel, and therefore $\origpk(k_1) = \origpk(k_2)$. An examination of \ref{ChannelPGSSS} - \ref{ChannelGGSSS} shows that the clones and zombies of a pre-critical ghost packet reside on the same channel exactly when the original packet is sent out of the process $D$, in which case it generates one clone and one zombie on each relevant channel. In this case $X=\plusminusG$ and the labeling forces the clone to precede the zombie, and therefore in this case $k_1$ is the clone and $k_2$ is the zombie of the same original packet.

\item $T$ comprises a trigger followed by a single multicast of flush packets. This happens if $T$ is triggered by the processing of an acknowledgement packet that cleared out \bwset{} while \fwset{} was already empty. In this case there is a single packet at each relevant channel and therefore $\origpk(k_1) = \origpk(k_2)$. An examination of \ref{ChannelPGSSS} - \ref{ChannelGGSSS} shows that distinct clones of a pre-critical flush packet reside on different channels, contrary to our assumption that $k_1$ and $k_2$ reside on the same channel. So this case does not occur.

\item $T$ comprises a trigger followed by a single multicast of ghost packets, followed by a single multicast of flush packets. This happens if $T$ is triggered by the processing of an acknowledgement packet that cleared out \fwset{} while \bwset{} was already empty, or if $T$ is triggered by the processing of a $\nr{P}$ notification that occurred when \fque{}[P], \fwset{} and \bwset{} were all empty. In this case there is in each relevant channel a single ghost packet followed by a single flush packet. There are two possibilities here. If $T$ occurs in process $D$ (in which case $X = \plusminusG$) then the flush packet generates no clones and the ghost packet is the original of both $k_1$ and $k_2$ which are a clone and zombie, respectively, of their shared original. If $X \ne \plusminusG$ then the ghost packet and the flush packet each generate a single clone (but no zombie) on each relevant packet, where the labeling forces the ghost clone to precede the flush clone. Therefore $k_1$ is the ghost clone and $k_2$ is the flush clone, and the clones are sent in the same order as their originals.

\item $T$ comprises a trigger and no other event. This happens if $T$ is triggered by the processing of a ghost packet, or if $T$ is triggered by a flush packet that did not cause the installation of a new view, or if $T$ is triggered by a flush packet that did cause the installation of a new view while \lque{} was empty, or if $T$ was triggered by the processing of an acknowledgement packet such that \fwset{} was not empty when processing started and remained non-empty when processing concluded, or if $T$ was triggered by an acknowledgement packet for an original message packet such that \fwset{} was empty and \bwset{} was not empty when processing started, and \bwset{} remained non-empty when processing concluded, or if $T$ is triggered by the processing of a $\nr{P}$ notification that occurred when \fque{}[P] was empty while \fwset{} was not empty. This case does not occur, of course, since we know that are least $\origpk(k_1)$ must have been queued to be sent as part of $T$.
\end{itemize}
\end{proof}

\begin{proof}[Proof of Theorem \ref{ReducedHistoryThm}]
Some of the axioms are easy to check. The \AxViewI{} and the \AxViewII{} are trivial. The \AxPackEventI{} and the \AxPackEventII{} are true by construction. The \AxNotEventI{} and the \AxNotEventIII{} are trivial. The \AxGMSI{} follows directly from the construction of notifications in $H^r$.

The \AxNotEventII{} says that view notifications  are processed in order. This property in induced directly from $H$ for all processes other than $\plusminusG$.
In the case of $\minusG$ the axiom follows because we explicitly added a processing event for each notification except for the last critical one.
In the case of $G$, all the post-critical notification processing events (including the critical one) exist because $H$ is transactional, while all the pre-critical ones are added explicitly. The processing order is induced directly from the labels of the events.

The \AxProcI{} has two parts. The first part claims that events within a single process are linearly ordered. The second part claims that multicast sets are finite.
To prove the first part, notice that at each process $P$ the constellation coordinate is made up of clean events at process $P$ in $H$, and this set is linearly ordered since the \AxProcI{} holds in $H$. Within each constellation the other coordinates are also taken from linearly ordered sets: this is trivial for the adjustment and side effect coordinates since the sets $\aspace$ and $\fspace$ are linearly ordered. As for the sub-transaction coordinate, we have the following cases:
\begin{itemize}
\item Non-critical donation and co-donation constellations have the same events in $H^r$ as they do in $H$, and have the same labels. Since the labels are linearly ordered in $H$, they are linearly ordered in $H^r$.
\item In a critical donation constellation, the side-effect events in $H^r$ consist of the sub-transaction events from $H$ (the co-donation queuing event is removed) and have additional trigger events. All of these events get their sub-transaction coordinates from the space $\{\psendevent{k} \in \dot{\eventset}^H \,\Vert\, k \in \channel{P}{D}^H \}$, where $P$ is the sender of the donation packet. This space is linearly ordered by the \AxProcI{} in $H$.
\item In a critical co-donation constellation, the side-effect events in $H^r$ consist of all the side-effect events in $H$ with the addition of new trigger events. All of these events get their sub-transaction coordinates from one of two sets:
\begin{align*}
C_1 & = \{\psendevent{k} \,\dot{\prec}\, \ell_{\vcrit} \in \dot{\eventset^H} \,\Vert\, k \in \channel{D}{P}^H \}  \\
C_2 & = \{\psendevent{k} \,\dot{\succ}\, \ell_{\vcrit} \in \dot{\eventset^H} \,\Vert\, k \in \channel{G}{G}^H \}
\end{align*}
Each of these sets is linearly ordered by the \AxProcI{} in $H$. Taken together they are still linearly ordered because all the elements in $C_1$ precede $\ell_{\vcrit}$ and all the elements in $C_2$ succeed it.
\item In a non-donation constellation, all the events have the $\hat{0}$ value for the sub-transaction coordinate. 
\end{itemize}
It follows that the $\lbl$ function linearly orders $\eventset^{H^r}_P$ at each process $P$. By Corollary \ref{CompEventCor} the first part of the axiom now follows.

The second part of the axiom follows immediately from Lemma \ref{SEContactSetLem}.

To verify the \AxProcII{} and the \AxProcIII{}, notice that they trivially hold for packets exchanged between non-$G$ processes. As for original packets exchanged between $G$ and non-$G$ processes, almost nothing changes except that from the point of view of the non-$G$ processes in $H^r$, process $G$ had been a member from the start. This does not invalidate either axiom.

The clones and zombies of timely packets are either queued at $\plusminusG$ or dequeued at $\plusminusG$ at essentially the same time that the original is queued or dequeued at $D$ (any differences in the label of the original versus the clone or zombie occur at the second or third coordinate). The same is true of untimely clones and zombies as far as queuing time is concerned. As for processing time, the first label coordinate for the processing event of untimely clones or zombies, when such an event exists, is either $\crittime{P}{G}$ or $\crittime{G}{P}$ for some non-$G$ process $P$. Since this label corresponds to the processing transaction of a $P$-donation at $G$ or a $G$-codonation at $P$ in the original history, the axioms follow in this case from the same axioms in the original history as they pertain to the donation or co-donation packet.

As for clones of a post-critical packet $k \in \channel{G}{G}$, such a clone is created only if the original packet is queued while $P$ appears alive to $G$, and its processing event label has $\crittime{G}{P}$ as a first coordinate, which verifies the axiom using the same argument as before. 

The \AxProcIV{} holds trivially for original packets $k$. The only cloned packets on self channels exist on $\channel{G}{G}$ and $\channel{\minusG}{\minusG}$. If such a cloned packet is pre-critical, then the original \AxProcIV{} guarantees that it is the clone of a timely packet and the axiom follows easily from that in that case. Since neither of the two channels contains clones of post-critical packets, we are done.

The \AxAppI{} is trivially induced from $H$.

The \AxOrderI{} follows from Corollary \ref{CompEventCor} as long as $\labelspace{}$ is itself very well founded. This in turn follows from the \AxOrderI{} in $H$.

The \AxOrderII{} is true in $H^r$ by definition (see \ref{ReducedHistoryOrder}).

The \AxHaltI{} would follow if we show that we only add a finite number of new events in $H^r$. Indeed, all the timely and untimely packets that we add are clones and zombies of packets that are queued before the critical view change notification. It follows from the \AxOrderI{} and the \AxViewII{} (in $H$) that there is only a finite number of such packets.

The post-critical packets that we add are all clones of post-critical packets in $\channel{G}{G}$ that are queued before $G$ processes a donation packet from some original process $P \in \viewset_{\vcrit}$ or before $G$ receives a removal notification for $P$. If $G$ halts then there is a finite number of such packets by the \AxHaltI{} (in $H$). If $G$ does not halt but $P$ is removed then $G$ processes the removal notification of $P$ (because $H$ is transactional) and therefore queues only a finite number of packets before that event by the \AxOrderI{}.

If $G$ does not halt and $P$ is not removed then $P$ does not halt (because $H$ is conforming and not stunted, see Lemma \ref{FirstFaultLem}). In that case $P$ must dequeue and process to completion the join notification of $G$ (because $H$ is transactional) and therefore must queue a donation packet $d$ to $G$. All the packets queued by $P$, including $d$, must be sent (by the \AxHaltIII{}) and none of them are dropped (because there are no dropped packets in a transactional history). Therefore $G$ receives all of them, and since $G$ does not halt it must dequeue all of them, including $d$ (by the \AxHaltIII{}). Therefore $G$ queues a finite number of packets prior to the $\preceiveevent{d} = \crittime{P}{G}$ event (by the \AxOrderI{}) and we are done.

The \AxHaltII{} carries over to $H^r$ for every non-$G$ process. The first part of the axiom holds for $\plusminusG$ because there are no dropped notifications in $H^r$. The second part holds for $\minusG$ because it halts for for $G$ because it processes all of its notifications.

The \AxHaltIII{} has three parts. The first two parts follow from the same axiom in $H$ and the fact that each clone and zombie $k$ which is not processed is dropped. The third part follows for the same reason, though not so trivially. Suppose that a packet $k$ meets the criteria of the third part of the axiom. Then $k$ is received and therefore is not dropped. If $k$ is a clone or a zombie then $k$ must be processed because all unprocessed clones and zombies are dropped. If it is original and all its predecessors are received, then surely its original predecessors are received, and therefore in $H$ the packet $k$ is processed, and therefore it is processed in $H^r$ as well.

The \AxHaltIV{} is trivially induced from $H$.

The main challenge is verifying the last remaining axiom, the \AxPackEventIII{}, which says that the channels are FIFO. This is crucial but unfortunately quite tedious. Our basic approach is the following. We take two packets $k_1$ and $k_2$ that meet the assumptions of the axiom. Namely, they belong to the same channel $\channel{X}{Y}^{H^r}$, with $\psendevent{k_1} \prec \psendevent{k_2}$, and with $\preceiveevent{k_2}$ existing. To verify the axiom we have to show that $\preceiveevent{k_1}$ exists and precedes $\preceiveevent{k_2}$. To do that we look at the original packets $\origpk(k_1)$ and $\origpk(k_2)$ and use the FIFO property in $H$ to glean information about their events and labels. From there we derive information about the labeling of the $k_1$ and $k_2$ events, which in turn determines their order in $H^r$ (see Corollary \ref{CompEventCor}).

We look at different cases according as $k_1$ and $k_2$ are pre-critical or post-critical. The condition $\psendevent{k_1} \prec \psendevent{k_2}$ implies that if $k_2$ is pre-critical, then so is $k_1$. Also notice that a packet $k$ is pre-critical if and only if $\origpk(k)$ is pre-critical.

\par{Case I: $k_1$ and $k_2$ are both pre-critical} \\ 
In this case we can use Lemma \ref{FIFOLem} and conclude that $\psendevent{\origpk(k_1)} \preceq \psendevent{\origpk(k_2)}$. According to the same lemma, if there is an equality then $\origpk(k_1) = \origpk(k_2)$, the common origin is a ghost packet and $X = \plusminusG$. In this case a routine inspection of the different cases of \ref{ChannelPGSSS} - \ref{ChannelGGSSS} shows that $\preceiveevent{k_1}$ exists whenever $\preceiveevent{k_2}$ does, and precedes the latter as well. So we can assume that
$\psendevent{\origpk(k_1)} \prec \psendevent{\origpk(k_2)}$ and therefore $\origpk(k_1)$ and $\origpk(k_2)$ are distinct. 

If $\origpk(k_2)$ is timely then the \AxPackEventIII{} for $H$ implies that $\preceiveevent{\origpk(k_1)}$ exists and moreover $\preceiveevent{\origpk(k_1)} \prec \preceiveevent{\origpk(k_2)}$. This in turn implies that $\origpk(k_1)$ is a timely packet and therefore $k_1$ is timely as well and therefore $\preceiveevent{k_1}$ exists. Let $T_i$ be the transaction triggered by $\preceiveevent{\origpk(k_i)}$. Then $\trig(T_1) \prec \trig(T_2)$ and so $\ell_{T_1} \,\dot{\prec}\, \ell_{T_2}$.

Going over \ref{ChannelPGSSS} - \ref{ChannelGGSSS} one can verify that if $k$ is a timely clone or zombie packet then $\lbl(\preceiveevent{k})$ shares the same constellation coordinate with $\lbl(\preceiveevent{\origpk(k)})$. Therefore the constellation coordinate of $\preceiveevent{k_i}$ is $\ell_{T_i}$ and it follows from Corollary \ref{CompEventCor} that $\preceiveevent{k_1} \prec \preceiveevent{k_2}$ and we are done.

We are left with the case where $\origpk(k_2)$ is untimely. We know that $\preceiveevent{k_2}$ exists. Going over the untimely cases of \ref{ChannelPGSSS} - \ref{ChannelGGSSS} one can verify that $\lbl(\preceiveevent{k_2}) = \lblggz{T}{\psendevent{\origpk(k_2)}}{*}$ where $\ell_T = \crittime{P}{G}$ or $\ell_T = \crittime{G}{P}$. In particular we know that $T$ exists. In this case if $\origpk(k_1)$ is timely then $\preceiveevent{k_1}$ exists and is timely and therefore precedes the untimely $\preceiveevent{k_2}$. If $\origpk(k_1)$ is untimely, it follows from the existence of $T$ that $\preceiveevent{k_1}$ also exists and has a label $\lblggz{T}{\psendevent{\origpk(k_1)}}{*}$. This implies that $\lbl(\preceiveevent{k_1}) < \lbl(\preceiveevent{k_2})$ and therefore $\preceiveevent{k_1} \prec \preceiveevent{k_2}$.

\par{Case II: $k_1$ is pre-critical and $k_2$ is post-critical} \\
The case where $k_1$ is timely is trivial, because for a timely packet $\preceiveevent{k_1}$ exists by definition and again by definition
$\preceiveevent{k_1} \prec \pnotifyevent{\vcrit}{Y} \prec \preceiveevent{k_2}$. So we can assume that $k_1$ is untimely. 

Since we do not add any new packets to non-$G$ channels, and since any channel involving $\minusG$ does not contain any post-critical packets, we can confine our attention to the channels $\channel{G}{P}$, $\channel{P}{G}$ and $\channel{G}{G}$, where $P$ is an non-$G$ process.

Look at the channel $\channel{P}{G}$. We do not add any new packets to this channel that are post-critical. Therefore $k_2$ must be original. Therefore we can assume that $k_1$ is not original (otherwise there is nothing to prove). We know that $\preceiveevent{k_2}$ exists. Denote by $T$ the transaction that is triggered by $\preceiveevent{k_2}$.

We know that the first post-critical packet (in $H$) that $P$ queues to $\channel{P}{G}$ is the donation packet $d$, and since $P$ queues some post-critical packet to $G$ in $H$ (namely $k_2$), it must also queue $d$. Moreover, $k_2 \ne d$, since $k_2$ survives in $H^r$ and $d$ does not (see \ref{ChannelPGSSS}). Therefore $\psendevent{d} \prec \psendevent{k_2}$ and the \AxPackEventIII{} in $H$ implies that $\preceiveevent{d}$ exists and precedes $\preceiveevent{k_2}$ in $H$. Therefore $\crittime{P}{G}$ exists and $\crittime{P}{G} \,\dot{\prec}\, \ell_T$. Since $k_1$ is a clone or a zombie, the existence of $\crittime{P}{G}$ implies that $\preceiveevent{k_1}$ exists and $\lbl(\preceiveevent{k_1}) = \lblcggz{P}{G}{*}{*}$. On the other hand $\lbl(\preceiveevent{k_2}) = \lblggz{T}{*}{*}$ and $\crittime{P}{G} \,\dot{\prec}\, \ell_T$. Therefore
$\lbl(\preceiveevent{k_1}) < \lbl(\preceiveevent{k_2})$ and therefore $\preceiveevent{k_1} \prec \preceiveevent{k_2}$ and we are done.

Look at the channel $\channel{G}{P}$. Since $\channel{G}{P}^H$ does not have pre-critical packets, $k_1$ cannot be original. We assume first that $k_2$ is original. We proceed just like in the case of $\channel{P}{G}$. We know that the first post-critical packet (in $H$) that $G$ queues to $\channel{G}{P}$ is the co-donation packet $c$, and since $G$ queues some post-critical packet to $\channel{G}{P}$ in $H$ (namely, $k_2$), it must also queue $c$. Moreover, $k_2 \ne c$, since $k_2$ survives in $H^r$ and $c$ does not (see \ref{ChannelGPSSS}). Therefore $\psendevent{c} \prec \psendevent{k_2}$ and the \AxPackEventIII{} in $H$ implies that $\preceiveevent{c}$ exists and precedes $\preceiveevent{k_2}$ in $H$. Therefore $\crittime{G}{P}$ exists and $\crittime{G}{P} \,\dot{\prec}\, \ell_T$, where $T$ is the transaction triggered by $\preceiveevent{k_2}$. Since $k_1$ is a clone or a zombie, the existence of $\crittime{G}{P}$ implies that $\preceiveevent{k_1}$ exists and $\lbl(\preceiveevent{k_1}) = \lblcgzz{G}{P}{*}{*}$. On the other hand $\lbl(\preceiveevent{k_2}) = \lblggz{T}{*}{*}$ and $\crittime{G}{P} \,\dot{\prec}\, \ell_T$. Therefore
$\lbl(\preceiveevent{k_1}) < \lbl(\preceiveevent{k_2})$ and therefore $\preceiveevent{k_1} \prec \preceiveevent{k_2}$ and we are done.

If $k_2$ is not original then $\origpk(k_2) \in \channel{G}{G}$ (this is the only channel that produces post-critical clones). By assumption $\preceiveevent{k_2}$ exists, so by \ref{ChannelGPSSS} $\crittime{G}{P}$ exists and
$$
\lbl(\preceiveevent{k_2}) = \lblcgzz{G}{P}{\psendevent{\origpk(k_2)}}
$$
In particular the co-donation packet $c \in \channel{G}{P}$ exists, and is processed, in $H$.

$k_1$ is an untimely clone or a zombie. The \ref{ChannelGPSSS} construction shows that when $\crittime{G}{P}$ exists, $\preceiveevent{k_1}$ exists and
$\lbl(\preceiveevent{k_1}) = \lblcgzz{G}{P}{\psendevent{\origpk(k_1)}}$. In the clean event order on $\dot{\eventset}$ we have
$$
\psendevent{\origpk(k_1)} \,\dot{\prec}\, \ell_{\vcrit} \,\dot{\prec}\, \psendevent{\origpk(k_2)}
$$
Therefore $\lbl(\preceiveevent{k_1}) < \lbl(\preceiveevent{k_2})$ and therefore $\preceiveevent{k_1} \prec \preceiveevent{k_1}$ and we are done in this case as well.

We are left with channel $\channel{G}{G}$. This channel does not contain untimely packets, so there is nothing to consider in this case.

\par{Case III: $k_1$ and $k_2$ are both post-critical} \\
If both $k_1$ and $k_2$ are original packets there is nothing to prove. The only post-critical clones (there are no zombies) occur in $\channel{G}{P}$ channels where $P$ is an non-$G$ process. So we can assume that $k_1, k_2 \in \channel{G}{P}$ and at least one of them is not original.

If $k_2$ is not original then it is a clone of a post-critical packet $\channel{G}{G}$ with $\lbl(\psendevent{\origpk(k_2)}) = \lblgzg{T}{*}{*}$. If $\crittime{P}{G}$ exists then $\ell_T \,\dot{\prec}\, \crittime{P}{G}$. Since $\preceiveevent{k_2}$ exists it follows by \ref{ChannelGPSSS} that $\crittime{G}{P}$ must exist and $\lbl(\preceiveevent{k_2}) = \lblcgzz{G}{P}{\psendevent{\origpk(k_2)}}$. This means in particular that the co-donation packet $c$ from $G$ to $P$ exists, and is processed, in $H$. It also means, a-fortiori, that $\crittime{P}{G}$ exists.

Suppose that $k_1$ is original. Since $c$ is the first packet, in $H$, that is sent from $G$ to $P$; and since $k_1 \ne c$ (it survives in $H^r$ and $c$ does not), we have $\psendevent{c} \prec \psendevent{k_1}$. Therefore
$$
\lbl(\psendevent{c}) = \lblczzg{P}{G}{\lCODONATE} < \lbl(\psendevent{k_1})
$$
On the other hand
$$
\lbl(\psendevent{\origpk(k_2)}) = \lblgzg{T}{*}{*} < \lblczzg{P}{G}{\lCODONATE} = \lbl(\psendevent{c})
$$
By \ref{ChannelGPSSS}, $\lbl(\psendevent{k_2}) = \lbl(\psendevent{\origpk(k_2)})$ and so $\lbl(\psendevent{k_2}) < \lbl(\psendevent{k_1})$ and therefore $\psendevent{k_2} \prec \psendevent{k_1}$ contrary to our assumption. So this case is not possible.

So we know that if $k_2$ is a clone then $k_1$ must be a clone as well, and since $\crittime{G}{P}$ exists we know by \ref{ChannelGPSSS} that $\preceiveevent{k_1}$ exists and
$\lbl(\preceiveevent{k_1}) = \lblcgzz{G}{P}{\psendevent{\origpk(k_1)}}$. Since $\origpk(k_1)$ and $\origpk(k_2)$ both belong to $\channel{G}{G}$ and since
$$
\lbl(\psendevent{\origpk(k_1)}) = \lbl(\psendevent{k_1}) < \lbl(\psendevent{k_2}) = \lbl(\psendevent{\origpk(k_2)})
$$
we have $\psendevent{\origpk(k_1)} \prec \psendevent{\origpk(k_2)}$ and so by \ref{SubTransSSS} $\lbl(\preceiveevent{k_1}) < \lbl(\preceiveevent{k_2})$ and so $\preceiveevent{k_1} \prec \preceiveevent{k_2}$ and we are done. 

We are left with the case where $k_2$ is original and $k_1$ is not. Using the same argument that we used with an original $k_1$, we can establish that $\psendevent{c} \prec \psendevent{k_2}$, where $c$ is the co-donation packet sent from $G$ to $P$ in $H$. We know by assumption that $\preceiveevent{k_2}$ exists. So we can conclude from the FIFO property in $H$ that $\preceiveevent{c}$ exists and $\preceiveevent{c} \prec \preceiveevent{k_2}$. In particular $\crittime{G}{P}$ exists. Let $T$ be the transaction triggered by $\preceiveevent{k_2}$. Then $\crittime{G}{P} \prec \ell_T$.

Since $k_1$ is a clone and $\crittime{G}{P}$ exists, we conclude from \ref{ChannelGPSSS} that $\preceiveevent{k_1}$ exists and
$$
\lbl(\preceiveevent{k_1}) = \lblcgzz{G}{P}{\psendevent{\origpk(k_1)}} < \lblggz{T}{*}{*} = \lbl(\preceiveevent{k_2})
$$
and therefore $\preceiveevent{k_1} \prec \preceiveevent{k_2}$. This concludes the proof. 
\end{proof}

\begin{thm}
$H^r$ is a conforming history.
\end{thm}

\begin{proof}
We already know that $H^r$ is a history, but we still have to confirm the conforming axioms, using the fact that $H$ is a conforming history.

\begin{description}
\item[The \CAxChannel{}.] \hfill \\
If either $P=(\minusG)$ or $Q=(\minusG)$ we are done because $\minusG$ halts and is removed in $H^r$. All the other channels in $H^r$ are original $H$ channels to which we added clones and zombies and from which we removed the critical donation and co-donation packets. Since the latter are finite in number, any original channel that is finite in $H^r$ is also finite in $H$, and it follows from the conformity of $H$ that either $P$ is removed in $H$ or $Q$ halts in $H$. These properties carry over directly to $H^r$.

\item[The \CAxPacket{}.] \hfill \\
Suppose that this axiom is violated. Then there are processes $X$ and $Y$ and a packet $k \in \channel{X}{Y}$ such that
\begin{itemize}
\item $\notify{i}{Y} = \nr{X}$ and $\pnotifyevent{i}{Y}$ exists.
\item $\preceiveevent{k}$ exists and $\pnotifyevent{i}{Y} \prec \preceiveevent{k}$
\end{itemize}
Since we do not change the non-$G$ channels, at least one of $X$ and $Y$ must be equal to $\plusminusG$. If $k$ is original then it does not violate the axiom. All the clones and zombies that we add to the $\channel{(\plusminusG)}{(\plusminusG)}$ channels are timely, and we add only pre-critical packet processing events in channels where $X = (\minusG)$. Therefore $k$ must be a clone or zombie on a $\channel{G}{P}$ channel or a $\channel{P}{(\plusminusG)}$ channel.

In the case of a $\channel{G}{P}$ channel $k$ must be untimely, and since it is processed the constellation label of $\preceiveevent{k}$ is $\crittime{G}{P}$ and the critical co-donation packet is processed at $P$. It follows from the \CAxPacket{} in $H$ that $\crittime{G}{P} \dot{\prec} \ell_{\vdeath{G}}$ and we are done.

The case of a $\channel{P}{(\plusminusG)}$ channel and an untimely $k$ is similar. If $k$ is timely then the constellation label of $\preceiveevent{k}$ is equal to the clean event $\preceiveevent{\origpk(k)}$. By the \CAxPacket{} in $H$ we know that $\preceiveevent{\origpk(k)} \dot{\prec} \ell_{\vdeath{P}}$ and we are done in this case as well.

\item[The \CAxNotify{}.] \hfill \\
By construction $H^r$ has no dropped notifications, so the axiom is vacuously true.

\item[The \CAxView{}.] \hfill \\
The new process $\minusG$ has a finite view interval. If $P \ne (\minusG)$ and $P$ halts then $P$ exists in $H$ and has a finite view interval there. This property carries over to $H^r$. 

\item[The \CAxParent{}.] \hfill \\
The history $H$ is transactional and therefore has no uninitialized processes. This property carries over to $H^r$ for all the processes with the exception of $\minusG$. But $\minusG$ itself is also initialized because it processes all of its notifications. Therefore this property holds vacuously at $H^r$.

\item[The \CAxHalt{}.] \hfill \\
This property carries over to $H^r$ for all the processes in $H$. The property holds for $\minusG$ as well because it halts.
\end{description}
\end{proof}

\section{The History Equivalence Theorem}
\subsection{Introduction}
We are now ready to prove the fundamental property of $H^r$, namely, that it performs the same calculation as $H$.

The proof proceeds by induction on the partially ordered constellations of $\labelspace{}$. We will formulate an inductive hypothesis that correlates the state of the processes in $H$ and in $H^r$ prior to a given constellation. We will show that under the hypothesis, if the processes of $H^r$ process the triggers of the constellation according to the \cbcast{} algorithm, each would generate and queue the exact same packets, and in the same order, that are observed in the $H^r$ transaction for that trigger. Moreover, the post-processing state of the processes in $H^r$ will continue to relate to the state of the processes in $H$ according to the inductive hypothesis once the whole constellation is processed.

The most important part of the rather elaborate inductive hypothesis is what it says about the eventual state of $H$ and $H^r$, namely that the state becomes identical. This means that the calculation carried out by the two histories is the same calculation. This makes it possible to carry over desirable properties like coherence and progress from $H^r$ to $H$. By repeating that step we can ultimately carry over these properties from relatively simple histories that do not have any process joins to the more intractable histories that have any finite number of such joins.

\newcommand{\HistoryEquivThm}{History Equivalence Theorem}
\begin{thm}[\HistoryEquivThm]
\label{HistoryEquivThm}
$H^r$ is the history of a \cbcast{} and \ulp{} computation that performs the same computation that $H$ does. Specifically
\begin{itemize}
\item $H^r$ delivers the same messages and view installations in the same order and at the same constellation as $H$ does at any process $P \ne \plusminusG$.
\item $H^r$ delivers the same messages and view installations in the same order and at the same constellation as $H$ does at process $G$ after the critical moment.
\item $H^r$ delivers the same messages and view installations at processes $\plusminusG$ in the same order and at the same constellation as $H$ does at process $D$ before the critical moment.
\item At some point the states of $H^r$ and $H$ become identical
\end{itemize}
Where a message delivery is an invocation of the ApplyMessage up-call and a view installation is an invocation of the ApplyJoin or ApplyRemoval up-calls.
\end{thm}

\subsection{Proof plan and preliminaries}
\label{ProofPlanSSS}
The claim of the theorem is a little subtle. $H$ is a history that arises naturally from an execution of \cbcast{} and \ulp{}. But $H^r$ is a synthetic history whose trigger events occur for no underlying reason. To prove the theorem we have to endow $H^r$ with an execution that gives rise to its arbitrary behavior. We do that inductively, one constellation at a time.

At the beginning of time we have each original process in $H^r$ initialized by invoking the \ref{StartCluster} call, using
$$
\fp{roster}^r = \fp{roster} \,\bigcup\, \{ \plusminusG \}
$$
as the roster. This causes the processes of $H^r$ to initialize to a specific initial state. We will show that this state is {\em similar} to the initial state of $H$, where similarity of state is a rather complex relationship that we will define later. This similarity forms the basis of our inductive process. The induction is by the partially ordered constellations of $\cspace{}$ (see Definition \ref{CSpaceDef}).

Recall from \ref{LabelSpaceSS} that $H$ and $H^r$ constellations were defined as sets of events that share the first coordinate, called the constellation coordinate, in their labels. Since each constellation coordinate is a clean $H$-trigger, there is a very close relationship between $H^r$-constellations and $H$-constellations. In fact in most cases a $H^r$-constellation is essentially an original $H$-transaction or a set of clones of an $H$-transaction. The exceptions are the critical donation and co-donation transactions of $H$. Each of these gets broken down into a sequence of $H^r$ transactions which correspond to $H$ sub-transactions and which can be distinguished by their sub-transaction labels.

For each constellation we assume the following about $H^r$ and $H$:
\begin{enumerate}
\item The starting states of $H^r$ and $H$ are similar.
\label{HET:StateSimilar}
\item The starting states of the \ulp{} thread in $H^r$ and $H$ are {\em identical}.
\label{HET:ULPIdentical}
\item The next execution interval of the \ulp{} thread will occur at exactly the same time in $H^r$ and $H$ at every process. In particular, since the constellation consists of at most a single transaction per process in $H$, the \ulp{} thread will not continue running in $H^r$ until the conclusion of the constellation, despite the fact that it may consist of multiple transactions.
\label{HET:ThreadSerial}
\end{enumerate}
The last assumption reflects a degree of freedom that we have in weaving the \ulp{} computation into $H^r$. After all, we only have to show that $H^r$ {\em can} arise as a history of a \cbcast{} and \ulp{} computation, not that it must arise.

The only difficulty with the timing of the \ulp{} thread occurs at its inception. In $H$, the \ulp{} thread at $G$ is launched when $G$ installs the critical view. In $H^r$ however the thread is launched at the beginning of time by the \ref{StartCluster} procedure, since $G$ is original in $H^r$. However the launch is asynchronous, meaning that the thread is not executed immediately but at some indeterminate point in the future. Since the lauch is earlier in $H^r$ we can simply assume that the execution is delayed long enough to coincide with the execution in $H$. At process $\minusG$ in $H^r$ the launch also occurs at the beginning of time but we can assume that the execution is delayed until $\minusG$ halts and therefore never occurs.

Under these assumptions we demonstrate the following conclusions:
\begin{itemize}
\item Using the \cbcast{} callbacks to execute the constellation in $H^r$ results in an ending state of $H^r$ that is similar to the state of $H$ at the end of the same constellation.
\item The side effects that are generated by the \cbcast{} callbacks are identical to the observed side effects in $H^r$. In other words the current constellation looks like it has come about as a result of a \cbcast{} execution rather than an arbitrary choice of side effects.
\item The message deliveries and view installations that are generated by the \cbcast{} callbacks in $H^r$ are identical to those generated in $H$, in the sense that was elaborated in the statement of the theorem. As a result any information that is visible to \ulp{} remains identical at the end of the constellation.
\end{itemize}
This looks like it is enough for carrying the induction forward, but it is not quite enough, because even though we have shown that the side effects of the current  constellation are generated by \cbcast{} rather than being arbitrary, we have not shown that any subsequent trigger events are non-arbitrary.

So suppose that $C$ is a constellation in $H^r$, and suppose that every preceding constellation has been shown to arise from a \cbcast{} and \ulp{} execution. Why should this execution give rise to the triggers of $C$? Look at any trigger in $C$:

If the trigger is a dequeuing of a notification then there is no problem, because \gms{} is part of $H^r$ and we are allowed to control its behavior arbitrarily.

If the trigger is a dequeuing event of a packet, then the packet was queued in a previous constellation and therefore was a result of a \cbcast{} and \ulp{} execution. The dequeuing of the packet at this point is the result of the timing (or labeling) that we built into $H^r$ explicitly in order to have packets dequeue at their destinations at improbable, but incredibly convenient times.

If the trigger is a dequeuing of a message broadcast request then its existence depends on \ulp{} actually having produced the same requests in $H^r$ and in $H$ at the same time. The only way for this to happen is for the \ulp{} thread to be exactly identical in both histories, and this can only be guaranteed by perfectly masking the differences between $H^r$ and $H$ from \ulp{}. For this to happen we need three conditions:
\begin{itemize}
\item The \ulp{} thread must have been in an identical state in both histories at all processes at the end of the previous constellation.
\item The \ulp{} thread must have seen the exact same information since that time.
\item The \ulp{} thread execution must have proceeded at the exact same speed at the exact same intervals in both histories and must not have intermingled with constellation executions.
\end{itemize}
We have shown that the first two conditions are met, and are at liberty to assume the third, as we have seen. Therefore we can conclude that \ulp{} could have issued the same requests\footnote{we could have simplified this argument by replacing \ulp{} in $H^r$ with a random oracle that by sheer luck broadcasts the same messages that \ulp{} issues. However we thought it was significant that the  argument could be carried forward with the same user application and without resorting to an artificial oracle.}

With these observations we can conclude that the triggers of $C$ are indeed a result of a \cbcast{} and \ulp{} execution and that the three inductive hypotheses (\ref{HET:StateSimilar}), (\ref{HET:ULPIdentical}) and (\ref{HET:ThreadSerial}) above continue to hold.

The definition of state similarity is somewhat complex and varies depending on the {\em period} in the life of the process into which the constellation falls. Each process goes through three periods, the {\em pre-critical}, {\em interim} and {\em convergent} periods. The pre-critical period includes all the transactions that occur prior to the critical notification and the convergent period includes all the transactions that occur after the critical view is installed. The interim period includes all the constellations that occur while the critical view is pending installation.

Here is the crux of the matter: in the convergent period, similarity becomes equality, and the two histories converge as claimed.

It turns out that desirable properties like Causal Order and Progress carry over from $H^r$ to $H$. By iterating the history reduction process we can ultimately carry over these properties from relatively simple histories that do not have any process joins to the more intractable histories that have any finite number of such joins. Later we will show that join-free histories enjoy the Causal Order Property and the Progress Property. As a result both properties hold for finite-join histories. We will show that the Causal Order property holds for histories with an infinite number of joins as well. This is not true for the Progress Property.

\subsection{Side Effects in $H^r$}
\label{ReducedSideEffectsSSS}
In our model each side effect is a queuing event. A queuing event is a multicast (in the case of a message, ghost or flush packet) or a unicast (in the case of an acknowledgement, donation or co-donation packet). In a queuing event a process $P$ queues a set of identical packets to outbound channels, bound for a target set of processes (see the \AxPackEventII{}). Before we can make inductive arguments, we must relate the observed side effects in $H^r$ to the observed side effects in $H$.

As we prove the History Equivalence Theorem, we will repeatedly make the argument that the execution of \cbcast{} in $H^r$ produces the observed $H^r$ side effects. Each observed side effect in $H^r$ has two characteristics: the type and content of the packet that is being queued, and the target set of the multicast or unicast. In each case we will have to show that the \cbcast{} code execution produces the observed type of packet with the observed content. As for the target set, all the multicasts in \cbcast{} (step \ref{BM:Members} of \ref{BroadcastMessage}, step \ref{RRN:Fwd} of \ref{RemovalNotification}, steps \ref{CF:SendGhost} and \ref{CF:SendFlush} of \ref{CheckFlush}) use \cset{} as the target set. We will establish later (see Lemma \ref{ReducedContactLem}) that in each case, the observed target set is equal to the value of \cset{} that exists at the process in $H^r$ at the time of the multicast. Unicasts will not present a similar problem.

\subsection{The inductive hypothesis} 
\label{InductiveHypSSS}
The inductive hypothesis relates the states of certain processes in $H$ and $H^r$. The complete set of variables that make up the state of a process is listed in \ref{ProcessState}. The inductive hypothesis is complex enough to warrant a preliminary discussion.

Thanks to labeling we have a common "timeline" for $H$ and $H^r$, namely the common constellation partial order. We divide this common timeline into three periods: The pre-critical period, the interim period and the convergence period. The pre-critical period is the interval of time up to the critical view change constellation $\ell_{\vcrit}$. The interim period ends at a process when that process installs the critical view. This boundary occurs at a different constellation at each process. Process $\minusG$ is removed at the end of the pre-critical period, so it does not have an interim period. The convergence period starts at the end of the interim period and continues indefinitely. The inductive hypothesis is divided into separate hypotheses for each time period. The most important of those is the convergent period, where the hypothesis is that $H$ and $H^r$ are identical.

To summarize, let $e \in \dot{\eventset}$ be a constellation and let $P$ be a process. Then
\begin{itemize}
\item the constellation belongs to the {\em pre-critical period} at $P$ if $e \,\dot{\prec}\, \ell_{\vcrit}$
\item the constellation belongs to the {\em interim period} at $P$ if $e \,\dot{\succeq}\, \ell_{\vcrit}$ and $\valuepre{\cview}{P}{e} < \vcrit{}$
\item the constellation belongs to the {\em convergent period} at $P$ if $e \,\dot{\succeq}\, \ell_{\vcrit}$ and $\valuepre{\cview}{P}{e} \ge \vcrit{}$
\end{itemize}

The most complex constellations are the critical donation and co-donation constellations, $\crittime{P}{G}$ and $\crittime{G}{P}$. Each of these is a single transaction in $H$, but becomes a sequence of transactions in $H^r$ - potentially even an empty sequence, in which case the constellation does not exist in $H^r$. To prove the inductive hypothesis for one of these constellations, we need to resort to a second level of induction. For this purpose we will formulate sub-hypotheses that relate the states of $H$ and $H^r$ at each sub-transaction.

We build on our observations in \ref{ReducedSideEffectsSSS} to define the following equivalence between side effects. We use this equivalence to separate the issue of side effect type and content from the issue target set, that will be treated separately.

\begin{defn}
\label{EquivSEDef}
Let $e = \psendevent{k}$ and $e^r = \psendevent{k^r}$ be two queuing events in $H$ and $H^r$, respectively. We say that $e$ and $e^r$ produce an {\bf equivalent} side effect if $k$ and $k^r$ have the same packet type and $\cont(k) \cong \cont(k^r)$ (see Definition \ref{EquivDef}).
\end{defn}

The inductive hypothesis is actually a set of related hypotheses and sub-hypotheses. The main hypotheses are:
\begin{itemize}
\item The First Pre-Critical Hypothesis, which describes how the state of a process $P \in \processset^H$ in $H$ is related to the state of the same process in $H^r$ at the start of a pre-critical constellation. Notice that $P \ne \plusminusG$ because $\minusG$ is not a process in $H$ while $G$ joins post-critically in $H$.
\item The Second Pre-Critical Hypothesis, which describes how the states of $D$, $G$ and $\minusG$ in $H^r$ are related to each other at the start of a pre-critical constellation. Notice that in this case the comparison is within $H^r$, not between $H$ and $H^r$.
\item The Interim Non-$G$ Hypothesis, which describes how the state of a process $P \ne G$ in $H$ is related to the state of the same process in $H^r$ at the start of a post-critical constellation that occurs before $P$ installs the critical view in $H$.
\item The Interim $G$ Hypothesis, which describes how the state of $G$ in $H$ is related to its state in $H^r$ at the start of a post-critical constellation that occurs before $G$ installs the critical view in $H$.
\item The Convergent Hypothesis, which claims that the state of a process $P$ in $H$ is identical to its state in $H^r$ at any time after $P$ installs the critical view in $H$.
\end{itemize}

The sub-hypotheses are:
\begin{itemize}
\item The Donation Sub-Hypothesis, which describes how the state of $G$ in $H$ relates to its state in $H^r$ at the start of each sub-transaction of each donation constellation at $G$.
\item The First Co-Donation Sub-Hypothesis, which describes how the state of $P \ne G$ in $H$ relates to its state in $H^r$ at the start of each untimely sub-transaction of each co-donation constellation in $P$.
\item The Second Co-Donation Sub-Hypothesis, which describes how the state of $P \ne G$ in $H$ relates to its state in $H^r$ at the start of each post-critical sub-transaction of each co-donation constellation in $P$.
\end{itemize}

\begin{indhyp}[First Pre-Critical Hypothesis]
Let $C$ be any pre-critical constellation and let $P$ be a process that exists in both $H$ and $H^r$ at the start of $C$. Then the state of $P$ in $H$ and the state of $P$ in $H^r$ are identical at that point, with the following exceptions:
\begin{align*}
\mset{}^r & = \mset{} \,\bigcup\, \{ \plusminusG \}   \\
\lset{}^r & = \lset{} \,\bigcup\, \{ \plusminusG \}   \\
\cset{}^r & = \cset{} \,\bigcup\, \{ \plusminusG \}   \\
\vtime{}^r[] & = \vtime{}[] \,\bigcup\, \left\{ [G] = 0, [\minusG] = 0 \right\}   \\
\rset{}^r & \cong \rset{}   \\
\fque{}^r[] & \cong \fque{}[] \,\bigcup\, \left\{ [G] = \emptyset, [\minusG] = \emptyset \right\}   \\
\wset{}^r & \cong \left\{
\langle \fp{msg}, \fp{index}, \partial(\fp{iset}) \rangle \,|\, \langle \fp{msg}, \fp{index}, \fp{iset} \rangle \in \wset{} \right\}  \\
& \quad \text{where } \partial(\fp{iset}) =
\begin{cases}
\fp{iset}      & D \notin \fp{iset}     \\
\fp{iset} \,\bigcup\, \left\{ \fp{iset}[\plusminusG] = \{ f =  \fp{iset}[D].f,\, b = 0 \} \right\}  & D \in \fp{iset}
\end{cases}   \\
\mpkin{}^r[X].f & =
\begin{cases}
\mpkin{}[X].f & X \ne \plusminusG  \\
\mpkin{}[D].f & X = \plusminusG 
\end{cases}   \\
\mpkin{}^r[X].b & =
\begin{cases}
\mpkin{}[X].b & X \ne \plusminusG  \\
0 & X = \plusminusG 
\end{cases}   \\
\gvec{}^r[] & = \gvec{}[] \,\bigcup\, \left\{ [G] = \gvec{}[D], [\minusG] = \gvec{}[D]\right\}   \\
\fvec{}^r[] & = \fvec{}[] \,\bigcup\, \left\{ [G] = \gvec{}[D], [\minusG] = \gvec{}[D]\right\}
\end{align*}
Notice that $\fvec{}^r[\plusminusG]$ is indeed inherited from $\gvec{}[D]$, not $\fvec[D]$!
\end{indhyp}

\begin{indhyp}[Second Pre-Critical Hypothesis]
Let $C$ be any pre-critical constellation. Then the states of $\plusminusG$ in $H^r$ at the start of $C$ are identical to the state of $D$ in $H^r$ at the same point, with the following exceptions:
\begin{align*}
\self{}^r(\plusminusG) & = \plusminusG \\
\bwset{}^r(\plusminusG) & = \emptyset   \\
\fwset{}^r(\plusminusG) & = \{ \langle \fp{msg}, \{  f = \fp{index}.f, b = 0 \}, \fp{iset} \rangle \,| \\
 & \langle \fp{msg}, \fp{index}, \fp{iset} \rangle \in \fwset{}^r(D) \}  \\
\lque{}^r(\plusminusG) & = \emptyset   \\
\sfh{}^r(\plusminusG) & = \sgh{}^r(D)   \\
\mpkout{}^r.b(\plusminusG) & = 0
\end{align*}
\end{indhyp}

\begin{indhyp}[Interim Non-$G$ Hypothesis]
Let $C$ be any post-critical constellation and let $P \ne G$ be a process that exists at the start of $C$ and has no yet installed the critical view in $H$. Then the state of $P$ in $H$ is identical to the state of $P$ in $H^r$ at that point with the following exceptions:
\begin{align*}
\mset{}^r & = \mset{} \,\bigcup\, \{ \plusminusG \}   \\
\pque{}^r & = \pque{} \text{ with } \langle \operatorname{REMOVE}, \minusG \rangle \text{ replacing } \langle \operatorname{JOIN}, G \rangle   \\
\vtime{}^r[] & = \vtime{}[] \,\bigcup\, \left\{ [G] = 0, [\minusG] = 0 \right\}   \\
\rset{}^r & \cong \rset{}   \\
\fque{}^r[] & \cong \fque{}[]   \\
\wset{}^r & \cong \wset{}
\end{align*}
\end{indhyp}

\begin{indhyp}[Interim $G$ Hypothesis]
Let $C$ be any post-critical constellation and suppose that $G$ exists at the start of $C$ and has not yet installed the critical view in $H$. Then the state of $G$ in $H$ is identical to the state of $G$ in $H^r$ at that point with the following exceptions:
\begin{align*}
\mset{}^r & = \mset{} \,\bigcup\, \{ \plusminusG \}   \\
\pque{}^r & = \pque{} \text{ with } \langle \operatorname{REMOVE}, \minusG \rangle \text{ replacing } \langle \operatorname{JOIN}, G \rangle   \\
\cset{}^r & = \lset{} \\
\vtime{}^r[] & = \vtime{}[] \,\bigcup\, \left\{ [G] = 0, [\minusG] = 0 \right\}   \\
\rset{}^r & \cong \rset{}   \\
\fque{}^r[] & \cong \fque{}[]   \\
\wset{}^r & \cong \wset{}
\end{align*}
\end{indhyp}

\begin{indhyp}[Donation Sub-Hypothesis]
Let $P$ be any process that sends a critical donation packet to $G$ in $H$ and let $k$ be any untimely packet in $\channel{P}{D}$. If $G$ processes the critical donation packet from $P$ in $H$ then the state of $G$ in $H$ at the start of the $\lblcggz{P}{G}{\psendevent{k}}{*}$ sub-transaction is identical to the state of $G$ in $H^r$ at the start of the matching transaction in $H^r$, with the following exceptions:
\begin{align*}
\mset{}^r & = \mset{} \,\bigcup\, \{ \plusminusG \}   \\
\pque{}^r & = \pque{} \text{ with } \langle \operatorname{REMOVE}, \minusG \rangle \text{ replacing } \langle \operatorname{JOIN}, G \rangle   \\
\cset{}^r & = \lset{} \\
\vtime{}^r[] & = \vtime{}[] \,\bigcup\, \left\{ [G] = 0, [\minusG] = 0 \right\}   \\
\rset{}^r & \cong \rset{}   \\
\fque{}^r[] & \cong \fque{}[]   \\
\wset{}^r & \cong \wset{} \\
\gvec{}[P] & \le \gvec{}^r[P] \le \valuepre{\sgh{}}{P}{\ell_{\vcrit}} \\
\fvec{}[P] & \le \fvec{}^r[P] \le \valuepre{\sfh{}}{P}{\ell_{\vcrit}}
\end{align*}
\end{indhyp}

\begin{indhyp}[First Co-Donation Sub-Hypothesis]
Let $P \ne G$ be any process that receives a critical co-donation packet from $G$ in $H$ and let $k$ be any untimely packet in $\channel{D}{P}$. If $P$ processes the critical co-donation in $H$ then the state of $P$ in $H$ at the start of the $\lblcggz{G}{P}{\psendevent{k}}{*}$ sub-transaction in $H$ is identical to the state of $P$ in $H^r$ at the start of the matching transaction in $H^r$, with the following exceptions:
\begin{align*}
\mset{}^r & = \mset{} \,\bigcup\, \{ \plusminusG \}   \\
\pque{}^r & = \pque{} \text{ with } \langle \operatorname{REMOVE}, \minusG \rangle \text{ replacing } \langle \operatorname{JOIN}, G \rangle   \\
\vtime{}^r[] & = \vtime{}[] \,\bigcup\, \left\{ [G] = 0, [\minusG] = 0 \right\}   \\
\rset{}^r & \cong \rset{}   \\
\fque{}^r[] & \cong \fque{}[]   \\
\wset{}^r & \cong \wset{} \\
\gvec{}[G] & \le \gvec{}^r[G] \le \valuepre{\sgh{}}{D}{\ell_{\vcrit}} \\
\fvec{}[G] & \le \fvec{}^r[G] \le \valuepre{\sgh{}}{D}{\ell_{\vcrit}}
\end{align*}
\end{indhyp}

\begin{indhyp}[Second Co-Donation Sub-Hypothesis]
Let $P \ne G$ be any process that receives a critical co-donation packet from $G$ in $H$ and let $k$ be any uncontacted packet in $\channel{G}{G}$. If $P$ processes the critical co-donation in $H$ then the state of $P$ in $H$ at the start of the $\lblcggz{G}{P}{\psendevent{k}}{*}$ sub-transaction is identical to the state of $P$ in $H^r$ at the start of the matching transaction in $H^r$, with the following exceptions:
\begin{align*}
\mset{}^r & = \mset{} \,\bigcup\, \{ \plusminusG \}   \\
\pque{}^r & = \pque{} \text{ with } \langle \operatorname{REMOVE}, \minusG \rangle \text{ replacing } \langle \operatorname{JOIN}, G \rangle   \\
\vtime{}^r[] & = \vtime{}[] \,\bigcup\, \left\{ [G] = 0, [\minusG] = 0 \right\}   \\
\rset{}^r & \cong \rset{}   \\
\fque{}^r[] & \cong \fque{}[]   \\
\wset{}^r & \cong \wset{} \\
\gvec{}[G] & \le \valuepre{\sgh{}}{D}{\ell_{\vcrit}} \le \gvec{}^r[G] \le \valuepre{\sgh{}}{G}{\crittime{P}{G}} \\
\fvec{}[G] & \le \valuepre{\sgh{}}{D}{\ell_{\vcrit}} \le \fvec{}^r[G] \le \valuepre{\sgh{}}{G}{\crittime{P}{G}}
\end{align*}
\end{indhyp}

\begin{indhyp}[Convergent Hypothesis]
For any post-critical constellation $C$ and any process $P$ that exists at the start of $C$, if $P$ has already installed the critical view at that point in $H$, then the state of $P$ in $H$ is identical to the state of $P$ in $H^r$ at that same point. Moreover, the state of $P$ does not contain any pre-critical messages.
\end{indhyp}

\subsection{Technical lemmas for the History Equivalence Theorem proof}

\begin{lem}
\label{ReducedContactLem}
Let $e = \psendevent{k}$ be a queuing event of a ghost, flush or message packet at process $P$ in $H^r$, and let $e_0 = \psendevent{\origpk(k)}$ be the original queuing event at process $R$ in $H$. Assume that the inductive hypothesis holds and that one of the following conditions hold as well:
\begin{enumerate}
\item $e$ is pre-critical, and $\valuepre{\cset{}^r}{P}{e} = \valuepre{\cset{}}{R}{e_0} \,\bigcup\, \{ \plusminusG \}$
\item $e$ is post-critical, $P \ne G$ and $\valuepre{\cset{}^r}{P}{e} = \valuepre{\cset{}}{R}{e_0}$
\item $e$ is post-critical, $P = G$ and $\valuepre{\cset{}^r}{P}{e} = \valuepre{\lset{}}{R}{e_0}$
\end{enumerate}
Then $T_e = \cset{}^r$
\end{lem}

\begin{proof}
Most of the cases follow immediately from Lemma \ref{SEContactSetLem} and the various inductive hypotheses. The one non-trivial case is when $e$ is post-critical, $P = G$ and the Convergent Hypothesis holds. In that case it follows from Lemma \ref{SEContactSetLem} and the inductive hypothesis that $T_e = \lset{}^r$. We have to show that $\cset{}^r = \lset{}^r$. It follows from Lemma \ref{OmnibusCBCASTLem}(\ref{OCL:lset}) that the equality holds for original processes. Since $G$ is an original process in $H^r$ we are done.
\end{proof}

\begin{lem}
\label{BroadcastEquivLem}
Let \fp{msg} be an unstamped message and let $P \ne \plusminusG$ be a process in $H$. Suppose that the states of $P$ in $H$ and $H^r$ satisfy the First Pre-Critical Hypothesis. Suppose as well that $D \in \lset{}$ and $\plusminusG \notin \lset{}$ at $P$ in $H$. Then executing the \ref{BroadcastMessage}(\fp{msg}) in both histories will result in equivalent side effects (see Definition \ref{EquivSEDef}) while preserving the Hypothesis.
\end{lem}

\begin{proof}
We follow the execution of the \ref{BroadcastMessage} procedure step by step.

The first step is a decision whether to proceed or to queue \fp{msg} to \lque{}. Both histories take the same decision here, resulting in identical changes to \lque{}. If $\vgap{} > 0$ we are done.

The next step increments $\mpkout{}.b$, keeping it identical.

The next two steps define a temporary vector $\vtime'$. Since $P \ne \plusminusG$, the resulting vector has values at $H$ and $H^r$ that bear the same relationship to each other as \vtime{} does.

The next three steps stamp the message. By Definition \ref{EquivDef} this produces equivalent messages.

The next step creates the queuing event. By Definition \ref{EquivSEDef} this event produces equivalent side effects.

The next step creates  the local variable \se{index} which has the same value in both histories.

The next step creates the local vector $\se{iset}[]$, which is related in $H$ and $H^r$ the same way $\mpkin{}[]$ is.

Since $D \in \lset{}$, we know from Lemma \ref{OmnibusCBCASTLem} that there is a $D$ coordinate in $\mpkin{}[]$ and therefore the record that is added to \bwset{} in the last step bears the required relation for \wset{} records. Therefore this step preserves the inductive hypothesis and we are done.
\end{proof}

\begin{lem}
\label{ScanEquivLem}
Let $P$ and $Q$ be two processes that use the same implementation of the ApplyMessage up-call. Further assume that the states of $P$ and $Q$ have the following relations:
\begin{align*}
\uldata{}(P) & = \uldata{}(Q) \\
\cview{}(P) & = \cview{}(Q) < \vcrit{} \\
\mset{}(P) & = \mset{}(Q) \,\bigcup\, \{ \plusminusG \}   \\
\rset{}(P) & \cong \rset{}(Q) \\
\vtime{}[](P) & = \vtime{}[](Q) \,\bigcup\, \left\{ [G] = 0, [\minusG] = 0 \right\}   \\
\end{align*}
Then executing the \ref{ScanCall} procedure in both processes will result in the same relations being preserved, with all other state variables remaining unchanged in both $P$ and $Q$. Moreover both processes invoke the ApplyMessage up-call at the same times and with the same messages.
\end{lem}

\begin{proof}
We follow the execution of the \ref{ScanCall} procedure step by step.

The first step sets the value of $\lv{deliverable\_messages\_found}$ to $\lv{false}$ at both $P$ and $Q$.

The next step is a loop over members of \rset{}. Since \rset{} is equivalent at $P$ and $Q$, the loop goes over the same messages in the same order in both executions. For each message, we check whether the message is a current view message and whether we have already delivered all the previous messages from the same source. This amounts to checking whether $\mview{\fp{msg}} = \cview{}$ and whether $\mvt{\fp{msg}}[\orig{\fp{msg}}] = \vtime{}[\orig{\fp{msg}}]+1$.

The first check gives the same result in $P$ and $Q$ because $\mview{\fp{msg}}$ is the same (due to equivalence) and \cview{} is assumed to be the same. Moreover, if the check is successful it means that $\mview{\fp{msg}} = \cview{} < \vcrit{}$ and therefore $\orig{\fp{msg}} \ne \plusminusG$. This is because $G$ starts life with $\cview{} + \vgap{} = \vcrit{}$ so it does not originate any messages before $\cview{} = \vcrit{}$ and $\vgap{} = 0$. $\minusG$ does not originate any messages at all.

The second check gives the same result because of equivalence and because of our assumption that $\orig{\fp{msg}} \ne \plusminusG$. Therefore $\mvt{\fp{msg}}[\orig{\fp{msg}}]$ and $\vtime{}[\orig{\fp{msg}}]$ are the same at $P$ and $Q$.

Therefore the conditional block is executed for the same messages, following the steps:
\begin{itemize}
\item \lv{all\_dependents\_delivered} is set to \lv{true} in both $P$ and $Q$.
\item A loop searches \mset{} for coordinates whose values will prevent delivery. Other than $\plusminusG$ the set \mset{} in $Q$ contains the same processes as in $P$ and the test yields the same results in both. The only way this loop could produce divergent results is if $Q$ decided that the message was deliverable while $P$ found an impediment that is related to $\fp{pid} = \plusminusG$. But the equivalence condition on \rset{} guarantees that $\mvt{\fp{msg}}[\plusminusG] = 0$ in $P$, so this cannot happen and the loop must exit with the same value of \lv{all\_dependents\_delivered} in both executions.
\item The decision to deliver the message is controlled by \lv{all\_dependents\_delivered} so both $P$ and $Q$ make the same decision. The message delivery takes the following steps:
\begin{itemize}
\item The value of $\lv{deliverable\_messages\_found}$ is set to $\lv{true}$ at both $P$ and $Q$.
\item The vector time is incremented at the $\orig{\fp{msg}}$ coordinate. This is the same coordinate in $P$ and $Q$ due to equivalence, and since it is not $\plusminusG$, we increment an identical value in \vtime{}, keeping it identical.
\item The message is removed from \rset{} and stripped of its vector time and view stamps, leaving it with only its origin stamp. Since the origin stamp is identical for equivalent messages, this leaves the message identical in $P$ and $Q$, in addition to leaving \rset{} equivalent.
\item The identical message is applied to the identical user data object \uldata{} in the same way. This leaves the data object identical.
\end{itemize}
\end{itemize}
The last step is a recursive call to \ref{ScanCall}, controlled by \lv{deliverable\_messages\_found}. We have shown that up to this point the assumed relationships have not changed and no variables changed other than some of the ones mentioned. Since $P$ and $Q$ take the same decision about making the recursive call, the lemma is now  proven by induction.
\end{proof}

\begin{lem}
Let $P$ and $Q$ be two processes whose states have the following relations:
\begin{align*}
\fwset{}(P) = \emptyset & \text{ iff } \fwset{}(Q) = \emptyset \\ 
\bwset{}(P) = \emptyset & \text{ iff } \bwset{}(Q) = \emptyset \\ 
\cview{}(P) & = \cview{}(Q) \\
\vgap{}(P) & = \vgap{}(Q) \\
\sgh{}(P) & = \sgh{}(Q) \\
\sfh{}(P) & = \sfh{}(Q) \\
\end{align*}
Then executing the \ref{CheckFlush} procedure in both processes has the following two results:
\begin{itemize}
\item The same relations are preserved, with all other state variables remaining unchanged in both $P$ and $Q$.
\item The same side effects occur in $P$ and $Q$.
\end{itemize}
\end{lem}

\begin{proof}
We follow the execution of \ref{CheckFlush} step by step.

The first step exits if \fwset{} is not empty. $P$ and $Q$ make the same decision here by assumption.

The next block compares \sgh{} to $\cview{} + \vgap{}$. If \sgh{} is low it updates it and broadcasts ghost packets of height $\cview{} + \vgap{}$. Since \sgh{}, \cview{} and \vgap{} are all the same in $P$ and $Q$, this block results in the same side effects and preserves the postulated relations without changing any other variables.

The rest of the procedure is a repeat of the first part with \bwset{} and \sfh{} replacing \fwset{} and \sgh{}, so the same argument holds.
\end{proof}

\begin{lem}
\label{CheckFlushEquivLemTwo}
Let $P$ and $Q$ be two processes whose states have the following relations:
\begin{align*}
\fwset{}(P) = \emptyset & \text{ iff } \fwset{}(Q) = \emptyset \\ 
\bwset{}(P) & = \emptyset \\
\cview{}(P) & = \cview{}(Q) \\
\vgap{}(P) & = \vgap{}(Q) \\
\sgh{}(P) & = \sgh{}(Q) \\
\sfh{}(P) & = \sgh{}(Q) \\
\end{align*}
Then executing the \ref{CheckFlush} procedure in both processes has the following two results:
\begin{itemize}
\item The same relations are preserved, with all other state variables remaining unchanged in both $P$ and $Q$.
\item If $Q$ broadcasts ghost packets of height $v$ then $P$ broadcasts ghost packets of height $v$ followed by flush packets of height $v$.
\item If $Q$ does not broadcast ghost packets, then $P$ has no side effects.
\end{itemize}
\end{lem}

\begin{proof}
We follow the execution of \ref{CheckFlush} step by step.

The first step exits if \fwset{} is not empty. $P$ and $Q$ make the same decision here by assumption. If they exit then neither has side effects and we are done.

The next block compares \sgh{} to $\cview{} + \vgap{}$. If \sgh{} is low it updates it and broadcasts ghost packets of height $\cview{} + \vgap{}$. Since \sgh{}, \cview{} and \vgap{} are all the same in $P$ and $Q$, this block results in the same side effects and preserves the postulated relations without changing any other variables.

The rest of the procedure is a repeat of the first part with \bwset{} and \sfh{} replacing \fwset{} and \sgh{}. If $P$ and $Q$ decided to broadcast ghost packets, then our assumptions will force $P$ to broadcast flush packets as well, as the lemma claims. Since the $P$ decision to broadcast flush packets follows $Q$'s decision to broadcast ghost packets, the value of \sfh{} in $P$ ends up being equal to the value final value of \sgh{} in $Q$, as claimed.
\end{proof}

\begin{cor}
\label{CheckFlushEquivCor}
\begin{enumerate}
\item Let $P$ be a process in $H$ and suppose that the state relationships in one of the inductive hypotheses or sub-hypotheses, excluding the Second Pre-Critical Hypothesis, hold with respect to states of $P$ in $H$ and $H^r$. Then executing the \ref{CheckFlush} procedure at $P$ in $H$ and $H^r$ preserves the same relations and causes the same side effects.
\item Suppose that the state relationships in the Second Pre-Critical Hypothesis hold at $D$, $G$ and $\minusG$ in $H^r$. Then executing the \ref{CheckFlush} procedure in all three processes preserves the same relations and causes side effects that are related as stated in Lemma \ref{CheckFlushEquivLemTwo}.
\end{enumerate}
\end{cor}

\begin{proof}
\end{proof}

\begin{lem}
\label{FirstTTIEquivLem}
Let $X$ be a process in $H$ that is in the middle of the execution of a constellation in both $H$ and $H^r$ (if $X = G$ then the constellation must be post-critical). Suppose that the First Pre-Critical, Interim non-$G$, Interim $G$ or Convergent Hypothesis holds with respect to the state of $X$ in $H$ and $H^r$ (if $X = G$ then the pre-critical case does not apply). Also assume that $\ulview{} = \cview{}^r$ at $X$ in $H^r$. Then executing the \ref{TryToInstall} procedure in both histories at $X$ has the following results:
\begin{itemize}
\item If $\cview{} < \vcrit{}$ in $H$ at the end of the execution then the same inductive hypothesis (First Pre-Critical or Interim) still holds. 
\item If $\cview{} \ge \vcrit{}$ in $H$ at the end of the execution then the Convergent Hypothesis holds.
\item If $P = G$ and the procedure installs the critical view then $G$ launches the \ulp{} thread at exactly the same moment in $H$ and $H^r$.
\item In all cases the side effects in $H^r$ are equivalent to the side effects in $H$.
\end{itemize}
\end{lem}

\begin{proof}
We follow the execution step by step:

The procedure starts with a loop that goes over the live processes, looking for an impediment to installation of the next view in the form of $\fvec[]$ values that are too low. In all the situations under consideration we have
\begin{align*}
\cview{}^r & = \cview{} \\
\vgap{}^r & = \vgap{} \\
\lset{}^r & \supset \lset{} \\
\fvec{}^r[Y] & = \fvec{}[Y] \quad \text{whenever } Y \in \lset{} 
\end{align*}
Therefore if there is an impediment to installation in $H$, the same impediment exists in $H^r$. The converse is trivially true in all but the pre-critical case, because that is the only case where $\lset{}^r \ne \lset{}$. However in that case we have, by the First Pre-Critical Hypothesis:
$$
\fvec{}^r[\plusminusG] = \gvec[D] \ge \fvec[D]
$$
where the inequality on the right follows from Lemma \ref{OmnibusCBCASTLem}(\ref{OCL:gf}). Therefore if $G$ or $\minusG$ form an impediment in $H^r$, then so does $D$ in $H$. Therefore both histories make the same decision on whether to proceed with installing all the pending views.

The next part of the procedure is a loop that installs all the pending views. It goes through  the following steps:
\begin{itemize}
\item The obsolete messages are removed from \rset{}. Because $\rset{}^r \cong \rset{}$, they have the same messages and each message has the same message view. Therefore the same messages are discarded in $H$ and $H^r$ and the hypothesis is preserved. In the case where the view being installed is the critical one, this step leaves \rset{} with no messages of pre-critical view. By the definition of equivalence (Definition \ref{EquivDef}) this means that $\rset{}^r = \rset{}$.
\item The obsolete messages are removed from $\fque{}[]$. This is very similar to the previous step. The only complication is that in the pre-critical case there are two queues, $\fque{}^r[\plusminusG]$, that do not exist in $H$. However these queues are empty according to the First Pre-Critical Hypothesis, and so this step preserves the inductive hypothesis as well. As in the previous case, if the view being installed is the critical one, then at this stage $\fque{}^r[X] = \fque{}[X]$ at every process $X$.
\item The next two steps increment \cview{} and decrement \vgap{}. This obviously preserves the inductive hypothesis. Also $\ulview{} = \cview{}^r - 1$. This will be fixed soon.
\item The next step pops the head of \pque{}. In most cases this trivially preserves the inductive hypothesis, as well as yielding an identical value for \lv{notification}. However in the case where the view being installed is the critical one, this step renders \pque{} identical in both histories while yielding different values for \lv{notification}. In $H$ we have a value that indicates that $G$ is joining while in $H^r$ we have a value that indicates that $\minusG$ is leaving.
\item The next step updates \mset{} according to the type of notification, and applies the change to the replicated application data. We have to consider the following cases:
\begin{enumerate}
\item The notification is for a non-critical joining of a process which is not the local process $X$. Since the first join in $H$ occurs at the critical view we have
$$
\cview{} \ge \vcrit{} > \magicview
$$
Both histories add the new process to \mset{}, preserving the inductive hypothesis. Then in $H$ the process $X$ invokes the up-call ApplyJoin@X, which modifies \uldata{} in an unspecified user-defined way. In $H^r$ the process $X$ invokes ApplyJoin$^r$@X. According to \ref{ReducedUAI}, ApplyJoin$^r$@X behaves the same way in this case, regardless of whether $X = G$ or not. It increments \ulview{} and restores the equality $\ulview{} = \cview{}^r$. Since \ulview{} is positive, the up-call invokes the original ApplyJoin@X. Therefore \uldata{} is modified in the same way in $H^r$ and the inductive hypothesis is preserved.
\item The notification is for a non-critical joining of the local process $X$. This case is just like the previous one but in addition process $X$ also launches the main \ulp{} thread, with the same identity parameter \fp{pid} in both histories. We may assume that the thread will start executing at the exact same time in both histories (because it is possible, not because it is probable).
\item The notification is for a non-critical removal of a process. In both histories $X$ removes the process from \mset{}, preserving the inductive hypothesis. Then in $H$ it invokes the up-call ApplyRemoval@X, which modifies \uldata{} in an unspecified user-defined way. In $H^r$ it invokes ApplyRemoval$^r$@X. If $X \ne G$ then according to \ref{ReducedUAI} the up-call simply increments \ulview{}, restoring the equality $\ulview = \cview{}^r$, and then calls ApplyRemoval@X, thus preserving the inductive hypothesis. If $X = G$ then ApplyRemoval$^r$@G increments \ulview{} (restoring the equality with $\cview{}^r$) and then invokes either ApplyRemoval@G or ApplyRemoval@D, according as $\ulview{} > \magicview$ or not. By Definition \ref{MagicDef} $\ulview{} > \magicview$ exactly when the constellation is post-critical, which must be the case here when $X = G$, therefore ApplyRemoval@G is invoked and \uldata{} is modified the same way in both histories.
 
\item The notification is for the critical view and $X \ne G$. In this case process $X$ in $H$ invokes the up-call ApplyJoin@X and in $H^r$ the up-call ApplyRemoval$^r$@X. By \ref{ReducedUAI} the ApplyRemoval$^r$@X call increments \ulview{}, thus restoring the equality $\ulview{} = \cview{}^r$. The up-call proceeds to invoke ApplyJoin@X. This modifies \uldata{} the same way as in $H$.
\item The notification is for the critical view and $X = G$. In this case process $G$ in $H$ invokes the up-call ApplyJoin@G and then launches the \ulp{} thread with $\fp{pid} = G$. In $H^r$ it invokes ApplyRemoval@G. By \ref{ReducedUAI} the ApplyRemoval$^r$@G call increments \ulview{}, thus restoring the equality $\ulview{} = \cview{}^r$. The up-call proceeds to invoke ApplyJoin@G with $\fp{pid} = G$ exactly as $G$ did in $H$. This preserves the equality of \uldata{}.

In $H^r$ the \ulp{} thread is not launched. But since $G$ is an original process in $H^r$ the thread was already launched, with the same paramter, at the beginning of time when the \ref{StartCluster} procedure was invoked. Since both invocations are asynchronous we may assume that by sheer luck the early invocation of the thread in $H^r$ is delayed so much that the thread starts execution at the exact same point in time in both histories.
\end{enumerate}
\item The vector time is reset. This means that the previous vector time is replaced with a vector of zeroes, one per process in \mset{}. This is easily seen to preserve the inductive hypothesis.

At this point we have to take stock of the case where the view installation was the critical one. This is the boundary between the Interim period, where the Interim Hypotheses are in force, and the Convergent period. Following the steps we took so far demonstrates that all the differences that existed in the state of the process in the two histories have now dissolved. The $\plusminusG$ difference in \mset{} has been bridged. As a result \vtime{} converged as well. \pque{} shed the one record that was different and became equal. \rset{} and $\fque{}[]$ have gone from being equivalent to being equal, as we have seen.

Since we managed to install the views we know that the self flush height is high namely
$$
\fvec{}[X] = \cview{} + \vgap{}
$$
From Lemma \ref{OmnibusCBCASTLem}(\ref{OCL:sendgfA}) it follows that $\sfh{} = \cview{} + \vgap{}$. Since we start out with $\vgap{} > 0$ (otherwise no views are installed) we know by Lemma \ref{OmnibusCBCASTLem}(\ref{OCL:sendgfB}) that \wset{} must be empty at this point in both histories and therefore equal as well. The only remaining possible difference is in \cset{}, in the case $X = G$ only. But this difference cannot exist here because Lemma \ref{OmnibusCBCASTLem}(\ref{OCL:gf}) implies that if $Y \in \lset{} \setminus \cset{}$ then $\fvec{}[Y] < \cview{} + \vgap{}$. This would have prevented the views from being installed. As a result the Convergent Hypothesis holds, and the computations have now converged for $X$.
\item the next step is an invocation of the \ref{ScanCall} procedure. By Lemma \ref{ScanEquivLem}, this step preserves the inductive hypothesis and creates no side effects
\end{itemize}

The last step involves broadcasting all the messages out of \lque{}.

At this point $\vgap{} = 0$ and therefore we cannot be in the interim period. Therefore we only need to consider the First Pre-Critical and the Convergent Hypotheses.
Under either hypothesis $\lque{}^r = \lque{}$ and by Lemma \ref{BroadcastEquivLem} each call to \ref{BroadcastMessage} preserves the hypothesis and generates equivalent side effects in both histories. From Lemma \ref{ReducedContactLem} it follows that the target set of each side effect is equal to $\cset{}^r$ in $H^r$.
\end{proof}

\begin{lem}
\label{SecondTTIEquivLem}
Suppose that $G$, $\minusG$ and $D$ are in the middle of the execution of a pre-critical constellation in $H^r$. Suppose that the Second Pre-Critical Hypothesis holds. Also suppose that $\ulview{} = \cview{}^r$ in all three processes. Then executing the \ref{TryToInstall} procedure in all three processes preserves the inductive hypothesis and produces no side effects in either $G$ or $\minusG$. 
\end{lem}

\begin{proof}
We follow the execution step by step:

The procedure starts with a loop that goes over the live processes, looking for an impediment to installation of the next view in the form of $\fvec[]$ values that are too low. In the situation under consideration we have the same values of \cview{}, \vgap{}, \lset{}  and $\fvec{}[]$. Therefore all three processes reach the same decision on installation of the pending views. Since \pque{} is also identical among the processes, the same views get installed.

The next part of the procedure is a loop that installs all the pending views. It goes through  the following steps:
\begin{itemize}
\item The obsolete messages are removed from \rset{} and $\fque{}[]$. Because both sets are identical in all three processes, the same messages are discarded in all three processes and the hypothesis is preserved.
\item The next two steps increment \cview{} and decrement \vgap{}. This obviously preserves the inductive hypothesis. It also results in $\ulview{} = \cview{}^r - 1$.
\item The next step pops the head of \pque{}. This trivially preserves the inductive hypothesis, as well as yielding an identical value for \lv{notification}.
\item The next step updates \mset{} according to the type of notification, and applies the change to the replicated application data. Since we are dealing with a pre-critical constellation, the view change must be a removal of a process, and by Definition \ref{MagicDef} $\cview{} \le \magicview$. Process $D$ executes ApplyRemoval$^r$@D while $\plusminusG$ execute ApplyRemoval$^r$@G. By \ref{ReducedUAI}, this results in all cases in incrementing \ulview{}, which restores the equality $\ulview{} = \cview{}^r$, and in invoking ApplyRemoval@D, which modifies \uldata{} in the same way in all three processes.
\item The vector time is reset. This means that the previous vector time is replaced with a vector of zeroes, one per process in \mset{}. This is easily seen to preserve the inductive hypothesis.
\item the next step is an invocation of the \ref{ScanCall} procedure. By Lemma \ref{ScanEquivLem}, this step preserves the inductive hypothesis and creates no side effects.
\end{itemize}

The last step involves broadcasting all the messages out of \lque{}. By the Second Pre-Critical Hypothesis \lque{} is empty in $\plusminusG$, so this step generates no side effects there, while generating a single original message broadcast per message in $D$. As far as the inductive hypothesis is concerned, for every message in \lque{} that is broadcast by $D$ the value of $\mpkout{}.b$ is incremented and a copy of the stamped message is attached to \bwset{} in $D$ together with an instability vector. In addition \lque{} is emptied out in $D$. These state changes do not violate the Second Pre-Critical Hypothesis, which only requires that \bwset{} and \lque{} be empty and $\mpkout{}.b$ be zero in $\plusminusG$.
\end{proof}

\subsection{Proof of the History Equivalence Theorem}

We start the proof at the beginning of time. At time zero, each member process is initialized by the \ref{StartCluster} procedure (see \ref{InitializationSSS}). Direct inspection shows that the state differences between $H$ processes and $H^r$ processes conform to the pre-critical inductive hypotheses. Notice also that after \ref{StartCluster} is executed in $H^r$
$$
\ulview{} = \cview{}^r = 0
$$
It is easy to check that the equality between \ulview{} and $\cview{}^r$ continues to hold as long as $H^r$ executes \cbcast{}. This is because the identity is passed from parent to child, and the only place where either value is changed is in the \ref{TryToInstall} procedure, where these two values are incremented in tandem.

Suppose that the inductive hypothesis holds at constellation label $L = \lblgzz{t}{s}$. Both $H$ and $H^r$ have a set of transactions or sub-transactions which share the constellation label $L$. But in $H$ these transactions are generated by the execution of \cbcast{} procedures as a reaction to triggers, whereas in $H^r$ the transactions are manufactured artificially. We have to show three things. First we must show that if $H^r$ executes the \cbcast{} protocol it will produce transactions that are identical to its observed ones. Then we have to show that if $H^r$ executes the \cbcast{} protocol, then once all the (sub-)transactions in the constellation conclude the inductive hypothesis continues to hold. As a final step we have to show that the triggers of the next constellation in $H^r$ arise naturally from the past behavior of $H^r$. This guarantees that $H^r$ continues to look like an execution of \cbcast{} and \ulp{} all the way up to the next constellation, and that the inductive hypothesis continues to hold until that point.

We divide our analysis at each stage into cases, depending on the type of constellation, and the type of side effects (See \ref{FirstSideEffectSSS} - \ref{LastSideEffectSSS}) we observe in $H$ within the constellation. In each case we try to deal with all of the inductive hypotheses together. In some cases only some of the inductive hypotheses are relevant. For example, join notifications do not occur pre-critically, by definition.

We also make a repeated use of arguments that show that certain variables that start out equivalent at the start of a constellation end up remaining equivalent at the end. In the Convergent period these variables are required to be identical and not just equivalent, so proving that they remain equivalent is not quite enough. But as long as the procedural parameters that are used by \cbcast{} to execute each transaction are identical in both histories, the state will remain identical at the end rather than just equivalent. We will point out the few cases where the procedural parameters are not equal, and otherwise gloss over this issue.

\subsubsection{Notification constellations}

\begin{itemize}
\item
\begin{description}
\item[\bf Type of Constellation:] A non-critical \gms{} removal notification of a process $R$.
\item[\bf $H$ Transactions:] One execution of the \ref{RemovalNotification} procedure at each surviving process.
\item[\bf Observed behavior in $H^r$:] Post-critically, the transactions occurring in $H^r$ are equivalent to the $H$ transactions occurring at the same processes. Pre-critically, equivalent transactions occur at the non-$G$ processes, and in addition there are two transactions at $\plusminusG$. These transactions have the same triggers as the other transactions in the constellation, and their side effects are related to the original transaction at $D$ according to the following cases (see Lemma \ref{ViewChangeSELem}):
\begin{itemize}
\item If the original $D$ transaction in $H$ queues a sequence of message multicasts then the $\plusminusG$ transactions in $H^r$ queue the equivalent message multicasts.  
\item If the original $D$ transaction in $H$ has no side effects then the $\plusminusG$ transactions in $H^r$ have no side effects.
\item If the original $D$ transaction in $H$ queues a ghost multicast (with or without a succeeding flush multicast) then the $\plusminusG$ transactions in $H^r$ queue a ghost multicast followed by a flush multicast, of the same height as the original ghost multicast.
\end{itemize} 

\item[\bf Execution of \cbcast{} in $H^r$:] We follow the execution step by step, dealing with all the inductive hypotheses at once. The first step increments $\vgap{}$ in all the executions. By the inductive hypothesis $\vgap{}$ is identical at all compared processes. After incrementation, $\vgap{}$ remains identical since all the involved processes execute the procedure. Therefore the first step preserves the inductive hypothesis. This step produces no side effects.

The next step adds a record to \pque{}. In the pre-critical and convergent periods $\pque{}^r = \pque{}$ and the added record is identical as well, so the First Pre-Critical Hypothesis and the Convergent Hypothesis are preserved. In the interim period $\pque{}^r$ is not identical to $\pque{}$ but the difference is preserved after the addition of the new record, so the Interim Hypotheses are preserved as well. Finally, pre-critically $\pque{}^r$ is identical at $D$, $G$ and $\minusG$ so the Second Pre-Critical Hypothesis is also preserved.

The next two steps remove $R$ from \lset{} and \cset{}. Pre-critically both $\lset{}^r$ and $\cset{}^r$ differ from $\lset{}$ and $\cset{}$ by the addition of $\plusminusG$, and both sets are identical at $D$, $G$ and $\minusG$ in $H^r$. Since $R \notin \{ D, G, \minusG \}$ in the pre-critical case, the First and Second Pre-Critical Hypotheses are preserved. Post-critically $\lset{}^r = \lset{}$ and this property is preserved. Away from $G$ the same is true for $\cset{}^r$. Therefore the Interim non-$G$ and Convergent Hypotheses are preserved. At $G$, $\cset{}^r = \lset{}$, so the Interim $G$ Hypothesis is preserved as well.

The next step is a loop that stabilizes messages in $\wset{}$ with respect to $R$. Post-Critically $\wset{}^r \cong \wset{}$ in all cases. The loop preserves this relationship, which guarantees the preservation of all the post-critical inductive hypotheses.

Pre-critically, $\bwset{}^r$ is empty at $\plusminusG$ while $\fwset{}^r$ is identical at $D$, $G$ and $\minusG$ with the exception of a slightly different $\fp{index}$ at each record. The loop preserves these relationships, which guarantees the preservation of the Second Pre-Critical Hypothesis.
 
Pre-critically at each non-$G$ process $\wset{}^r$ is very close to being equivalent to $\wset{}$ with the only differences being larger instability vectors. We know that in the pre-critical case $R \notin \{ D, G, \minusG \}$ and therefore removing $R$ from the instability sets in $\wset{}^r$ and $\wset{}$ does not disturb their hypothesized relationship. Finally this fact also guarantees that each instability set will empty out in $\wset{}$ if and only if it empties in $\wset{}^r$ so the loop will always make the same decisions on record removal from $\wset{}$ in both histories. Therefore the First Pre-Critical Hypothesis is preserves as well.

The next step discards $\mpkin{}[R]$. This preserves all the hypotheses.

The next step forwards all the messages out of the forwarding queue $\fque{}[R]$. In all periods and at all participating processes $\fque{}^r[R] \cong \fque{}[R]$, while pre-critically $\fque{}^r[R]$ is identical at $D$, $G$ and $\minusG$. Therefore the loop proceeds over equivalent or identical messages (depending on which hypothesis we are checking), executing the following steps at each message:
\begin{itemize}
\item Popping the head of $\fque{}[R]$, yielding a message \fp{msg} - this preserves all the hypotheses. $\fp{msg}^r \cong \fp{msg}$ at all processes in all periods. Pre-critically at $D$, $G$, and $\minusG$ the value of $\fp{msg}^r$ is the same.
\item Incrementing $\mpkout{}.f$ - this preserves all the hypotheses.
\item Creating $\se{index} = \mpkout{}$. This yields $\se{index}^r = \se{index}$ for all processes in all periods. Pre-critically at $D$, $G$ and $\minusG$ this yields:
\begin{align*}
\se{index}^r.f(\plusminusG) & = \se{index}^r.f(D) \\
\se{index}^r.b(\plusminusG) & = 0 
\end{align*}
\item Creating an instability vector $\se{iset} = \mpkin{}[]$. This yields $\se{iset}^r = \se{iset}$ for all processes in all the post-critical periods. Pre-critically it yields an identical value of $\se{iset}^r$ at $D$, $G$ and $\minusG$. Pre-critically this also yields the following relation at each non-$G$ process:
\begin{align*}
\se{iset}^r[X].f & =
\begin{cases}
\se{iset}[X].f & X \ne \plusminusG  \\
\se{iset}[D].f & X = \plusminusG 
\end{cases}   \\
\se{iset}^r[X].b & =
\begin{cases}
\se{iset}[X].b & X \ne \plusminusG  \\
0 & X = \plusminusG 
\end{cases}
\end{align*}
\item Queuing a message multicast containing the message $\fp{msg}$ and targeted at \cset{}. Since $\fp{msg}^r \cong \fp{msg}$ this side effect matches the observed side effect in $H^r$. Since the inductive hypothesis holds at this point, Lemma \ref{ReducedContactLem} guarantees that the target set matches the observed target set.
\item Adding a record $\langle \fp{msg}, \se{index}, \se{iset}[] \rangle$ to \fwset{}. Post critically $\fwset{}^r = \fwset{}$ at every process and the new record is equivalent as well, which preserves the relevant hypothesis.

Pre-critically at $D$, $G$ and $\minusG$ in $H^r$, $\fwset{}^r$ is almost identical, except for a zeroed out $.b$ coordinate at each index at $\plusminusG$. The new record $\langle \fp{msg}^r, \se{index}^r, \se{iset}^r[] \rangle$ exhibits exactly this behavior at the three processes so the Second Pre-Critical Hypothesis is preserved.

Pre-critically at any non-$G$ process $\wset{}^r$ is nearly equivalent to $\wset{}$, with the addition of $\plusminusG$ coordinates to $\se{iset}^r$ wherever a $D$ coordinate exists in $\se{iset}$. The value of $\se{iset}^r$ in the new record in $H^r$ exactly matches the required difference from the value of $\se{iset}$ in the new record in $H$, while $\fp{msg}^r \cong \fp{msg}$ and $\se{index} = \se{index}$. Therefore the First Pre-Critical Hypothesis is preserved.
\end{itemize}

In the next few steps we discard $\fque{}[R]$, $\gvec{}[R]$ and $\fvec{}[R]$, actions which preserve all the inductive hypotheses. The one thing to notice is that this follows from the fact that pre-critically $R \notin \{ D, G, \minusG \}$.

In the last step we call \ref{CheckFlush} which by Corollary \ref{CheckFlushEquivCor} preserves the inductive hypotheses. Moreover the same corollary guarantees that the side effects produced by \ref{CheckFlush} are equal, and therefore equivalent, to the observed side effects. Lemma \ref{ReducedContactLem} guarantees that the target set produced by the execution of \ref{CheckFlush} for each side effect, namely $\cset{}^r$, is equal to the observed target set.
\end{description}

\item
\begin{description}
\item[\bf Type of Constellation:] A non-critical \gms{} join notification of a process $J$, with parent process $K$. By definition this must be a post-critical constellation with $J \ne G$.
\item[\bf $H$ Transactions:] One execution of the \ref{JoinNotification} procedure at each existing process. One execution of the \ref{Run} procedure at $J$. $J$ and $K$ have identical states at the start of the transaction.
\item[\bf Observed behavior in $H^r$:] The transactions at all the processes, including $J$, are equivalent to the $H$ transactions at the same processes.
\item[\bf Execution of \cbcast{} in $H^r$:] $J$ and $K$ start with the same state, and the same calls are executed in $H^r$ as in $H$, at the same processes. Since there is no pre-critical case here, we only have to verify the preservation of the Interim non-$G$, Interim $G$ and Convergent Hypotheses. We follow the execution step by step, starting with \ref{JoinNotification}.

The first few steps increment $\vgap{}$, add a record to \pque{}, add $J$ to \lset{} and \cset{} and create a new and empty entry in \fque{} For the process $J$. These steps can be easily seen to preserve all three Hypotheses where they apply.

The next step goes over \wset{}, and creating an instability coordinate for $J$ whenever such a coordinate exists for the parent process $K$. Only the forwarding part of the coordinate is copied. The broadcasting part is zeroed. Post-critically, $\wset{}^r \cong \wset{}$ so the loop proceeds identically in $H$ and $H^r$ and preserves the applicable inductive hypotheses.

The next three steps create $J$ coordinates in the $\gvec{}[]$, $\fvec{}[]$ and $\mpkin{}[]$ vectors. These steps can be easily seen to preserve the applicable hypotheses.

The next two steps create and queue a donation packet to $J$. Since all the ingredients of the donation vector \se{donation} are equivalent, the resulting side effect in $H^r$ is equivalent to the corresponding side effect in $H$, as observed in $H^r$. The target set of the packet is equal to $\{ J \}$ in both histories.

In the last step we call \ref{CheckFlush} which by Corollary \ref{CheckFlushEquivCor} preserves the inductive hypotheses. Moreover the same corollary guarantees that the side effects produced by \ref{CheckFlush} are equal, and therefore equivalent, to the observed side effects. Lemma \ref{ReducedContactLem} guarantees that the target set produced by the execution of \ref{CheckFlush} for each side effect, namely $\cset{}^r$, is equal to the observed target set.

We now look at the execution steps in \ref{Run}. These steps are taken by $J$ in both $H$ and $H^r$. Recall from \ref{AlgOutlineSS} that a new process starts out with a state identical to its parent $K$. Therefore, prior to the execution of \ref{Run} the states of $J$ in $H$ and $H^r$ conform to the inductive hypothesis for the post-critical state of $K$, which can vary according as $K = G$ or $K \ne G$. However, if one looks at these two cases in the inductive hypothesis, one sees that they claim the same things, except that in the case $K \ne G$, $\cset{}^r = \cset{}$ whereas in the case $K = G$, $\cset{}^r = \lset{}$. This will make no difference in the end because the \ref{Run} procedure recomputes \cset{} in a way that renders its value identical in both $H$ and $H^r$ as is expected since $J \ne G$.

The first three steps in \ref{Run} increment \vgap{} and update \pque{} and \lset{}. These steps are easily seen to preserve either the Interim Non-$G$ Hypothesis or the Convergent Hypothesis, as the case may be. The Interim $G$ Hypothesis is mostly preserved except that now $\cset{}^r \ne \lset{}$. This violation is corrected in the next step.

The next step resets \cset{} to include the process identifier of $J$ only. This step forces \cset{} to be identical in $H$ and $H^r$. This step preserves the Interim Non-$G$ and the Convergent Hypotheses. If the Interim $G$ Hypothesis applied at the start of the constellation, now the non-$G$ Hypothesis would apply, which is the appropriate hypothesis for $J$.

All the following steps except the last one rather trivially preserve the appropriate Hypothesis (either the Non-$G$ Hypothesis or the Convergent Hypothesis).

In the last step we call \ref{CheckFlush} which by Corollary \ref{CheckFlushEquivCor} preserves the inductive hypotheses. Moreover the same corollary guarantees that the side effects produced by \ref{CheckFlush} are equal, and therefore equivalent, to the observed side effects. Lemma \ref{ReducedContactLem} guarantees that the target set produced by the execution of \ref{CheckFlush} for each side effect, namely $\cset{}^r$, is equal to the observed target set.
\end{description}

\item
\begin{description}
\item[\bf Type of Constellation:] The constellation is the critical \gms{} notification.
\item[\bf $H$ Transactions:] One execution of the \ref{JoinNotification} procedure at each existing process. One execution of the \ref{Run} procedure at $G$. $G$ and $D$ have identical states at the start of the transaction.
\item[\bf Observed behavior in $H^r$:] The transactions at all the processes, including $G$, are equivalent to their counterparts in $H$ with the exceptions that in $H^r$ the trigger notification is a removal notification rather than a join notification, and the donation packet queuing event at non-$G$ processes in $H$ is missing in $H^r$.
\item[\bf Execution of \cbcast{} in $H^r$:] The surviving processes, including $G$, execute the \ref{RemovalNotification} procedure.

Not surprisingly, this case is more subtle than the other cases. First, the processes execute a different call in $H$ and $H^r$. Second, the critical notification separates the pre-critical period from the interim period. Therefore, when we argue that the constellation preserves the inductive hypothesis, we are really saying that if the First and Second Pre-Critical Hypotheses hold before the execution of the respective calls, then the Interim $G$ and Interim non-$G$ Hypotheses hold after the execution.

There is one last complication. Every process that participates in the critical constellation executes exactly one of three procedures. In $H^r$, every process executes the \ref{RemovalNotification} procedure. In $H$, every process except $G$ executes the \ref{JoinNotification} procedure while $G$ executes the \ref{Run} procedure.

The fact that \ref{RemovalNotification} replaces \ref{JoinNotification} at non-$G$ processes explains the observed absence of a queuing of donation packets at those processes in $H^r$.

Each of these calls ends with a call to \ref{CheckFlush}, which is known to preserve all the inductive hypotheses (see Corollary \ref{CheckFlushEquivCor}). Therefore in the following analysis we are going to ignore the call to \ref{CheckFlush}. We will show that the inductive hypothesis is preserved if each process pauses just before executing this call. Then Corollary \ref{CheckFlushEquivCor} together with Lemma \ref{ReducedContactLem} will complete the argument.

Due to the fact that different code is executed in $H$ and $H^r$ we cannot use our usual method of following the parallel execution step by step. Instead we prove this case one state variable at a time. For each variable we show that if the pre-critical hypotheses hold at the start of the execution, then the interim hypotheses ($G$ and non-$G$) hold for that variable at the point where \ref{CheckFlush} is invoked.

\begin{description}
\item[\cview{}] \hfill \\
Non-$G$ case: Initially $\cview{}^r = \cview{}$ according to the First Pre-Critical Hypothesis. Both the \ref{JoinNotification} and the \ref{RemovalNotification} procedures leave it unchanged, so it remains equal at the end of the constellation, as required by the Interim non-$G$ Hypothesis.

$G$ case: Initially
$$
\cview{}^r(G) = \cview{}^r(D)
$$
according to the Second Pre-Critical Hypothesis. At the same time
$$
\cview{}^r(D) = \cview{}(D)
$$
according to the First Pre-Critical Hypothesis, and
$$
\cview{}(D) = \cview{}(G)
$$
because $G$ starts life with a state that is identical to the state of its parent (see \ref{AlgOutlineSS}). Therefore
$$
\cview{}^r(G) = \cview{}(G)
$$
at the start of the constellation. Since the current view is not changed by any of the three procedures, it remains equal at the end of the constellation, as required by the Interim $G$ Hypothesis.

Since this argument will come up repeatedly, we will refer to it as the {\em sameness argument}.

\item[\vgap{}] \hfill \\
Non-$G$ case: Initially $\vgap{}^r = \vgap{}$. Both the \ref{JoinNotification} and the \ref{RemovalNotification} procedures increment the view gap, so it remains equal as required.

$G$ case: Initially $\vgap{}^r(G) = \vgap{}(G)$ by the sameness argument. Both the \ref{Run} and the \ref{RemovalNotification} procedures increment the view gap, so it remains equal as required.

\item[\mset{}] \hfill \\
Non-$G$ case: Initially $\mset{}^r = \mset{} \,\bigcup\, \{ \plusminusG \}$. Neither the \ref{JoinNotification} procedure nor the \ref{RemovalNotification} procedure change this variable, so the same relation continues to hold at the end of the constellation, as required.

$G$ case: Initially $\mset{}^r = \mset{} \,\bigcup\, \{ \plusminusG \}$ as can be easily seen by a slight modification of the sameness argument. Neither the \ref{Run} procedure nor the \ref{RemovalNotification} procedure change this variable, so the same relation continues to hold at the end of the constellation, as required.

\item[\pque{}] \hfill \\
Non-$G$ case: Initially $\pque{}^r = \pque{}$. The \ref{JoinNotification} procedure appends a $\langle \operatorname{JOIN}, G \rangle$ record to \pque{}, while the \ref{RemovalNotification} procedure adds a $\langle \operatorname{REMOVE}, \minusG \rangle$ record, which is exactly what is required by the Interim non-$G$ Hypothesis.

$G$ case: Initially $\pque{}^r = \pque{}$ by the sameness argument. The \ref{Run} procedure appends a $\langle \operatorname{JOIN}, G \rangle$ record to \pque{}, while the \ref{RemovalNotification} procedure adds a $\langle \operatorname{REMOVE}, \minusG \rangle$ record to $\pque{}^r$, which is exactly what is required by the Interim $G$ Hypothesis.

\item[\lset{}] \hfill \\
Non-$G$ case: Initially $\lset{}^r = \lset{} \,\bigcup\, \{ \plusminusG \}$. The \ref{JoinNotification} procedure adds $G$ to \lset{}, while the \ref{RemovalNotification} procedure removes $\minusG$ from $\lset{}^r$, so $\lset{}^r$ becomes equal to $\lset{}$ as required by the Interim Non-$G$ Hypothesis.

$G$ case: Initially $\lset{}^r = \lset{} \,\bigcup\, \{ \plusminusG \}$ as can be easily seen by a slight modification of the sameness argument. The \ref{Run} procedure adds $G$ to \lset{} while the \ref{RemovalNotification} procedure removes $\minusG$ from $\lset{}^r$, $\lset{}^r$ becomes equal to $\lset{}$ as required.

\item[\cset{}] \hfill \\
Non-$G$ case: Initially $\cset{}^r = \cset{} \,\bigcup\, \{ \plusminusG \}$. The \ref{JoinNotification} procedure adds $G$ to \cset{}, while the \ref{RemovalNotification} procedure removes $\minusG$ from $\cset{}^r$, so $\cset{}^r$ becomes equal to $\cset{}$ as required.

$G$ case: It follows from Lemma \ref{OmnibusCBCASTLem}(\ref{OCL:lset}) that initially $\cset{}^r = \lset{}^r$. The \ref{RemovalNotification} procedure removes $\minusG$ from both $\lset{}^r$ and $\cset{}^r$, keeping them equal.

We already know that at the end of the constellation $\lset{}^r = \lset{}$ and therefore we end up with $\cset{}^r = \lset{}$, as required by the Interim $G$ Hypothesis.

\item[$\vtime{}{[]}$] \hfill \\
Non-$G$ case: Initially $\vtime{}^r[] = \vtime{}[] \,\bigcup\, \left\{ [G] = 0, [\minusG] = 0 \right\}$. Neither the \ref{JoinNotification} procedure nor the \ref{RemovalNotification} procedure make changes to $\vtime{}[]$, so the relation remains the same, as required.

$G$ case: Initially $\vtime{}^r[] = \vtime{}[] \,\bigcup\, \left\{ [G] = 0, [\minusG] = 0 \right\}$ as can be easily seen by a slight modification of the sameness argument. Neither the \ref{Run} procedure nor the \ref{RemovalNotification} procedure make changes to $\vtime{}[]$, so the relation remains the same, as required.
\item[\uldata{}] \hfill \\
Non-$G$ case: Initially $\uldata{}^r = \uldata{}$. Neither the \ref{JoinNotification} procedure nor the \ref{RemovalNotification} procedure invoke any of the user up-calls (see \ref{UserApplication}), so the relation remains the same, as required.

$G$ case: Initially $\uldata{}^r = \uldata{}$ by the sameness argument. Neither the \ref{Run} procedure nor the \ref{RemovalNotification} procedure invoke any of the user up-calls, so the relation remains the same, as required.
\item[\rset{}] \hfill \\
Non-$G$ case: Initially $\rset{}^r \cong \rset{}$. Neither the \ref{JoinNotification} procedure nor the \ref{RemovalNotification} procedure make changes to $\rset{}$, so the relation remains the same, as required.

$G$ case: Initially $\rset{}^r \cong \rset{}$ by the sameness argument. Neither the \ref{Run} procedure nor the \ref{RemovalNotification} procedure make changes to $\rset{}$, so the relation remains the same, as required.

\item[$\fque{}{[]}$] \hfill \\
Non-$G$ case: Initially $\fque{}^r[] \cong \fque{}[] \,\bigcup\, \left\{ [G] = \emptyset, [\minusG] = \emptyset \right\}$. In $H^r$ the \ref{RemovalNotification} procedure forwards all the messages in $\fque{}^r[\minusG]$ and removes it. Since this queue is empty, the overall effect is simply the removal of the empty queue without any side effects. In $H$ the \ref{JoinNotification} procedure adds an empty queue for $G$. Therefore $\fque{}^r$ becomes equivalent to $\fque{}$ as required.

$G$ case: Initially $\fque{}^r[] \cong \fque{}[] \,\bigcup\, \left\{ [G] = \emptyset, [\minusG] = \emptyset \right\}$ by the sameness argument. In $H^r$ the \ref{RemovalNotification} removes the empty queue for $\minusG$ without creating any side effects. In $H$ the \ref{Run} procedure adds an empty queue for $G$. Therefore in this case as well the Interim $G$ Hypothesis holds and no side effects are created.

\item[\wset{}, which includes \bwset{} and \fwset{}] \hfill \\
Non-$G$ case: By the First Pre-Critical Hypothesis, the wait set contains equivalent messages in $H^r$ and $H$, with a slight difference in their instability sets. Namely that whenever the instability in $H$ has a $D$ coordinate, the instability in $H^r$ has $\plusminusG$ coordinates in addition. These coordinates have the same $f$ field as the $D$ coordinate, and a zero $b$ field. In $H$, the \ref{JoinNotification} procedure adds the exact same $G$ instability wherever a $D$ instability exists, while in $H^r$ the \ref{RemovalNotification} procedure removes the $\minusG$ instability. If a message becomes stable as a result the \ref{RemovalNotification} procedure removes the message from \wset{}. This however cannot happen because the First Pre-Critical Hypothesis guarantees that the $\minusG$ instability is accompanied by $G$ and $D$ instabilities. Therefore $\wset{}^r$ becomes equivalent to $\wset{}$ as required by the Interim Non-$G$ Hypothesis.

$G$ case: In this case we have to analyze the two parts of \wset{} separately.

We start with \bwset{}. By the Second Pre-Critical Hypothesis $\bwset{}^r$ starts out empty. The \ref{RemovalNotification} procedure does not change this fact, and therefore $\bwset{}^r$ ends up empty. On the other hand in $H$ process $G$ executes the \ref{Run} procedure which empties \bwset{}. So both $\bwset{}^r$ and $\bwset{}$ end up empty and therefore equivalent, as required by the Interim $G$ Hypothesis.

The case of \fwset{} is more complicated. First we verify that $\fwset{}^r$ and $\fwset{}$ contain equivalent messages. Indeed it follows from the Second Pre-Critical Hypothesis that at the start of the transaction $\fwset{}^r(G)$ and $\fwset{}^r(D)$ contain the same messages. By the First Pre-Critical Hypothesis $\fwset{}^r(D)$ and $\fwset{}(D)$ contain equivalent messages. The \ref{RemovalNotification} procedure does not remove any messages from $\fwset{}^r(G)$ because any message in $\fwset{}^r(G)$ that has a $\minusG$ instability also has $D$ and $G$ instabilities. On the other hand $\fwset{}(G)$ is initially equal to $\fwset{}(D)$ and therefore has messages equivalent to the ones in $\fwset{}^r(G)$. The \ref{Run} procedure does not remove any messages from $\fwset{}(G)$. Therefore at the end of the constellation $\fwset{}^r(G)$ and $\fwset{}(G)$ contain equivalent messages.

Let \fp{msg} be any message in $\fwset{}(D)$ at the critical moment. Suppose that its record there is $\langle \fp{msg}, \fp{index}, \fp{iset} \rangle$. By the \AxProcIV{} \fp{iset} does not contain any $D$ instability and therefore it follows from the First Pre-Critical Hypothesis that the record of \fp{msg} in $\fwset{}^r(D)$ is equivalent. By the Second Pre-Critical Hypothesis, the record of \fp{msg} at $\fwset{}^r(G)$ is equivalent to $\langle \fp{msg}, \{  f = \fp{index}.f, b = 0 \}, \fp{iset} \rangle$. The \ref{RemovalNotification} procedure does not change this record since it does not contain $\minusG$ instability and so this remains the record at $\fwset{}^r(G)$ at the end of the constellation. The record of \fp{msg} at $\fwset{}(G)$ starts out equal to the record at $\fwset{}(D)$, and then the \ref{Run} procedure zeroes out the $\fp{index}.b$ component of the record, leaving it equivalent to the final value of the record at $\fwset{}^r(G)$, as required by the Interim $G$ Hypothesis.
\item[\lque{}] \hfill \\
non-$G$ case: Initially $\lque{}^r = \lque{}$. Neither the \ref{JoinNotification} procedure nor the \ref{RemovalNotification} procedure make changes to $\lque{}$, so it remains the same, as required.

$G$ case: Initially in $H$ process $G$ inherits its launch queue from $D$, while in $H^r$ process $G$ has an empty launch queue according to the Second Pre-Critical Hypothesis. The \ref{RemovalNotification} procedure does not change the value of $\lque{}^r$ while the \ref{Run} procedure empties $\lque{}$. So $\lque{}^r$ becomes identical to $\lque{}$ as required.

\item[\sgh{} and \sfh{}] \hfill \\
Non-$G$ case: Initially $\sgh{}^r = \sgh{}$ and $\sfh{}^r = \sfh{}$. Neither the \ref{JoinNotification} nor the \ref{RemovalNotification} procedures make changes to either value (remember that we ignore the \ref{CheckFlush} invocation at the end), so both values remain the same, as required.

$G$ case: Initially $G$ inherits its state from $D$, so $\sgh{}(G) = \sgh{}(D)$ and $\sfh{}(G) = \sfh{}(D)$. The Second Pre-Critical Hypothesis implies that $\sgh{}^r(G) = \sgh{}^r(D)$ and $\sfh{}^r(G) = \sgh{}^r(D)$. Therefore initially $\sgh{}^r(G) = \sgh{}(G)$. The \ref{RemovalNotification} procedure does not make changes to either variable in $H^r$ while the \ref{Run} procedure resets the value of \sfh{} to be equal to \sgh{}, therefore at the end of the constellation:
\begin{multline*}
\sfh{}(G) = \sgh{}(G) = \sgh{}^r(G) = \\
{} = \sgh{}^r(D) = \sfh{}^r(G)
\end{multline*}
as required by the Interim $G$ Hypothesis.

\item[\mpkout{}] \hfill \\
Non-$G$ case: Initially $\mpkout{}^r = \mpkout{}$. The \ref{JoinNotification} procedure does not change this value in $H$, while the \ref{RemovalNotification} procedure increments $\mpkout{}^r.f$ for each message in $\fque{}[\minusG]$. Since that queue is empty, the value of $\mpkout{}^r$ does not change either, and so it remains the same at the end of the constellation, as required by the Interim Non-$G$ Hypothesis.

$G$ case: Initially $\mpkout{}^r.f = \mpkout{}.f$ by the sameness argument and it remains so at the end of the constellation because neither the \ref{Run} procedure nor the \ref{RemovalNotification} procedure change that value (in the latter case because $\fque{}^r[\minusG]$ is empty). Initially $\mpkout{}.b(G) = \mpkout{}.b(D)$ while the Second Pre-Critical Hypothesis implies that initially $\mpkout{}^r.b(G) = 0$. The \ref{RemovalNotification} procedure does not change that value in $H^r$ while the \ref{Run} procedure zeroes it in $H$. Therefore $\mpkout{}^r.b(G)$ becomes equal to $\mpkout{}.b(G)$.

As a result $\mpkout{}^r = \mpkout{}$ as required by the Interim $G$ Hypothesis.

\item[$\mpkin{}{[]}$] \hfill \\
Non-$G$ case: We start with the coordinates that exist in $H$. By the First Pre-Critical Hypothesis, these coordinates have identical values in $H$ and $H^r$. These coordinates are not touched by either the \ref{RemovalNotification} procedure or the \ref{JoinNotification} procedure, and so they remain identical at the end of the constellation as required by the Interim non-$G$ Hypothesis.

The $\mpkin{}^r[\minusG]$ coordinate is removed by the \ref{RemovalNotification} procedure, leaving a $G$ coordinate whose value is
$$
\mpkin{}^r[G] = \{ f = \mpkin{}^r[D].f ; b = 0 \}
$$
according to the First Pre-Critical Hypothesis. The \ref{JoinNotification} procedure creates a new $G$ coordinate whose value is
$$
\mpkin{}[G] = \{ f = \mpkin{}^r[D].f ; b = 0 \}
$$
Since we have already seen that $\mpkin{}^r[D].f = \mpkin{}[D].f$,  the requirement of the Interim non-$G$ Hypothesis is met.

$G$ case:
Initially $\mpkin{}^r[](G) = \mpkin{}^r[](D)$ according to the Second Pre-Critical Hypothesis. The \ref{RemovalNotification} procedure removes the $\minusG$ coordinate in both at $G$ and $D$ and therefore at the end of the constellation the equality
$$
\mpkin{}^r[](G) = \mpkin{}^r[](D)
$$
remains valid.

On the other hand the initial relation in $H$ is $\mpkin{}[](G) = \mpkin{}[](D)$ because $G$ is created with the same state as $D$.
The \ref{JoinNotification} procedure creates a new $G$ coordinate at $D$ with value
$$
\mpkin{}[G](D) = \{ f = \mpkin{}[D].f(D) ; b = 0 \}
$$
The \ref{Run} procedure creates a new $G$ coordinate at $G$ with value
$$
\mpkin{}[G](G) = \{ f = \mpkout{}.f(G) ; b = 0 \}
$$
The \AxProcIV{} and Lemma \ref{OrigPCountLem} guarantee that at the start of the constellation $\mpkin{}[D](D) = \mpkout{}(D)$. Since $G$ starts out with the same state as $D$ we also have $\mpkout{}.f(G) = \mpkout{}.f(D)$. Therefore $\mpkin{}[G](G) = \mpkin{}[G](D)$ and the vector as a whole remains identical at $G$ and $D$.

From the Non-$G$ case we know that $\mpkin{}^r[](D) = \mpkin{}[](D)$ at the end of the constellation. It follows that $\mpkin{}^r[](G) = \mpkin{}[](G)$ at the end of the constellation as required by the Interim $G$ Hypothesis.

\item[$\gvec{}{[]}$ and $\fvec{}{[]}$] \hfill \\
Non-$G$ case: Initially, the First Pre-Critical Hypothesis implies that
\begin{align*}
\gvec{}^r[] & = \gvec{}[] \,\bigcup\, \left\{ [G] = \gvec{}[D], [\minusG] = \gvec{}[D] \right\}   \\
\fvec{}^r[] & = \fvec{}[] \,\bigcup\, \left\{ [G] = \gvec{}[D], [\minusG] = \gvec{}[D] \right\}
\end{align*}
The \ref{RemovalNotification} procedure removes the $\minusG$ coordinate from both vectors in $H^r$. The \ref{JoinNotification} procedure creates a $G$ coordinate in both vectors in $H$ and sets both of them to the value of $\gvec{}[D]$. As a result $\gvec{}^r[]$ and $\fvec{}^r[]$ become equal to $\gvec{}[]$ and $\fvec{}[]$ as required by the Interim non-$G$ Hypothesis.

$G$ case: Initially $\gvec{}[](G) = \gvec{}[](D)$ and $\fvec{}[](G) = \fvec{}[](D)$. By the Second Pre-Critical Hypothesis $\gvec{}^r[](G) = \gvec{}^r[](D)$ and $\fvec{}^r[](G) = \fvec{}^r[](D)$. Using the First Inductive Hypothesis relation between the $H$ and $H^r$ vectors in $D$ we conclude that
\begin{align*}
\gvec{}^r[](G) & = \gvec{}[](G) \,\bigcup\, \left\{ [G] = \gvec{}[D](G), [\minusG] = \gvec{}[D](G) \right\}   \\
\fvec{}^r[](G) & = \fvec{}[](G) \,\bigcup\, \left\{ [G] = \gvec{}[D](G), [\minusG] = \gvec{}[D](G) \right\}
\end{align*}
The \ref{RemovalNotification} procedure removes the $\minusG$ coordinate from both vectors in $H^r$. In $G$, the \ref{Run} procedure creates a $G$ coordinate in both vectors and sets both of them to $\sgh{}(G)$. Since $G$ starts out with the same state as $D$, we have $\sgh{}(G) = \sgh{}(D)$. It follows from the \AxProcIV{} and Lemma \ref{FCountLem} that at the start of the constellation $\sgh{}(D) = \gvec{}[D](D)$. Therefore at the end of the constellation
\begin{align*}
\gvec{}^r[G](G) & = \gvec{}[D](G) = \fvec^r[G](G) \\
\gvec{}[G](G) & = \sgh{}(G) = \fvec[G](G) \\
\gvec{}[D](G) & = \gvec{}[D](D) = \sgh{}(D) = \sgh{}(G)
\end{align*}
and therefore $\gvec{}^r[](G) = \gvec{}[](G)$ and $\fvec{}^r[](G) = \fvec{}[](G)$ as required by the Interim $G$ Hypothesis.
\end{description}
\end{description}
\end{itemize}

\subsubsection{Message broadcast request constellations}

\begin{itemize}
\item
\begin{description}
\item[\bf Type of Constellation:] A message broadcast request dequeuing event at a process $P$. We assume (see \ref{UserApplication}) that \ulp{} at a process $P$ does not issue such a request before its Main() function is executed in a separate thread.
\item[\bf $H$ Transactions:] An execution of the \ref{BroadcastMessage} procedure at $P$.
\item[\bf Observed behavior in $H^r$:] An equivalent transaction at $P$.
\item[\bf Execution of \cbcast{} in $H^r$:]
In $H^r$, process $P$ executes the \ref{BroadcastMessage} procedure, just as it does in $H$. In the pre-critical case $P \ne \plusminusG$. Also notice that $\requestset^{H^r} = \requestset^H$ (see \ref{HistReduxPrelim}) and therefore $\plusminusG$ do not generate message broadcast requests pre-critically in $H^r$.

Since the message is not yet stamped we have an actual equality $\fp{msg}^r = \fp{msg}$ and not just an equivalence.

We start by verifying that in the pre-critical case the Second Pre-Critical Hypothesis is preserved when $P = D$. Processes $\plusminusG$ do not participate in the constellation in the pre-critical case. It is easy to observe that the \ref{BroadcastMessage} call does not change any state variables other than \lque{}, \bwset{} and $\mpkout{}.b$. All of these variables are stipulated by the Second Pre-Critical Hypothesis to be equal to zero at $\plusminusG$ and so we are done with this case.

To verify the other inductive hypotheses and verify the side effects, we follow the execution step by step.

In the first step, we check if $\vgap{}$ is zero. If it is not, we append the message to \lque{}. By the relevant inductive hypothesis (either First Pre-Critical, Interim non-$G$, Interim $G$ or Convergent Hypothesis) $\vgap{}^r = \vgap{}$ and $\lque{}^r = \lque{}$. Since $\fp{msg}^r = \fp{msg}$ this step is executed identically in $H$ and $H^r$ and preserves the hypotheses.

The next and final step is executed if $\vgap{} = 0$. Notice that during the interim period $\vgap{} > 0$ and therefore for this step only the First Pre-Critical and Convergent Hypotheses are relevant. This step contains the following sub-steps:
\begin{itemize}
\item The counter $\mpkout{}.b$ is incremented. $\mpkout{}^r.b = \mpkout{}.b$ and therefore the inductive hypotheses are preserved.
\item The next few lines stamp the message. Here the inductive hypotheses requires that the stamping be equivalent, rather than equal, in $H$ and $H^r$. Since $\self{}^r = \self{}$ and $\cview{}^r = \cview{}$ under all the hypotheses these parts of the stamp end up being equal as required.

As for the message vector time, it is computed from the process vector time, and is adjusted at the $\self{}$ coordinate. The Convergent Hypothesis is in force when $\cview{} \ge \vcrit{}$ (see \ref{InductiveHypSSS}) and it stipulates that all the variables in sight are equal in $H$ and $H^r$, resulting in $\vtime{}'^r[] = \vtime{}'[]$. This matches the equivalence requirement (see Definition \ref{EquivDef}).

The First Pre-Critical Hypothesis is in force when $\cview{} < \vcrit{}$ and it stipulates that
$$
\vtime{}^r[] = \vtime{}[] \,\bigcup\, \left\{ [G] = 0, [\minusG] = 0 \right\}
$$
The $\self{}$ coordinate of $\vtime{}'$ is adjusted by $\mpkout.b - \mpkin{}[\self{}].b$. If $P \ne G$ then this adjustment is the same in $H$ and $H^r$ and does not affect the $G$ coordinate. But we already know that in the pre-critical case $P \ne G$, so the hypothesis is preserved in this case.

\item In the next step we queue the message multicast to all the members of \cset{}. Therefore we have an equivalent side effect to the one in $H$ as observed and Lemma \ref{ReducedContactLem} guarantees that the computed target set $\cset{}^r$ is equal to the observed target set.
\item In the next and last steps, a record for the message is added to \bwset{}.
In the pre-critical case, the message is equivalent and
\begin{align*}
\se{index}^r & = \mpkout{}^r = \mpkout{} = \se{index} \\
\se{iset}^r[X].f & =
\begin{cases}
\mpkin{}^r[X].f = \mpkin{}[X].f = \se{iset}[X].f & \quad \text{if } X \ne \plusminusG \\
\mpkin{}^r[\plusminusG].f = \mpkin{}[D].f = \se{iset}[D].f & \quad \text{if } X = \plusminusG 
\end{cases} \\
\se{iset}^r[X].b & =
\begin{cases}
\mpkin{}^r[X].b = \mpkin{}[X].b = \se{iset}[X].b & \quad \text{if } X \ne \plusminusG \\
\mpkin{}^r[\plusminusG].b = 0 & \quad \text{if } X = \plusminusG 
\end{cases} 
\end{align*}
as required by the First Pre-Critical Hypothesis.
\end{itemize}
\end{description}
\end{itemize}

\subsubsection{Message and acknowledgment packet constellations}

\begin{itemize}

\item
\begin{description}
\item[\bf Type of Constellation:] An acknowledgement packet, sent by a process $Q$ in response to an original broadcast of a message $m$, is received at the originating process $P$. Note that neither $G$ nor $\minusG$ broadcast original messages during the pre-critical period or the interim period, therefore we can assume $P \ne \plusminusG$ in these cases.
\item[\bf $H$ Transactions:] A single transaction, with an execution of the \ref{ReceiveAck} procedure at $P$.
\item[\bf Observed behavior in $H^r$:] Post-critically, or when $Q \ne D$, there is a single transaction at $P$ with an equivalent trigger and equivalent side effects (see Definition \ref{EquivSEDef}). Pre-critically when $Q = D$, there are three transactions at $P$, triggered, in that order, by an acknowledgement packet from $\minusG$, $G$ and $D$. The first two have no side effects. The third one has side effects equivalent to the original transaction in $H$.
\item[\bf Execution of \cbcast{} in $H^r$:] We divide the analysis into several cases:
\begin{itemize}
\item The constellation is post-critical.
\item The constellation is pre-critical, with $P, Q \ne D$
\item The constellation is pre-critical with $Q \ne D$ and $P = D$
\item The constellation is pre-critical, with $Q = D$ and $P \ne D$
\item The constellation is pre-critical with $P = Q = D$
\end{itemize}
In the first three cases, where the constellation is post-critical or $Q \ne D$, the $H^r$ execution is a single invocation of \ref{ReceiveAck} by $P$. Only in the third of these cases (when $P = D$) is there anything to prove about Second Pre-Critical Hypothesis. We will ignore that part for the moment and return to it later. We will deal with the remaining two cases (where $Q = D$) later as well.

The Convergent case is not entirely trivial here because the trigger is equivalent but perhaps not equal ($\pks{m^r} \cong \pks{m}$). In fact the trigger is equal but showing that requires the History Equivalence Theorem so we will not rely on this fact and use a weaker argument instead.

To verify the preservation of the First Pre-Critical Hypothesis, Interim non-$G$ Hypothesis, Interim $G$ Hypothesis or Convergent Hypothesis, as the case may be, we go step by step over the procedure.

It starts with a check for the existence of a record for $m$ in \wset{}. The record is guaranteed to be there but we have not proved that. However it does follow from the First Pre-Critical and from the two Interim Hypotheses that the record is there in $H$ if and only if it is there in $H^r$.

The Convergent case is the tricky one here. We know that $m^r \cong m$. If $m$ is pre-critical then it follows from the Convergent Hypothesis that the record is {\em not} found in either $H$ or $H^r$. If $m$ is post-critical then $m = m^r$ and the record is found in $H$ if and only if it is found in $H^r$. Therefore the same decision is made in both histories regardless of which hypothesis is in force, and in the Convergent case if the record is found then the $H$ and $H^r$ triggers are equal.

If the record is not found then the procedure call exits and we are done. With the check successful, the call continues through the following steps:
\begin{itemize}
\item The $Q$ entry is removed from the instability set of the message. In the pre-critical case $Q \ne \plusminusG$ , and the $Q$ entry exists in both $H$ and $H^r$ as we have shown. Therefore removing it in both histories preserves the inductive hypothesis.
\item If the instability set becomes empty, the record is removed from \wset{} and the \ref{CheckFlush} procedure is called. In the post-critical case the instability vector is identical and therefore either empty in both histories or non-empty in both. In the pre-critical case, even though the instability sets are different, they are still empty or non-empty together. Removing the record from \wset{} preserves the inductive hypothesis and by Corollary \ref{CheckFlushEquivCor}, so does invoking \ref{CheckFlush}. Moreover the same corollary guarantees that the side effects produced by \ref{CheckFlush} are equal, and therefore equivalent, to the observed side effects. Lemma \ref{ReducedContactLem} guarantees that the target set produced by the execution of \ref{CheckFlush} for each side effect, namely $\cset{}^r$, is equal to the observed target set.
\end{itemize}
To verify the Second Pre-Critical Hypothesis in the third case, notice that since the message is original, any changes in its record in \wset{} only affect the \bwset{} portion of \wset{}. Since the inductive hypothesis only requires that \bwset{} be empty in $\plusminusG$, it remains valid until the invocation of \ref{CheckFlush}.
In this case we cannot rely on Corollary \ref{CheckFlushEquivCor} because $\plusminusG$ do not participate in the constellation and therefore do not execute \ref{CheckFlush}. Instead we rely on Lemma \ref{OmnibusCBCASTLem}.

If $\fwset{}^r(D) \ne \emptyset{}$ when $D$ calls \ref{CheckFlush}, the procedure exists immediately and we are done. If $\fwset{}^r(D) = \emptyset{}$ then as we have seen this implies that the set was empty since the beginning of the transaction. This in turn implies that $\sgh{}^r(D) = \cview{}^r(D) + \vgap{}^r(D)$. This follows from Lemma \ref{OmnibusCBCASTLem}(\ref{OCL:sendgfB}) in the case $\vgap{}^r(D) > 0$. If $\vgap{}^r(D) = 0$ then Lemma \ref{OmnibusCBCASTLem}(\ref{OCL:gf}) implies that $\fvec{}^r[D](D) = \cview{}^r(D) + \vgap{}^r(D)$ and the equation follows from Lemma \ref{OmnibusCBCASTLem}(\ref{OCL:sendgfA}).

Since $\sgh{}^r(D) = \cview{}^r(D) + \vgap{}^r(D)$ the \ref{CheckFlush} procedure does not produce a ghost side effect and as a result does not change the value of $\sgh{}^r(D)$. Therefore the Second Pre-Critical Hypothesis remains valid.

We are left with the two pre-critical cases where $Q = D$.

We start by verifying the preservation of the First Pre-Critical Hypothesis. The execution in $H^r$ consists of three executions of \ref{ReceiveAck} at $P$, triggered by the receipt of a packet from $\minusG$, $G$ and $D$ respectively.

The first two executions remove the $\minusG$ and $G$ entries from the instability set, but do not empty out the set because it still possesses a $D$ entry.  Since \ref{CheckFlush} is not executed this guarantees that the two executions have no side effects, as expected. However the state at $P$ is now in violation of the First Pre-Critical Hypothesis because $m$ has instability at $D$ without matching instabilities at $\plusminusG$. The third execution restores the hypothesis by removing $D$ from the instability set. The remaining part of the third execution preserves the inductive hypothesis and produces the expected side effects for the exact same reasons as the previous cases that we already analyzed.

With regard to the preservation of the Second Pre-critical Hypothesis, it is only relevant in the fifth case where $P = D$. Just as in the third case, the changes in \wset{} do not affect the validity of the inductive hypothesis, and it continues to be valid until the moment that $D$ calls \ref{CheckFlush}. Just as before, Lemma \ref{OmnibusCBCASTLem} guarantees that this invocation does not produce a ghost broadcast and therefore preserves the Second Pre-Critical Hypothesis.
\end{description}

\item
\begin{description}
\item[\bf Type of Constellation:] An acknowledgement packet, sent by a process $Q$ in response to the forwarding of a message $m$, is received at the forwarding process $P$, in $H$.
\item[\bf $H$ Transactions:] A single transaction, with an execution of the \ref{ReceiveAck} procedure at $P$.
\item[\bf Observed behavior in $H^r$:] There are a number of separate cases:
\begin{itemize}
\item Post-critically, or when $P, Q \ne D$, there is a single transaction at $P$ with the same trigger and side effects.
\item Pre-critically when $P = D$ and $Q \ne D$, there are three transactions at $D$, $G$ and $\minusG$ respectively, all triggered by acknowledgement packets from $Q$ for the same forwarded message. The observed side effects at $D$ are equal to the original side effects in $H$. At $\plusminusG$ the side effects depend on the $D$ side effects in the following way. If there is a ghost broadcast out of $D$ then there is a same height ghost broadcast followed by a a same height flush broadcast out of $\plusminusG$. If there is no ghost broadcast out of $D$, then there are no side effects at $\plusminusG$. 
\item Pre-critically when $P \ne D$ and $Q = D$, there are three transactions at $P$, triggered, in that order, by an acknowledgement packet from $\minusG$, $G$ and $D$ for the same forwarded message. The first two have no side effects. The third one has the same side effect as the original transaction in $H$.
\item Pre-critically when $P = Q = D$, there are nine transactions in total, three each at $D$, $G$ and $\minusG$, triggered at each process, in that order, by acknowledgement packets from $\minusG$, $G$ and $D$, all for the same forwarded message. The first two transactions in each process have no side effects. The third transaction at $D$ has the side effects as the original transaction at $D$. The third transaction at each of $\plusminusG$ has side effects that depend on the original side effects at $D$ in the same way as in the second case: If there is a ghost broadcast out of $D$ then there is a same height ghost broadcast followed by a a same height flush broadcast out of $\plusminusG$. If there is no ghost broadcast out of $D$, then there are no side effects for the third transaction at $\plusminusG$.
\end{itemize}
\item[\bf Execution of \cbcast{} in $H^r$:] As before we divide the analysis into the different cases:
\begin{itemize}
\item The constellation is post-critical.
\item The constellation is pre-critical, with $P, Q \ne D$
\item The constellation is pre-critical with $Q \ne D$ and $P = D$
\item The constellation is pre-critical, with $Q = D$ and $P \ne D$
\item The constellation is pre-critical with $P = Q = D$
\end{itemize}

The analysis of cases one, two and four is identical to the analysis of similar cases in the case of an acknowledgement for an original message, and we will not repeat it here.

In the third case there are three executions of the \ref{ReceiveAck} procedure, one execution at each of $D$, $G$ and $\minusG$. The analysis of the First Pre-Critical Hypothesis proceeds exactly like the analysis of this case for an original message acknowledgement. As for the Second Pre-Critical Hypothesis, the record for the message in \wset{} is in the forwarding part \fwset{}, which contains the same messages in $\plusminusG$ as it does in $D$. Moreover, for each message it contains the same instability set. As a result the execution of the \ref{ReceiveAck} proceeds in the same way in all three of these processes and the Second Pre-Critical Hypothesis is preserved. By Corollary \ref{CheckFlushEquivCor}, the invocation of \ref{CheckFlush} produces side effects that match the observed side effects and Lemma \ref{ReducedContactLem} guarantees that these side effects have the observed target set.

In the fifth case there are three executions of the \ref{ReceiveAck} procedure at each process, triggered by the receipt of an acknowledgement packet from $\minusG$, $G$ and $D$, in this order. The analysis of the First Pre-Critical Hypothesis proceeds exactly like the analysis of this case for an original message acknowledgement. As for the Second Pre-Critical Hypothesis, the record for the message $m$ in \wset{} is in the forwarding part \fwset{}, which contains the same messages in $\plusminusG$ as it does in $D$. Moreover, for each message it contains the same instability set. As a result, all three executions of the \ref{ReceiveAck} proceed in the same way in all three of these processes. By Corollary \ref{CheckFlushEquivCor}, the invocation of \ref{CheckFlush} preserves the Second Pre-Critical Hypothesis and produces side effects that match the observed side effects and Lemma \ref{ReducedContactLem} guarantees that these side effects have the observed target set. Therefore overall the Second Pre-Critical Hypothesis is preserved and the side effects in $\plusminusG$ match the observed side effects.
\end{description}

\item
\begin{description}
\item[\bf Type of Constellation:] An original message packet containing message $m$ is received at process $P$ in $H$.
\item[\bf $H$ Transactions:] A single transaction, with an execution of the \ref{ReceiveMessage} procedure at $P$.
\item[\bf Observed behavior in $H^r$:] There are two cases
\begin{itemize}
\item Post-critically, or when $P \ne D$, there is a single transaction at $P$ with the same trigger and side effect as the original transaction in $H$.
\item Pre-critically, when $P = D$, there are three transactions, one each at $D$, $G$ and $\minusG$, with equivalent triggers and side effects as the original transaction in $H$.
\end{itemize}
\item[\bf Execution of \cbcast{} in $H^r$:] \hfill
\begin{itemize}
\item In the post-critical and $P \ne D$ cases, process $P$ invokes the \ref{ReceiveMessage} procedure. The execution proceeds through the following steps:

First an acknowledgment is sent back to the sender. This ensures that the side effect in $H^r$ is the expected one.

The next step determines whether the message is original. Since the value of $\orig{m}$ is equal in both histories, the determination resolves the same way in both histories, and since the message is original in $H$ it is deemed to be original in $H^r$ as well. As a result $\mpkin{}[\orig{m}].b$ is incremented.

In the pre-critical case $\orig{m} \ne \plusminusG$ and therefore incrementing the counter preserves the First Pre-Critical Hypothesis. In the post-critical case, the $\mpkin{}[]$ vector is identical at $H$ and $H^r$, thus preserving all the relevant hypotheses (Interim $G$, Interim non-$G$ and Convergent Hypotheses). The Second Pre-Critical Hypothesis is automatically preserved because neither $D$ nor $\plusminusG$ participate in the constellation in this case.

The next few steps check for duplicates. These checks rely on values, including \cview{}, \rset{}, the $\orig{m}$ coordinate in \vtime{} and $\mvt{m}$, that are all guaranteed by all the relevant inductive hypotheses to be identical or equivalent at $H$ and $H^r$. Therefore, the duplicate check yields the same results in both histories.

At this point a comment is due about the Convergent case. The trigger in this case is equivalent ($\pkm{m^r} \cong \pkm{m}$) but not necessarily equal. Unlike the case of an acknowledgement packet, a non-equal trigger may actually occur in the Convergent period. For that to happen it must be that $\mview{m} < \vcrit{}$ (see Definition \ref{EquivDef}). Since $P$ is convergent, $\cview{} \ge \vcrit{}$. This means that the message is obsolete. Since labeled step \ref{RMP:DiscardObsolete} of the procedure discards obsolete messages and exits, there is no contamination of the state with non-equal values and the inductive hypothesis holds.

If the message is not a duplicate, the next and last steps are:
\begin{itemize}
\item The message $m$ is added to \rset{}. Since both the message $m$ and \rset{} are equivalent in both histories, this step preserves the equivalence of \rset{} and the inductive hypothesis.
\item The message $m$ is appended to the tail of the sender's forwarding queue. In this case as well $\fque{}[\fp{sender}]$ is equivalent in $H$ and $H^r$, and the step preserves the equivalence and the relevant inductive hypothesis.
\item The \ref{ScanCall} procedure is invoked. We want to use Lemma \ref{ScanEquivLem} to prove that this call preserves all the inductive hypotheses and has no side effects. But the lemma requires that the ApplyMessage up-call be identical at $P$ in $H$ and $H^r$. According to \ref{ReducedUAI} this is the case for $P \ne G$. This covers the pre-critical case since $P$, as an $H$ process, cannot be equal to $G$. In the post-critical case the definition of ApplyMessage$^r$@G shows that it behaves, post-critically, like ApplyMessage@G, and we are done.
\end{itemize}
\item In the pre-critical case when $P = D$, all three processes $D$, $G$ and $\minusG$ invoke the \ref{ReceiveMessage} procedure. The verification of the First Pre-Critical Hypothesis proceeds in this case exactly as it does in the $P \ne D$ case. As for the Second Pre-Critical Hypothesis, we retrace the execution steps:

First an acknowledgment is sent back to the sender. This ensures that the side effects at $G$ and $\minusG$ are the expected ones.

The next step determines whether the message is original. Since the value of $\orig{m}$ is equal in all three processes, the determination resolves the same way in all three, and since the message is original at $D$ it is deemed to be original at $\plusminusG$ as well. As a result $\mpkin{}[\orig{m}].b$ is incremented at all three processes.

Since the $\mpkin{}[]$ vector is identical at all three processes, incrementing the counter at the sender coordinate preserves the inductive hypothesis.

The next few steps check for duplicates. These checks rely on values, including \cview{}, \rset{}, \vtime{} and $\mvt{m}$, that are all guaranteed by the Second Pre-Critical Hypothesis to be identical at all three processes. Therefore, the duplicate check yields the same results in all three.

If the message is not a duplicate, the next and last steps are:
\begin{itemize}
\item The message $m$ is added to \rset{}. Since both the message $m$ and \rset{} are identical in all three processes, this step preserves the Second Pre-Critical Hypothesis.
\item The message $m$ is appended to the tail of the sender's forwarding queue. In this case as well $\fque{}[]$ is identical in all three processes and the Hypothesis is preserved.
\item The \ref{ScanCall} procedure is invoked. The ApplyMessage$^r$@G up-call behaves, pre-critically, like ApplyMessage@D, which in turn behaves like ApplyMessage$^r$@D (see \ref{ReducedUAI}). Therefore ApplyMessage behaves the same way in $D$ and in $\plusminusG$  during the pre-critical period so it follows from Lemma \ref{ScanEquivLem} that this call preserves all the inductive hypotheses and has no side effects.
\end{itemize}
\end{itemize}
\end{description}

\item
\begin{description}
\item[\bf Type of Constellation:] A message packet containing a forwarded message $m$ is received at process $P$ from process $Q$ in $H$.
\item[\bf $H$ Transactions:] A single transaction, with an execution of the \ref{ReceiveMessage} procedure at $P$.
\item[\bf Observed behavior in $H^r$:] There are several cases:
\begin{itemize}
\item Post-critically, or when $P, Q \ne D$, there is a single transaction at $P$ with an equivalent trigger and side effect as the original transaction in $H$.
\item Pre-critically, when $P = D$ and $Q \ne D$, there are three transactions, one each at $D$, $G$ and $\minusG$, with triggers and side effects equivalent to the original transaction in $H$.
\item Pre-critically, when $P \ne D$ and $Q = D$, there are three transactions at $P$, triggered by the receipt of a message packet from $D$, $\minusG$, and $G$, in that order, each with the side effect of an acknowledgement packet being sent to the respective sender. Each transaction is equivalent to the original one in $H$.
\item Pre-critically, when $P = Q = D$, there are nine transactions, three each at $D$, $G$ and $\minusG$, triggered by the receipt of a message packet from $D$, $\minusG$ and $G$, in that order, each with the side effect of an acknowledgement packet being send to the respective sender. Each transaction is equivalent to the original one in $H$.
\end{itemize}
\item[\bf Execution of \cbcast{} in $H^r$:] \hfill
\begin{itemize}
\item The first two cases proceed just like the equivalent cases of an original message packet, except that the message counter is incremented at the $.f$ coordinate rather than the $.b$ coordinate.

\item In the pre-critical case where $P \ne D$ and $Q = D$, the processes $D$ and $\plusminusG$ do not participate in the constellation and therefore the Second Pre-Critical Hypothesis is automatically preserved. The process $P$ executes the \ref{ReceiveMessage} three times. First it executes the call for the message that it receives from $D$, and then for the same message being received from $\minusG$ and $G$. The first execution proceeds in parallel with the original execution in history $H$. The analysis of this first execution proceeds in almost the same way as the analysis of the receipt of an original message does. There is an important difference, however. When $\mpkin{}[D].f$ is incremented the First Pre-Critical Hypothesis is violated because $\mpkin{}[\plusminusG].f$ is not incremented at the same time. This violation does not, however, affect the rest of the analysis, and remains an isolated violation. This includes the invocation of the \ref{ScanCall} procedure since Lemma \ref{ScanEquivLem} does not require this particular part of the inductive hypothesis.

The subsequent two executions proceed through the following steps:

First, an acknowledgement packet is sent to the respective sender, generating the expected side effect.

The next step increments the forwarded message counter ($\mpkin{}[\minusG].f$ in the second execution and $\mpkin{}[G].f$ in the third execution). These steps remove the violation of the First Pre-Critical Hypothesis.

The next steps check for duplicates. These steps cause the call to exit in both executions because the message is obviously a duplicate - it was already placed in \rset{} and possibly even delivered by the first call.

\item In the pre-critical case when $P = Q = D$, each of the processes $D$, $G$ and $\minusG$ execute the \ref{ReceiveMessage} three times. First for the message that each receives from $D$ and then for the same message that each receives from $\minusG$ and then from $G$. The analysis of the state at $D$ in $H$ and $H^r$ proceeds exactly as in the previous case, where $P \ne D$ and $Q = D$, proving that the First Pre-Critical Hypothesis is preserved in this case. All we have to show is that the three executions at $G$ and at $\minusG$ result in the preservation of the Second Pre-Critical Hypothesis. This part of the analysis is just a recap of previous cases. The analysis of the first execution at all three processes (where the message is received from $D$) proceeds exactly like the same part of the analysis of the case of an original message that is received at $D$. The next two executions proceed identically at $D$ and at $G$ and $\minusG$, resulting in a determination that the message is a duplicate.
\end{itemize}
\end{description}
\end{itemize}

\subsubsection{Ghost and flush packet constellations}

\begin{itemize}

\item
\begin{description}
\item[\bf Type of Constellation:] A ghost packet of height $v$ is received at a process $P$ from a process $Q$ in $H$.
\item[\bf $H$ Transactions:] A single transaction, with an execution of the \ref{ReceiveGhost} procedure at process $P$ and no side effects.
\item[\bf Observed behavior in $H^r$:] There are several cases:
\begin{itemize}
\item Post-critically, or when $P, Q \ne D$, there is a single transaction at $P$ with the same trigger and no side effects.
\item Pre-critically, when $P = D$ and $Q \ne D$, there are three transactions, one each at $D$, $G$ and $\minusG$, each with the same trigger and no side effects.
\item Pre-critically, when $P \ne D$ and $Q = D$, there are five transactions at $P$, triggered by the receipt of identical ghost packets of the same height from $\minusG$ and then $G$, followed by flush packets of the same height from $\minusG$ and then $G$ and finally a ghost packet of the same height from $D$. This order is induced from the adjustment coordinate of the labels of the triggers. None of these transactions have any side effects.
\item Pre-critically, when $P = Q = D$, there are fifteen transactions, five each at $D$, $G$ and $\minusG$, triggered at each process by the receipt of two ghost packets of the same height from $\minusG$ and then $G$, followed by flush packets of the same height from $\minusG$ and then $G$ and finally a ghost packet of the same height from $D$. None of these transactions have any side effects.
\end{itemize}
\item[\bf Execution of \cbcast{} in $H^r$:] \hfill
\begin{itemize}
\item In the post-critical case $P$ executes the \ref{ReceiveGhost} procedure which results in raising $\gvec{}[Q]$ to $v$ in both $H$ and $H^r$. This preserves the relevant one among the Interim non-$G$, Interim $G$ and Convergent Hypotheses, because all of them stipulate that the $\gvec{}[]$ vector is equal in $H$ and $H^r$. The procedure does not generate any side effects, which matches the observed lack of side effects in $H^r$.

\item In the pre-critical case where $P, Q \ne D$, process $P$ executes the \ref{ReceiveGhost} procedure which results in raising $\gvec{}[Q]$ to $v$ in both $H$ and $H^r$. This preserves the Second Pre-Critical Hypothesis because $P \ne D$ and so no changes occur to the states of either $D$, $G$ or $\minusG$. To see that the First Pre-Critical Inductive Hypothesis is preserved observe that $Q \ne \plusminusG$ because $Q$ exists in $H$ and $\plusminusG$ do not, and since $Q \ne D$ by assumption, neither $\gvec{}[\plusminusG]$ nor $\gvec{}[D]$ is affected. The procedure does not generate any side effects, which matches the observed lack of side effects in $H^r$.

\item  In the pre-critical case when $P = D$ and $Q \ne D$, the three processes $D$, $G$ and $\minusG$ each execute the \ref{ReceiveGhost} procedure, which raises the $\gvec{}[Q]$ level in each process to $v$. In this case the Second Pre-Critical Hypothesis is preserved because this hypothesis stipulates that $\gvec{}[]$ is equal in all three processes, and they all perform the same change. The First Pre-Critical Hypothesis is preserved for the same reason as in the previous case. The procedure does not generate any side effects, which matches the observed lack of side effects in $H^r$.

\item In the pre-critical case when $P \ne D$ and $Q = D$, the process $P$ executes the \ref{ReceiveGhost} procedure twice, once for each ghost packet trigger from $\minusG$ and $G$, then the \ref{ReceiveFlush} procedure twice, once for each flush packet trigger from $\minusG$ and $G$, and finally the \ref{ReceiveGhost} procedure once more for the ghost packet trigger from $D$. Since $D$ and $\plusminusG$ do not participate in the constellation, the Second Pre-Critical Hypothesis is automatically preserved. At $P$, the five calls raise the values of $\gvec[\minusG]$, $\gvec[G]$, $\fvec[\minusG]$, $\fvec[G]$ and  $\gvec[D]$ to $v$. This conforms with the First Pre-Critical Hypothesis. In addition, the calls to \ref{ReceiveFlush} cause \ref{TryToInstall} to be invoked.

By the First Pre-Critical Hypothesis the value of \vgap{} is the same at $P$ in $H$ and $H^r$ at the start of the constellation.

If $\vgap{} = 0$ at the start of the constellation then by Lemma \ref{OmnibusCBCASTLem}(\ref{OCL:lque}) $\lque{} = \emptyset$ at $P$. The first block of \ref{TryToInstall} does nothing, and the execution of the call skips to the second and last block, where the messages in \lque{} are broadcast. Since $\lque{}$ is empty the first execution of \ref{TryToInstall} has no side effects and does not change the process state. The second execution of \ref{TryToInstall} encounters the same values of \vgap{} and \lque{} as the first execution, and therefore it also does nothing. As a result the whole constellation preserves both pre-inductive hypotheses and produces the observed lack of side effects in $H^r$.

The case where $\vgap{} > 0$ is more subtle. First, it follows from the \AxProcIII{} in $H$ that $v \le \cview{} + \vgap{}$.

By Lemma \ref{FCountLem} at the start of the constellation at $P$ in $H$, $\gvec{}[D] < v$

Furthermore it follows from Lemma \ref{OmnibusCBCASTLem}(\ref{OCL:gf}) that at the start of the constellation we have at $P$, in $H$:

$$
\fvec{}[D] \le \gvec{}[D] < v \le \cview{} + \vgap{}
$$

The First Pre-Critical Hypothesis insures that both inequalities true in $H^r$ as well.

The constellation does not change the value of $\fvec{}[D]$, because it does not include the receipt of a flush packet from $D$. Therefore the test loop at labeled step \ref{TTI:CheckFlush} of \ref{TryToInstall} fails and the call exits without changing the process state and without any side effects. As a result the same happens in the second call to \ref{TryToInstall}, and the constellation as a whole preserves the First Pre-Critical Hypothesis and conforms with the observed behavior of $H^r$.

\item In the pre-critical case when $P = D$ and $Q = D$, the process $D$ executes the \ref{ReceiveGhost} procedure twice, once for each ghost packet trigger from $\minusG$ and $G$, then the \ref{ReceiveFlush} procedure twice, once for each flush packet trigger from $\minusG$ and $G$, and finally the \ref{ReceiveGhost} procedure once more for the ghost packet trigger from $D$. The exact same thing happens at $G$ and $\minusG$, for a grand total of fifteen function calls.

The exact same analysis as in the previous case proves that the First Pre-Critical Hypothesis is preserved with respect to $D$, and that the resulting side effects in $D$ conform with the observed behavior in $H^r$ (namely, no side effects).

As for the execution at $\plusminusG$ in $H^r$, the Second Pre-critical Hypothesis insures that $\plusminusG$ have the same values of $\gvec{}[]$, $\fvec{}[]$ and \vgap{} as $D$. Therefore the execution of the constellation in $\plusminusG$ proceeds in the exact same way as in $D$, with the same changes to the state and the same lack of side effects.
\end{itemize}
\end{description}

\item
\begin{description}
\item[\bf Type of Constellation:] A flush packet of height $v$ is received at a process $P$ from a process $Q$ in $H$.
\item[\bf $H$ Transactions:] A single transaction, with an execution of the \ref{ReceiveFlush} procedure at process $P$.
\item[\bf Observed behavior in $H^r$:] There are two cases:
\begin{itemize}
\item Post-critically, or when $P \ne D$, there is a single transaction at $P$ with the same trigger and equivalent side effects.
\item Pre-critically, when $P = D$, there are three transactions at $D$, $G$ and $\minusG$, each triggered by an identical flush packet from $Q$, and with the side effects at $D$ being equivalent to the ones of the original transaction in $H$, and with the transactions in $G$ and $\minusG$ having no side effects.
\end{itemize}
\item[\bf Execution of \cbcast{} in $H^r$:] \hfill
\begin{itemize}
\item Post-critically, or when $P \ne D$, the process $P$ executes the \ref{ReceiveFlush} procedure in $H^r$.

The first step in this call increments the value of $\fvec{}[Q]$. Post-critically, this step obviously preserves the Interim $G$, Interim non-$G$ and Convergent Hypotheses, where relevant. Pre-critically, the Second Pre-Critical Hypothesis is preserved because $P \ne D$ and so no changes occur at either $D$, $G$ or $\minusG$. To see why the First Pre-Critical Hypothesis is preserved, notice first that in the pre-critical case $Q \ne \plusminusG$ because $Q$ exists in $H$ and $\plusminusG$ do not, which is why the $\fvec{}[\plusminusG]$ coordinates are not affected in $H^r$. Notice also that the value of $\gvec{}[D]$ is not affected in $H$. Taken together, these two observations explain why the hypothesis is preserved by the first step in this case.

The next and last step the procedure invokes the \ref{TryToInstall} procedure. By Lemma \ref{FirstTTIEquivLem} the \ref{TryToInstall} call preserves the inductive hypothesis in this case and generates equivalent side effects. It follows from Lemma \ref{ReducedContactLem} that the side effects have the
observed target set.

\item In the pre-critical case when $P = D$, the analysis of the First Pre-Critical Hypothesis is identical to the analysis of the previous case. For the Second Pre-Critical Hypothesis, there are three transactions in $H^r$, at $D$, $G$ and $\minusG$, each one executing the \ref{ReceiveFlush} procedure. This call increments the value of $\fvec{}[Q]$, which preserves the inductive hypothesis since $\fvec{}[]$ is identical in all three processes this case. In the second and last step the invocation of \ref{TryToInstall} preserves the inductive hypothesis and generates the expected side effects according to Lemma \ref{SecondTTIEquivLem}.
\end{itemize}
\end{description}
\end{itemize}

\subsubsection{Donation and co-donation packet constellations}

\begin{itemize}

\item
\begin{description}
\item[\bf Type of Constellation:] A non-critical donation packet is received at a process $J$ from a process $P$ in $H$.
\item[\bf $H$ Transactions:] A single transaction, with an execution of the \ref{ReceiveDonation} procedure at process $J$.
\item[\bf Observed behavior in $H^r$:] A single $H^r$ transaction at $J$ that is equivalent to the original $H$ transaction.
\item[\bf Execution of \cbcast{} in $H^r$:] The process $J$ executes the \ref{ReceiveDonation} procedure in $H^r$. Notice that since the donation is not critical then by definition $J \ne G$. Also by definition, donation packets are only sent post-critically. Therefore the only two possible operative inductive hypotheses are the Interim non-$G$ Hypothesis or the Convergent Hypothesis.

The \ref{ReceiveDonation} call starts with $J$ adding $P$ to \cset{}. According to either the Interim Non-$G$ Hypothesis or by the Convergent Hypothesis $\cset{}^r = \cset{}$ and so the first step preserves the inductive hypothesis.

In the next two steps, $J$ sends a co-donation to $P$ in both histories. The contents of \se{co\_donation} are made up of state variables that are equivalent in $H$ and $H^r$ regardless of whether the Interim non-$G$ or the Convergent Hypothesis is in force. Therefore this side effect is equivalent in both histories.

The next few steps construct the ordered set \se{UNT}. It follows from the two Hypotheses and from the fact that $\se{donation}^r \cong \se{donation}$ that $\se{UNT}^r \cong \se{UNT}$, since all the ingredients used in the construction of this set are equivalent in both histories. In addition the same hypotheses imply that $\mpkin{}^r[] = \mpkin{}[]$ and therefore the subsequent loop on the members of \se{UNT} invokes the same procedures (either \ref{ReceiveMessage} or \ref{ReceiveAck}) in the same order, and with equivalent parameters.

To see that the \ref{ReceiveMessage} call preserves the inductive hypotheses and generates equivalent side effects, simply follow the reasoning earlier in this proof for the case of a constellation that is triggered by the post-critical receipt of an original message. To see that the \ref{ReceiveAck} call preserves the inductive hypotheses and generates equivalent side effects, follow the reasoning earlier of this proof for case of a constellation that is triggered by the post-critical receipt of an acknowledgement in response to an original message. 

The last two steps of the call update $\gvec{}[P]$ and $\fvec{}[P]$ at $J$ from the donated values of \sgh{} and \sfh{}, respectively. The fact that $\se{donation}^r \cong \se{donation}$ guarantees that the latter two values are the same in $H$ and $H^r$. As a result $\gvec{}[P]$ and $\fvec{}[P]$ remain identical in $H$ and $H^r$ as required by both the Interim non-$G$ and the Convergent Hypotheses.
\end{description}

\item
\begin{description}
\item[\bf Type of Constellation:] A non-critical co-donation packet is received at a process $P$ from a process $J$ in $H$.
\item[\bf $H$ Transactions:] A single transaction, with an execution of the \ref{ReceiveCoDonation} procedure at process $P$.
\item[\bf Observed behavior in $H^r$:] A single $H^r$ transaction at $P$ that is equivalent to the original $H$ transaction.
\item[\bf Execution of \cbcast{} in $H^r$:] The process $P$ executes the \ref{ReceiveCoDonation} procedure in $H^r$. By definition, co-donation packets are only sent post-critically. However it is possible that $P = G$. Therefore the possible operative inductive hypotheses are the Interim non-$G$ Hypothesis, Interim $G$ Hypothesis or the Convergent Hypothesis.

The \ref{ReceiveCoDonation} call begins with constructing the ordered set \se{UNT}. It follows from the three Hypotheses and from the fact that $\se{co\_donation}^r \cong \se{co\_donation}$ that $\se{UNT}^r \cong \se{UNT}$, since all the ingredients used in the construction of this set are equivalent in both histories. In addition the same hypotheses imply that $\mpkin{}^r[] = \mpkin{}[]$ and therefore the subsequent loop on the members of \se{UNT} invokes the same procedures (either \ref{ReceiveMessage} or \ref{ReceiveAck}) in the same order, and with equivalent parameters.

To see that the \ref{ReceiveMessage} call preserves the inductive hypotheses and generates equivalent side effects, simply follow the reasoning earlier in this proof for the case of a constellation that is triggered by the post-critical receipt of an original message. To see that the \ref{ReceiveAck} call preserves the inductive hypotheses and generates equivalent side effects, follow the reasoning earlier of this proof for case of a constellation that is triggered by the post-critical receipt of an acknowledgement in response to an original message. 

The next two steps of the call update $\gvec{}[J]$ and $\fvec{}[J]$ at $P$ from the co-donated values of \sgh{} and \sfh{}, respectively.

The fact that $\se{co\_donation}^r \cong \se{co\_donation}$ guarantees that the latter two values are the same in $H$ and $H^r$. As a result $\gvec{}[J]$ and $\fvec{}[J]$ remain identical in $H$ and $H^r$ as required by the three Hypotheses.

In the last step the \ref{TryToInstall} procedure is invoked. By Lemma \ref{FirstTTIEquivLem}, this call preserves the inductive hypotheses and generates equivalent side effects. It follows from Lemma \ref{ReducedContactLem} that the side effects have the correct target set, namely $\cset{}^r$.
\end{description}

\item
\begin{description}
\item[\bf Type of Constellation:] A critical donation packet is received at $G$ from a process $P$ in $H$.
\item[\bf $H$ Transactions:] A single transaction, with an execution of the \ref{ReceiveDonation} procedure at process $G$.
\item[\bf Observed behavior in $H^r$:]
The original trigger (the receipt of a donation packet) disappears, since we removed that packet in $H^r$. The queuing of the co-donation packet also disappears, because the co-donation packet is removed as well. What remains are the side effects of all the sub-transactions, and in addition there is a new trigger event for each untimely packet in the $\channel{P}{D}$ channel, with the exception of the acknowledgement packets for original messages, since these packets do not get cloned. The new triggers and the existing side effects line up perfectly. By Corollary \ref{DonOrderCor} there is an order preserving one-to-one correspondence between the sub-transactions and the untimely message packets and forwarded acknowledgement packet in the $\channel{P}{D}$ channel. We carefully labeled the side effects and the new triggers so that they pair up according to that correspondence to make complete transactions. In addition, there are untimely ghost and flush packets in the $\channel{P}{D}$ channel whose clones give rise to additional triggers for which there are no corresponding sub-transactions in $H$. This is not a problem, these naked triggers simply represent transactions that do not have side effects. Our only burden is to prove that if $H^r$ executes the \cbcast{} algorithm, then these ghost and flush triggers will indeed produce no side effects and preserve the inductive hypotheses.
\item[\bf Execution of \cbcast{} in $H^r$:]
The critical donation packet is the first post-critical packet from $P$ that $G$ processes in $H$. We carefully labeled the untimely packets in $\channel{P}{G}^{H^r}$ so that they are all processed within the critical donation constellation. Therefore in $H^r$ process $G$ does not process any packets from $P$ in the interval between the critical constellation and the $P$-donation constellation. As a result at the start of the constellation we have in $H^r$, according to Lemma \ref{OmnibusCBCASTLem}(\ref{OCL:gfxmitrcv}, \ref{OCL:sendgfA} and \ref{OCL:height})
\begin{multline*}
\gvec{}[P](G) = \valuepost{\gvec{}[P]}{G}{\ell_{\vcrit}} \le \valuepre{\sgh{}}{P}{\ell_{\vcrit}} = {}  \\
{} = \valuepre{(\cview{} + \vgap{})}{P}{\ell_{\vcrit}} < \vcrit{}
\end{multline*}
\begin{multline*}
\fvec{}[P](G) = \valuepost{\fvec{}[P]}{G}{\ell_{\vcrit}} \le \valuepre{\sfh{}}{P}{\ell_{\vcrit}} = {}  \\
{} = \valuepre{(\cview{} + \vgap{})}{P}{\ell_{\vcrit}} < \vcrit{}
\end{multline*}

Therefore when the $P$-donation constellation starts, the value of $\fvec{}[P]$ at $G$ is too low for $G$ to have already installed view $\vcrit$. Therefore $G$ is still in its interim period (see \ref{ProofPlanSSS}), and therefore the operating inductive hypothesis is the Interim $G$ Hypothesis.

All process $G$ does in $H^r$ is process all the clones of untimely packets in $\channel{P}{D}$. In $H$, process $G$ executes the \ref{ReceiveDonation} procedure. This means that it adds $P$ to \cset{}, then queues a co-donation to $P$, and only then proceeds to process a similar sequence of packets, with the ghosts and flushes excluded. The addition of $P$ to \cset{} does not violate the Interim $G$ Hypothesis. The fact that a co-donation is queued in $H$ but not in $H^r$ is what we expect since we explicitly removed the critical donation packets from $H^r$.

So before either history processes the first untimely packet the Interim $G$ Hypothesis still holds and the side effects meet our expectations. We have to show that the Donation Sub-Hypothesis also holds at this point. We have already shown that the values of $\gvec[P]$ and $\fvec[P]$ at $G$ are no higher than the critical values of \sgh{} and \sfh{}, respectively, at $P$. This proves that the Donation Sub-Hypothesis holds at that point.

We are going to show, by induction, that up to the $i^{th}$ untimely clone the history $H^r$ looks like an execution of \cbcast{} and the Donation Sub-Hypothesis holds. This means that we assume that in $H^r$, process $G$ has processed all the clones of untimely packets in channel $\channel{P}{D}$ up to, but not including the $i^{th}$ clone. We assume that in $H$, process $G$ has executed all the sub-transactions that correspond to the same clones, excluding the ghost and flush clones which do not have corresponding sub-transactions. We also assume that at that point the Donation Sub-Hypothesis holds. Let us look now at the $i^{th}$ clone. What does $G$ do with it in $H$ and $H^r$? We look at each case separately.

\begin{itemize}
\item {\bf The clone is a message packet}. In this case $G$ invokes the \ref{ReceiveMessage} procedure in both histories. The message can be either original or forwarded. We have already analyzed similar cases earlier in the proof (post-critical message receipt) and showed that in both cases the Interim $G$ Hypothesis is preserved. Since the Donation Sub-Hypothesis only differs from the Interim $G$ Hypothesis with respect to variables that are not used by the \ref{ReceiveMessage} procedure, the same analysis holds without change in this case as well.

\item {\bf The clone is an acknowledgement packet}. In this case $G$ invokes the \ref{ReceiveAck} procedure in both histories. The packet has to be in response to a forwarded message. We have already analyzed a similar case earlier in the proof (post-critical receipt of a forwarded acknowledgement packet) and showed that the Interim $G$ Hypothesis is preserved. Since the Donation Sub-Hypothesis only differs from the Interim $G$ Hypothesis with respect to variables that are not used by the \ref{ReceiveAck} procedure, the same analysis holds without change in this case as well.

\item {\bf The clone is a ghost packet}. In this case $G$ invokes the \ref{ReceiveGhost} procedure in $H^r$, but there is no corresponding action in $H$. The procedure produces no side effects, which is what we would expect (since $G$ does not do anything at all in $H$), but it does increase $\gvec[P]$. This deviation is allowed by the Donation Sub-Hypothesis, as long as the value $\gvec[P]$ does not exceed the critical value of \sgh{} at $P$. This is guaranteed by Lemma \ref{FCountLem}, since the packet being processed is untimely and therefore queued by $P$ pre-critically.

\item {\bf The clone is a flush packet}. This case is argued similarly to the ghost case as far as state is concerned. However a flush packet may cause side effects. Since $G$ does nothing in $H$, there must not be any side effects in $H^r$. For there to be any side effects in $H^r$ the flush value must be high, i.e. it must be equal to $\cview{} + \vgap{}$. But this is impossible since Lemma \ref{FCountLem} guarantees that the flush value is at most the critical value of \sfh{} at $P$, and Lemma \ref{OmnibusCBCASTLem} guarantees that the latter value is lower than $\vcrit{}$,
while the current value of $\cview{} + \vgap{}$ at $G$ is at least $\vcrit{}$.
\end{itemize}
We have established that once the process $G$ is done processing all the sub-transactions (in $H$) and all the untimely packets (in $H^r$), the Donation Sub-Hypothesis still holds.

This is the end of the constellation in $H^r$. What are the values of $\gvec[P]$ and $\fvec[P]$ at $G$ in this history? Since all the untimely packets are now processed, it follows from Lemma \ref{FCountLem} that these values are equal to the critical values of $\sgh{}$ and $\sfh{}$, respectively, at $P$.

The donation transaction is not yet over in $H$, however. As the last step in the \ref{ReceiveDonation} procedure, $G$ updates its values for $\gvec[P]$ and $\fvec[P]$ from the donated values of $\sgh{}$ and $\sfh{}$ that it receives as a donation from $P$. Since the $P$ donation reflects its critical state, this last step restores the Interim $G$ Hypothesis and we are done.
\end{description}

\item
\begin{description}
\item[\bf Type of Constellation:] A critical co-donation packet is received at a process $P$ from $G$ in $H$.
\item[\bf $H$ Transactions:] A single transaction, with an execution of the \ref{ReceiveCoDonation} procedure at process $P$.
\item[\bf Observed behavior in $H^r$:]
Like the case of a critical donation, the observed behavior of a critical co-donation constellation in $H^r$ is made up of triggers and sub-transaction side effects that line up perfectly:
\begin{itemize}
\item Side effects of the co-donation transaction in $H$ that carry over to $H^r$.
\item Trigger events, made up of the processing events of cloned and zombied packets:
\begin{itemize}
\item A trigger for each clone and zombie of an untimely packet in the $\channel{D}{P}$ channel. Original message packets and flush packets do not get cloned on this channel, so they have no corresponding triggers. Ghost packets get both cloned and zombied, so each of these packets have two corresponding triggers.
\item A trigger for each clone of a post-critical, pre-$P$-donation packet in the $\channel{G}{G}$ channel. Acknowledgement packets on this channel do not get cloned, so they do not have a corresponding trigger.
\end{itemize}
\end{itemize}
According to Lemma \ref{DonationSELem}, if the invocation of the \ref{TryToInstall} procedure at the end of the \ref{ReceiveCoDonation} procedure installs the pending views, then $\se{UNT}$ is empty and there are no sub-transactions. According to Corollary \ref{CoDonOrderCor}, whatever sub-transactions do exist line up perfectly with the $H^r$ triggers for message and acknowledgement packets. According to Lemma \ref{CoDonFlushLem} the side effects of \ref{TryToInstall} line up with a final, post-critical cloned flush packet trigger.

Therefore, if \ref{TryToInstall} does not install pending views, then all the ghost and flush triggers produce no side effects. If \ref{TryToInstall} does install the pending views, then there are no triggers other than ghost and flush triggers, and none of them produces any side effects except possibly the last one. Our challenge is to prove that the \cbcast{} protocol produces exactly those side effects while preserving the inductive hypothesis.

\item[\bf Execution of \cbcast{} in $H^r$:]
When $G$ receives the critical donation packet from $P$, it adds $P$ to its \cset{} and queues the critical co-donation packet to $P$. Prior to receiving the donation, $P$ is not in $G$'s \cset{} (see Lemma \ref{OmnibusCBCASTLem}(\ref{OCL:lset})) and therefore it does not queue any packets to $P$. As a result, the co-donation packet is the first packet that $P$ receives from $G$ in $H$. Therefore the value of $\gvec{}[G]$ and $\fvec{}[G]$ at $P$ remains equal to the original value that $P$ sets when it executes the critical invocation of \ref{JoinNotification}.
$$
\valuepre{\gvec{}[G]}{P}{\crittime{G}{P}} = \valuepre{\fvec{}[G]}{P}{\crittime{G}{P}} = \valuepre{\gvec{}[D]}{P}{\ell_{\vcrit}} \le \valuepre{\sgh{}}{D}{\ell_{\vcrit}}
$$
Where the last inequality follows from Lemma \ref{OmnibusCBCASTLem}(\ref{OCL:gfxmitrcv}). We also know from Lemma \ref{OmnibusCBCASTLem} that $\valuepre{\sgh{}}{D}{\ell_{\vcrit}} < \vcrit$ and therefore $P$ does not have a sufficient value of $\fvec{}[G]$ at the start of the constellation to allow for the installation of view $\vcrit{}$. As a result the operating inductive hypothesis at the onset of the constellation is the Interim non-$G$ Hypothesis, and from the inequality above it follows that the First Co-Donation Sub-Hypothesis holds as well.

To complete the proof for this constellation, we first use induction to show that the First Co-Donation Sub-Hypothesis holds, and the expected side effects occur, up to the last untimely clone or zombie. Then we use a second induction to deal with the post-critical clones and the \ref{TryToInstall} side effects.

We start by showing that up to the $i^{th}$ untimely clone the history $H^r$ looks like an execution of \cbcast{} and the First Co-Donation Sub-Hypothesis holds. This means that we assume that in $H^r$, process $P$ has processed all the clones and zombies of untimely packets in channel $\channel{D}{P}$ up to, but not including the $i^{th}$ clone or zombie. We assume that in $H$, process $P$ has executed all the sub-transactions that correspond to the same clones and zombies, excluding the ghost clones and flush zombies which do not have corresponding sub-transactions. Let us look now at the $i^{th}$ clone or zombie. What does $P$ do with it in $H$ and $H^r$? We look at each case individually.

\begin{itemize}
\item {\bf The clone is a message packet}. In this case $P$ invokes the \ref{ReceiveMessage} procedure in both histories.  The packet has to be a forwarded message packet. We have already analyzed similar cases earlier in the proof (post-critical message receipt) and showed that in both cases the Interim Non-$G$ Hypothesis is preserved. Since the First Co-Donation Sub-Hypothesis only differs from the Interim Non-$G$ Hypothesis with respect to variables that are not used by the \ref{ReceiveMessage} procedure, the same analysis holds without change in this case as well.

\item {\bf The clone is an acknowledgement packet}. In this case $P$ invokes the \ref{ReceiveAck} procedure in both histories. We have already analyzed a similar case earlier in the proof (post-critical receipt of a forwarded acknowledgement packet) and showed that the Interim Non-$G$ Hypothesis is preserved. Since the First Co-Donation Sub-Hypothesis only differs from the Interim Non-$G$ Hypothesis with respect to variables that are not used by the \ref{ReceiveAck} procedure, the same analysis holds without change in this case as well.

\item {\bf The clone is a ghost packet}. In this case $P$ invokes the \ref{ReceiveGhost} procedure in $H^r$, but there is no parallel action in $H$. The procedure produces no side effects, which is what we would expect (since $P$ does not do anything at all in $H$), but it does increase $\gvec{}^r[G]$. This deviation is allowed by the First Co-Donation Sub-Hypothesis, as long as the value $\gvec{}^r[G]$ does not exceed $\valuepre{\sgh{}}{D}{\ell_{\vcrit}}$. This is guaranteed by Lemmas \ref{FCountLem} and \ref{OmnibusCBCASTLem}, since the packet being processed is untimely and therefore queued by $D$ pre-critically.

\item {\bf The zombie is a flush packet}. This case is argued similarly to the ghost case as far as state is concerned with the additional fact that
$$
\valuepre{\sfh{}}{D}{\ell_{\vcrit}} \le \valuepre{\sgh{}}{D}{\ell_{\vcrit}} < \vcrit{}
$$
which follows Lemma \ref{OmnibusCBCASTLem}. The inequality implies that in $H^r$ the value of $\fvec{}[G]$ remains low and as a result there are no side effects. This is what we expect since $P$ does nothing in $H$ and therefore does not produce side effects there.
\end{itemize}
We have established that once the process $P$ is done processing all the pre-critical sub-transactions (in $H$) and all the untimely packets (in $H^r$), the First Co-Donation Sub-Hypothesis still holds and all the side effects are as expected. At this point in $H^r$ the process $P$ has already processed the last pre-critical ghost and flush packets from $G$. Therefore by Lemma \ref{FCountLem}
\begin{align*}
\gvec{}^r[G] & = \valuepre{\sgh{}^r}{G}{\ell_{\vcrit}} = \valuepre{\sgh{}^r}{D}{\ell_{\vcrit}} = \valuepre{\sgh{}}{D}{\ell_{\vcrit}} \\
\fvec{}^r[G] & = \valuepre{\sfh{}^r}{G}{\ell_{\vcrit}} = \valuepre{\sgh{}^r}{D}{\ell_{\vcrit}} = \valuepre{\sgh{}}{D}{\ell_{\vcrit}}
\end{align*}
Where the two rightmost equalities in each line follow from the Second Pre-Critical and First Pre-Critical Hypotheses, respectively. We also know that
$$
\valuepre{\sgh{}^r}{G}{\ell_{\vcrit}} \le \valuepre{\sgh{}^r}{G}{\crittime{P}{G}} = \valuepre{\sgh{}}{G}{\crittime{P}{G}}
$$
To see why the left inequality is true, notice that $\crittime{P}{G} \,\dot{\prec}\, \crittime{G}{P}$ (this is mediated by the critical co-donation packet) and therefore we can assume by induction that $G$ executes \cbcast{} up to constellation $\crittime{P}{G}$. Now the inequality follows from the fact that \sgh{} is monotone increasing and $\ell_{\vcrit} \,\dot{\prec}\, \crittime{P}{G}$. The right equality follows from the Interim $G$ Hypothesis.

Taken together these relations prove that the Second Co-Donation Sub-Hypothesis now holds.

We divide the rest of the analysis into two parts. First, assume that the \ref{TryToInstall} invocation at the end of the \ref{ReceiveCoDonation} procedure fails to install the pending views.

We will show that up to the $i^{th}$ post-critical clone the history $H^r$ looks like an execution of \cbcast{} and the Second Co-Donation Sub-Hypothesis holds. This means that we assume that in $H^r$, process $P$ has processed all the clones of post-critical packets in channel $\channel{G}{G}$ up to, but not including the $i^{th}$ clone. We assume that in $H$, process $P$ has executed all the sub-transactions that correspond to the same clones, excluding the ghost and flush which do not have corresponding sub-transactions. Let us look now at the $i^{th}$ clone. What does $P$ do with it in $H$ and $H^r$? We look at each case individually.

\begin{itemize}
\item {\bf The clone is a message packet}. In this case $P$ invokes the \ref{ReceiveMessage} procedure in both histories. We have already analyzed similar cases earlier in the proof (post-critical message receipt) and showed that in both cases the Interim Non-$G$ Hypothesis is preserved. Since the Second Co-Donation Sub-Hypothesis only differs from the Interim Non-$G$ Hypothesis with respect to variables that are not used by the \ref{ReceiveMessage} procedure, the same analysis holds without change in this case as well.

\item {\bf The clone is an acknowledgement packet}. This case does not occur since post-critical acknowledgement packets do not get cloned.

\item {\bf The clone is a ghost packet}. In this case $P$ invokes the \ref{ReceiveGhost} procedure in $H^r$, but there is no parallel action in $H$. The procedure produces no side effects, which is what we would expect (since $P$ does not do anything at all in $H$), but it does increment $\gvec{}^r[G]$. This deviation is allowed by the Second Co-Donation Sub-Hypothesis, as long as the value $\gvec{}^r[G]$ does not exceed $\valuepre{\sgh{}}{G}{\crittime{P}{G}}$. This is guaranteed by Lemma \ref{FCountLem}, since the packet being processed is queued by $G$ before $\crittime{P}{G}$.

\item {\bf The clone is a flush packet\footnote{Untimely flushes in $\channel{G}{P}$ are zombies, but post-critical flushes are clones}}. This case is argued similarly to the ghost case as far as state is concerned with the additional fact that
$$
\valuepre{\sfh{}}{G}{\crittime{P}{G}} \le \valuepre{\sgh{}}{G}{\crittime{P}{G}}
$$
which follows from Lemma \ref{OmnibusCBCASTLem}(\ref{OCL:sendgfA}). However a flush packet may cause side effects. Since $P$ does nothing in $H$, there must not be any side effects in $H^r$.

Suppose that the flush packet does have side effects. For this to happen, the \ref{TryToInstall} invocation in \ref{ReceiveFlush} must install the pending views. To do that it is required that for every $Q \in \lset{}^r$ we have $\fvec{}^r[Q] = \cview{}^r + \vgap{}^r$. Since the Second Co-Donation Sub-Hypothesis holds, we know that 
\begin{align*}
\lset{}^r & = \lset{} \\
\cview{}^r & = \cview{} \\
\vgap{}^r & = \vgap{} \\
\fvec{}^r[Q] & = \fvec{}[Q] \quad \text{for all } Q \in \lset{} \setminus \{ G \}
\end{align*}
Therefore in $H$, for any live $Q \ne G$ we have $\fvec{}[Q] = \cview{} + \vgap{}$

Since the flush packet had side effects in $H^r$ we know that the packet was $\pkf{v}$ with $v = \cview{} + \vgap{}$. Since the original packet was queued by $G$ before it processed the donation from $P$ it follows that
$$
v \le \valuepre{\sfh{}}{G}{\crittime{P}{G}}
$$
In $H$, right before calling \ref{TryToInstall} in \ref{ReceiveCoDonation}, process $P$ updates its value of $\fvec{}[G]$, setting it to
\begin{multline*}
\fvec{}[G](P) = \fp{co\_donation}.\sfh{} = \valuepre{\sfh{}}{G}{\crittime{P}{G}} \ge \\
{} \ge v = \cview{} + \vgap{}
\end{multline*}
It follows that the invocation of \ref{TryToInstall} in the \ref{ReceiveCoDonation} procedure does succeed in installing the pending views in $H$, contrary to our assumptions.
\end{itemize}
We have established that once the process $P$ is done processing all the post-critical sub-transactions (in $H$) and all the post-critical packets (in $H^r$), the Second Co-Donation Sub-Hypothesis still holds and all the side effects are as expected.

At this point in $H$ process $P$ proceeds to update its values of $\gvec{}[G]$ and $\fvec{}[G]$:
\begin{align*}
\gvec{}[G] & = \fp{co\_donation}.\sgh{} = \valuepre{\sgh{}}{G}{\crittime{P}{G}} \\
\fvec{}[G] & = \fp{co\_donation}.\sfh{} = \valuepre{\sfh{}}{G}{\crittime{P}{G}}
\end{align*}

In $H^r$, $P$ has processed all the post-critical clones, including the last of the post-critical ghost and flush packets from $G$. It follows from Lemma \ref{FCountLem} that at this point
\begin{align*}
\gvec{}^r[G] & = \valuepre{\sgh{}}{G}{\crittime{P}{G}} \\
\fvec{}^r[G] & = \valuepre{\sfh{}}{G}{\crittime{P}{G}}
\end{align*}

Therefore $\gvec{}^r[G] = \gvec{}[G]$ and $\fvec{}^r[G] = \fvec{}[G]$ and the Interim Non-$G$ Hypothesis is restored. The last step in $H$ is an invocation of the \ref{TryToInstall} procedure, that by assumption fails to install the pending views. Therefore it produces no side effects and does not change any state variables.  

We are finally left with the case where \ref{TryToInstall} does succeed in installing the pending views. It follows from Lemma \ref{DonationSELem} that in this case $\se{UNT}$ is empty and no sub-transactions are executed in $H$. In addition Lemma \ref{CoDonFlushLem} shows that the last regular packet queued by $G$ prior to processing the donation from $P$ is $k_{\fp{last}} = \pkf{v}$ with $v = \cview{} + \vgap{}$.

It follows that the only post-critical clones in $H^r$ are ghost and flush packets, and their processing, up to $k$, causes only allowable state changes and generates no side effects (this is shown in exactly the same way as in the previous case, where \ref{TryToInstall} does not install pending views).

So at last we are at the point where in $H$ process $P$ is about to update its $\gvec{}[G]$ and $\fvec{}[G]$ values and invoke \ref{TryToInstall}, which is going to succeed in installing the pending views. In $H^r$ process $P$ is about to process $k = \pkf{v}$, the last cloned packet from $G$. The Second Co-Donation Sub-Hypothesis still holds.

Since \ref{TryToInstall} succeeds in installing the pending views, it follows that in $H$, after $P$ updates $\fvec{}[G]$ we have $\fvec{}[G] = \cview{} + \vgap{}$. Therefore
$$
v = \cview{} + \vgap{} = \fp{co\_donation}.\sfh{} = \valuepre{\sfh{}}{G}{\crittime{P}{G}}
$$
By Lemma \ref{OmnibusCBCASTLem} we have 
\begin{multline*}
\cview{} + \vgap{} \ge \fp{co\_donation}.\sgh{} = \valuepre{\sgh{}}{G}{\crittime{P}{G}} \ge \\
{} \ge \valuepre{\sfh{}}{G}{\crittime{P}{G}} = \cview{} + \vgap{}
\end{multline*}
and therefore $P$ also updates $\gvec{}[G] = \cview{} + \vgap{}$.

In $H^r$ the first step in \ref{ReceiveFlush} sets $\fvec{}^r[G] = v = \cview{} + \vgap{}$ and it follows from Lemma \ref{OmnibusCBCASTLem}(\ref{OCL:gf}) that at that point we already have
$$
\gvec{}^r[G] = v = \cview{} + \vgap{}
$$
Therefore, if we hold both $H$ and $H^r$ at the point where both histories are about to invoke the  \ref{TryToInstall} procedure (in the \ref{ReceiveCoDonation} and the \ref{ReceiveFlush} procedures, respectively), the Interim Non-$G$ Hypothesis is restored.

It follows from Lemma \ref{FirstTTIEquivLem} that the execution of \ref{TryToInstall} in both histories generates the expected side effects and causes the interim period to end and the Convergent Hypothesis to hold. Lemma \ref{ReducedContactLem} ensures that the side effects have the correct target set. This concludes the proof of the theorem and shows that $H$ and $H^r$ perform the same computation.
\end{description}
\end{itemize}

\section{Causality and Progress With The \cbcast{} Algorithm}
\subsection{ The Goals of the Analysis}
To show that the algorithm works as expected, we have to prove two things. One is the preservation of causality,
meaning that a message is only delivered at a process after all the messages on which it depends have already been delivered.
The second one is the guarantee of progress - that is to say that all messages eventually get delivered. It is well known (see \cite{fischer85}) that progress cannot be guaranteed in the presence of failures, and it follows from \cite{schiper94} that a similar limitation exists in our model of dynamic membership, even without failures. Therefore we can only prove somewhat less than a full guarantee of progress.

\subsection{The Causal Order Property and The Progress Property}
\begin{defn}
\label{FamiliarMessageDef}
We say that a message $m$ is {\bf familiar} to a process $P$ if either
\begin{itemize}
\item $m$ is delivered at $P$.
\item There is a sequence of processes
$$
P_0, P_1, \ldots, P_n = P \quad \text{where } n > 0
$$
such that
\begin{enumerate}
\item For each $i < n$, the process $P_i$ is the parent of process $P_{i+1}$, i.e. $\exists \pnotifyevent{j(P_{i+1})}{P_i} = \nj{P_{i+1}}{P_i}$.
\item $n$ is delivered at $P_0$ prior to the $\pnotifyevent{j(P_1)}{P_0}$ event.
\end{enumerate}
\end{itemize}
\end{defn}

\begin{defn}
\label{COP}
We say that a history $H$ of the \cbcast{} protocol has the {\bf Causal Order Property} if messages in $H$ are delivered in an order that respects their causal relationships. Technically, this means that if a process $P$ originally broadcasts a message $m$ that is eventually delivered at a process $Q$ then
\begin{enumerate}
\item every message that was broadcast by $P$ prior to broadcasting $m$ is already familiar to $Q$ at the delivery of $m$.
\label{COP:bcast}
\item every message that is familiar to $P$ at the broadcasting $m$ is already familiar to $Q$ at the delivery of $m$.
\label{COP:fam}
\end{enumerate}
\end{defn}
\begin{defn}
We say that a history $H$ of the \cbcast{} protocol has the {\bf Progress Property} if every message that is originally broadcast by a non-halting process in $H$ eventually becomes familiar to every non-halting process in $H$.
\end{defn}

\begin{thm}[Causal Order Theorem]
\label{CausalOrderThm}
The Causal Order Property holds for every conforming history of the \cbcast{} protocol.
\end{thm}

\begin{thm}[Restricted Progress Theorem]
\label{RestrictedProgressThm}
The Progress Property holds for every finite-join conforming history of the \cbcast{} protocol.
\end{thm}

\subsection{Causal Order And Progress In Reduced Histories}
As a first step in proving the two theorems, we show that if any of the two properties holds in the reduction of a transactional history, then it holds in the original history as well. This allows us to ignore any finite number of process join notifications that may occur in the course of the history. In our proofs we make use of the fact that $H$ and $H^r$ have been labeled using a common label space $\labelspace{}$. This allows us to compare the timing of events that occur in the two histories. As usual we denote by $G$ the first joining process in $H$, with $D$ denoting its parent. We will casually refer to "time" when we technically mean "constellation label". The critical time is the constellation $\ell_{\vcrit}$ in the joint label space, when $G$ joins in $H$.

\begin{defn}
We say that a message  $m$ is {\bf taken up} at process $P$ if the process moves the message $m$ into \rset{} (this takes place at labeled step \ref{RMP:ReceiveEntry} of the \ref{ReceiveMessage} procedure). We will see later that the same message can be received by the same process many times, but is only taken up once.
\end{defn}

We need the following basic facts:
\begin{itemize}
\item By construction $H^r$ and $H$ have the same message broadcast requests ($\requestset{}^r = \requestset{}$) and those requests are labeled the same way in both histories. This means that the processes in $\processset^{H^r}$ originate the same message broadcasts, at the same time, as they do in $H$.
\item We know from the \HistoryEquivThm{} (Theorem \ref{HistoryEquivThm}) that
every process other than $G$ delivers the same messages, at the same time, in both histories; that the same is true for $G$ post-critically; and that pre-critically process $G$ delivers the same messages in $H^r$ that $D$ delivers in $H$, at the same time. As a result $G$ is familiar, at any post-critical time, with the same messages in both histories. For other processes in $H$ this is true without restriction.
\item It follows from the proof of the \HistoryEquivThm{} that every process other than $G$ takes up the same messages, at the same time, in both histories; that the same is true for $G$ post-critically; and that pre-critically process $G$ takes up the same messages in $H^r$ that $D$ takes up in $H$, at the same time. See in particular the parts of the proof that relate to receiving message packets, donation packets and co-donation packets. Notice that cloned message packets always result in the message being discarded rather than being taken up.
\end{itemize}

\begin{thm}
Let $H$ be a transactional history that includes at least one join notification, and let $H^r$ be its reduction. If The Causal Order Property holds in $H^r$ then it holds in $H$ as well.
\end{thm}

\begin{proof}
Let $P$ and $Q$ be any processes in $H$ and let $m$ be a message that originates at $P$ and is eventually delivered at $Q$. Let $n$ be a message is either originated by $P$ prior to originating $m$ or is familiar to $P$ when it originates $m$.

If $P$ originates $n$ prior to $m$ in $H$, then $P$ also originates $n$ prior to $m$ in $H^r$. If $P$ is familiar with $n$ when it originates $m$ in $H$, it is also familiar with it in $H^r$ when it originates $m$. Since $Q$ delivers $m$ in $H$, it also delivers $m$ in $H^r$. Because the Causal Order Property holds in $H^r$ this implies that $Q$ is familiar with $n$ at the time that it delivers $m$ in $H^r$. If that time is post-critical, then $Q$ is also familiar with $n$ at the same time in $H$ and we are done. If the time is pre-critical, then $Q \ne G$ (because $G$ does not exist in $H$ pre-critically) and therefore $Q$ is familiar with $n$ at the same time in $H$ regardless.
\end{proof}

\begin{thm}
Let $H$ be a transactional history that includes at least one join notification, and let $H^r$ be its reduction. If The Progress Property holds in $H^r$ then it holds in $H$ as well.
\end{thm}

\begin{proof}
Let $P$ and $Q$ be two processes in $H$ that never halt, and let $m$ be a message that is originated by $P$. Then $P$ and $Q$, as processes in $H^r$ are also non-halting, and $P$ originates $m$ in $H^r$ as well. Since the Progress Property holds in $H^r$, $Q$ eventually becomes familiar with $m$ in $H^r$. If this happens post-critically or if $Q \ne G$ then $Q$ also becomes familiar with $m$ in $H$. If it happens pre-critically and $Q = G$ then after the critical time $G$ becomes familiar in $H$ with all the messages that $G$ was familiar with in $H^r$ pre-critically, since it does not halt. Therefore in all cases $Q$ eventually becomes familiar with $m$ in $H$.
\end{proof}

Transactional histories are artificial and do not arise naturally the way conforming histories do. In addition, the reduction of a transactional history is rarely transactional. Therefore we need the following lemma.

\begin{lem}
Let $H$ be a conforming history and let $\tr(H)$ be its transactional closure (see the \FaultThm{}, Theorem \ref{FaultThm}). If $\tr(H)$ has the Causal Order Property then so does $H$. If $\tr(H)$ has the Progress Property then so does $H$.
\end{lem}

\begin{proof}
We start with the Causal Order Property. Suppose that in $H$ a message $m$ originates at process $P$ and is delivered at process $Q$. Suppose that a message $n$ is either originated by $P$ prior to originating $m$, or else is familiar to $P$ at the time that it originates $m$. History $\tr(H)$ contains all the events of $H$, and by construction an $H$ dequeuing event is processed the same way in $\tr(H)$ as it is in $H$. Therefore $m$ originates at $P$ and is delivered at $Q$ in $\tr(H)$ and $n$ is originated at $P$ or is familiar to $P$ in $\tr(H)$ prior to the origination of $m$. Since $\tr(H)$ has the causal order property, $n$ must be delivered at $Q$ prior to the delivery of $m$, in $\tr(H)$. By construction $\tr(H)$ does not contain any new events that precede existing $H$ events. Therefore $n$ must be delivered at an original $H$ event, which means that $n$ is delivered in $H$ as well. Therefore $H$ has the Causal Order Property.

To prove the Progress Property suppose that in $H$ a message $m$ originates at a non-halting process $P$ and suppose that $Q$ is also a non-halting process, not necessarily distinct from $P$. Then in $\tr(H)$ both $P$ and $Q$ are non-halting and $P$ originates $m$. Since $\tr(H)$ has the Progress Property $m$ is delivered at $Q$ in $\tr(H)$, as part of the processing of some event $e$. If $e$ is an original $H$ event then $m$ is delivered in $H$ and we are done. Otherwise $e$ is a vacuum event that is generated by the vacuum loop. But the vacuum loop is only applied to halting processes, and $Q$ does not halt. This concludes the proof.
\end{proof}

\subsection{Stunting}

A {\em join-free history} is a history $H$ where no processes join. In other words all the processes in $\processset^H$ are members of view zero, and all view changes are removals of existing members. The discussion so far implies that if only we could prove that all join-free transactional histories enjoy the Causal Order Property and the Progress Property, then any finite-join conforming history would enjoy these properties as well. We are going to prove exactly that in the next section. As far as the Restricted Progress Theorem is concerned, nothing more is claimed. However, the Causal Order Theorem claims that the Causal Order Property holds for all conforming histories, not only finite-join ones. We plug this gap by introducing the notion of {\em stunting}.

\begin{defn}
Let $H$ be a transactional history and let $0 < v < \numview{}^H$ be any view change in $H$. Then the {\bf stunting} of $H$ at view $v$, denoted $H^{<v}$, is the stunted history (see Definition \ref{StuntedHistoryDef}) that is obtained from $H$ by taking only that part of $H$ that precedes the view change constellation $\ell_v$. We now define the stunting as follows:
\begin{align*}
\processset^{H^{<v}} & = \bigcup_{u<v} \viewset^H_u \\
\processset_h^{H^{<v}} & = \processset^{H^{<v}} \\
\numview^{H^{<v}} & = v \\
\viewset_i^{H^{<v}} & = \viewset_i^H \quad \text{for all } i < v \\
\packetset^{H^{<v}} & = \left\{ k \in \packetset^H | \view(\psendevent{k}) < v \right\} \\
\channel{P}{Q}^{H^{<v}} & = \channel{P}{Q}^H \,\bigcap\, \packetset^{H^{<v}} \quad \text{for all } P,Q \in \processset^{H^{<v}} \\ 
\eventset^{H^{<v}} & = \left\{ e \in \eventset^H | \view(e) < v \right\} \\
\prec^{H^{<v}} & = \prec^H \,\bigcap\, \left( \eventset^{H^{<v}} \times \eventset^{H^{<v}} \right) \\ 
\notificationset^{H^{<v}} & = \left\{ \notify{i}{P} \in \notificationset \vert i < v \right\} \\
\requestset^{H^{<v}} & = \left\{ r \in \requestset \vert \pappevent{r} \in \eventset^{H^{<v}} \right\}
\end{align*}
\end{defn}
All the packets and notifications have the same faulting characteristics that they inherit from $H$. Since $H$ is transactional there are no dropped packets or notifications.

The proof that $H^{<v}$ is a conforming history is tedious but straightforward and we omit it. Moreover, if all the processes in $H^{<v}$ start with the same internal state they had in $H$ then $H^{<v}$ becomes a history of the execution of the \cbcast{} protocol and at any point in time up to $\ell_v$ the internal state in the processes in $H^{<v}$ is identical to the internal state of the same processes in $H$.

\begin{thm}
Assume that the Causal Order Property holds for all join-free transactional histories. Then it holds for all conforming histories.
\end{thm}

\begin{proof}
We have already shown that under the assumption the property holds for all finite-join conforming histories. Let $H$ be an infinite-join conforming history and let $P$ and $Q$ be processes in $H$. Suppose that $P$ originates a message $m$ and that $m$ is delivered at $Q$. Assume also that a message $n$ is either originated by $P$ prior to the origination of $m$ or that $n$ is familiar to $P$ at the time that it originates $m$. Since $H$ is an infinite-join history, there is a view $v$ for which the view change notification event occurs after $Q$ delivers $m$.

Look at the $v$ stunting of $H$. Since $H^{<v}$ has the same state as $H$ up to time $\ell_v$ and since $m$ is delivered at $Q$ before time $\ell_v$, $m$ is also delivered at $Q$ at the same time in $H^{<v}$. The origination of $m$ by $P$ occurs before $Q$ delivers $m$, so it also occurs before $\ell_v$ and therefore occurs in $H^{<v}$.

For the same reasons message $n$ is either originated by $P$ prior to $m$ or is familiar to $P$ at the time $m$ is originated, in the stunted history. Since the stunted history is a finite-join conforming history, the Causal Order Property holds there and therefore $Q$ is familiar with $n$ at the time that it delivers $m$. Therefore the same happens in $H$ and we are done.
\end{proof}

All we have left to do is prove the Causal Order Theorem (Theorem \ref{CausalOrderThm}) and the Restricted Progress Theorem (Theorem \ref{RestrictedProgressThm}) in the case of join-free transactional histories. For the rest of the paper we will consider only such histories. In particular, we will assume without further comment that every view notification is a removal, and we will not refer to message familiarity (see Definition \ref{FamiliarMessageDef}) since it is now synonymous with message delivery. Any definitions (and most importantly, the notion of effective route) will only be assumed to make sense for join-free transactional histories.

In a join-free history of \cbcast{} the values of \lset{}, \cset{} and \mset{} have the relations
\begin{align*}
\lset{} & = \cset{}  \\
\lset{} & \subset \mset{}
\end{align*}
Therefore we will assume that all live processes are contacted and are members of the current view.

\subsection{The Central Lemma}
Both theorems rely heavily on the a lemma which we refer to as the Central Lemma, and which we will introduce after a few definitions.

\begin{defn}
The {\bf installation gap} of view $v$ at process $P$ is the gap between $v$ and the highest view known to $P$ at the time that $v$ is installed at $P$. It is the value of $\vgap$ when the view $v$ is installed at process $P$. More precisely, it is the value of $\vgap$ at labeled step \ref{TTI:Installation} of the \ref{TryToInstall} procedure when $\cview = v$. The installation gap of a view $v$ at process $P$ is denoted by $\gap_v(P)$.
\end{defn}

\begin{lem}
\label{InstallHeightLem}
Let $P$ and $Q$ be processes and suppose that $P$ installs view $v$ when $Q$ is still alive ($Q \in \lset{}$). Then $P$ must have received a $\pkf{v + \gap_v(P)}$ from $Q$ prior to installing the view. 
\end{lem}

\begin{proof}
Looking at the \ref{TryToInstall} procedure one can verify that $\fvec{}[Q] = v + \gap_v(P)$ when view $v$ is installed at $P$. The claim now follows from Lemma \ref{FCountLem} and from the monotonicity of $\fvec{}[Q]$ (Lemma \ref{OmnibusCBCASTLem}(\ref{OCL:gfmon})).
\end{proof}

\begin{defn}
\label{DefEffective}
Due to forwarding, each message can be received and acknowledged multiple times by a process, but the message is moved into \rset{} and \fque{}[] at most once (See labeled step \ref{RMP:NoDuplicates} of \ref{ReceiveMessage}). Additionally, since a forwarding queue becomes empty after forwarding (see labeled step \ref{RN:DiscardFque} of the \ref{RemovalNotification} procedure), a sender will only send a message once to any target. Therefore for any message that is received by a target there is at most one {\bf effective sender} - namely the sender that managed to get its message packet not just acknowledged but also appended to the forwarding queue of the receiver - and only one {\bf effective packet} that is sent by the effective sender. A message gets forwarded only as a result of a notification of the removal of the effective sender. As a result, for every process that receives a message, there is at most one {\bf effective route} of retransmissions of effective packets leading from the message originator to the process.
\end{defn}

\begin{lem}
\label{EffectiveDeathLem}
Let $n$ be a message and let
$$
\orig{n} = R_0 \rightarrow R_1 \rightarrow R_2  \rightarrow R_3 \rightarrow \dotsb \rightarrow R_k = T \quad \text{ where } k \ge 0
$$
Be the effective route of $n$ from its originator $R_0$ to process $T$. Then for every $0 \le i <= k$ the process $R_i$ is a member of view $\mview{n}$ and
$$
\vdeath{R_0} < \vdeath{R_1} < \dotsb < \vdeath{R_{k-1}}
$$
\end{lem}
\begin{proof}
Messages are only originated or forwarded by members of the message view to members of the message view (see labeled step \ref{BM:Members} of the \ref{BroadcastMessage} procedure and labeled step \ref{RRN:Members} of the \ref{RemovalNotification} procedure). That proves the first claim. Forwarding is triggered by the removal notification of the effective sender, and view change notifications are only queued to members of that view. That proves the second claim.
\end{proof}

\begin{lem}
\label{UniqueEffectiveLem}
When a message is received by a process for the first time it is moved into \rset{} and \fque{}[]. Therefore for every message $n$ that is received by a process $R$ there is exactly one effective route from the originator of $n$ to $R$.
\end{lem}

\begin{proof}
Following the flow of \ref{ReceiveMessage} we see that there are three cases where a message can be discarded without being moved into \rset{} and \fque{}[]. The second and third cases occur when the message has already been delivered and when the message is still in \rset{}. Neither of these cases can occur the first time the message is received. The first case occurs when the message is obsolete. To prove the lemma we have to show that in this case as well the message is not received for the first time. Suppose therefore that a process $S$ sent a message packet $\pkm{n}$ to process $R$, which then discarded the message as obsolete. Either $S$ is the originator of the message or it forwards $n$ out of its \fque{}[]. Either way there exists in this case an effective route

\begin{equation}
G = G_0 \rightarrow G_1 \rightarrow \dotsb \rightarrow G_k = S \quad\text{ where } k \ge 0
\end{equation}
Where $G = \orig{n}$ is the originator of $n$. Notice that $k = 0$ covers the case where $S = G$, i.e. the case where $S$ is the originator of $n$.

From Lemma \ref{EffectiveDeathLem} we know that $\vdeath{G_0} < \vdeath{G_1} < \dotsb < \vdeath{G_{k-1}}$. In addition $\vdeath{G_{k-1}} < \vdeath{S}$ because $S$ forwards $n$ (when $k > 0$) as a result of the removal of the effective sender $G_{k-1}$. We also know that $\mview{n} < \vdeath{G}$ because $G$ originates $n$. So for all $i$, we know that $\vdeath{G_i} > \mview{n}$ and therefore $G_i \in \viewset_{\mview{n}}$.

By assumption, message $n$ was obsolete when received by $R$ from $S$. This means that $R$ had already installed view $\mview{n}+1$.
It follows from Lemma \ref{InstallHeightLem} that process $R$
must have received a $\pkf{> \mview{n}}$ packet from every member of $\mview{n}$ for which $R$ had not yet received a removal notification. In other words, before $R$ can install $\mview{n}+1$ it must have received either a $\pkf{> \mview{n}}$ packet or a $\nr{P}$ notification from/about every member $P \in \mview{n}$. This is true in particular for the processes in the effective route above, all of whom are members of $\mview{n}$.

We know that $R$ does not receive a $\nr{S}$ notification before installing view $\mview{n}+1$. This is because it receives the message packet $\pkm{n}$ from $S$ after installing the view, and by the \CAxPacket{} this can only happen if $R$ still considers $S$ to be alive. Let $i$ be the index of the first process $G_i$ in the effective route for which $R$ does not receive a $\nr{G_i}$ notification prior to installing view $\mview{n}+1$.

If $i=0$, then $G_i$ is the originator of $n$. Since $R$ does not receive a $\nr{G_0}$ notification, it must receive a $\pkf{> \mview{n}}$ packet from $G_0$. But $G_0$ only sends such a packet when $\cview{} + \vgap{} > \mview{n}$. Being the originator of $n$, process $G_0$ broadcasts $n$ while $\cview = \mview{n}$ and $\vgap = 0$ (messages are not broadcast when $\vgap > 0$). This means that $R$ receives $n$ from $G_0$ before receiving the flush packet and therefore before installing view $\mview{n}+1$, so the obsolete $\pkm{n}$ packet from $S$ is not the first receipt of $n$.

If $i > 0$, then $R$ receives $\nr{G_{i-1}}$ prior to installing view $\mview{n}+1$. When that happens, $\cview + \vgap = \vdeath{G_{i-1}}$ (See labeled step \ref{RRN:NewView} of \ref{RemovalNotification}). It also implies that $\vdeath{R} > \vdeath{G_{i-1}}$. Afterwards, $R$ does not proceed to install $\mview{n}+1$ before receiving $\pkf{\ge \vdeath{G_{i-1}}}$ from all the surviving members of $\mview{n}$.

Now look at process $G_i$. it forwards its $\pkm{n}$ when it receives $\nr{G_{i-1}}$. It then places $n$ in its \wset{}, together with an instability set that includes $R$ (since $\vdeath{R} > \vdeath{G_{i-1}}$). It then waits for \wset{} to clear before sending its first $\pkf{\ge \vdeath{G_{i-1}}}$. We know that this flush packet is eventually sent by $G_i$, by definition of $i$, so $n$ does indeed leave \wset{}. We also know that the flush packet is sent while $G_i$ still considers $R$ to be alive, according to the \AxProcII{}. Therefore $n$ must leave \wset{} as a result of $G_i$ receiving a $\pks{n}$ packet from $R$. This means that $R$ receives $n$ from $G_i$ before installing view $\mview{n}+1$, so the obsolete $\pkm{n}$ packet from $S$ is not the first receipt of $n$ in this case either.
\end{proof}

\begin{lem}[Central Lemma]
If process $P$ installs view $v_m$ with $\gap_{v_m}(P) = 0$, then any message $n$ with $\view(n) < v_m$ that is received by $P$ is also received by all other processes that install view $v_m$. 
\end{lem}
We will prove the central lemma a little later on. First, we use it to prove the Causal Order Theorem and the Restricted Progress Theorem.
\begin{proof}[Proof of the Causal Order Theorem]
We start with a note on message origination and broadcasting. The act of originating a message requires two separate steps whenever $\vgap{} > 0$. First, placing the message on \lque{} (see labeled step \ref{BM:LaunchQueue} of the \ref{BroadcastMessage} procedure), and then later when \vgap{} becomes zero, queuing packets containing the message (see labeled step \ref{TTI:NoGapLaunch} of the \ref{TryToInstall} procedure). In the context of this proof, we are talking about the second step whenever we talk about {\em broadcasting} a message. This definition of the term "broadcast" can create phantom causal relationships between messages. Specifically, if message $n$ is delivered after message $m$ is placed on \lque{} but before a packet containing $m$ is queued, we consider $m$ to be dependent on $n$, though they are obviously independent. This does not invalidate our arguments, of course, but it does imply that the algorithm serializes message delivery to a greater degree than seems necessary. This suboptimal behavior is inherent in the protocol, because every view change is a global synchronization point that generates excess serialization.

Let $P$ and $Q$ be processes, and let $m$ be a message that is broadcast by $P$ and is delivered at $Q$. We have to demonstrate parts \ref{COP:bcast} and \ref{COP:fam} of Definition \ref{COP}. Remember that message familiarity is now synonymous with message delivery.

We start by observing that whenever a process $P$ broadcasts a message $n$, the message is delivered at $P$ itself before it processes any further notifications from \gms{}. This follows from the \AxProcIV{}. This axiom only implies that $n$ is {\em received} by $P$ prior to processing any further notifications from \gms{}. However, when $n$ is received by $P$ it is immediately deliverable. This is easy to verity when you notice that $P$ stamps $n$ with its own vector time before broadcasting it. Moreover, the \ref{ReceiveMessage}() procedure includes a call to \ref{ScanCall}(), resulting in an immediate delivery of $n$. 
As a result, if $P$ broadcasts $n$ before broadcasting $m$, and if $\mview{n} < \mview{m}$, then $n$ is familiar to $P$ when it broadcasts $m$, and we can lump this case under part \ref{COP:fam} of definition \ref{COP}. This leaves the case $\mview{n} = \mview{m}$. In this case we have $\mvt{n} < \mvt{m}$ because $n$ gets stamped by $P$ before $m$ does.

When $P$ broadcasts $m$, it has view
\begin{align*}
\cview & = \mview{m} \\
\vgap{} & = 0
\end{align*}
and therefore $\gap_{\mview{m}}(P) = 0$.

Also, since $m$ is delivered at $Q$ we know that $Q$ installs $\mview{m}$. Therefore $P$ and $Q$ meet the requirements of the Central Lemma.
We divide the messages that are broadcast by $P$ or delivered at $P$ prior to the broadcast of $m$ into the following subsets:
\begin{enumerate}
\item For each $0 < k < \mview{m}$, the set
$$
\kmsg{} = \{ n|\mview{n} = k \text{ and } n \text{ is delivered at } P \}
$$
\label{CviewMessages}\item All the messages $n$ with $\mview{n} = \mview{m}$ and $\mvt{n} < \mvt{m}$
\end{enumerate}

Our previous observation shows that these sets cover all the messages $n$ of both part \ref{COP:bcast} and part \ref{COP:fam} of Definition \ref{COP}. The messages in the second category must be delivered at $Q$ prior to the delivery of $m$ (see section 5.1 of \cite{birman1991lightweight}).
For a given $0 < k < \mview{m}$, the Central Lemma implies that all the messages in $\kmsg{}$ are received by $Q$.

Suppose that there is a message in $\kmsg{}$ that is not delivered at $Q$, and let $n$ be such a message with a minimal vector time. Let $n'$ be any message, not necessarily in $\kmsg{}$, such that $\mview{n'} = k$ and $\mvt{n'} \prec \mvt{n}$. Because $n \in \kmsg{}$, $n$ is delivered at $P$. Therefore $n'$ must be delivered at $P$ before $n$ is delivered (see section 5.1 of \cite{birman1991lightweight}), and therefore $n' \in \kmsg{}$. By the minimality of $n$, we know that $n'$ is delivered at $Q$. Look at the latest of the following points in time:
\begin{itemize}
\item The point at which $n$ enters $\rset{}(Q)$ (at labeled step \ref{RMP:ReceiveEntry} of \ref{ReceiveMessage})
\item For each $n'$ with $\mview{n'} = k$ and $\mvt{n'} \prec \mvt{n}$, the point at which $n'$ is delivered at $Q$ (labeled step \ref{SC:Delivery} of \ref{ScanCall})
\item The point at which $Q$ installs view $k$ (labeled step \ref{TTI:Installation} of \ref{TryToInstall})
\end{itemize}
This latest point in time has the following two properties. First, $\ref{ScanCall}()$ is called shortly after that point, as can be verified by tracing each of the three code locations. Second, the message $n$ is present in $\rset{}(Q)$ at that point. To see that, notice that by definition $n$ must have already entered $\rset{}(Q)$ (because of the first point in time), so we only have to show that $n$ has not yet been removed. There are two places where $n$ can be removed. The first is labeled step \ref{SC:Delivery} of \ref{ScanCall}, which is where $n$ is delivered. Since all the above points in time must occur before $n$ can be delivered, this case can be excluded. The second is labeled step \ref{TTI:RemoveRset} of \ref{TryToInstall}, where $n$ is removed as an obsolete message. But all the points in time that we listed must occur before $Q$ installs a view that is higher than $k$, so this case can be excluded as well. Therefore the message $n$ is present in $\rset{}(Q)$ at the latest point in time. Moreover, $n$ is deliverable at this point because $Q$ has already installed view $k$ and all the messages of lower vector time have been delivered. Therefore the imminent call to \ref{ScanCall}() will result in the immediate delivery of $n$ to $Q$ (in the case of labeled step \ref{SC:Delivery} of \ref{ScanCall} it may happen even earlier, within the loop). Due to their low view, the messages in $\kmsg{}$ must be delivered at $Q$ before the installation of $\mview{m}$ which in turn takes place before the delivery of $m$.
\end{proof}

In order to prove the Restricted Progress Theorem, we need the following lemma.

\begin{lem}
\label{NonHaltInstallLem}
Non-halting processes install all the views.
\end{lem}

\begin{proof}
The proof depends on the finiteness condition in a fundamental way. Let $v_L$ be the last view. By definition, the non-halting processes are exactly the members of that view. Let $P$ be a non-halting process. At some point in time $P$ dequeues a notification of the last view which by assumption is a notification $\nr{X}$ of the removal of the last halting process $X$.

When $P$ processes that notification it forwards all the messages in $\fque{}[X]$ in order, and moves them to its \wset{}. Thereafter, as long as $P$ does not install view $v_L$, $P$ does not add any new messages to \wset{}. This is because messages are added to \wset{} only after broadcasting or forwarding a message. But since $\vgap > 0$ during the period at hand, $P$ does not broadcast any messages, and since no further view change notifications occur, $P$ does not forward any messages either. All the messages in \wset{}, having been sent by the non-halting process $P$, are eventually acknowledged by all the non-halting processes.

Since $P$ has already received notice of the removal of all the halting processes, this implies that all the messages in \wset{} eventually stabilize and therefore \wset{} empties. Once that happens, $P$ sends a $\pkf{v_L}$ packet to all the non-halting processes. This means that every non-halting process eventually receives a $\pkf{v_L}$ packet from all the non-halting processes. Once that happens, all the views up to an including $v_L$ are installed immediately.
\end{proof}

\begin{proof}[Proof of the Restricted Progress Theorem]
Let $P$ and $Q$ be two non-halting processes and suppose $P$ broadcasts a message $m$. A join-free history can only have a finite number of view changes so we have to show that $m$ is eventually delivered at $Q$.

Let $v_L$ be the last view. By Lemma \ref{NonHaltInstallLem}, both $P$ and $Q$ install $v_L$, and by necessity they do it with
$\gap_{v_L}(P) = \gap_{v_L}(Q) = 0$. By the Central Lemma this implies that all messages that were received
by $P$ prior to view $v_L$ are also received by $Q$ and vice versa. More generally, all the non-halting processes receive the same messages prior to installing view $v_L$.

To see that the same happens with messages that are received after installing view $v_L$, observe that since view $v_L$ is comprised exclusively of non-halting processes, all messages sent in view $v_L$ are broadcast by non-halting processes and therefore must be received by all members of $v_L$. Therefore, all the non-halting processes receive the same messages.

We will show now that non-halting processes also deliver the same messages. Let $n$ be a minimal message is delivered at a non-halting process $P$ but not delivered at
a non-halting process $Q$. We know by the discussion above that $n$ is received by $Q$, and therefore by Lemma \ref{UniqueEffectiveLem} $n$ enters $Q$'s \rset{}. We know by the Causal Order Theorem all the messages that message $n$ depends on are delivered at $P$ and by minimality are also delivered at $Q$. Why would $n$ not be delivered? Looking at the \ref{ScanCall} and \ref{TryToInstall} procedures, we see that $n$ enters \rset{} when $\cview \le \mview{n}$, and $n$ is not delivered as long as $Q$ does not install $\mview{n}$. In addition $n$ is removed from \rset{} without being delivered if it has not been delivered by the time $Q$ installs $\mview{n}+1$. From lemma \ref{NonHaltInstallLem} we know that $Q$ installs view $\mview{n}$. Look at the latest of the following three points in time:
\begin{enumerate}
\item $Q$ has received the effective packet $\pkm{n}$, and has just placed $n$ in \rset{}. (Occurs at the conclusion of labeled step \ref{RMP:ReceiveEntry} of \ref{ReceiveMessage}.)
\item $Q$ has just delivered the last message that $n$ depends on. (Occurs at the conclusion of labeled step \ref{SC:Delivery} of \ref{ScanCall}.)
\item $Q$ has just installed $\mview{n}$. (Occurs at labeled step \ref{TTI:Installation} of \ref{TryToInstall}.)
\end{enumerate}
When that latest moment occurs, two conditions are met. First, $Q$ has $\cview = \mview{n}$, because at least the third moment occurs when $\cview = \mview{n}$, and none of the other two moments can happen after a higher view is installed. Second, the message $n$ is in \rset{} because it enters the set at the first moment and neither way out of the set (delivery of $n$ or installation of a higher view) is available until after all three moments in time have occurred. As a result, the message $n$ becomes deliverable at that moment. Looking at the relevant code locations, one can see that the \ref{ScanCall} procedure is either invoked or (in the second case) continues to be executed, resulting in the delivery of $n$. 

We are almost done. We know that all the non-halting processes deliver the same messages, but we have to show that they deliver {\em all} the messages that originate from non-halting processes. To do that, we will show that every process delivers all the messages that it originates itself. Since all non-halting processes deliver the same messages, this implies the desired result.

Suppose that the non-halting process $P$ originates a message $n$. Since $P$ is non-halting it receives $n$ and by Lemma \ref{UniqueEffectiveLem} it must be placed in $P$'s \rset{}. As we already observed in the proof of the Causal Order Theorem, when $n$ is placed in \rset{} of its own originator it is immediately deliverable, because its $\mvt{n}$ is derived from $P$'s vector time in exactly the fashion that makes it deliverable. Since message placement in \rset{} is followed by a call to \ref{ScanCall}() (labeled steps \ref{RMP:ReceiveEntry} and \ref{RMP:Scan} of \ref{ReceiveMessage}), $n$ is delivered at $P$ and we are done.
\end{proof}

\subsection{Proving the Central Lemma}
We start with a technical lemma that we will use repeatedly in the proof of the Central Lemma.
\begin{lem}
\label{TechLemma1}
Let $A$, $B$, $D$ and $F$ be processes (not necessarily distinct), and let $n$ be a message, meeting the following criteria:
\begin{enumerate}
\item\label{TL1:eff}	Process $D$ sends an effective $\pkm{n}$ packet to process $F$ as it broadcasts or forwards $n$.
\item\label{TL1:B}	Process $B$ is a member of $\mview{n}$.
\item\label{TL1:death}	$\vdeath{D} < \vdeath{F}$.
\item\label{TL1:flush}	Process $A$ receives a $\pkf{f}$ packet from process $F$, with $\vdeath{D} \le f < \vdeath{B}$
\end{enumerate}
Then process $B$ receives $n$ before process $F$ queues a $\pkf{f}$ packet to process $A$ and before $B$ dequeues a notification of view $\vdeath{F}$.
\end{lem}
\begin{proof}
Original messages are broadcast to all the members of the message view (see labeled step \ref{BM:Members} of the \ref{BroadcastMessage} procedure. When no processes join, $\cset = \lset$ and since $\vgap = 0$ here, $\mset{} = \lset{}$), and are forwarded to all the surviving members of the message view (see labeled step \ref{RRN:Members} of \ref{RemovalNotification}). By condition \ref{TL1:B} process $B$ is a member of $\mview{n}$ and condition \ref{TL1:flush} implies that $\vdeath{D} < \vdeath{B}$, so process $D$ never receives a $\nr{B}$ notification. Therefore, when $D$ broadcasts or forwards message $n$, it queues a $\pkm{n}$ packet to process $B$. By condition \ref{TL1:eff} we know that $D$ also queues a $\pkm{n}$ packet to $F$, and that this packet is effective. That means that when process $F$ receives the packet, it appends $n$ to its forwarding queue, with $D$ as the sender (see labeled step \ref{RMP:Sender} of \ref{ReceiveMessage}). Condition \ref{TL1:death} implies that $F$ receives a $\nr{D}$ notification. We divide the rest of the proof into two cases. The easy case is when $n$ is still in the forwarding queue when the $\nr{D}$ notification is dequeued by $F$. The harder case is when $n$ has already been removed from the forwarding queue.

If $n$ is still in the forwarding queue, then $F$ forwards $n$ to all the surviving members of view $\mview{n}$, which include $B$ because $\vdeath{D} < \vdeath{B}$. It then moves $n$ to its wait set, together with an instability set that includes $B$ (see labeled step \ref{RRN:WaitSet} of \ref{RemovalNotification}). The moment that $F$ dequeues the removal notification of $D$ is also the first point in time at which $\cview{} + \vgap{} \ge \vdeath{D}$ in $F$ (see labeled step \ref{RRN:Vgap} of \ref{RemovalNotification}). Therefore, the flush packet that $F$ sends to $A$ according to condition \ref{TL1:flush} is queued later (see labeled step \ref{CF:SendFlush} of \ref{CheckFlush}). Since flush packets are only queued when \wset{} is empty (see condition in ambient block of same location), the flush packet must be queued after message $n$ has left the wait set, i.e. message $n$ must have stabilized in the meantime. But process $B$ was initially in the instability set of message $n$. Message $n$ can stabilize with respect to process $B$ either by receipt of a removal notification for process $B$, or by receipt of an acknowledgment of message $n$ by $B$. In the latter case $B$ must send the acknowledgment to $F$ before becoming aware of view $\vdeath{F}$ and we are done. The former case is impossible because condition \ref{TL1:flush} implies that the flush packet was sent before process $F$ dequeued a notification of the removal of process $B$.

If $n$ is not in the forwarding queue when a $\nr{D}$ notification is received by $F$, it means that it has already been removed. The only possible cause of removal is the installation, by $F$, of view $\mview{n}+1$ (see labeled step \ref{TTI:RemoveFset} of the \ref{TryToInstall} procedure). Let
$$
e = \mview{n}+1+\gap_{\mview{n}+1}(F)
$$
$F$ installs view $\mview{n}+1$ only after receiving $\pkf{\ge e}$ from all the surviving members of view $\mview{n}$ (see lemma \ref{InstallHeightLem}). In the case at hand process $D$ is one of the surviving members, because by our assumption $F$ receives the $\nr{D}$ notification after installing the new view. This also implies that $e < \vdeath{D}$ and therefore, by condition \ref{TL1:flush}, $e < f$, and therefore $F$ sends the $\pkf{f}$ packet to $A$ after installing view $\mview{n}+1$.

If we are lucky, then $D$ sends the effective $\pkm{n}$ packet before sending $\pkf{\ge e}$ to $F$. In that case $n$ enters the wait set of $D$ with process $B$ in its instability set, and it must exit the wait set before the flush is sent. Since $\vdeath{D} < \vdeath{B}$, this can only happen if $B$ receives and acknowledges $n$. Process $B$ acknowledges the receipt of $n$ before $D$ sends a $\pkf{\ge e}$ packet to $F$, which is received by $F$ before it installs view $\mview{n}+1$, which in turn happens before $F$ sends a $\pkf{f}$ packet to $A$. Process $B$ acknowledges the receipt of $n$ before becoming aware of $\vdeath{D}$ and $\vdeath{D} < \vdeath{F}$, so the acknowledgment is sent before $B$ is aware of the death of $F$ and so we are done in this case.

If we are not lucky, then $D$ sends a $\pkf{\ge e}$ to $F$ and only later sends an effective packet $\pkm{n}$ to $F$. Let
$$
D_0 \rightarrow D_1 \rightarrow \dotsb \rightarrow D_k = D \rightarrow F \quad \text{ where } k \ge 0
$$
Be the effective route of the message $n$ from its originator $D_0$ through $D$ to $F$. Let $i$ be the smallest integer such that $D_i$ sends a $\pkf{\ge e}$ packet to $F$ before sending an effective $\pkm{n}$ packet along the effective route. Then $i > 0$ because the originator $D_0$ broadcasts $n$ while $\vgap = 0$, so it could not have sent a $\pkf{\ge e}$ beforehand, since $e > \mview{n}$. Process $D_i$ forwards $n$ along the effective route after receiving a $\nr{D_{i-1}}$ notification, which arrives, by the definition of $i$, after $D_i$ already sent a $\pkf{\ge e}$ packet to $F$. This implies that $\vdeath{D_{i-1}} > e$, which in turn implies (by lemma \ref{InstallHeightLem}) that $F$ does not install view $\mview{n}+1$ before receiving a $\pkf{\ge e}$ from $D_{i-1}$. By definition, $D_{i-1}$ must send the effective $\pkm{n}$ along the effective route before sending the $\pkf{\ge e}$ packet. In addition, by lemma \ref{EffectiveDeathLem} $\vdeath{D_{i-1}} < \vdeath{D} < \vdeath{B}$ so $D_{i-1}$ sends message $n$ to process $B$ and when it moves $n$ to its wait set it includes process $B$ in the instability set of $n$. As a result, $D_{i-1}$ cannot send a $\pkf{\ge e}$ packet to $F$ before $n$ stabilizes with respect to process $B$, and since $\vdeath{D_{i-1}} < \vdeath{B}$, this stabilization can only happen as a result of a receipt of a $\pks{n}$ packet from $B$. So in this case as well $B$ receives and acknowledges message $n$, and it must happen before $F$ installs view $\mview{n}+1$, which in turn happens before $F$ sends a $\pkf{f}$ packet to $A$. In addition $B$ must send the acknowledgment of receipt of $n$ before becoming aware of the death of $D_{i-1}$ and since $\vdeath{D_{i-1}} < \vdeath{D} < \vdeath{F}$ this means that $B$ sends the acknowledgment before becoming aware of view $\vdeath{F}$. So we are done in this case as well.
\end{proof}

\begin{proof}[Proof of the Central Lemma]
Assume that there is a message $n$ with $\mview{n} < v_m$ that is received by process $P$,
and let $Q$ be some process that installed view $v_m$.
Denote $v_n = \mview{n}$.
Let $S = \orig{n}$, the process that originally broadcast $n$. Then there is an effective route (see Definition \ref{DefEffective} and Lemma \ref{UniqueEffectiveLem}) of $\pkm{n}$ packets:
\begin{equation}
S \rightarrow R_1 \rightarrow R_2  \rightarrow R_3 \rightarrow \dotsb \rightarrow R_k \rightarrow P \quad \text{ where } k \ge 0
\end{equation}
To streamline our arguments, we will occasionally refer to $S$ as $R_0$ and to $P$ as $R_{k+1}$. By Lemma \ref{EffectiveDeathLem} all the processes in the effective route must be members of $v_n$ and
$$
\vdeath{S} < \vdeath{R_1} < \dotsb < \vdeath{R_k}
$$
By assumption, $P$ installs $\view(v_m)$, and $\gap_{v_m}(P) = 0$. From Lemma \ref{InstallHeightLem} we know that in order to install view $v_m$ process $P$ must first receive a $\pkf{v_m+\gap_{v_m}(P)}$ packet from each member of \lset{}. It is easy to check that $\lset{} = \mset{}$ whenever $\vgap{} = 0$. Therefore $P$ must receive a $\pkf{v_m}$ packet from every member of $\view(v_m)$. We will use that fact several times in the argument below.

$R_{k+1} = P$ is a member of view $v_m$. Let $i$ be the smallest integer such that $R_i$ is a member of $v_m$. We will divide the proof to two cases: $i > 0$ and $i = 0$.

We start with the case $i > 0$.

Since $i > 0$ there exists a process $R_{i-1}$ that sends an effective packet $\pkm{n}$ to $R_i$, and $R_{i-1}$ is not a member of view $v_m$. Since both are members of view $v_n$ and $R_i$ is a member of view $v_m$, $\vdeath{R_{i-1}}$ is in the membership interval of $R_i$ and so it must receive a $\nr{R_{i-1}}$ packet from the membership service.
We know that $v_n < v_m < \vdeath{R_i}$. We look at the following cases:

Case I: $\vdeath{Q} \le \vdeath{R_i}$.

We know that $Q$ installs view $v_m$, and by the assumption of the current case, $Q$ never receives a removal notification for $R_i$. It follows from Lemma \ref{InstallHeightLem} that $Q$ must receive a $\pkf{f}$ packet from $R_i$ where $f \ge v_m$ before installing view $v_m$. Assigning $A=B=Q$, $D=R_{i-1}$ and $F=R_i$, we can check that all the conditions of Lemma \ref{TechLemma1} are met. The only difficulty is with condition \ref{TL1:flush}. There we have
\begin{description}
\item[] $\vdeath{D} = \vdeath{R_{i-1}} \le v_m \le f$
\item[] $R_i$ would not send a $\pkf{f}$ packet to $Q$ if it were aware of its removal, therefore $f~<~\vdeath{Q} = \vdeath{B}$
\end{description}
Therefore process $Q$ receives $n$.

Case II: $\vdeath{Q} > \vdeath{R_i}$.

Since process $R_i$ is a member of view $v_m$, $P$ must receive a $\pkf{f}$ from $R_i$, where $f = v_m$, before it installs view $v_m$. Assigning $A=P$, $B=Q$, $D=R_{i-1}$ and $F=R_i$, we can check that all the conditions of Lemma \ref{TechLemma1} are met. Again the difficulty is with condition \ref{TL1:flush}. There we have
\begin{description}
\item[] $\vdeath{D} = \vdeath{R_{i-1}} \le v_m = f$
\item[] $f < \vdeath{R_i} < \vdeath{Q} = \vdeath{B}$
\end{description}
Therefore process $Q$ receives $n$ in this case as well.

We now turn to the case $i = 0$.

$i = 0$ means that process $S$, the originator of message $n$, is a member of view $v_m$. Therefore $S$ must send a $\pkf{v_m}$ packet to $P$. Since $Q$ is a member of view $v_m$, this flush packet must be sent while $Q$ is in $\lset{}(S)$.

When $S$ broadcasts $n$, $\cview{}(S) = v_n$ and $\vgap{}(S) = 0$, because message broadcasts do not occur when $\vgap{} > 0$. As a result the broadcast occurs before $S$ sends the $\pkf{v_m}$ packet to $P$, and while $Q$ is in $\lset{}(S)$. This implies that $S$ sends a $\pkm{n}$ packet to $Q$ as part of the broadcast, and it includes $Q$ in the instability set of $n$ as it adds its record to $\wset{}(S)$. The message $n$ must stabilize with respect to $Q$ before $S$ sends the $\pkf{v_m}$ packet. This flush packet is sent before $S$ receives a removal notification for $Q$. Therefore the stabilization can only occur as a result of the receipt of a $\pks{n}$ packet from $Q$. In other words, $Q$ must receive $n$.

This concludes the proof of the case $i = 0$ and of the Central Lemma.
\end{proof}


\end{document}